%% file: master.tex
\documentclass[a4paper,11pt]{memoir}
\pdfoutput=1

\newif\ifmonochrome
\monochromefalse

\usepackage[T1]{fontenc}
\usepackage[utf8]{inputenc}

\usepackage[german,ngerman,english]{babel}
\usepackage{csquotes}

\usepackage{acronym}
\usepackage[shortlabels]{enumitem}

\usepackage{xstring}

\ifmonochrome\usepackage[monochrome]{color}
\else\usepackage{color}\fi

\usepackage{amsmath}
\usepackage{mathtools}
\usepackage[amsmath,thmmarks,hyperref]{ntheorem}
\usepackage{tensor}

\usepackage{units}

\usepackage{tikz}

\input{font.def}

\settypeblocksize{48\baselineskip}{32pc}{*}
\setlrmargins{*}{*}{1}
\setulmargins{*}{*}{.8}
\checkandfixthelayout

\nouppercaseheads
\aliaspagestyle{title}{empty}
\aliaspagestyle{part}{empty}
\copypagestyle{headings}{plain}
\ifdraftdoc
  \makeevenhead{headings}{\thepage\hskip2em\leftmark}{}{DRAFT}
  \makeoddhead{headings}{DRAFT}{}{\rightmark\hskip2em\thepage}
\else
  \makeevenhead{headings}{\thepage\hskip2em\leftmark}{}{}
  \makeoddhead{headings}{}{}{\rightmark\hskip2em\thepage}
\fi
\makeevenfoot{plain}{}{}{}
\makeoddfoot{plain}{}{}{}

\setpnumwidth{2em}
\setrmarg{3em}

\renewcommand{\afterpartskip}{}

\makechapterstyle{thesis}{
  \setlength{\beforechapskip}{0pt}
  \setlength{\midchapskip}{28pt}
  \setlength{\afterchapskip}{28pt plus 20pt minus 20pt}

}
\chapterstyle{thesis}

\setsecheadstyle{\normalfont\Large\sffamily}
\setbeforesecskip{-3.5ex plus -2ex minus -.2ex}
\setaftersecskip{2.3ex plus .2ex}
\setsubsecheadstyle{\normalfont\large\sffamily}
\setbeforesubsecskip{-3.25ex plus -1ex minus -.2ex}
\setaftersubsecskip{1.5ex plus .2ex}
\newcommand{\dotsfparastyle}[1]{\normalfont\sffamily #1.}
\setparaheadstyle{\dotsfparastyle}

\maxsecnumdepth{subsection}

\captionnamefont{\normalfont\bfseries\sffamily}
\captiondelim{.\hskip.8em}
\feetbelowfloat

\numberwithin{equation}{chapter}

\definecolor{hypercolor}{rgb}{0,0.2,0.7}

\usepackage[backend=bibtex,bibencoding=ascii,style=numeric-comp,uniquename=init,firstinits=true,doi=false,url=false,isbn=false,defernumbers=true]{biblatex}

\makeatletter
\renewbibmacro*{in:}{%
\def\tempb{article}%
\ifx\abx@field@entrytype\tempb%
  \else%
  \printtext{%
    \bibstring{in}\intitlepunct}%
 \fi%
}
\makeatother

\DeclareFieldFormat{sentencecase}{\MakeSentenceCase{#1}}

\renewbibmacro*{title}{%
  \ifthenelse{\iffieldundef{title}\AND\iffieldundef{subtitle}}
    {}
    {\ifthenelse{\ifentrytype{article}\OR\ifentrytype{inbook}%
      \OR\ifentrytype{incollection}\OR\ifentrytype{inproceedings}%
      \OR\ifentrytype{inreference}\OR\ifentrytype{thesis}%
      \OR\ifentrytype{masterthesis}\OR\ifentrytype{phdthesis}}
      {\printtext[title]{%
        \printfield[sentencecase]{title}%
        \setunit{\subtitlepunct}%
        \printfield[sentencecase]{subtitle}}}%
      {\printtext[title]{%
        \printfield[titlecase]{title}%
        \setunit{\subtitlepunct}%
        \printfield[titlecase]{subtitle}}}%
    \newunit}%
  \printfield{titleaddon}}

\newbibmacro{string+doiurlisbn}[1]{%
  \iffieldundef{doi}{%
    \iffieldundef{url}{%
      \iffieldundef{isbn}{%
        \iffieldundef{issn}{%
          #1%
        }{%
          \href{http://www.worldcat.org/issn/\thefield{issn}}{#1}%
        }%
      }{%
        \href{http://www.worldcat.org/isbn/\thefield{isbn}}{#1}%
      }%
    }{%
      \href{\thefield{url}}{#1}%
    }%
  }{%
    \href{http://dx.doi.org/\thefield{doi}}{#1}%
  }%
}

\DeclareFieldFormat{title}{\usebibmacro{string+doiurlisbn}{\mkbibemph{#1}}}
\DeclareFieldFormat[article,incollection,thesis]{title}%
    {\usebibmacro{string+doiurlisbn}{\mkbibquote{#1}}}

\DeclareNameAlias{sortname}{last-first}
\DeclareNameAlias{default}{last-first}

\bibliography{references.bib}

\DeclareBibliographyCategory{own}
\addtocategory{own}{siemssen:2011,dappiaggi:2013,fewster:2014,pinamonti:2015,pinamonti:2015a}

\input{index.def}

\newcommand{\hyp}{-\penalty0\hskip0pt\relax}
\makeatletter
\newcommand{\allowhyphenation}{\nobreak\hskip\z@skip}
\makeatother
\newcommand{\optword}[1]{(#1\allowhyphenation\discretionary{-)}{}{)}\allowhyphenation}
\newcommand{\optwordhy}[1]{(#1\allowhyphenation\discretionary{-)}{}{-)}\allowhyphenation}

\newcommand{\latin}{\textit}
\newcommand{\cf}{\latin{cf.}}
\newcommand{\eg}{\latin{e.g.}}

\newcommand{\etc}{\latin{etc.}}
\newcommand{\ie}{\latin{i.e.}}

\newcommand{\viz}{\latin{viz.}}

\input{math.def}

\theoremnumbering{arabic}
\theoremstyle{plain}
\theorembodyfont{\itshape}
\theoremheaderfont{\normalfont\sffamily\bfseries}
\theoremseparator{.}
\qedsymbol{\ensuremath{\quad\Box}}
\newtheorem{theorem}{Theorem}[chapter]
\newtheorem{proposition}[theorem]{Proposition}
\newtheorem{lemma}[theorem]{Lemma}
\newtheorem{corollary}[theorem]{Corollary}

\newtheorem{remark}[theorem]{Remark}
\newtheorem{definition}[theorem]{Definition}
\newtheorem{conjecture}[theorem]{Conjecture}

\theoremstyle{nonumberplain}
\theorembodyfont{\normalfont}
\theoremheaderfont{\normalfont\sffamily\bfseries}
\theoremsymbol{\ensuremath{\quad\Box}}
\newtheorem{proof}{Proof}

\usepackage[final]{hyperref}
\newsubfloat{figure}

\usepackage[capitalise]{cleveref}
\usepackage{bookmark}

\crefname{chapter}{Chap.}{Chaps.}
\crefname{section}{Sect.}{Sects.}
\crefname{figure}{Fig.}{Figs.}
\crefname{theorem}{Thm.}{Thms.}
\crefname{proposition}{Prop.}{Props.}
\crefname{definition}{Def.}{Defs.}
\crefname{lemma}{Lem.}{Lems.}
\crefname{corollary}{Cor.}{Cors.}

\title{The Semiclassical Einstein Equation on Cosmological Spacetimes}
\author{Daniel Siemssen}

\hypersetup{
  unicode,
  colorlinks,
  linkcolor=hypercolor,
  urlcolor=hypercolor,
  citecolor=hypercolor,
  bookmarksnumbered,
  bookmarksdepth=2,
  pdfauthor={Daniel Siemssen},
  pdftitle={The Semiclassical Einstein Equation on Cosmological Spacetimes},
  pdfsubject={Quantum Field Theory in Curved Spacetimes}
}

\begin{document}

\frontmatter

\begin{titlingpage}
  \renewcommand*{\thepage}{title-\arabic{page}}
  \newlength{\covermargin}
  \calccentering{\covermargin}

  \begin{adjustwidth}{\covermargin-8mm}{-\covermargin-8mm}
    \raggedright

    {\Huge \theauthor}\\

    \vspace{40mm}

    \begin{Spacing}{1.1}
      \Huge\bfseries\thetitle
    \end{Spacing}

    \vspace{15mm}

    {\LARGE
      Ph.D. Thesis\\[\onelineskip]
      Dipartimento di Matematica\\
      Università degli Studi di Genova
    }
  \end{adjustwidth}
\ifmonochrome\end{titlingpage}
\else\clearpage\fi

\thispagestyle{empty}

\null\vfill

\begin{flushleft}
  {\large\bfseries \thetitle}\\
  \smallskip
  Ph.D. thesis submitted by \href{mailto:siemssen@dima.unige.it}{\theauthor}\\
	Genova, February 2015

	\bigskip

	\begin{tikzpicture}
		\node[anchor=west,inner sep=0] at (0,0)
			{\ifmonochrome\includegraphics[width=1cm]{unige-grey.pdf}\else\includegraphics[width=1cm]{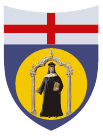}\fi};
		\node[anchor=west,inner sep=0] at (1.3,0) [align=left]{Dipartimento di Matematica\\ Università degli Studi di Genova};
	\end{tikzpicture}

  \bigskip

  Supervisor: Prof. Dr. Nicola Pinamonti\\
  Examiner: Prof. Dr. Valter Moretti
\end{flushleft}
\ifmonochrome\else\end{titlingpage}\fi

\cleardoublepage

\thispagestyle{empty}

\pdfbookmark{Abstract}{sec:abstract}
\begin{abstract}
  The subject of this thesis is the coupling of quantum fields to a classical gravitational background in a semiclassical fashion.
  It contains a thorough introduction into quantum field theory on curved spacetime with a focus on the stress-energy tensor and the semiclassical Einstein equation.
  Basic notions of differential geometry, topology, functional and microlocal analysis, causality and general relativity will be summarised, and the algebraic approach to quantum field theory on curved spacetime will be reviewed.
  The latter part contains an introduction to the framework of locally covariant quantum field theory and relevant quantum states: Hadamard states and, on cosmological spacetimes, adiabatic states.
  Apart from these foundations, the original research of the author and his collaborators will be presented:

  Together with Fewster, the author studied the up and down structure of circular and linear permutations using their decomposition into so-called atomic permutations.
  The relevance of these results to this thesis is their application in the calculation of the moments of quadratic quantum fields in the quest to determine their probability distribution.

  In a work with Pinamonti, the author showed the local and global existence of solutions to the semiclassical Einstein equation in flat cosmological spacetimes coupled to a massive conformally coupled scalar field by solving simultaneously for the quantum state and the Hubble function in an integral-functional equation.
  The theorem is proved with the Banach fixed-point theorem using the continuous functional differentiability and boundedness of the integral kernel of the integral-functional equation.

  Since the semiclassical Einstein equation neglects the quantum nature of the stress-energy tensor by ignoring its fluctuations, the author proposed in another work with Pinamonti an extension of the semiclassical Einstein equations which couples the moments of a stochastic Einstein tensor to the moments of the quantum stress-energy tensor.
  In a toy model of a Newtonianly perturbed exponentially expanding spacetime it is shown that the quantum fluctuations of the stress-energy tensor induce an almost-scale-invariant power spectrum for the perturbation potential and that non-Gaussianties arise naturally.
\end{abstract}

\cleardoublepage

\pdfbookmark{Contents}{sec:contents}
\microtypesetup{protrusion=false}
\tableofcontents*
\microtypesetup{protrusion=true}

%------------------------------------------------------------------------------%

\mainmatter

\input{introduction}

\renewcommand{\afterpartskip}{
\vspace*{8\onelineskip}
\begin{quotation}\begin{otherlanguage}{german}
  \textit{Zudem ist es ein Irrtum zu glauben, daß die Strenge in der Beweisführung die Feindin der Einfachheit wäre. An zahlreichen Beispielen finden wir im Gegenteil bestätigt, daß die strenge Methode auch zugleich die einfachere und leichter faßliche ist. Das Streben nach Strenge zwingt uns eben zur Auffindung einfacherer Schlußweisen; auch bahnt es uns häufig den Weg zu Methoden, die entwickelungsfähiger sind als die alten Methoden von geringerer Strenge.}
  \sourceatright{--- David Hilbert, ``Mathematische Probleme'' (1900), p. 257.}
\end{otherlanguage}\end{quotation}
}
\part{Foundations}
\input{differential_geometry}
\input{lorentzian_geometry}
\input{analysis}
\input{combinatorics}

\renewcommand{\afterpartskip}{
\vspace*{4\onelineskip}
\begin{quotation}
  \textit{Is the purpose of theoretical physics to be no more than a cataloging of all the things that can happen when particles interact with each other and separate? Or is it to be an understanding at a deeper level in which there are things that are not directly observable (as the underlying quantized fields are) but in terms of which we shall have a more fundamental understanding?}
  \sourceatright{--- Julian S. Schwinger, ``Quantum Mechanics'' (2001), p. 24 f.}
\end{quotation}
\begin{quotation}
  \textit{First, in order to achieve the greatest possible generality we continue our total boycott of the canonical formalism [...].}
  \sourceatright{--- Bryce S. DeWitt, J. Math. Phys. 3 (1962), p. 1073.}
\end{quotation}
}
\part{Quantum field theory}
\input{locally_covariant}
\input{states}

\renewcommand{\afterpartskip}{
  \vspace*{4\onelineskip}
  \begin{quotation}
    \textit{It is shown in a quite general manner that the quantization of a given system implies also the quantization of any other system to which it can be coupled.}
    \sourceatright{--- Bryce S. DeWitt, in ``Gravitation: An Introduction to Current Research'' (1962), p. 272}
  \end{quotation}
  \begin{quotation}
    \textit{Oh gravity, thou art a heartless bitch.}
    \sourceatright{--- Sheldon Cooper, Season 1, Episode 2, The Big Bang Theory}
  \end{quotation}
}
\part{Semiclassical gravity}
\input{einstein}
\input{solutions}
\input{fluctuations}

\bookmarksetupnext{level=-1}

\input{conclusions}

\input{acknowledgements}

%------------------------------------------------------------------------------%

\backmatter

%% BIBLIOGRAPHY
\printbibheading
\label{cha:bibliography}
\begingroup
\emergencystretch=1em

\section*{Works by the author}
\printbibliography[heading=none,category=own]

\section*{Works by other authors}
\printbibliography[heading=none,notcategory=own]
\endgroup

%% INDEX
% \printindex
\input{compiled-index}

\end{document}

%% file: introduction.tex
%!TEX root = master.tex

\chapter*{Introduction}
\addcontentsline{toc}{chapter}{Introduction}
\markboth{Introduction}{Introduction}
\label{cha:introduction}

The subject of this thesis is the interplay between quantum matter and gravity, \ie, the coupling of quantum fields to a classical gravitational background in a \emph{semiclassical} fashion.
Semiclassical gravity describes physics midway between the classical regime covered by the Einstein equation and a full-fledged quantum gravity.
However, while no theory of quantum gravity is universally accepted, quantum field theory on curved spacetimes offers an approach to the `low'-energy and `small'-curvature regime based on the firm foundation of quantum field theory and Lorentzian geometry.
Despite of the employed approximation, it has made several successful and relevant predictions like the Fulling--Davies--Unruh effect \cite{fulling:1973,davies:1975,unruh:1976}, the Hawking effect \cite{hawking:1975,fredenhagen:1990}, cosmological particle creation \cite{parker:1968} and the generation of curvature fluctuations during inflation \cite{mukhanov:1981,guth:1982,hawking:1982,starobinsky:1982}.
It is generally expected that a successful theory of quantum gravity also describes these phenomena and, in fact, they are used as criteria to select candidate theories.

Although quantum gravity is one motivation for studying quantum field theory on curved spacetime, it is not the only reason.
While quantum field theory is typically formulated on a Minkowski background, the Universe appears well-described by a curved spacetime and Minkowski spacetime provides only a local approximation.
However, even the slightest gravitational interaction causes many of the basic assumptions of `standard' quantum field theory on Minkowski spacetime to fail and in important situations, like inflation, the departure from a flat background is not small but causes important effects that cannot be neglected.
From this point of view it would be conceptually very unsatisfying if it was not possible to successfully formulate quantum field theory on a curved spacetime in such a way that it reduces to standard QFT in the case of a flat background.

Attempts to formalize quantum field theory in a mathematically exact manner have led to many significant insights into the structure of quantum fields: the CPT theorem, the spin-statistics connection, and superselection sectors to name a few, see \eg\ \cite{streater:1964,haag:1996}.
By studying aspects of semiclassical gravity and quantum field theory on curved spacetimes in the rigorous framework of algebraic quantum field theory, one hopes to gain deep and novel insights into the subtle nature of quantum fields on curved spacetimes and at the same time often prove theorems that have also a purely mathematical value.
Moreover, as a consequence of the correspondence principle, it is highly plausible that a careful investigation of the semiclassical theory gives us further hints about the structure of an eventual theory of quantum gravity.
In particular, one can expect that observations in cosmology are already described to high precision within semiclassical Einstein gravity and that tight limits can be placed on the creation of extreme objects such as wormholes in generic spacetimes.

In the formulation of quantum field theory on Minkowski spacetime one usually starts the with the unique Poincaré-invariant vacuum state as the ground state in a Fock space motivated by the particle interpretation.
On a generic spacetime, due to the absence of any symmetries, no such distinguished state can exist and, as illustrated in the Unruh and the Hawking effect, no unique particle interpretation is available.
This suggests that the starting point for a quantum field theory on curved spacetimes should be a formulation that does not require a preferred state.
For this reason, rigorous quantum field theory on curved spacetimes is often discussed within the \emph{algebraic approach} to quantum field theory \cite{haag:1996,haag:1964,haag:1984}.
In the algebraic approach one begins by considering an abstract algebra of quantum fields, which respects conditions of locality (quantum fields only depend on the local structure of the spacetime) and causality (causally separated quantum fields (anti)commute).

A modern formulation of QFT on curved spacetimes is the framework of \emph{locally covariant quantum field theory} \cite{brunetti:2003}.
It helped the development of several important contributions to QFT on curved spacetimes like renormalization and perturbative algebraic quantum field theory \cite{brunetti:2009a,fredenhagen:2012,fredenhagen:2013}, superselection sectors on curved spacetimes \cite{brunetti:2007a}, abstract and concrete results on gauge theories \cite{fewster:2013a,benini:2014,benini:2014a,benini:2014b,dappiaggi:2012,sanders:2014} and many other results.
In this framework one considers quantum field theories as covariant functors from a category of background structures to a category of physical systems.
In most simple examples the background structure is given by a category of globally hyperbolic spacetimes with so-called hyperbolic embeddings as morphisms, but the background structure can be replaced by anything reasonable that allows for a categorical formulation, see \cite{pinamonti:2009,fewster:2014} for examples of alternative choices.
A suitable category representing physical systems for algebraic quantum field theory is a category of ${}^*$-algebras so that a quantum field theory maps a spacetime to an algebra of observables in that spacetime.

However, while states, \viz, positive linear functionals on a ${}^*$-algebra, are not necessary for the formulation of the theory, they are indispensable if one wants to make quantitative predictions.
Given a state it is possible to return to a Hilbert space picture as a representation of the ${}^*$-algebra of observables via the Gel'fand--Naimark--Segal theorem.
Not all possible states on the algebra of quantum fields are of equal physical importance.
Physically and mathematically preferred states are the so-called Hadamard states, which have an ultraviolet behaviour analogous to that of the Minkowski vacuum.
Hadamard states are for example required for a reasonable semiclassical Einstein equation; otherwise the fluctuations of the quantum stress-energy tensor are not even distributions and the semiclassical Einstein equation becomes physically meaningless, because we equate a quantity with a probabilistic interpretation and `diverging' fluctuations with a classical non-fluctuating quantity.
Major advances in quantum field theory on curved spacetime were achieved after it was realized in~\cite{radzikowski:1996} that all Hadamard states satisfy a constraint on the wavefront set of the $n$-point functions of the state.
This constraint was called microlocal spectrum condition in allusion to the condition from Wightman quantum field theory on Minkowski spacetime.
In particular, this discovery led to the formulation of a rigorous theory of renormalization and a concept of normal ordering on curved spacetimes \cite{brunetti:1996,hollands:2001,hollands:2002,brunetti:2000}.

The developments of quantum field theory on curved spacetimes were often driven by problems related to the \emph{semiclassical Einstein equation}.
In the semiclassical Einstein equation contains instead of a classical stress-energy tensor the expectation value of a quantum stress-energy tensor~$\norder{T_{ab}}$ in a certain state~$\omega$:
\begin{equation*}
  G_{ab} + \Lambda g_{ab} = \frac{8 \uppi \mathrm{G}}{\mathrm{c}^4} \omega(\norder{T_{ab}}).
\end{equation*}
The quantum stress-energy tensor may be obtained by replacing the products of classical fields in the classical stress-energy tensor by normally ordered products of quantum fields.
This requires the notion of normal ordering on curved spacetimes mentioned above.
The resulting quantum stress-energy is not uniquely fixed but, due to the non-uniqueness of the normal ordering prescription, subject to a renormalization freedom, which is a polynomial of local geometric quantities, whose coefficients are called renormalization constants.

Quantum field theory on curved spacetimes is best understood in a few special cases of highly symmetric spacetimes.
In particular quantum fields on Friedmann--Lemaître--Robertson--Walker spacetimes are well-studied as they are important in quantum cosmology.
Nevertheless, already in this simplified case many interesting effects occur, for example the creation of particles in an expanding spacetime \cite{parker:1968}.
Due to the developments discussed above, in recent years computations of quantum field theoretic effects in cosmological spacetimes and their backreaction to the spacetime via the semiclassical Einstein equation have come into the reach of the algebraic approach to QFT on curved spacetimes.

A first step towards doing cosmology in algebraic QFT on curved spacetimes is often the construction of appropriate states.
Noteworthy recent works are the holographic (or bulk-to-boundary) construction \cite{dappiaggi:2006,dappiaggi:2009a,dappiaggi:2009b,dappiaggi:2011a,moretti:2006a,moretti:2008,dappiaggi:2011c,dappiaggi:2011d,dappiaggi:2013} and the states of low energy \cite{olbermann:2007,them:2013} (see also \cite{degner:2010,degner:2013}) which are a Bogoliubov transformation of adiabatic states \cite{luders:1990,junker:2002,parker:1969}.
Given a state, one can study the semiclassical Einstein equation to study the backreaction effects of quantum matter fields; this has been done, for example, in \cite{hack:2010,hack:2013,dappiaggi:2008,dappiaggi:2010}.
Going one step further, one can attempt to solve this semiclassical Einstein equation, \ie, finding a spacetime and a state on that spacetime so that the equation holds.
This problem was analyzed for cosmological spacetimes in \cite{pinamonti:2011,eltzner:2011,pinamonti:2015}.
Other works studied linearized gravity \cite{fewster:2012d,hack:2014}, inflation~\cite{eltzner:2013} (see also \cite{pinamonti:2015a} for a non-standard approach) and other cosmological models \cite{zschoche:2014} in the algebraic framework.
Furthermore, several researches have studied thermal aspects of quantum fields on curved spacetime, \cite{schlemmer:2010,eltzner:2013a,buchholz:2007,schlemmer:2008} to name a few, which are arguably of importance to quantum cosmology.

In this thesis several aspects of the works cited above will be summarized and, when necessary, developed further.
To give this work a clearer structure, it is divided into three parts.

The first part is mostly intended to lay the foundations of the remaining two parts.
In \cref{cha:differential_geometry} a rapid summary of subjects from differential geometry relevant to QFT on curved spacetimes is presented but it also contains a few sections and remarks on subjects which are usually not covered in standard text books on differential geometry, \eg\ bitensors.
\Cref{cha:lorentzian_geometry} focuses on the particular case of Lorentzian geometry including notions of causality, the classical Einstein equation and cosmology.
Analysis, in the broadest sense, will be the subject of \cref{cha:analysis} and in that chapter various results on topology, ${}^*$-algebras, functional derivatives and their relation to the Banach fixed-point theorem, microlocal analysis and wave equations will be summarised.
In favour of not jumping back and forth between different subjects in these three sections I chose a rather unpedagogical order and the reader should be aware that there are many interrelations between the various sections.
This should, however, not be a too large an obstacle for the reader.
\Cref{cha:combinatorics} concerns the enumerative combinatorics of permutations and appears somewhat unrelated to most of this thesis.
However, combinatorics is very important in many applications of quantum field theory and the results presented in this chapter are important in the moment problem for quadratic quantum fields \cite{fewster:2012b}.
The contents of this last chapter represent work by the author in collaboration with Fewster and were published in~\cite{fewster:2014a}.

In the second part of this thesis several aspects of quantum field theory on curved spacetimes will be discussed.
It begins with an introduction to the categorical framework of locally covariant quantum field theory with an emphasize on the field algebra of Klein--Gordon-like quantum fields in \cref{cha:locally_covariant}.
In \cref{cha:states} we discuss quantum states and in particular the construction of adiabatic states of cosmological spacetimes and the holographic construction of Hadamard states on asymptotically flat spacetimes.

The third and last part of this thesis represents the largest portion of novel research done during the authors Ph.D. studies; here the semiclassical Einstein equation will be analyzed in detail.
The basic notions including the stress-energy tensor for the scalar field and its renormalization will be introduced in \cref{cha:einstein}.
In \cref{cha:einstein_solutions} the proof of the author and Pinamonti \cite{pinamonti:2015} on the local and global existence of solutions to the semiclassical Einstein equation will be presented.
Finally, in the last chapter of this thesis (\cref{cha:fluctuations}), the fluctuations of the stress-energy tensor will be analyzed and how their backreaction to the metric may be accounted for; this work in collaboration with Pinamonti was already presented in~\cite{pinamonti:2015a}.

This thesis will be concluded with some final remarks on the research presented above in the \hyperref[cha:conclusions]{conclusions}.
Following this the cited references may be found in the \hyperref[cha:bibliography]{bibliography}, starting with the works (co)authored by the present author.

%% file: differential_geometry.tex
%!TEX root = master.tex

\chapter{Differential geometry}
\label{cha:differential_geometry}

\IdxRanBegin{diff-geo}

% \begin{quotation}
%   Manifolds are a bit like pornography: hard to define, but you know one when you see one.
%   \sourceatright{S. Weinberg}
% \end{quotation}

\section*{Summary}

This chapter is \emph{mostly} a summary of some common definitions and standard results on differential geometry and most of its content can be safely skipped by a reader well acquainted with the topic.
Proofs are omitted everywhere except in the last section and may be found in any text book on differential geometry.
Nevertheless, the author has attempted to present the material in such a way that many statements should become self-evident, although, as always, care should be taken.

In the first section (\cref{sec:differentiable_manifolds}) the basic theory of differentiable manifolds and vector bundles is summarized.
Here the notions of coordinates, maps between manifolds, vector bundles and sections, the (co)tangent bundle, (co)vectors and (co)vector fields, curves, tensor and exterior tensor product bundles, bundle metrics, frames, differential operators, and the index notation are explained.
The second section (\cref{sec:connections}) is concerned with the definition of connections on vector bundles and the objects that follow from this.
That is, it discusses the notions of curvature, geodesics, and the slightly unrelated concept of Killing vector fields.
Differential forms and integration are introduced in the third section (\cref{sec:differential_forms}).
In particular we will introduce the de Rham cohomology, the Hodge star and the dual of the exterior derivative, the codifferential, which leads to the definition of the Laplace--de Rham operator, and close with a short discussion of integral manifolds.

In the presentation of these three sections the author follows partially that of~\cite{lee:2013} and also~\cite{abraham:1988}, but these standard definitions may be found in many places in the literature.

The fourth and last section (\cref{sec:bitensors}) treats a more obscure topic: bitensors.
Bitensors are already introduced in an abstract manner in \cref{sub:dg_tensors}; a simple, yet important, example are biscalars: functions on the product $M \times M$ of a manifold~$M$.
In this last section concrete and important cases of bitensors such as Synge's world function and the van Vleck--Morette determinant are discussed.
It will also form the foundation for the discussion of the Hadamard coefficients in \cref{sub:hadamard_parametrix}.
An excellent resource on bitensors is the review article \cite{poisson:2011}, which contains most of the first part of this section.
The second part of this section is concerned with computational methods that help the calculation of coincidence limits of bitensors.
Here we will discuss the semi-recursive Avramidi method developed in~\cite{ottewill:2011}.
We close this section with a recursive method to calculate the coefficients of an asymptotic expansion of Synge's world function in coordinate separation.
To the authors knowledge, this simple and efficient method has never been fully developed but traces of it may be found in~\cite{ottewill:2009}.

%-- Differentiable manifolds --------------------------------------------------%
\section{Differentiable manifolds and vector bundles}
\label{sec:differentiable_manifolds}

A \IdxMain{top-man}\emph{topological manifold of dimension~$n$} is a second-countable Hausdorff space~$M$ that is locally homeomorphic to~$\RR^n$ (\ie, each point of~$M$ has a neighbourhood that is homeomorphic to an open subset of~$\RR^n$).
We often omit the dimension of the manifold and simply say: $M$ is a \emph{topological manifold}.

Since a topological manifold $M$ is locally homeomorphic to~$\RR^n$, we can assign coordinates to points of~$M$ in each open neighbourhood $U \subset M$:
A \IdxMain{chart}\emph{(coordinate) chart} of~$M$ is a pair~$(U, \varphi)$ which gives exactly such a homeomorphism
\begin{equation*}
  \varphi : U \to \varphi(U) \subset \RR^n;
\end{equation*}
we call~$U$ its \IdxMain{coord-neighbourhood}\emph{coordinate neighbourhood}.
The component functions ${(x^1, \dotsc, x^n)} = \varphi$ are called \emph{(local) coordinates} on~$U$.

An \IdxMain{atlas}\emph{atlas}~$A$ of~$M$ is a family of charts~$(U_i, \varphi_i)_{i \in \NN}$ which cover~$M$.
If any two overlapping charts~$(U, \varphi), (V, \psi)$ in an atlas are \IdxMain{smoothly-compat}\emph{smoothly compatible}, \viz, the \IdxMain{trans-map}\emph{transition map}
\begin{equation*}
  \psi \circ \varphi^{-1} : \varphi(U \cap V) \to \psi(U \cap V)
\end{equation*}
is a smooth, bijective map with a smooth inverse (\cref{fig:smooth_atlas}), we say that the atlas is \IdxMain{smooth-atlas}\emph{smooth}.
We further say that a smooth atlas~$A$ on~$M$ is \IdxMain{max-atlas}\emph{maximal} if it is not properly contained in any larger smooth atlas so that any chart which is smoothly compatible with the charts of~$A$ is already contained in~$A$.

\begin{figure}
  \centering
  \includegraphics{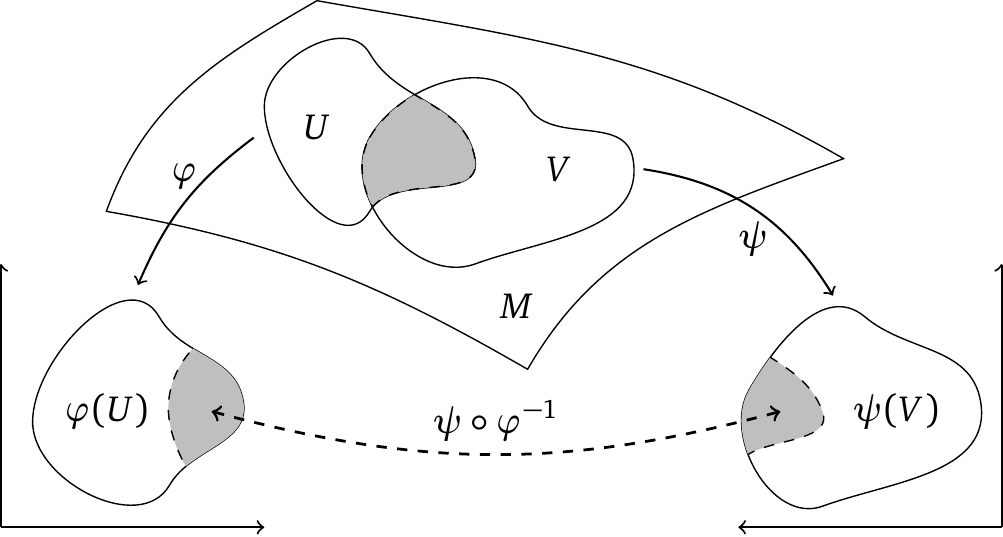}
  \caption{Two overlapping charts and their transition map.}
  \label{fig:smooth_atlas}
\end{figure}

Finally, a \IdxMain{smooth-man}\emph{smooth manifold} is a pair~$(M, A)$, where~$M$ is a topological manifold and~$A$ a maximal smooth atlas.
A maximal smooth atlas might not exist and, if it exists, it is not necessarily unique as shown, \eg, by the existence of exotic $\RR^4$.
Nevertheless, usually a canonical smooth atlas is understood from context.
Then we omit the explicit mention of the maximal smooth atlas~$A$ and say: $M$ is a smooth manifold.
One can replace the requirement of the transition maps in an atlas to be smooth by requiring that the transition maps are~$C^k$, \optwordhy{real}analytic or complex-analytic (if $\dim M = 2n$, we have $\RR^{2n} \simeq \CC^n$) thus arriving at the notions of~$C^k$, \optwordhy{real}analytic and complex-analytic manifolds.

\subsection{Smooth maps}
\label{sub:dg_maps}

A real-valued function $f : M \to \RR$ on a smooth manifold~$M$ is a \IdxMain{fun-on-man}\emph{$C^k$, smooth or analytic function} if there exists a chart~$(U, \varphi)$ containing~$x$ at every~$x \in M$ such that the composition ${f \circ \varphi^{-1}}$ is~$C^k$, smooth or analytic on the image~$\varphi(U)$; the spaces of these function are denoted $C^k(M)$, $C^\infty(M)$ and~$C^\omega(M)$ respectively.

More generally, a map $f : M \to N$ between two smooth manifolds~$M$ and~$N$ is~$C^k$, smooth or analytic, if there exist charts~$(U, \varphi)$ at~$x \in M$ and~$(V, \psi)$ at~$f(x) \in N$ such that the composition
\begin{equation*}
  \psi \circ f \circ \varphi^{-1} : \varphi(U) \to \psi(V)
\end{equation*}
is~$C^k$, smooth or analytic (\cref{fig:manifold_map}).
If $M$ and~$N$ have equal dimension and~$f$ is a homeomorphism such that $f$ and its inverse~$f^{-1}$ are smooth, we call $f : M \to N$ a \IdxMain{diffeomorph}\emph{diffeomorphism}.
Whenever there exists such a diffeomorphism between, they are \emph{diffeomorphic}; in symbols $M \simeq N$.

\begin{figure}
  \centering
  \includegraphics{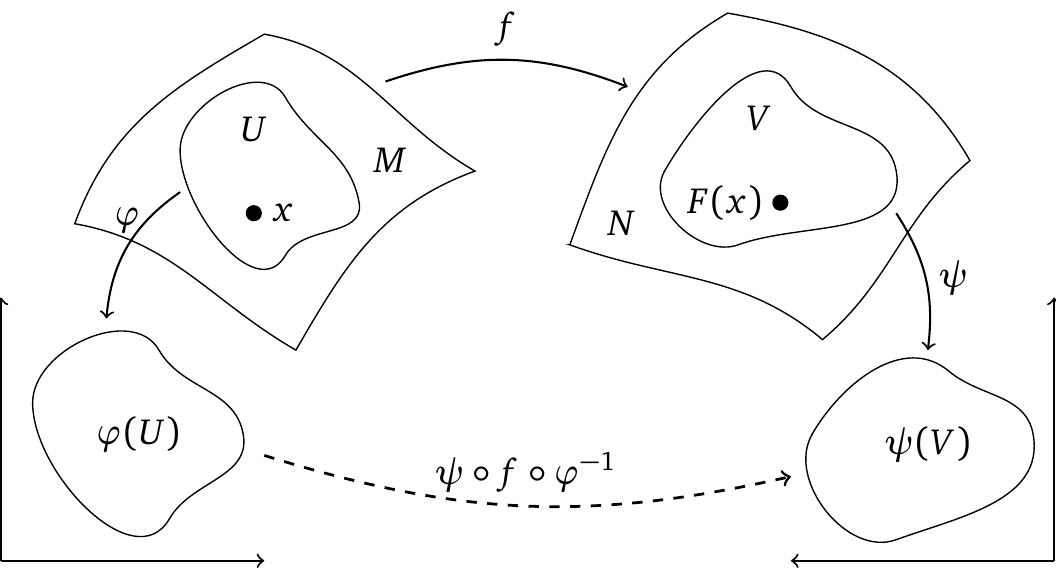}
  \caption{A map between two manifolds.}
  \label{fig:manifold_map}
\end{figure}

\subsection{Vector bundles}
\label{sub:dg_bundles}

\begin{figure}[b]
  \centering
  \includegraphics[scale=1.02]{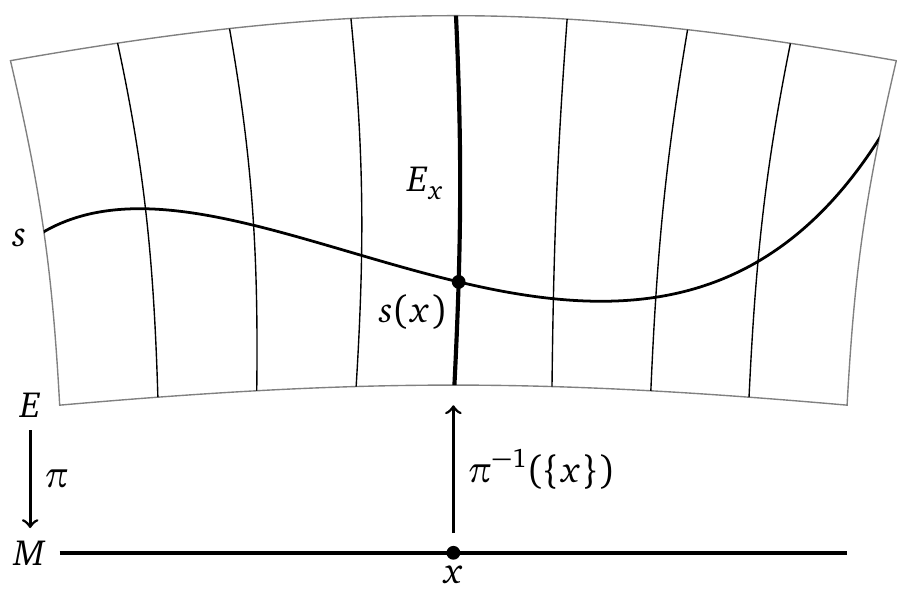}
  \caption{A vector bundle and a section.}
  \label{fig:vector_bundle}
\end{figure}

A (smooth) \IdxMain{vec-bund}\emph{$\KK$-vector\footnote{$\KK$ will always be either $\RR$ or $\CC$. In particular, $\KK$ is a field of characteristic $0$.} bundle} of dimension~$n$
\begin{equation*}
  \pi : E \to M
\end{equation*}
consists of two smooth manifolds~$E$, the \IdxMain{total-space}\emph{total space}, and $M$, the \IdxMain{base-space}\emph{base (space)}, and a smooth surjection~$\pi$, the \IdxMain{bund-proj}\emph{bundle projection}, that associates to every~$x \in M$ a $n$-dimensional $\KK$-vector space~$E_x = \pi^{-1}(\{x\})$, the \IdxMain{fibre}\emph{fibre} of~$E$ at~$x$ (\cref{fig:vector_bundle}).
Moreover, we require that around every~$x$ there exists an open neighbourhood~$U \subset M$ and a diffeomorphism $\varphi: \pi^{-1}(U) \to U \times E_x$ such that its projection to the first factor gives the bundle projection: $\pr_1 \circ\, \varphi = \pi$; $\varphi$ is called \IdxMain{loc-triv}\emph{local trivialization} of the vector bundle.
A trivialization of a vector bundle over its whole base is called a \IdxMain{glob-triv}\emph{global trivialization}.

A smooth map $f : E \to F$ between two $\KK$-vector bundles~$\pi_M : E \to M$ and~$\pi_N : F \to N$ is a \IdxMain{bund-morph}\emph{(vector) bundle homomorphism} if there exists a smooth map $g : M \to N$ such that $\pi_N \circ f = g \circ \pi_M$ and the restriction of the map to each fibre $f \restriction_{E_x} : E_x \to F_{g(x)}$ is $\KK$-linear.
In other words we require that the diagram
\begin{center}
  \includegraphics{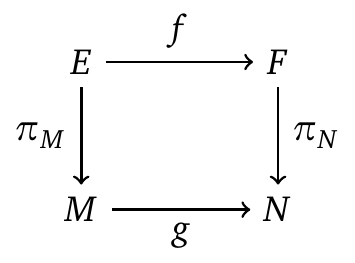}
\end{center}
commutes and that $f$ is $\KK$-linear map on each fibre.
By definition, if $f$ is a bundle homomorphism, then $g$ is given by $g = \pi_N \circ f \circ \pi_M^{-1}$.

Given an open subset~$U \subset M$, we can restrict $\pi : E \to M$ to a vector bundle~$\pi_{U} : E_{U} \to U$ by setting $E_{U} \defn \pi^{-1}(U)$ and $\pi_{U} \defn \pi \restriction_{U}$.
More generally, a subset~$E' \subset E$ such that $\pi \restriction_{E'} : E' \to M$ is a vector bundle and $E' \cap \pi^{-1}(x)$ is a vector subspace in~$\pi^{-1}(x)$ for all $x \in M$ is called a \IdxMain{vec-subbund}\emph{vector subbundle} of~$E$.
If each fibre of~$E'$ has dimension $k$, we say that $E'$ is a rank-$k$ subbundle of~$E$.
For a subbundle $\pi' : E' \to M$ of $E$ we define the \IdxMain{quot-bund}\emph{quotient bundle}~$E / E'$ as the disjoint union $\bigsqcup_{x \in M} E_x / E'_x$ of the quotient spaces of the fibres.

If $M, N$ are smooth manifolds with a smooth map~$\psi: M \to N$ and a vector-bundle~$E \to N$, we can define on~$M$ the \IdxMain{pb-bund}\emph{pullback bundle}~$\pb{\psi} E$ as the bundle whose fibres over~$M$ are given by $(\pb{\psi} E)_x \defn E_{\psi(x)}$ for each~$x \in M$.

A \IdxMain{section}\emph{section} of a vector bundle~$E$ is a continuous map~$f : M \to E$ such that $\pi \circ f = \id_M$ (\cref{fig:vector_bundle}); the space of sections of a vector bundle~$E$ is denoted by~$\Gamma(E)$.
We denote by~$\Gamma^n(E)$ the space of~$C^n$ sections ($C^n$ maps~$f : M \to E$) of the vector bundle~$E$, while the spaces of compactly supported sections are indicated by a subscript $0$, \eg, $\Gamma_0^\infty(E)$.
Furthermore, a \IdxMain{loc-sect}\emph{local section} over an open subset~$U \subset M$ is a section of the vector bundle~$E_U$.

If $N$ is another smooth manifold with a vector-bundle~$F \to N$ and there exists a smooth map~$\psi: M \to N$, the \IdxMain{pb-sect}\emph{pullback section} $\pb{\psi} f \in \pb{\psi} F$ of~$f \in \Gamma(F)$ is defined as the section $\pb{\psi} f \defn f \circ \psi$.
The opposite of the pullback is achieved by the pushforward if $\psi$ is a diffeomorphism.
Namely, we will say that $\pf{\psi} h \defn h \circ \psi^{-1}$ is the \IdxMain{pf-sect}\emph{pushforward section} of~$h \in \Gamma(E)$.
If $h$ is compactly supported in a region~$U \subset M$, it suffices that $\psi$ is a diffeomorphism onto the image of $U$; outside of $\psi(U)$ we set $\pf{\psi} h$ identical to zero.

\subsection{Tangent bundle}
\label{sub:dg_tangent}

Let $M$ be a smooth manifold.
A \IdxMain{vector}\emph{(tangent) vector at~$x \in M$} is a linear map $v : C^\infty(M) \to \RR$ that satisfies the Leibniz `product' rule
\begin{equation*}
  v(f g) = f(x)\, v g + g(x)\, v f
\end{equation*}
for all $f, g \in C^\infty(M)$.
The set of all tangent vectors constitutes a vector space~$T_x M$ called the \IdxMain{tan-space}\emph{tangent space to~$M$ at~$x$}; it has the same dimension as the base manifold~$M$ for each~$x \in M$.
Note that, if a smooth manifold~$M$ is also a vector space, then we can identify it with its tangent space, \ie, $M \simeq T_x M$ at each point~$x$, which justifies to the geometric visualization of the tangent space (\cref{fig:tangent_space}).

\begin{figure}
  \centering
  \includegraphics{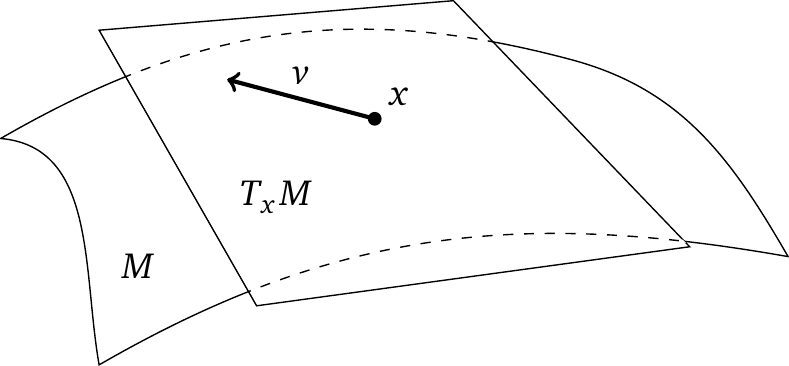}
  \caption{The tangent space of a manifold at a point and a tangent vector.}
  \label{fig:tangent_space}
\end{figure}

The vector bundle $TM \to M$ with fibres $T_x M$ at $x$ is the \IdxMain{tan-bund}\emph{tangent bundle of $M$}.
In the special situation where the $n$-dimensional smooth manifold~$M$ can be covered by a single chart, $TM$ is diffeomorphic to~$M \times \RR^n$.

The sections $\Gamma(TM)$ of the vector bundle are called \IdxMain{vec-field}\emph{vector fields}.
Applying a vector field~$v \in \Gamma(TM)$ to a function~$f \in C^\infty(M)$, we obtain a new function $(v f)(x) \mapsto v_x f$, \viz, a vector field defines a linear automorphism on the smooth functions called a \emph{derivation}.
This is exactly the \IdxMain{lie-deriv}\emph{Lie derivative} $\lieD_v f$ of a function $f \in C^\infty(M)$ along a vector field~$v$: $(\lieD_v f)(x) = (v f)(x)$.
The Lie derivative $\lieD_v w$ of a differentiable vector field~$w$ with respect to another differentiable vector field~$v$ is another vector field such that
\begin{equation*}
  \lieD_v w \defn [v, w] \defn v \circ w - w \circ v
\end{equation*}
when applied to smooth functions.
It satisfies the Leibniz rule and the Jacobi identity
\begin{equation*}
  \lieD_v (f w) = (\lieD_v f) w + f \lieD_v w,
  \quad
  \lieD_u [v, w] = [\lieD_u v, w] + [v, \lieD_u w]
\end{equation*}
for all $f \in C^\infty(M)$ and vector fields $u, v, w \in \Gamma^\infty(TM)$.

Given two smooth manifolds~$M, N$ and a smooth map $F : M \to N$, we can define at each point~$x \in M$ the \IdxMain{tan-map}\emph{tangent map} or \emph{differential of~$F$ at~$x$} as the linear map
\begin{equation*}
  T_x F : T_x M \to T_{F(x)} N,
\end{equation*}
see also \cref{fig:differential}, which is for every~$v \in T_x M$ and~$f \in C^\infty(M)$ the derivation
\begin{equation*}
  T_x F (v)(f) = v(f \circ F).
\end{equation*}
The tangent maps of~$F$ at all points taken together form $T F : T M \to T N$, the \IdxMain{glob-tan-map}\emph{(global) tangent map} or \emph{differential of~$F$}.
If $N = \KK$, \ie, $F$ is a smooth function on a manifold, we note that $\dif F \defn T F$ is the usual differential.

If $F$ is even a diffeomorphism, then $T F$ defines a bijection between the vector fields on~$M$ and~$N$.
In this case one can define the \IdxMain{pf-vec}\emph{pushforward} $\pf{F} v$ of a vector field~$v$ on~$M$ by~$F$ as the vector field on~$N$ given at each $x \in N$ by
\begin{equation*}
  (\pf{F} v)_x \defn T_{F^{-1}(x)} F \big( v_{F^{-1}(x)} \big).
\end{equation*}
Clearly, this is generally not well-defined if $F$ is not a diffeomorphism.
Using that $F$ is invertible, a \IdxMain{pb-vec}\emph{pullback} $\pb{F} w$ of a vector field~$w$ on~$N$ can be defined as the inverse of the pushforward, namely,
\begin{equation*}
  \pb{F} w \defn \pf{(F^{-1})} w.
\end{equation*}

\begin{figure}
  \centering
  \includegraphics{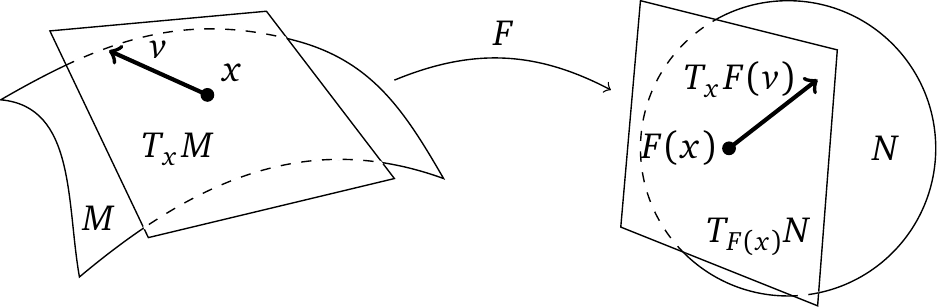}
  \caption{The tangent map at a point.}
  \label{fig:differential}
\end{figure}

The tangent map allows us to single out an important type of maps between manifolds:
\IdxMain{immers}\emph{Immersions} are maps $F : M \to N$ such that $T F$ is injective; if in addition $F$ is injective and a homeomorphisms onto its image, then it is called an \IdxMain{embed}\emph{embedding}.
Consequently we say that a subset~$S \subset M$ is an \IdxMain{imm-subman}\emph{immersed submanifold} if it is a topological manifold and the inclusion~$S \hookrightarrow M$ is an immersion; if the topology of~$S$ is the subspace topology and the inclusion is an embedding, then $S$ is called an \IdxMain{emb-subman}\emph{embedded submanifold}.
Observe that the tangent space~$T_x S$ of a submanifold~$S \subset M$ is a subspace of~$T_x M$ at every point~$x \in S$.
We can then say that a vector field~$v$ is \IdxMain{tan-space-subman}\emph{tangent to~$S$} if $v_x \in T_x S$ at every point~$x \in S$.

\subsection{Curves}
\label{sub:dg_curves}

A \IdxMain{param-curve}\emph{parametrized curve} on a smooth manifold~$M$ is a map $\gamma : I \to M$ from a connected, usually open, interval~$I \subset R$ into the manifold.

A curve $\gamma : [a, b) \to M$ is called \IdxMain{inext-curve}\emph{inextendible} if there exists a sequence $t_n$ converging to $b$ such that $\gamma(t_n)$ does not converge.
This notion readily extends to left-open domains and open domains.

The \IdxMain{velocity}\emph{velocity} $\dot\gamma(t)$ at $t$ of a differentiable parametrized curve $\gamma : I \to M$, where $I \subset \RR$ is an interval, is the vector
\begin{equation*}
  \dot\gamma(t) \defn T_t\gamma\bigg(\od{}{t} \bigg|_t \bigg) \in T_{\gamma(t)}M.
\end{equation*}
Working in the opposite direction, we can try to find a curve~$\gamma$, whose velocity at every point is determined by a given vector field~$v$:
\begin{equation}\label{eq:integral_curve}
  \dot\gamma(t) = v_{\gamma(t)}.
\end{equation}
Such a curve~$\gamma$ is called an \IdxMain{int-curve}\emph{integral curve} of~$v$.

For a sufficiently small interval~$I$ around~$0$ a unique integral curve starting at a point $x \in M$ can always be found by solving the differential equation~\eqref{eq:integral_curve} in a coordinate neighbourhood of~$x$.
The domain of an integral curve cannot necessarily be extended to the entire real line.
We say that a vector field is \IdxMain{compl-vec}\emph{complete} if the domain of all of its \IdxMain{max-int-curve}\emph{maximal integral curves}, \ie, the integral curves whose domain cannot be extended, is the entire real line.

The \IdxMain{flow}\emph{flow} $\psi_t$ of a complete vector field~$v$ is
\begin{equation*}
  \psi_t(x) = \gamma_x(t),
\end{equation*}
where $\gamma_x$ is the maximal integral curve starting at $x \in M$.
This defines for every $t$ a diffeomorphism $\psi_t : M \to M$ and the collection of all these diffeomorphisms is a group $\{ \psi_t \}_{t \in \RR}$ with unit $\psi_0$, multiplication $\psi_s \circ \psi_t = \psi_{s + t}$ and inverse $\psi_t^{-1} = \psi_{-t}$.
If $v$ is not complete, then one can still define a \IdxMain{loc-flow}\emph{local flow} $\psi_t$ around a point~$x$ with domain $U$, where $U$ is a neighbourhood of $x$, and $t \in I$ is restricted to an interval around~$0$.

\subsection{Cotangent bundle}
\label{sub:dg_cotangent}

At each point~$x$ of a smooth manifold~$M$, we define the \IdxMain{cotan-space}\emph{cotangent space to~$M$ at~$x$}, denoted $T_x^* M$, as the dual space of tangent space at the same point, namely,
\begin{equation*}
  T_x^* M \defn (T_x M)^*.
\end{equation*}
The elements of the cotangent space~$T_x^* M$ are called \IdxMain{covector}\emph{(tangent) covectors}; naturally they are linear functionals on the tangent space.
Taking the cotangent spaces at each point as fibres, we obtain $T^*\!M$, the \IdxMain{cotan-bund}\emph{cotangent bundle of $M$}.
The sections of the cotangent bundle are called \IdxMain{covec-field}\emph{covector fields} or \emph{one-forms} (see \cref{sec:differential_forms}).
The differentials (at a point) of functions discussed above are examples of covectors resp. covector fields.

A smooth map $F : M \to N$ between smooth manifolds $M, N$ induces at each point $x \in M$ the map $T_x F : T_x M \to T_{F(x)} N$ between the tangent spaces.
By duality one can find the transpose map $T^*_x F : T_{F(x)}^* N \to T_x^* M$, called the \IdxMain{cotan-map}\emph{cotangent map of~$F$ at~$x$}, between the cotangent spaces at~$F(x)$ and~$x$, which is given for each $v \in T_x M$ and $\omega \in T^*_{F(x)} N$ as
\begin{equation*}
  \big(T^*_x F(\omega)\big)(v) = \omega\big(T_x F(v)\big).
\end{equation*}
This map then gives the \IdxMain{glob-cotan-map}\emph{(global) cotangent map} $T^*\!F : T^*\!N \to T^*\!M$ and thus facilitates the definition of the \IdxMain{pb-covector}\emph{pullback} $\pb{F} \omega$ of a covector field~$\omega$ on~$N$ by~$F$ as the covector field on~$M$ given at each $x \in M$ by
\begin{equation*}
  (\pb{F} \omega)_x = T^*_x F \big( \omega_{F(x)} \big).
\end{equation*}
Note that, different than the pushforward, the pullback is even defined if $F$ is not a diffeomorphism.
However, if $F$ is a diffeomorphism, we can define the \IdxMain{pf-covector}\emph{pushforward} of a covector field as the pullback via the inverse $F^{-1}$.
That is,
\begin{equation*}
  \pf{F} \eta \defn \pb{(F^{-1})} \eta,
\end{equation*}
where $\eta$ is a covector field on~$M$.
The pushforward of a compactly supported covector field can also be defined if $F$ is an embedding \emph{and} $\dim M = \dim N$ by setting the $\pf{F} \eta$ to zero outside the image of~$F$.

\subsection{Tensors product bundles}
\label{sub:dg_tensors}

Given two $\KK$-vector bundles~$\pi : E \to M$ and~$\rho : F \to M$ over the same smooth manifold~$M$, we can define the \IdxMain{tensor-prod-bund}\emph{tensor product bundle} $\pi \otimes \rho : E \otimes F \to M$, which is just the fibrewise tensor product of vector spaces.
Namely, the fibre of $E \otimes F$ at $x \in M$ is $(E \otimes F)_x = (\pi \otimes \rho)^{-1}(\{x\}) = E_x \otimes F_x$.

In particular, we denote by~$T^p_q(E)$ the \emph{tensor bundle of type~$(p, q)$} of a vector bundle~$E$:
\begin{equation*}
  T^p_q(E) \defn E^{\otimes p} \otimes (E^*)^{\otimes q};
\end{equation*}
the tensor bundle of the \optword{co}tangent bundle is simply denoted by~$T^p_q M$ instead of~$T^p_q (T M)$.
The sections of the tensor product bundle~$T^p_q M$ are called \IdxMain{tensor}\emph{tensor (fields)}, if $p = 0$, we say that a tensor is \IdxMain{covariant}\emph{covariant}, while we say that it is \IdxMain{contravariant}\emph{contravariant} if $q = 0$.

The \Idx{lie-deriv}Lie derivative defined in \cref{sub:dg_tangent} generalizes to covariant tensor fields by duality in the following way:
Given a covariant tensor field $S \in \Gamma^\infty(T^0_q M)$ and vector fields $v, w_1, \dotsc, w_q \in \Gamma^\infty(TM)$, the $\lieD_v S$ of $S$ along $v$ can be defined by
\begin{equation*}\begin{split}
  (\lieD_v S)(w_1, \dotsc, w_q) & \defn v\big( S(w_1, \dotsc, w_q) \big) - S(\lieD_v w_1, w_2, \dotsc, w_q) \\&\quad - \dotsb - S(w_1, \dotsc, w_{q-1}, \lieD_v w_q).
\end{split}\end{equation*}
Note that it satisfies the Leibniz rule
\begin{equation*}
  \lieD_v (S \otimes T) = (\lieD_v S) \otimes T + S \otimes \lieD_v T,
\end{equation*}
where $T$ is any other differentiable covariant tensor field.
Since the Lie derivative on covariant tensor fields defines a Lie derivative on covector fields, this may be used to define a Lie derivative on mixed tensor field.

The pullback or pushforward of a mixed tensor field $S \in \Gamma(T^p_q N)$ or $T \in \Gamma(T^p_q M)$ by a diffeomorphism $F : M \to N$, is defined as the $(p+q)$-fold tensor product of the pullback or pushforward map for (co)vector fields.
If $S$ is covariant (\ie, $p=0$), then the pullback is also defined if $F$ is not a diffeomorphism.

If $F = E$, we can define the \IdxMain{symm-ten}\emph{symmetric} and \IdxMain{antisymm-ten}\emph{antisymmetric tensor product bundle} $E \odot E$ and $E \wedge E$ as the quotient bundles under the fibrewise equivalence relations $v \otimes w \sim \pm v \otimes w$ for all $v, w \in E_x$.
More generally, we denote by~$\mathrm{S}^p(E)$ and~$\bigwedge^p(E)$ the $p$-th \optword{anti}symmetric tensor product bundle which satisfy fibrewise the relation
\begin{align*}
  v(x_1, \dotsc, x_p) & = v(x_{\sigma(1)}, \dotsc, x_{\sigma(p)})
  \qquad \text{or} \\
  v(x_1, \dotsc, x_p) & = (\sign \sigma) v(x_{\sigma(1)}, \dotsc, x_{\sigma(p)})
\end{align*}
for $v \in E_x^{\otimes p}$ and all $\sigma \in \perms_p$, the symmetric group of $p$~elements, \cf\ \cref{sub:perm_linear}.
That is, the fibres of~$\mathrm{S}^p(E)$ and~$\bigwedge^p(E)$ are the $p$-th symmetric (resp. exterior) power of the fibres of~$E$.
Maps $\Sym : T^p_0(E) \to \mathrm{S}^p(E)$ and $\Alt : T^p_0(E) \to \bigwedge^p(E)$ extend from the fibrewise maps
\begin{align*}
  \Sym \big(v(x_1, \dotsc, x_p)\big) & = \frac{1}{p!} \sum_{\sigma \in \perms_p} v(x_{\sigma(1)}, \dotsc, x_{\sigma(p)})
  \qquad \text{and} \\
  \Alt \big(v(x_1, \dotsc, x_p)\big) & = \frac{1}{p!} \sum_{\sigma \in \perms_p} (\sign \sigma) v(x_{\sigma(1)}, \dotsc, x_{\sigma(p)}),
\end{align*}
where $v \in E_x^{\otimes p}$.
Moreover, fibrewise products $\odot : \mathrm{S}^p(E)_x \times \mathrm{S}^q(E)_x \to \mathrm{S}^{p+q}(E)_x$ and $\wedge : \bigwedge^p(E)_x \times \bigwedge^q(E)_x \to \bigwedge^{p+q}(E)_x$ are defined as
\begin{equation*}
  v \odot w = \Sym(v \otimes w)
  \quad \text{and} \quad
  v \wedge w = \frac{(p+q)!}{p!\, q!} \Alt(v \otimes w).
\end{equation*}

Another possibility to combine two vector bundles~$E$ and~$F$ is the \IdxMain{ext-ten-prod-bund}\emph{exterior tensor product} $E \boxtimes F \to M \times M$.
It is defined as the vector bundle over~$M \times M$ with fibre
\begin{equation*}
  (E \boxtimes F)_x = \pi^{-1}(\{x\}) \otimes \rho^{-1}(\{x'\}) = E_x \otimes F_{x'}
\end{equation*}
over the point~$(x, x')$.
The sections of the exterior tensor product~$T^p_q M \boxtimes T^r_s M$ are called \IdxMain{bitensor}\emph{bitensor (fields)}.

\subsection{Metrics}
\label{sub:dg_metrics}

Every vector bundle~$E$ has a \IdxMain{dual-bund}\emph{dual bundle}~$\pi^* : E^* \to M$ which has as its fibre~$E^*_x$, the dual vector spaces of the fibres~$E_x$.
We call the natural pairing $f(v)$ of an element $f \in E^*_x$ in the dual fibre on an element $v \in E_x$ in the corresponding fibre a \IdxMain{contract}\emph{contraction}.

A canonical isomorphism between~$E$ and~$E^*$ can be constructed if $E$ carries a \IdxMain{bund-metric}\emph{(bundle) metric}, \ie, a map
\begin{equation*}
  \metricp{\cdot\,}{\cdot} : E \times_M E \to \KK
\end{equation*}
such that the restriction~$\metricp{\cdot\,}{\cdot}_x$ to each fibre~$E_x$ is a fibrewise non-degenerate bilinear form; a (positive-definite) bundle metric can always be constructed.
In other words, a metric on~$E$ is a section in the tensor product bundle $\Gamma(E^* \otimes E^*)$.
Thus a metric induces a \IdxMain{metric-contract}\emph{metric contraction} between two elements of the same fibre~$E_x$.

The \IdxMain{dual-metric}\emph{dual metric} $\metricp{\cdot\,}{\cdot}^*$ is the unique metric on~$E^*$ such that
\begin{equation*}
  \metricp{\omega}{\eta}^* = \metricp{v}{w}
  \quad \text{with} \quad
  \omega = \metricp{v}{\cdot}, \eta = \metricp{w}{\cdot}
\end{equation*}
for all $v, w \in \Gamma(E)$.
Moreover, given metrics~$\metricp{\cdot\,}{\cdot}^E$ and $\metricp{\cdot\,}{\cdot}^F$ on~$E$ and~$F$, they induce metrics on the tensor product bundle ${E \otimes F}$ and the exterior tensor product bundle ${E \boxtimes F}$.
For the tensor product bundle it is defined fibrewise for all $v \in E_x, w \in F_x$ by
\begin{equation*}
  \metricp{v \otimes w}{v \otimes w}^{E \otimes F}_x = \metricp{v}{v}^E_x \metricp{w}{w}^F_x
\end{equation*}
and can be extended to arbitrary pairings by polarization and linearity; an analogous construction works for the exterior tensor product.

In pseudo-Riemannian geometry we find the tangent bundle equipped with a continuous symmetric metric usually denoted
\begin{equation*}
  g : TM \times_M TM \to \RR
\end{equation*}
with dual metric~$g^*$ on the cotangent bundle.
The canonical isomorphism induced by~$g$ between~$TM$ and~$T^*\!M$ is given by the \IdxMain{music}\emph{musical isomorphisms} \IdxMain{music-flat}\emph{`flat'} $\flat : TM \to T^*\!M$ and its inverse \IdxMain{music-sharp}\emph{`sharp'} $\sharp : T^*\!M \to TM$:
\begin{equation*}
  v^\flat \defn g(v, \cdot), \qquad
  \omega^\sharp \defn g(\omega, \cdot).
\end{equation*}
A tuple~$(M, g)$ of a smooth manifold $M$ with a metric~$g$ on its tangent bundle is called a \IdxMain{pseudo-riem-man}\emph{pseudo-Riemannian manifold}.

The maximal dimension of subspaces of~$T_x M$ where $g_x$ is negative-definite is called the \IdxMain{metric-ind}\emph{index}~$\Ind(g)$ of~$g$; since $g$ is continuous and non-degenerate, the index constant over the manifold.
We distinguish in particular two cases:
\begin{enumerate}[(a)]
  \item If $\Ind(g) = 0$ or, in other words, $g$ is pointwise positive-definite, we say that $g$ is a \IdxMain{metric-riem}\emph{Riemannian metric}.
  \item If $\Ind(g) = 1$ (and the manifold at least two-dimensional), we say that $g$ is a \IdxMain{metric-lor}\emph{Lorentzian metric}.\footnote{This \emph{choice} corresponds to the $\mathord{-}\mathord{+}\mathord{+}\mathord{+}$ convention, which we will use. In particle physics one often adopts the opposite convention $\mathord{+}\mathord{-}\mathord{-}\mathord{-}$, \ie, $(M, g)$ is Lorentzian if $\Ind(g) = n-1$.}
\end{enumerate}
We say that $(M, g)$ is a \IdxMain{riem-man}\emph{Riemannian} (\IdxMain{lor-man}\emph{Lorentzian}) \emph{manifold} if the metric~$g$ is Riemannian (Lorentzian).

The two prototypical examples for a Riemannian and a Lorentzian manifold are \IdxMain{euclid-space}\emph{Euclidean space} and \IdxMain{mink-spacetime}\emph{Minkowski space(time)}:
$n$-dimensional Euclidean space is the smooth manifold over $\RR^n$ with a the global \IdxMain{euclid-chart}chart $(\RR^n, \id)$, coordinate functions $(x^1, \dotsc, x^n)$ and with the \IdxMain{euclid-metric}\emph{Euclidean metric}
\begin{equation*}
  \delta \defn \sum_{i = 1}^n \dif x^i \otimes \dif x^i.
\end{equation*}
In the conventions chosen here, $(1+n)$-dimensional Minkowski spacetime is a smooth manifold modelled on $\RR^{1+n}$ with the single \IdxMain{mink-chart}chart $(\RR^{1+n}, \id)$, coordinate functions $(t, x^1, \dotsc, x^n)$ and with the \IdxMain{mink-metric}\emph{Minkowski metric}
\begin{equation}\label{eq:minkowski}
  \eta \defn - \dif t \otimes \dif t + \sum_{i = 1}^n \dif x^i \otimes \dif x^i
\end{equation}
in the coordinate frame.
Occasionally one uses the shorthand~$\MM$ to denote four\hyp{}dimensional Minkowski space $(\RR^4, \eta)$.

Let $(M, g_M)$ and $(N, g_N)$ be pseudo-Riemannian manifolds and $\psi: M \to N$ a diffeomorphism.
The map $\psi$ is called a \IdxMain{conf-iso}\emph{conformal isometry} if
\begin{equation*}
  \pb{\psi} g_N = \Omega^2 g_M
\end{equation*}
for some positive function $\Omega \in C^\infty_+(M)$.
Conformal isometries preserve the angles between vectors; in particular $g_N(v, w)_{\psi(x)} \gtreqlessslant 0$ implies $g_M(\pb{\psi} v, \pb{\psi} w)_(x) \gtreqlessslant 0$ for all $v, w \in \Gamma(TN)$ and $x \in M$.
In the special case that $\Omega \equiv 1$, we say that $\psi$ is an \IdxMain{isometry}\emph{isometry}.
These notions generalize straightforwardly to immersions and embeddings; and we call these maps \IdxMain{emb-conf}\IdxMain{imm-conf}\emph{conformal immersions/embeddings} and \IdxMain{emb-iso}\IdxMain{imm-iso}\emph{isometric immersions/embeddings}.

\subsection{Frames}
\label{sub:dg_frames}

A \IdxMain{loc-frame}\emph{local frame} of a vector bundle~$E$ on~$M$ is a set~$\{ e_\mu \}$ of smooth local sections~$\Gamma^\infty(E_{U})$ on a domain~$U \subset M$ contained in a neighbourhood of a local trivialization such that~$\{ e_\mu(x) \}$ forms a basis for each fibre~$E_x$ over~$U$.
The \IdxMain{dual-frame}\emph{dual frame} is the set~$\{ e^\mu \}$ in $\Gamma^\infty(E_{U}^*)$ that satisfies $e^\mu(e_\nu) = \delta^\mu_\nu$, \ie, the $e^\mu$ are the dual basis to~$e_\nu$.
Naturally, given frames~$\{ e_\mu(x) \}$ and~$\{ f_{\nu}(x) \}$ of vector bundles~$E, F$ on~$M$, they induce frames $\{ e_\mu(x) \otimes f_{\nu}(x) \}$ on the tensor product bundle $E \otimes F$; the generalization to the exterior tensor product bundle $E \boxtimes F$ is immediate.

If $E$ is equipped with a metric~$\metricp{\cdot\,}{\cdot}$, we say that the frame~$\{ e_\mu \}$ is \IdxMain{ortho-frame}\emph{orthogonal} ($\KK = \RR$) or \IdxMain{unit-frame}\emph{unitary} ($\KK = \CC$) if it forms an orthonormal basis for each fibre over~$U$, \ie, if
\begin{equation*}
  \metricp{e_\mu}{e_\nu} = \delta_{\mu \nu} \varepsilon_\mu
  \quad \text{with} \quad
  \varepsilon_\mu \defn \metricp{e_\mu}{e_\mu} = \pm 1.
\end{equation*}

Frames allow us to perform calculations in component form, \viz, given a section~$s \in \Gamma(E)$ its \IdxMain{components}\emph{components} in the frame~$\{ e_\mu \}$ on~$U$ are given via the dual frame~$\{ e^\mu \}$ as $s^\mu = s(e^\mu)$ such that $s = s^\mu e_\mu$.
This is the first instance where we used the \IdxMain{sum-convent}\emph{summation convention}: Unless otherwise noted, summation over \emph{balanced} indices (one upper and one lower) is always implied.

If $s$ is the section of a (exterior) tensor product bundle and frames on the single bundles are given, we use multiple indices to denote the sections.
For example, if $g \in \Gamma(T^*\!M \otimes T^*\!M)$, we can write
\begin{equation*}
  g = g_{\mu\nu}\, \dif x^{\mu} \otimes \dif x^{\nu}
\end{equation*}
in terms of the coordinate covectors.

The atlas of a manifold induces natural local frames on the tangent and the cotangent bundle.
If $(U, \varphi)$ is a smooth chart on~$M$ in a neighbourhood of~$x$, then the \IdxMain{coord-vec}\emph{coordinate vectors}\footnote{Often we will use the shorthand~$\partial_\mu$ for~$\partial/\partial x^\mu$.}
\begin{equation*}
  \partial_\mu|_x \defn (T_x \varphi)^{-1} \left( \partial_\mu|_{\varphi(x)} \right)
\end{equation*}
define a basis on~$T_x M$ because $T_x \varphi : T_x M \to T_{\varphi(x)} \RR^n \simeq \RR^n$ is an isomorphism.
Together with the coordinate functions these coordinate vectors at every point induce natural coordinates on~$(TM)_U$.
A local frame for $(T^*\!M)_U$ is then simply the dual frame $\{ \dif x^{\mu} \}$.

Note that the coordinate vector fields~$\{ \partial_\mu \}$ associated to a coordinate chart~$\{ x^\mu \}$ form a very special local frame of the tangent bundle.
Whereas the commutator~$[\partial_\mu, \partial_\nu]$ vanishes, this is no longer true in every frame~$\{ e_\mu \}$, where
\begin{equation*}
  [e_\mu, e_\nu] = c\indices{^\rho_{\!\mu \nu}} e_\rho
\end{equation*}
has in general non-vanishing \IdxMain{comm-coeff}\emph{commutation coefficients}~$c\indices{^\rho_{\!\mu \nu}}$.

\subsection{Differential operators}
\label{sub:df_differential_operators}

Given a vector bundle $E \to M$ of $n$-dimensional smooth manifold $M$, a \IdxMain{diff-op}\emph{linear differential operator of order~$m$} (with smooth coefficients) is a linear map $\mathrm{P} : \Gamma^\infty(E) \to \Gamma^\infty(E)$ which, in local coordinate $\{x^\mu\}$ on $U$, is given by
\begin{equation*}
  \mathrm{P} \restriction_U = \sum_{\abs{\alpha} \leq m} a_\alpha(x) \partial^\alpha,
\end{equation*}
where $\alpha = (\alpha_1, \dotsc, \alpha_n)$ are multi-indices with $\partial^\alpha = \partial^{\alpha_1}_1 \dotsm\, \partial^{\alpha_n}_n$ and the coefficients $a_\alpha : \Gamma^\infty(E) \to \Gamma^\infty(E)$ are linear maps.\footnote{Although we will not explicitly state this, sometimes we will use differential operators with non-smooth coefficients. In that case the coefficients map into $C^k$ sections.}
That is, $\mathrm{P}$ is locally defined as a polynomial in the partial derivatives~$\{\partial_\mu\}$.

The polynomial
\begin{equation*}
  p(x, \xi) = \sum_{\abs{\alpha} \leq m} a_\alpha(x) \xi^\alpha,
\end{equation*}
where $\xi^\alpha = \xi^{\alpha_1}_1 \dotsm\, \xi^{\alpha_n}_n$ and $\xi$ is a covector field with components $\xi = \xi_\mu \dif x^\mu$, is called the \IdxMain{tot-sym}\emph{total symbol} of~$\mathrm{P}$.
The leading term of $p(x, \xi)$,
\begin{equation*}
  \sigma_{\mathrm{P}}(x, \xi) = \sum_{\abs{\alpha} = m} a_\alpha(x) \xi^\alpha,
\end{equation*}
is the \emph{principal part} or \IdxMain{pri-sym}\emph{principal symbol} of~$\mathrm{P}$.
While this is not true for the total symbol, one can check that the principal symbol is covariantly defined as a function on the cotangent bundle: $\sigma_{\mathrm{P}} : \Gamma^\infty(\mathrm{S}^m(T^*\!M) \otimes E) \to \Gamma^\infty(E)$.

Suppose that $(M, g)$ is a pseudo-Riemannian manifold.
If the principal symbol~$\sigma_{\mathrm{P}}$ of a differential operator~$\mathrm{P}$ is given by the metric
\begin{equation*}
  \sigma_{\mathrm{P}}(x, \xi) = - g_x(\xi, \xi) \id_{E_x},
\end{equation*}
we say that $\mathrm{P}$ is \IdxMain{norm-hyp}\emph{normally hyperbolic} or, alternatively, that it is a \IdxMain{wave-op}\emph{wave operator}.
Normally hyperbolic operators on globally hyperbolic spacetime have a well-posed Cauchy problem and therefore play an important role in quantum field theory on curved spacetime.

Given a second differential operator~$\mathrm{Q}$ on the same vector bundle~$E$, the composition $\mathrm{P} \circ \mathrm{Q}$ is also a differential operator and the principal symbol of the composed operator is given by the composition of the principal symbols: $\sigma_{\mathrm{P} \circ \mathrm{Q}} = \sigma_{\mathrm{P}} \circ \sigma_{\mathrm{Q}}$.
We say that $\mathrm{P}$ is \IdxMain{pre-norm-hyp}\emph{pre-normally hyperbolic} if there exists $\mathrm{Q}$ such that $\mathrm{P} \circ \mathrm{Q}$ is normally hyperbolic.
One can show that also $\mathrm{Q} \circ \mathrm{P}$ is normally hyperbolic and thus $\mathrm{Q}$ is pre-normally hyperbolic too \cite{muhlhoff:2011}.

\subsection{Index notation}
\label{sub:dg_index_notation}

We just saw that in a frame $\{ e_\mu \}$ of a vector bundle $E \to M$ we can calculate the components of sections with respect to that frame.
For example, given two vectors fields $v, w \in \Gamma(TM)$ on a pseudo-Riemannian manifold~$(M, g)$, we can write their metric contraction in terms of their components with respect to the coordinate \optword{co}vector fields~$\{ \partial_\mu \}$ and~$\{ \dif x^\mu \}$ in a coordinate neighbourhood:
\begin{equation*}
  g(v, w) = g_{\mu\nu} v^\mu w^\nu.
\end{equation*}
Repeated indices imply summation by the Einstein summation convention as usual.

Usually the frame is not explicitly mentioned but instead implicitly given by a selection of letters for the indices.
Henceforth the small Greek letters $\mu, \nu, \lambda, \rho, \sigma$ will always be indices for a coordinate frame of the \optword{co}tangent bundle in the \IdxMain{concrete-ind}\emph{concrete index notation}.

When calculating contractions between more complicated tensors the notation in terms of indices is often over the abstract index-free notation which quickly becomes unwieldy.
Moreover, if the horizontal position of indices is kept fixed, we can use a metric to lower and raise indices, \eg, returning to our example, we write
\begin{equation*}
  g(v, w) = v_\mu w^\mu.
\end{equation*}
That is, we identify the components of the vector field $v$ with the components of the associated covector field $g(v, \cdot)$.
Contracting a tensor $S \in \Gamma(T^2_1 M)$ with the vector fields $v, w$, we see the advantage of this notation
\begin{equation*}
  S\big(g(v,\cdot), g(w,\cdot)\big)
  = g_{\nu\rho} g_{\smash\lambda\sigma} S\indices{_\mu^{\!\smash\rho\sigma}} v^\mu v^\nu w^\lambda
  = S_{\mu\nu\smash\lambda} v^\mu v^\nu w^\lambda
\end{equation*}
The `natural' position of the indices of the tensor $S$ must, however, be agreed upon beforehand.

Note that the formal aspects of this notation do not necessitate the existence of a frame.
This leads to Penrose's \IdxMain{penrose-ind}\emph{abstract index notation}.
Even in the absence of a concrete frame, we write for example
\begin{equation*}
  S\big(g(v,\cdot), g(w,\cdot)\big)
  = g_{bd} g_{\vphantom{b}ce} S_{\vphantom{b}a}{}^{de} v^a v^b w^c
  = S_{abc} v^a v^b w^c
\end{equation*}
Now, an index only labels a slot in the index-free expression and does not carry any numerical value.
In particular, Einstein summation convention does not apply to abstract indices -- it would not even make sense -- and double indices only imply (metric) contractions.
We will often use the small Latin letters $a, b, c, d, e$ as abstract indices for the \optword{co}tangent bundle.

Both for abstract and concrete index notation it is useful to introduce some shorthands.
Symmetrization and antisymmetrization of tensors are denoted by parentheses and brackets:
\begin{equation*}
  S_{(ab)} \defn \tfrac{1}{2}(S_{ab} + S_{ba}),
  \quad
  S_{[ab]} \defn \tfrac{1}{2}(S_{ab} - S_{ba})
\end{equation*}
with the obvious generalization to higher-order tensors.
Partial and covariant derivatives (see below) are sometimes indicated by comma and semicolon:
\begin{equation*}
  (\cdots)_{,a} \defn \partial_a (\cdots),
  \quad
  (\cdots)_{;a} \defn \nabla_a (\cdots).
\end{equation*}

%-- Connections and curvature -------------------------------------------------%
\section{Connections and curvature}
\label{sec:connections}

A (Koszul) \IdxMain{connection}\emph{connection}~$\nabla^E$ on a $m$-dimensional $\KK$-vector bundle~$E \to M$ is a $\KK$-linear map $\nabla^E : \Gamma^\infty(E) \to \Gamma^\infty(E \otimes T^*\!M)$ that satisfies the Leibniz rule
\begin{equation*}
  \nabla^E(\varphi f) = \varphi\, \nabla^E f + f \otimes \dif \varphi
\end{equation*}
for all $\varphi \in C^\infty(M)$ and $f \in \Gamma^\infty(E)$.
Every vector bundle admits a connection.
Henceforth we will often drop the superscript indicating the vector bundle that the connection acts on and simply denote it by~$\nabla$.
Given a vector field~$v$, the connection $\nabla$ defines the \IdxMain{cov-deriv-vec}\emph{covariant derivative along~$v$} as $\nabla_v : \Gamma^1(E) \to \Gamma(E)$ with
\begin{equation*}
  \nabla_v\,\cdot \defn (\nabla\,\cdot)(v).
\end{equation*}
If, in addition, the vector bundle~$E$ is equipped with a $C^\infty$ bundle metric~$\metricp{\cdot\,}{\cdot}$, we say that $\nabla$ is a \IdxMain{metric-conn}\emph{metric connection} if
\begin{equation*}
  v \metricp{f}{h} = \metricp{\nabla_v f}{h} + \metricp{f}{\nabla_v h}
\end{equation*}
holds for all $f, h \in \Gamma^1(E)$.
A connection~$\nabla$ on~$E = TM$ is \IdxMain{torsion-free-conn}\emph{torsion-free} if the Lie bracket of two vector fields~$v, w$ is given by~$[v, w] = \nabla_v w - \nabla_w v$.

Let $E, F$ be two vector bundles with connections~$\nabla^E, \nabla^F$ and sections $f \in \Gamma^1(E)$, $h \in \Gamma^1(F)$ and $u \in \Gamma^1(E^*)$.
A connection on the the tensor product bundle~$E \otimes F$ is defined by
\begin{equation*}
  \nabla^{E \otimes F}_v (f \otimes h) = (\nabla^E_v f) \otimes h + f \otimes (\nabla^F_v h).
\end{equation*}
Moreover, a \IdxMain{dual-conn}\emph{dual connection} is obtained from
\begin{equation*}
  \big(\nabla^{E^*}_v u\big)(f) = v \big(u(f)\big) - u \big(\nabla^E_v f \big).
\end{equation*}
This can be used to extend a connection~$\nabla$ on a vector bundle~$E$ to its dual bundle~$E^*$ and more generally to the tensor bundle~$T^p_q E$.

If $\psi : M \to N$ is a diffeomorphism between two manifolds and $E \to N$ a vector bundle with a connection $\nabla^E$, then we automatically find a unique connection $\pb{\psi} \nabla$ on the pullback bundle $\pb{\psi} E$.
The \IdxMain{pb-conn}\emph{pullback connection} is given by
\begin{equation*}
  (\pb{\psi} \nabla)_v (\pb{\psi} f) \defn \pb{\psi} \big( \nabla^E_{\dif\psi(v)} f \big)
\end{equation*}
for all $f \in \Gamma^1(E)$ and $v \in \Gamma(TM)$.

\subsection{Levi-Civita connection}
\label{sub:conn_levi_civita}

The \IdxMain{levi-civita}\emph{Levi-Civita connection} is the unique metric connection~$\nabla$ on~$TM$ with~$C^\infty$ metric~$g$ that is torsion free.
Therefore it satisfies the \IdxMain{koszul-formula}\emph{Koszul formula}
\begin{equation}\begin{split}\label{eq:koszul_formula}
  2 g(\nabla_u v, w) & = u g(v, w) + v g(u, w) - w g(u, v) \\&\quad - g(u, [v, w]) - g(v, [u, w]) - g(w, [u, v]).
\end{split}\end{equation}
Usually we will call the covariant derivative associated to the Levi-Civita connection of~$(M, g)$ simply \emph{the} covariant derivative.

In a chart~$(U, \varphi)$ of~$M$, we have
\begin{equation*}
  \nabla_{\partial_\mu} \partial_\nu = \Gamma\indices{^\rho_{\!\mu \nu}} \partial_\rho
\end{equation*}
and we call $\Gamma\indices{^\rho_{\!\mu \nu}}$ the \IdxMain{christoffel}\emph{Christoffel symbols}\footnote{
Note that $\nabla$ is not a tensor and thus the Christoffel symbols do not transform as tensors.} of the Levi-Civita connection~$\nabla$ for the chart~$(U, \varphi)$.
The Christoffel symbols are symmetric in their lower indices $\Gamma\indices{^\rho_{\!\mu \nu}} = \Gamma\indices{^\rho_{\!\nu \mu}}$ because $[\partial_\nu, \partial_\nu] = 0$.
Given two vector fields~$v, w \in \Gamma^1(TM)$, written in the coordinate basis as $v = v^\mu \partial_\mu$ and $w = w^\mu \partial_\mu$, $w$ has covariant derivative
\begin{equation*}
  \nabla_v w = v^\mu \bigg(\! \pd{w^\rho}{x^\mu} + \Gamma\indices{^\rho_{\!\mu \nu}} w^\nu \!\bigg) \partial_\rho
\end{equation*}
along~$v$.
The Kozul formula \eqref{eq:koszul_formula} yields a formula for the Christoffel symbols
\begin{equation*}
  \Gamma\indices{^\rho_{\!\mu \nu}} = \frac{1}{2} g^{\rho\lambda} \bigg(\! \pd{g_{\mu\smash\lambda}}{x^\nu} + \pd{g_{\nu\smash\lambda}}{x^\mu} - \pd{g_{\mu \nu}}{x^\lambda} \!\bigg).
\end{equation*}

\subsection{Curvature of a connection}
\label{sub:conn_curvature}

Different from ordinary second derivatives, covariant second derivatives
\begin{equation*}
  (\nabla^2 f)(v, w) = \nabla_v \nabla_w f - \nabla_{\nabla_v w} f,
\end{equation*}
where~$v, w \in \Gamma^1(TM)$ and $f \in \Gamma^2(E)$, do not commute in general.
The \IdxMain{curvature}\emph{curvature} $F \in \Gamma(T^*\!M \otimes T^*\!M \otimes E \otimes E^*)$ of a vector bundle~$E$ with connection $\nabla$ quantifies this failure of the covariant second derivative to commute and we define it as
\begin{equation*}
  F(v, w) \defn \nabla_v \nabla_w - \nabla_w \nabla_v - \nabla_{[v, w]}.
\end{equation*}
We say that the connection is \IdxMain{flat-conn}\emph{flat} its curvature $F$ vanishes.
If $E = TM$ and $\nabla$ is the Levi-Civita connection of a pseudo-Riemannian manifold $(M, g)$ with smooth metric~$g$, we denote the curvature by~$R(v, w)$ instead and call it \IdxMain{riem-curv}\emph{Riemann curvature (tensor)} of~$(M, g)$.

Let $u, v, w \in \Gamma^2(TM)$ be vector fields and $f \in \Gamma^2(E)$ a section of~$E$.
By definition the curvature is skew-hermitian: $F(v, w) = - F(w, v)$.
If the connection~$\nabla$ is compatible with a metric ${\metricp{\cdot\,}{\cdot}}$ on~$E$, then we also have that~$F$ is skew-adjoint, $\metricp{F(v, w) f}{f} = - \metricp{f}{F(v, w) f}$, and that it satisfies the \IdxMain{fst-bianchi}\emph{first Bianchi identity}
\begin{equation*}
  (\nabla_u F)(v, w) + (\nabla_v F)(w, u) + (\nabla_w F)(u, v) = 0.
\end{equation*}
Furthermore, if $E = TM$ and $\nabla$ the Levi-Civita connection, also the \IdxMain{snd-bianchi}\emph{second Bianchi identity}
\begin{equation*}
  R(u, v) w + R(v, w) u + R(w, u) v = 0
\end{equation*}
holds.

Several other curvature tensors can be derived from the Riemann curvature.
The \IdxMain{ricci-curv}\emph{Ricci curvature (tensor)} is defined as the symmetric $(0, 2)$ tensor $\mathrm{Ric}(v, w) = {\tr (u \to R(v, u)(w))}$, where $u, v, w \in \Gamma^2(TM)$; if it vanishes, we say that $(M, g)$ is \emph{Ricci-flat}.
Furthermore, we obtain the \IdxMain{ricci-scalar}\emph{Ricci scalar} as the contraction of the Ricci tensor $\mathrm{Scal} \defn \tr_g \mathrm{Ric} = \tr \mathrm{Ric}^\sharp$.
Using the coordinate \optword{co}vectors, we can express the Riemann tensor in component form as:
\begin{equation*}
  R\indices{^\sigma_{\!\mu\nu\smash\lambda}} = \partial^{\vphantom{\sigma}}_{\nu} \Gamma\indices{^\sigma_{\!\mu\smash\lambda}} - \partial^{\vphantom{\sigma}}_{\smash\lambda} \Gamma\indices{^\sigma_{\!\mu\nu}} + \Gamma\indices{^\sigma_{\!\rho\smash\lambda}} \Gamma\indices{^{\smash\rho}_{\!\mu\nu}} - \Gamma\indices{^\sigma_{\!\rho\nu}} \Gamma\indices{^{\smash\rho}_{\!\mu\smash\lambda}}.
\end{equation*}
The Ricci tensor and Ricci scalar are then written as $R_{\mu\nu} = -\delta^\lambda_\sigma R\indices{^\sigma_{\!\mu\nu\smash\lambda}}$ and $R = g^{\mu\nu} R_{\mu\nu}$.

\subsection{Geodesics}
\label{sub:conn_geodesics}

Let $\gamma : I \to M$ be a smooth curve parametrized by an interval $I \subset \RR$ and $f \in \Gamma^1(\pb{\gamma} E)$ a section in the pullback of the vector bundle~$E \to M$.
The \IdxMain{cov-deriv-curve}\emph{covariant derivative of~$f$ along~$\gamma$} is given by
\begin{equation*}
  \frac{\nabla f}{\dif t} \defn (\pb{\gamma} \nabla_{\dot\gamma(t)}) f,
\end{equation*}
\ie, it is a special case of a pullback connection.
If the covariant derivative of~$f$ vanishes along~$\gamma$, we say that~$f$ is \emph{parallel} to~$\gamma$.
Therefore connections facilitate the notion of \IdxMain{parallel-trans}\emph{parallel transport} along curves (\cref{fig:parallel_transport}, see also \cref{sub:bi_parallel}).
Namely, given a $s_0 \in E_{\gamma(t_0)}$ at the point~$\gamma(t_0)$, the parallel transport of~$f_0$ along~$\gamma$ is the unique solution~$f$ of the ordinary differential equation $\nabla (\pb{\gamma} f) / \dif t = 0$ with initial condition $(f \circ \gamma)(t_0) = f_0$.

\begin{figure}
  \centering
  \includegraphics{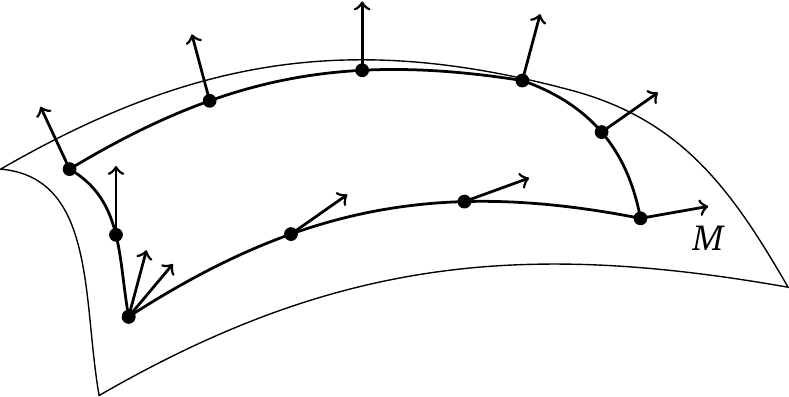}
  \caption{The parallel transport of a vector along a closed path.}
  \label{fig:parallel_transport}
\end{figure}

Auto-parallel curves, \ie, curves that satisfy
\begin{equation*}
  \frac{\nabla}{\dif t}\, \pb{\gamma} \dot{\gamma}(t) = 0,
\end{equation*}
are called \IdxMain{geodesic}\emph{geodesics of the connection $\nabla$}.
These are the usual geodesics (local minimizers of arc length if the metric is Riemannian) with respect to a metric~$g$ of a pseudo-Riemannian manifold~$(M, g)$ if $\nabla$ is the Levi-Civita connection with respect to $g$.

It follows from the theory of ordinary differential equations that, given a point~$x \in M$ and any vector~$v \in T_x M$, there exists a unique geodesic~$\gamma_v$ such that $\gamma_v(0) = x$ and $\dot\gamma_v(0) = v$.
Let $\Upsilon_x$ be the set of vectors~$v$ at~$x$ that give an inextendible geodesic~$\gamma_v$ defined (at least) on the interval~$[0, 1]$.
We say that $M$ is \IdxMain{geodesic-compl}\emph{geodesically complete} if $\Upsilon_x = T_x M$ at every $x \in M$.
Nevertheless, if there is a vector~$v \nin \Upsilon_x$, there exists $\varepsilon \in \RR$ such that $\varepsilon v \in \Upsilon_x$.

\begin{figure}[t]
  \centering
  \includegraphics{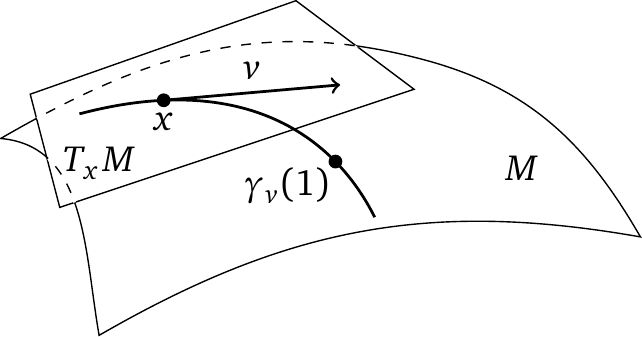}
  \caption{The exponential map applied to a vector.}
  \label{fig:exponential_map}
\end{figure}

The \IdxMain{exp-map}\emph{exponential map} at~$x$ is defined as the map (\cref{fig:exponential_map})
\begin{equation*}
  \exp_x : \Upsilon_x \to M, \quad v \mapsto \gamma_v(1).
\end{equation*}
That is, remembering that geodesics can be linearly reparametrized, the exponential map $\exp_x$ maps vectors at~$x$ to geodesics through~$x$.
For each~$x \in M$ there exists an open neighbourhood~$U' \subset T_x M$ of the origin on which $\exp_x$ is a diffeomorphism into an open neighbourhood~$U \subset M$ of $x$.
If $U'$ is starshaped\footnote{A \emph{starshaped neighbourhood} $S$ of a vector space $V$ is an open neighbourhood~$S \subset V$ of the origin such that $t v \in S$ for all~$t \in [0, 1]$ and~$v \in S$.}, then we say that $U$ is \IdxMain{geodesic-star}\emph{geodesically starshaped} ($U$ is a \emph{normal neighbourhood}) with respect to~$x$.
Moreover, if $U \subset M$ is geodesically starshaped with respect to all of its points, it is called \IdxMain{geodesic-convex}\emph{geodesically convex}.

\subsection{Killing vector fields}
\label{sub:conn_killing}

Given an $n$-dimensional pseudo-Riemannian manifold $(M, g)$, a \IdxMain{killing}\emph{Killing vector field} is a vector field~$\kappa$ such that
\begin{equation*}
  \lieD_\kappa g = 0.
\end{equation*}
In terms of the Levi-Civita connection on $(M, g)$, this equation may also be written as $\nabla_a \kappa_b - \nabla_b \kappa_a = 0$ in the abstract index notation.
More generally, a \IdxMain{killing-conf}\emph{conformal Killing vector field} is a vector field~$\kappa$ such that
\begin{equation*}
  \lieD_\kappa g = \omega g
  \quad \text{with} \quad
  \omega = \frac{2}{n} \tr(\nabla \kappa).
\end{equation*}

An equivalent definition of the Lie derivative of a covariant tensor $S \in \Gamma^1(T^0_q M)$ along a differentiable vector field~$\kappa$ is given by
\begin{equation*}
  \lieD_\kappa S = \lim_{t \to 0} \frac{1}{t} \big( \pb{\psi_t} S - S \big)
\end{equation*}
where $\psi_t$ is the (local) flow of~$\kappa$.
Therefore, a vector field $\kappa$ is (conformal) Killing vector field if and only if the flow that it generates is a family of local (conformal) isometries.
In other words, $\kappa$ encodes a \IdxMain{symm}\emph{(conformal) symmetry} of $(M, g)$.

Now, if $\gamma$ is a geodesic and $\kappa$ a Killing vector field, then it holds that $g(\kappa, \dot\gamma)$ is constant along $\gamma$.
That is, the geodesics of $(M, g)$ correspond to conserved quantities under the symmetry given by~$\kappa$.
If $\kappa$ is a \emph{conformal} Killing vector field, then $g(\kappa, \dot\gamma)$ is constant only if $g(\dot\gamma, \dot\gamma) = 0$.

%-- Differential forms and integration ----------------------------------------%
\section{Differential forms and integration}
\label{sec:differential_forms}

As already noted above, sections of the cotangent bundle~$T^*\!M$ are called \IdxMain{diff-form}\emph{(differential) $1$-forms}.
More generally, sections of the $p$-th antisymmetric tensor bundle~$\bigwedge^p(T^*\!M)$ are called \emph{(differential) $p$-forms}.
The set of all smooth $p$-forms on~$M$ is usually denoted $\Omega^p(M) \defn \Gamma^\infty(\bigwedge^p T^*\!M)$.

\subsection{Exterior derivative}
\label{sub:df_exterior_derivative}

The \IdxMain{ext-deriv}\emph{exterior derivative} $\dif: \Omega^p(M) \to \Omega^{p+1}(M)$ is the unique generalization of the differential of functions such that:
\begin{enumerate}[(a)]
  \item $\dif f$ for $0$-forms (\ie, functions) $f \in \Omega^0(M)$ is the usual differential,
  \item $\dif$ is a $\wedge$-antiderivation, \ie, it satisfies the product rule
    \begin{equation*}
      \dif(\omega \wedge \eta) = \dif \omega \wedge \eta + (-1)^p \omega \wedge \dif \eta,
    \end{equation*}
    where $\omega \in \Omega^p(M)$ and $\eta \in \Omega^l(M)$,
  \item $\dif^2 = \dif \circ \dif = 0$,
  \item $\dif$ commutes with pullbacks.
\end{enumerate}
We say that a form $\omega$ is \IdxMain{form-closed}\emph{closed} if $\dif \omega = 0$ and \IdxMain{form-exact}\emph{exact} if $\omega = \dif \eta$ for some form $\eta$.
While a closed form is in general not exact, the opposite is obviously always true.
The extend to which closed forms fail to be exact is measured by the \IdxMain{de-rham}\emph{de Rham cohomology groups}~$H_{\mathrm{dR}}^p(M)$ of the smooth manifold~$M$
\begin{equation*}
  H_{\mathrm{dR}}^p(M)
  \defn \frac{\{\omega \in \Omega^p(M) \mid \omega \; \text{closed}\}}
             {\{\omega \in \Omega^p(M) \mid \omega \; \text{exact}\}}.
\end{equation*}
Replacing $p$-forms with compactly supported forms in the definition above, we obtain the related notion of the \IdxMain{de-rham-compact}\emph{de Rham cohomology group with compact support}~$H_{\mathrm{dR},0}^p(M)$.

It is not difficult to show that the de Rham cohomology is a homotopy invariant and thus a topological invariant.
This is quite astonishing considering that its definition relies on the smooth structure of the manifold.
Note that, if $M$ is \IdxMain{contractible}\emph{contractible}, \ie, homotopy equivalent to a point, then all its de Rham cohomology groups vanish.

\subsection{Integration}
\label{sub:df_integration}

A smooth $n$-form $\mu \in \Omega^n(M)$ on a smooth $n$-dimensional manifold~$M$ is called a \IdxMain{volume}\emph{volume form} if $\mu(x) \neq 0$ for all $x \in M$.
If such a volume form~$\mu$ exists, we say that $M$ is \IdxMain{orientable}\emph{orientable} because $\mu$ assigns a consistent \IdxMain{orientation}\emph{orientation} to all of~$M$.
Here, we say that a basis $v_1, \dotsc, v_n \in T_x M$ is \IdxMain{orient-pos}\IdxMain{orient-neg}\emph{positively (negatively) oriented} at $x$ with respect to $\omega \in \Omega^n(M)$ if $\omega(x)(v_1 \otimes \dotsc \otimes v_n) \gtrless 0$.
If $M$ is an orientable smooth manifold equipped with a metric tensor~$g$, there exists a unique volume form~$\mu_g$, the \IdxMain{volume-metric}\emph{$g$-volume}, which satisfies
\begin{equation*}
  \mu_g(x)(v_1 \otimes \dotsm \otimes v_n) = \sqrt{\abs{g}} \defn \sqrt{\abs{\det[g_x(v_i, v_j)]}}
\end{equation*}
for a positively oriented basis $v_1, \dotsc, v_n \in T_x M$.

Once integration over~$\RR^n$ is defined, \ie, given a domain~$U \subset \RR^n$
\begin{equation*}
  \int_U \nu = \int_U f\, \dif x^1 \wedge \dotsm \wedge \dif x^n = \int_U f\, \dif x^1 \dotsm \dif x^n
\end{equation*}
for some $f \in C^\infty(\RR^n)$ such that $\nu = f\, \dif x^1 \wedge \dotsm \wedge \dif x^n$, we can use the a partition of unity, local charts and linearity of the integral to extend the notion to general smooth manifold.
That is, if $\omega \in \Omega^n(M)$ is a form of maximal degree on an orientable smooth manifold~$M$ which is compactly supported in the chart~$(U, \varphi)$, then
\begin{equation*}
  \int_M \omega = \pm\! \int_{\varphi(U)} \pf{\varphi} \omega,
\end{equation*}
where the sign depends on the orientation of the chart~$(U, \varphi)$ with respect to $\omega$.
Thus we obtained a method of \IdxMain{int-on-man}\emph{integration on manifolds} that is invariant under orientation-preserving diffeomorphism invariant.

Integrating a metric~$\metricp{\cdot\,}{\cdot}$ on a vector bundle~$E$ over an orientable pseudo\hyp{}Riemannian manifold~$(M, g)$, yields a natural inner product~$\innerp{\cdot\,}{\cdot}$ on the sections of~$E$:
Given $f, h \in \Gamma(E)$, we define
\begin{equation*}
  \innerp{f}{h} \defn \int_M \metricp{f}{h}\, \mu_g,
\end{equation*}
whenever the integral exists.

Arguably the most important result on integration on manifolds is \IdxMain{stokes}\emph{Stokes' theorem} -- a generalization of the fundamental theorem of calculus.
It states that the integral of an exact form $\dif \omega \in \Omega^n(M)$ over a relatively compact open subset $U \subset M$ of a $n$-dimensional oriented manifold $M$ is given by the integral of $\omega$ over the $C^1$-boundary $\partial U$:
\begin{equation*}
  \int_U \dif \omega = \int_{\partial U} \pb{\iota} \omega,
\end{equation*}
where $\iota : \partial U \hookrightarrow M$ denotes the inclusion map.
Note that the classical theorems of Gauss (also called the divergence theorem) and Green are special cases of Stokes' theorem.

\subsection{Hodge star and codifferential}
\label{sub:df_hodge}

Let~$(M, g)$ be an oriented pseudo-Riemannian $n$-dimensional manifold for the remainder of this section.

We can now introduce a smooth bundle homomorphism $\hodge : \bigwedge^p(T^*\!M) \to \bigwedge^{n-p}(T^*\!M)$, the \IdxMain{hodge-star}\emph{Hodge star (operator)}; it is the unique bijection such that
\begin{equation*}
  \omega \wedge \hodge \eta = g(\omega, \eta) \mu_g
\end{equation*}
for all $\omega, \eta \in \Omega^p(M)$.
This implies the properties
\begin{equation*}
  \hodge 1 = \mu_g,
  \quad
  \hodge \mu_g = (-1)^{\Ind(g)},
  \quad
  \hodge \hodge \omega = (-1)^{\Ind(g) + p (n - p)} \omega.
\end{equation*}
It also follows that we can rewrite the inner product $\innerp{\cdot\,}{\cdot}$ between differential forms induced by the metric as
\begin{equation*}
  \innerp{\omega}{\eta} = \int_M \omega \wedge \hodge \eta.
\end{equation*}
for all $\omega, \eta \in \Omega^p(M)$ for which the integral is defined.

The \Idx{formal-adj}formal adjoint of the exterior derivative with respect to this pairing is the \IdxMain{codiff}\emph{codifferential} $\codif : \Omega^{p+1}(M) \to \Omega^p(M)$:
\begin{equation*}
  \innerp{\omega}{\codif \eta} \defn \innerp{\dif \omega}{\eta}
\end{equation*}
or, equivalently,
\begin{equation*}
  \codif \omega \defn (-1)^{n p + 1 + \Ind(g)} \hodge \dif \hodge \omega.
\end{equation*}
The codifferential is \emph{not} a derivation and thus it does not satisfy the Leibniz rule.
Note that for one-forms~$\eta$ the codifferential satisfies $\codif \eta = - \div \eta^\sharp$, \ie, it is equal to minus the divergence of the related vector field.

Analogously to the case of the exterior derivative, a form~$\omega$ is called \IdxMain{diff-coclosed}\emph{coclosed} if $\codif \omega = 0$ and \IdxMain{diff-coexact}\emph{coexact} if $\omega = \codif \eta$ for some form~$\eta$.
As a consequence of the bijectivity of the Hodge star one finds\Idx{de-rham}
\begin{equation*}
  H_{\mathrm{dR}}^{n-p}(M)
  \cong \frac{\{\omega \in \Omega^p(M) \mid \omega \; \text{coclosed}\}}
             {\{\omega \in \Omega^p(M) \mid \omega \; \text{coexact}\}},
\end{equation*}
and, in particular,\Idx{de-rham-compact}
\begin{equation*}
  H_{\mathrm{dR}}^p(M) \cong H_{\mathrm{dR,0}}^{n-p}(M)
\end{equation*}
if $H_{\mathrm{dR}}^p(M)$ is finite-dimensional; the latter relation is a consequence of the \IdxMain{poincare-dual}\emph{Poincarè duality} theorem.

A normally hyperbolic differential operator generalizing the usual d'Alembert operator, the \IdxMain{laplace-de-rham}\emph{Laplace--de Rham operator}, can be obtained by composition of the codifferential and exterior derivative; we define it as
\begin{equation*}
  \Box \defn \dif \circ \codif + \codif \circ \dif.
\end{equation*}
In abstract tensor notation the Laplace--de Rham operator acting on a (smooth) $p$-form $\omega$ is given by the Weitzenböck type formula \cite[Eq. (10.2)]{lichnerowicz:1961}
\begin{equation*}
  (\Box \omega)_{a_1 \dotsm a_p}
  = - \nabla^b \nabla_b \omega_{a_1 \dotsm a_p} + \sum_m R_{a_m b} \omega_{a_1 \dotsm}{}^b{}_{\dotsm a_p} + \sum_{m \neq n} R_{a_m b a_n c} \omega_{a_1 \dotsm}{}^b{}_{\dotsm}{}^c{}_{\dotsm a_p}.
\end{equation*}
Consequently the Laplace--de Rham operator on scalar functions ($0$-forms), which we will also call the \IdxMain{dAlembert}\emph{d'Alembert operator}, is defined as
\begin{equation*}
  \Box \defn \codif \dif = -\nabla^a \nabla_a,
\end{equation*}
which differs from the definition in some publications of the author~\cite{pinamonti:2015a,pinamonti:2015} by a sign.

\subsection{Integral submanifolds}
\label{sub:df_integral_manifolds}

While a vector field can always be locally integrated to give an integral curve, it is not true that to every rank $k>1$ subbundle of the tangent bundle $TM$ there exists a $m$-dimensional submanifold of a smooth manifold~$M$.

Let $D \subset TM$ be a smooth rank-$k$ subbundle, a \IdxMain{plane-field}\emph{plane field}, and $N \subset M$ an immersed submanifold.
$N$ is an \IdxMain{int-man}\emph{integral manifold} of $D$ if
\begin{equation*}
  T_x N = D_x
\end{equation*}
for every $x \in N$ and we say that $D$ is \IdxMain{int-plane-field}\emph{integrable}.

The plane field~$D$ is locally spanned by smooth vector fields $v_1, \dotsc, v_k$, \ie, each local frame around a point $x \in M$ is given by the vector fields $v_i$ such that $v_1 \restriction_x, \dotsc, v_k \restriction_x$ is a basis of~$D_x$.
\IdxMain{frobenius}\emph{Frobenius' theorem} states that $D$ is integrable if and only if $[v_i, v_j]_x \in D_x$ for all $i,j$ at every point~$x$.

A plane field~$D$ is equivalently locally specified in a neighbourhood $U \subset M$ by a collection of covector fields $\omega_1, \dotsc, \omega_{n-k}$ such that
\begin{equation*}
  D_x = \ker \omega_1 \restriction_x \cap \;\dotsb \cap \ker \omega_{n-k} \restriction_x
\end{equation*}
for every $x \in U$.
The dual to Frobenius theorem is then: $D$ is integrable if and only if there exist smooth covector fields $\{ \eta_{i,j} \mid i,j = 1, \dotsc, n-k \}$ such that
\begin{equation*}
  \dif \omega_i = \sum_{j=1}^{n-k} \omega_j \wedge \eta_{i,j},
\end{equation*}
in other words, $D$ is \IdxMain{invol-plane-field}\emph{involutive}.
For example, a smooth covector field~$\omega$ specifies a codimension~$1$ integral manifold of~$M$ if and only if
\begin{equation*}
  \dif \omega \wedge \omega = 0.
\end{equation*}

A \IdxMain{foliation}\emph{foliation of dimension~$k$} of a smooth manifold~$M$ is a collection~$\mathcal{F} = (N_i)$ of disjoint, connected, (non-empty) immersed $k$-dimensional submanifolds of~$M$, the \IdxMain{foliation-leaves}\emph{leaves of the foliation}, such that their union is the entire manifold~$M$ and each point $x \in M$ has a chart $(U, \varphi)$ with coordinates $(x_1, \dotsc, x_n)$ such that for every leaf $N_i$ that intersects $U$ the connected components of the image $\varphi(U \cap N_i)$ are given by the equations $x_{k+1} = \const, \dotsc, x_n = \const$.

If $\mathcal{F}$ is a foliation of~$M$, then the tangent spaces of the leaves form a plane field of~$M$ which is involutive.
Conversely, the \IdxMain{global-frobenius}\emph{global Frobenius theorem} states that the maximal integral manifolds of an involutive plane field~$D$ on~$M$ form a foliation of~$M$.

%-- Bitensors -----------------------------------------------------------------%
\section{Bitensors}
\label{sec:bitensors}
\IdxRanBegin{bitensor}

Let $(M, g)$ be a pseudo-Riemannian $n$-dimensional manifold for the remainder of this section.
Henceforth we will use unprimed $a, b, \dotsc$ and primed $a', b', \dotsc$ abstract indices to distinguish between objects that transform like tensors at $x$ and $x'$ respectively.
For example, a bitensor $T \in \Gamma(T^2_0 M \boxtimes T^0_1 M)$ can be denoted by~$T\indices{^{ab}_{c'}}(x, x')$, where relative position of primed and unprimed indices is arbitrary.
Note that $T\indices{^{ab}_{c'}}$ transforms like a contravariant $2$-tensor at~$x$ and as a covector at~$x'$.
While many of the operations presented in this section generalize to sections of arbitrary exterior tensor bundles, we limit the discussion to bitensors, \ie, sections of $T^p_q M \boxtimes T^r_s M$, in favour of concreteness rather than generality.

Taking covariant derivatives of bitensors, we further notice that derivatives with respect to $x$ and $x'$ commute with each other.
That is, every (sufficiently regular) bitensor~$T(x, x')$ satisfies the identity
\begin{equation*}
  T_{; a b'} = T_{; b' a},
\end{equation*}
where we have suppressed any other indices (we will do the same in the next two paragraphs).
Partial derivatives commute as always.

Often one is interested in the limiting behaviour $x' \to x$.
This limit is called the \emph{coincidence} or \IdxMain{coincidence-lim}\emph{coinciding point limit} and can be understood as a section of a tensor bundle of~$M$ whenever the limit exists and is independent of the path $x' \to x$.
If a unique limit exists, we adopt \IdxMain{synge-bracket}\emph{Synge's bracket notation} \cite[Chap.~II.2]{synge:1960}
\begin{equation*}
  [T](x) \defn \lim_{x' \to x} T(x, x').
\end{equation*}

An important result on coincidence limits of bitensors is \IdxMain{synge-rule}\emph{Synge's rule}:
Let $T(x, x')$ be a bitensor as above, then \cite[Chap.~I.4.2]{poisson:2011}
\begin{equation}\label{eq:synges_rule}
  [T]_{; a} = [T_{; a}] + [T_{; a'}]
  \quad \text{or} \quad
  [T_{; a'}] = [T]_{; a} - [T_{; a}]
\end{equation}
whenever the limits exist and are unique.
The second equality is a useful tool to turn unprimed derivatives into primed ones and vice-versa.

\subsection{Synge's world function}
\label{sub:bi_synge}

In a geodesically convex neighbourhood~$U$ we can define the \IdxMain{geodesic-dist}\emph{geodesic distance} between two points~$x, x' \in U$ as the arc length of the unique geodesic~$\gamma$ joining $x = \gamma(t_0)$ and $x' = \gamma(t_1)$.
It is given by
\begin{equation*}
  d(x, x')
  \defn \int_{t_0}^{t_1} g_{\gamma(t)}(\dot\gamma, \dot\gamma)^{1/2} \, \dif t
  = (t_1 - t_0) g_x(\dot\gamma, \dot\gamma)^{1/2}
\end{equation*}
because the integrand is constant along the geodesic as a consequence of the geodesic equation.
A slightly more useful function is \IdxMain{world-function}\emph{Synge's world function}, introduced in \cite{ruse:1930,ruse:1931,synge:1931}, which gives half the squared geodesic distance between two points (thus it is sometimes also called the \emph{half, squared geodesic distance}).
Namely,
\begin{equation}\label{eq:world_function}
  \sigma(x, x')
  \defn \tfrac{1}{2} d(x, x')^2
  = \tfrac{1}{2} (t_1 - t_0)^2 g_x(\dot\gamma, \dot\gamma).
\end{equation}
In terms of the exponential map the geodesic distance and the world function can be expressed as
\begin{align*}
  d(x, x') & = g_x\big( \exp_x^{-1}(x'), \exp_x^{-1}(x') \big)^{1/2} \\
  \sigma(x, x') & = \tfrac{1}{2} g_x\big( \exp_x^{-1}(x'), \exp_x^{-1}(x') \big).
\end{align*}
Note that geodesic distance~$d$ and Synge's world function~$\sigma$ are invariant under linear reparametrizations of the geodesic~$\gamma$.
Furthermore, they are both examples of bitensors -- in fact, \emph{biscalars}.

For the covariant derivatives of the world function it is common to write
\begin{equation*}
  \sigma_{a_1 \dotsm a_p b_1' \dotsm b_q'} \defn \nabla_{a_1} \, \dotsm \, \nabla_{a_p} \nabla_{b_1'} \, \dotsm \, \nabla_{b_q'} \sigma,
\end{equation*}
\ie, we always omit the semicolon on the left-hand side.
From \eqref{eq:world_function} one can compute
\begin{equation*}
  \sigma_a(x, x') = (t_1 - t_0) g_{a b} \dot\gamma^b
  \quad \text{and} \quad
  \sigma_{a'}(x, x') = (t_0 - t_1) g_{a' b'} \dot\gamma^{b'},
\end{equation*}
where the metric and the tangent vector are evaluated at $x$ and $x'$ respectively.
Consequently we have $[\sigma_a] = [\sigma_{a'}] = 0$.

According to \eqref{eq:world_function}, the norm of these covectors is given by the fundamental relation\IdxMain{trans-synge}
\begin{equation}\label{eq:world_function_tangent}
  \sigma_{a} \sigma^{a} = 2 \sigma = \sigma_{a'} \sigma^{a'}.
\end{equation}
Therefore, $\sigma_{a}$ and $\sigma_{a'}$ are nothing but tangent vectors at $x$ and $x'$ to the geodesic $\gamma$ with length equal to the geodesic distance between these points.
Actually, \eqref{eq:world_function_tangent} together with the initial conditions
\begin{equation}\label{eq:coincidence_sigma}
  [\sigma] = 0
  \quad \text{and} \quad
  [\sigma_{a b}] = [\sigma_{a' b'}] = g_{a b}
\end{equation}
can be taken as the definition of the world function.
Coincidence limits for higher derivatives of~$\sigma$ can be obtained by repeated differentiation of~\eqref{eq:world_function_tangent} in combination with Synge's rule~\eqref{eq:synges_rule}.
In particular one finds \cite{poisson:2011}
\begin{equation}\label{eq:coincidence_sigma3}
  [\sigma_{a b c}] = [\sigma_{a b c'}] = [\sigma_{a b' c'}] = [\sigma_{a' b' c'}] = 0.
\end{equation}

An important biscalar that can be constructed from the world function is the \IdxMain{van-vleck-morette}\emph{van Vleck--Morette determinant} defined as \cite{moretti:1999a}
\begin{equation}\label{eq:vanvleck}
  \Delta(x, x') \defn \sign(\det g_x) (\det g_x \det g_{x'})^{-1/2} \det\big( - \sigma_{a b'}(x, x') \big) > 0.
\end{equation}
The van Vleck--Morette determinant expresses geodesic \optword{de}focusing: $\Delta > 1$ implies that geodesics near~$x$ and~$x'$ undergo focusing whereas $\Delta < 1$ implies that these geodesics undergo defocusing \cite{poisson:2011}.

From~\eqref{eq:vanvleck} one can derive the transport equation\IdxMain{trans-vanvleck}
\begin{equation*}
  n = \sigma^a (\ln \Delta)_{,a} + \sigma^a{}_a = \sigma^{a\mathrlap{'}} (\ln \Delta)_{,a'} + \sigma^{a\mathrlap{'}}{}_{a'}.
\end{equation*}
This, together with the initial condition $[\Delta] = 1$, can also be used as an alternative definition of the van Vleck--Morette determinant.

\subsection{Parallel propagator}
\label{sub:bi_parallel}

Another important biscalar is the \IdxMain{parallel-prop}\emph{parallel propagator}.
It does exactly what its name suggests: in a geodesically convex neighbourhood~$U$ it transports a vector of a fibre at a point~$x' \in U$ to a vector at~$x \in U$ along the unique geodesic~$\gamma$ joining the two points.
We will denote it by the same symbol as the metric.
Here we define it only for the tangent bundle, where we write
\begin{equation*}
  v^{a} \defn g\indices{^{a}_{a'}} v^{a'}
\end{equation*}
with $v^a \in T_x M$ and $v^{a'} \in T_{x'} M$.

The parallel propagator satisfies the transport equation\IdxMain{trans-parallel}
\begin{equation*}
  g\indices{^{a}_{b' ; c}} \sigma^{c} = 0 = g\indices{^{a}_{b' ; c'}} \sigma^{c'}
\end{equation*}
because $\sigma^{a}$ and $\sigma^{a'}$ are tangent to the geodesic~$\gamma$ at~$x$ and~$x'$ respectively.
Together with the initial condition
\begin{equation}\label{eq:coincidence_parallel}
  [g\indices{^{a}_{b'}}] = g\indices{^a_b} = \delta^a_b
\end{equation}
this equation may be taken as the definition of the parallel propagator.
With the parallel propagator we can raise and lower the \optword{un}primed index in~$g\indices{^{a}_{b'}}$ with the usual metric and use it to transport \optword{un}primed indices of tensors to their opposite.

\subsection{Covariant expansion}
\label{sub:bi_covariant_expansion}
\IdxMain{covar-exp}

It is possible to generalize Taylor's series expansion method to bitensors (and thereby also to tensors).
The covariant expansion was originally developed for ordinary tensors \cite{ruse:1930,ruse:1931} but can be easily extended to bitensors, see for example \cite[Chap. I.6]{poisson:2011}.

While Taylor's method is used to expand a function~$f(x')$ around a point~$x$ in terms of powers of the distance~$x'-x$ with series coefficients given by the derivatives of~$f$ at~$x$, the covariant expansion method replaces functions with bitensors, distance with geodesic distance as given by~$\sigma_a(x, x')$ and ordinary differentiation with the covariant derivative.
That is, given a bitensor~$T_{\underline{a}}(x, x')$, where $\underline{a} = a_1 \dotsm a_m$ is a list of unprimed indices, we perform the expansion
\begin{equation}\label{eq:covariant_expansion}
  T_{\underline{a}}(x, x') = \sum_{k=0}^\infty \frac{(-1)^k}{k!} t_{\underline{a} b_1 \dotsm b_k}(x) \sigma^{b_1} \dotsm \sigma^{b_k}.
\end{equation}
We can then solve for the expansion coefficients by repeated covariant differentiation and taking the coinciding point limit.
Namely, it follows from~\eqref{eq:covariant_expansion} that
\begin{equation*}
  t_{\underline{a}} = [T_{\underline{a}}], \quad
  t_{\underline{a} b_1} = [T_{\underline{a} ; b_1}] - t_{\underline{a} ; b_1}, \quad
  t_{\underline{a} b_1 b_2} = [T_{\underline{a} ; b_1 b_2}] - t_{\underline{a} ; b_1 b_2} - t_{\underline{a} b_1 ; b_2} - t_{\underline{a} b_2 ; b_1},
\end{equation*}
\etc\ If $T$ is a bitensor that also has primed indices, one first has to transport the primed indices to unprimed ones using the parallel propagator~$g\indices{^a_{b'}}$.

Of course, just like any other Taylor expansion, the covariant expansion of a smooth bitensor does in general not converge and if it converges it is not guaranteed to converge to the bitensor that is being expanded.
Only if the bitensor and the metric are analytic, we will not encounter these difficulties.
However, even then the covariant expansion can be a very useful \emph{asymptotic} expansion and we usually truncate it after a finite number of terms.
Therefore, in following we will work with the covariant series in a formal way and whenever we write an infinite sum we implicitly mean that the series is to be truncated and a finite remainder term is to be added.

\subsection{Semi-recursive Avramidi method}
\label{sub:bi_avramidi}
\IdxMain{avramidi}

As described, \eg, in~\cite{decanini:2006,ottewill:2011}, the `naïve' recursive approach to calculate the expansion coefficients of a covariant expansion, as briefly sketched above, is inefficient and does not scale well to higher orders in the expansion because the calculation of coincidence limits becomes computationally prohibitive.
An alternative, non-recursive and elegant method for the calculation of these coefficient was proposed by Avramidi~\cite{avramidi:1986,avramidi:2000}.
From the computational perspective, however, also this approach is not necessarily optimal as it does not always make good use of intermediate result leading to an algorithm which is space but not time efficient.
A middle way, that we will present here, was implemented in~\cite{ottewill:2011} using a `semi-recursive' method.

Avramidi's approach rests on the power series solution approach to solving differential equations.
Therefore it can only be applied where the bitensor solves a differential equation, the \IdxMain{trans-eq}\emph{transport equations}.

With the \IdxMain{trans-op}\emph{transport operators}
\begin{equation}\label{eq:transport_operators}
  \nabla_\sigma^{\vphantom{'}} \defn \sigma^a \nabla_a
  \quad \text{and} \quad
  \nabla_\sigma' \defn \sigma^{a'} \nabla_{a'}.
\end{equation}
the world function~\eqref{eq:world_function_tangent} can be expressed as $(\nabla_\sigma^{\vphantom{'}} - 2) \sigma = 0$ or $(\nabla_\sigma' - 2) \sigma = 0$.
One can construct additional transport equations by differentiating these equations and commuting covariant derivatives at the cost of introducing curvature tensors.
In particular we find \cite{ottewill:2011}
\begin{subequations}\label{eq:avramidi_transport}
  \begin{align}
    \nabla_\sigma' \xi\indices{^{a\mathrlap{'}}_{b'}} & = \xi\indices{^{a\mathrlap{'}}_{b'}} - \xi\indices{^{a\mathrlap{'}}_{c'}} \xi\indices{^{c\mathrlap{'}}_{b'}} - R\indices{^{a\mathrlap{'}}_{c' b' d'}} \sigma^{c'} \sigma^{d'},
    \\
    \nabla_\sigma' \Delta^{1/2} & = \tfrac{1}{2} \Delta^{1/2} \big( n - \xi\indices{^{a\mathrlap{'}}_{a'}} \big),
  \end{align}
\end{subequations}
where we have defined $\xi\indices{^{a\mathrlap'}_{b'}} \defn \sigma\indices{^{a\mathrlap'}_{b'}}$ to avoid confusion later on.

In the next step we take the formal covariant expansion of the bitensor and use the transport equation to find relations between the expansion coefficients.
Let us again take a bitensor~$T_{\underline{a}}(x, x')$ with the expansion~\eqref{eq:covariant_expansion}.
Applying to it the transport operator~$\nabla_\sigma'$, we obtain formally
\begin{equation*}\begin{split}
  \nabla_\sigma' T_{\underline{a}}(x, x')
  & = \sum_{m=0}^\infty \frac{(-1)^m}{m!} \nabla_\sigma' t_{\underline{a} b_1 \dotsm b_m}(x) \sigma^{b_1} \dotsm \sigma^{b_m} \\
  & = \sum_{m=1}^\infty \frac{(-1)^m}{m!} m t_{\underline{a} b_1 \dotsm b_m}(x) \sigma^{b_1} \dotsm \sigma^{b_m},
\end{split}\end{equation*}
\ie, applying $\nabla_\sigma'$ to the $m$-th term is equivalent to multiplying it by~$m$:
\begin{equation*}
  (\nabla_\sigma' T)_{(m)} = k T_{(m)}.
\end{equation*}
Contracting two bitensors~$S_{\underline{a}}(x, x')$ and~$T_{\underline{b}}(x, x')$, we further find
\begin{equation*}
  U_{\underline{a} \underline{b}} = S_{\underline{a} c \vphantom{b}} T^c{}_{\!\underline{b}} = \sum_{m=0}^\infty \frac{(-1)^m}{m!} \sum_{k=0}^k \binom{m}{k}\, s_{\underline{a} c d_1 \dotsm d_k} t^c{}_{\!\underline{b} d_{k+1} \dotsm d_m} \sigma^{d_1} \dotsm \sigma^{d_m},
\end{equation*}
\ie, the $m$-th expansion coefficient of~$U$ is obtained from the lower order coefficients of~$S$ and~$T$:
\begin{equation*}
  U_{\underline{a} \underline{b} \,(m)} = \sum_{k=0}^m \binom{m}{k}\, S_{\underline{a} c \,(k)} T^c{}_{\!\underline{b} \,(m-k)}.
\end{equation*}
Moreover, we follow Avramidi and introduce
\begin{equation*}
  \mathcal{K}\indices{^a_{b \,(m)}} \defn R\indices{^a_{\!(c_1 | b | c_2 ; c_3 \dotsm c_m)}} \sigma^{c_1} \dotsm \sigma^{c_m}
\end{equation*}
so that we can write the covariant expansion of $R\indices{^{a\mathrlap{'}}_{c' b' d'}} \sigma^{c\mathrlap{'}} \sigma^{d\mathrlap{'}}$ as
\begin{equation*}
  g\indices{^a_{a'}} g\indices{_b^{b'}} R\indices{^{a\mathrlap{'}}_{c' b' d'}} \sigma^{c\mathrlap{'}} \sigma^{d\mathrlap{'}}
  = \sum_{m=2}^\infty \frac{(-1)^m}{(m-2)!} \mathcal{K}\indices{^a_{b \,(m)}}.
\end{equation*}

These formal manipulations and definitions can now be applied to find the covariant expansion coefficients of a bitensor.
Here we will apply the method to~$\xi\indices{^{a\mathrlap{'}}_{b'}}$ and~$\Delta^{1/2}$ with their transport equations \eqref{eq:avramidi_transport}.
It follows from~\eqref{eq:coincidence_sigma} and~\eqref{eq:coincidence_sigma3} that $\xi\indices{^a_{b \,(0)}} \!= \delta^a_b$ and $\xi\indices{^a_{b \,(1)}} \!= 0$.
Then, for $m \geq 2$, we can easily find the relation
\begin{equation*}
  -(m+1) \xi\indices{^a_{b \,(0)}} \!= \sum_{k=2}^{m-2} \binom{m}{k}\, \xi\indices{^a_{c \,(k)}} \xi\indices{^c_{b \,(m-k)}} + m (m-1) \mathcal{K}_{(m)}.
\end{equation*}
For the coefficients of the square-root of the van Vleck--Morette determinant one gets $\Delta^{1/2}_{(0)} = 1$ from~\eqref{eq:coincidence_sigma} and~\eqref{eq:coincidence_parallel} and
\begin{equation*}
  \Delta^{1/2}_{(m)} = - \frac{1}{2m} \sum_{k=0}^{m-2} \binom{m}{k}\, \Delta^{1/2}_{(k)} \xi\indices{^a_{a \,(m-k)}}
\end{equation*}
for $m > 0$.
Equivalent relations were found in~\cite{ottewill:2011} and implemented in the Mathematica\texttrademark\ package \emph{CovariantSeries}.

\subsection{Coordinate expansion of the world function}
\label{sub:bi_coordinate_synge}

Our aim in this section is to find an expansion of Synge's world function~$\sigma$ in a coordinate neighbourhood.
There are different ways to achieve this (see for example \cite{eltzner:2011} for an approach using Riemannian normal coordinates).
Here we use a (formal) power series Ansatz and write
\begin{equation*}
  \sigma(x, x') = \sum_{m=0}^\infty \frac{1}{m!} \varsigma_{\mu_1 \dotsm \mu_m}(x) \delta x^{\mu_1} \dotsm \delta x^{\mu_m},
\end{equation*}
where we denote by $\delta x^\mu = (x' - x)^\mu$ the coordinate separation of the points $x, x'$ in a chart.

The transport equation~\eqref{eq:world_function_tangent} can then be applied to obtain relations between the coefficients $\varsigma_{\mu_1 \dotsm \mu_m} = \varsigma_{(\mu_1 \dotsm \mu_m)}$.
As consequence of~\eqref{eq:coincidence_sigma} and $[\sigma_\mu] = 0$, the first three coefficients are
\begin{equation*}
  \varsigma = 0,
  \quad
  \varsigma_{\mu} = 0
  \quad \text{and} \quad
  \varsigma_{\mu \nu} = g_{\mu \nu}.
\end{equation*}
This can be used to derive the relation
\begin{align}
  2(1-m) \varsigma_{\mu_1 \dotsm \mu_m} & = \sum_{k=2}^{m-2} \binom{m}{k}\, g^{\nu \rho} (\varsigma_{(\mu_1 \dotsm \mu_k|, \nu} - \varsigma_{(\mu_1 \dotsm \mu_k| \nu}) (\varsigma_{\mu_{k+1} \dotsm \mu_m), \rho} - \varsigma_{\mu_{k+1} \dotsm \mu_m) \rho}) \notag\\&\quad\label{eq:coordinate_sigma} - 2m \varsigma_{(\mu_1 \dotsm \mu_{m-1}, \mu_m)}
\end{align}
for $m > 2$ after a cumbersome but straightforward calculation.

It is not difficult to implement \eqref{eq:coordinate_sigma} efficiently in a modern computer algebra system; it involves multiplication, transposition and symmetrization of multidimensional arrays and partial derivatives on the components of these arrays.
Making use of the symmetry of the coefficients, can reduce computation time and memory usage for higher order coefficients, especially for `complicated' metrics, significantly.
Even if these symmetries are not used, this method appears to be more efficient than the method of~\cite{eltzner:2011}, as it does not require to compute Riemann normal coordinates first.
In fact, because the Riemann normal coordinates are given by the derivative of the world function, the coefficients $\varsigma_{\mu_1 \dotsm \mu_m}$ can also be used to compute the expansion of Riemann normal coordinates:
\begin{equation*}
  \sigma_\mu(x, x') = g_{\mu \nu} \delta x^\nu + \tfrac{1}{2} (\varsigma_{\nu \rho, \mu} + \varsigma_{\mu \nu \rho}) \delta x^\nu \delta x^\rho + \tfrac{1}{3!} (\varsigma_{\nu \rho \lambda, \mu} + \varsigma_{\mu \nu \rho \lambda}) \delta x^\nu \delta x^\rho \delta x^\lambda + \dotsb
\end{equation*}
Again, this method appears more direct and faster than the one described in~\cite{brewin:2009}.

For example, for the Friedmann\hyp{}Lemaître\hyp{}Robertson\hyp{}Walker metric in cosmological time (see \cref{sub:cos_flrw} for more details),\IdxMain{world-fun-flrw}
\begin{equation*}
  g = - \dif t \otimes \dif t + a(t)^2 \delta_{ij}\, \dif x^i \otimes \dif x^j,\end{equation*}
one finds with $H = \dot{a} / a$
\begin{equation*}\begin{split}
  2 \sigma(x, x') & = -\delta t^2 + a^2 \delta\vec{x}^2 \big( 1 + H \delta t + \tfrac{1}{3} (H^2 + \dot{H}) \delta t^2 + \tfrac{1}{12} a^2 H^2 \delta\vec{x}^2 \\&\quad + \tfrac{1}{12} (2 H \dot{H} + \ddot{H}) \delta t^3 + \tfrac{1}{12} a^2 H (2 H^2 + \dot{H}) \delta t\, \delta\vec{x}^2 \\&\quad + \tfrac{1}{180} (-4 H^4 - 8 H^2 \dot{H} + 2 \dot{H}^2 + 6 H \ddot{H} + 3 \dddot{H}) \delta t^4 \\&\quad + \tfrac{1}{360} a^2 (48 H^4 + 74 H^2 \dot{H} + 8 \dot{H}^2 + 9 H \ddot{H}) \delta t^2 \delta\vec{x}^2 \\&\quad + \tfrac{1}{360} a^4 H^2 (4 H^2 + 3 \dot{H}) \big) \delta\vec{x}^4 + \bigO(\delta x^7)
\end{split}\end{equation*}
in agreement with \cite{eltzner:2011}.
Above we denoted by a dot derivatives with respect to cosmological time~$t$ and defined the coordinate separation $\delta t = (x' - x)^0$ and $\delta\vec{x}^2 = \delta_{ij} (x' - x)^i (x' - x)^j$.
Transforming to conformal time, so that
\begin{equation*}
  g = a(\tau)^2 \big( - \dif\tau \otimes \dif\tau + \delta_{ij}\, \dif x^i \otimes \dif x^j \big),
\end{equation*}
one finds with $\mathcal{H} = a' / a$
\begin{equation*}\begin{split}
  2 a^{-2} \sigma(x, x') & = - \delta\tau^2 + \delta\vec{x}^2 - \mathcal{H} \delta\tau^3 + \mathcal{H} \delta\tau \delta\vec{x}^2 - \tfrac{1}{12} (7 \mathcal{H}^2 + 4 \mathcal{H}') \delta\tau^4 \\&\quad + \tfrac{1}{6} (3 \mathcal{H}^2 + 2 \mathcal{H}') \delta\tau^2 \delta\vec{x}^2 + \tfrac{1}{12} \mathcal{H}^2 \delta\vec{x}^4 - \tfrac{1}{12} (3 \mathcal{H}^3 + 5 \mathcal{H} \mathcal{H}' + \mathcal{H}'') \delta\tau^5 \\&\quad + \tfrac{1}{12} (2 \mathcal{H}^3 + 4 \mathcal{H} \mathcal{H}' + \mathcal{H}'') \delta\tau^3 \delta\vec{x}^2 + \tfrac{1}{12} \mathcal{H} (\mathcal{H}^2 + \mathcal{H}') \delta\tau \delta\vec{x}^4 \\&\quad - \tfrac{1}{360} (31 \mathcal{H}^4 + 101 \mathcal{H}^2 \mathcal{H}' + 28 \mathcal{H}'{}^2 + 39 \mathcal{H} \mathcal{H}'' + 6 \mathcal{H}''') \delta\tau^6 \\&\quad + \tfrac{1}{360} (15 \mathcal{H}^4 + 61 \mathcal{H}^2 \mathcal{H}' + 20 \mathcal{H}'{}^2 + 30 \mathcal{H} \mathcal{H}'' + 6 \mathcal{H}''') \delta\tau^4 \delta\vec{x}^2 + \\&\quad + \tfrac{1}{360} (15 \mathcal{H}^4 + 37 \mathcal{H}^2 \mathcal{H}' + 8 \mathcal{H}'{}^2 + 9 \mathcal{H} \mathcal{H}'') \delta\tau^2 \delta\vec{x}^4 \\&\quad + \tfrac{1}{360} (\mathcal{H}^4 + 3 \mathcal{H}^2 \mathcal{H}') \delta\vec{x}^6 + \bigO(\delta x^7),
\end{split}\end{equation*}
where we denoted by a prime derivatives with respect to conformal time~$\tau$ and defined $\delta\tau = (x' - x)^0$ and $\delta\vec{x}$ as above.

The method described above is sufficiently fast to be applied to more complicated spacetimes such as Schwarzschild or Kerr spacetimes.
Nevertheless, for such spacetimes the expansions become so long, that they easily fill more than a page if printed up to sixth order.
Therefore we will omit them here and suggest that the interested reader implements \eqref{eq:coordinate_sigma} in a computer algebra system.

\IdxRanEnd{bitensor}
\IdxRanEnd{diff-geo}

%% file: lorentzian_geometry.tex
%!TEX root = master.tex

\chapter{Lorentzian geometry}
\label{cha:lorentzian_geometry}

\IdxRanBegin{lor-geo}

\section*{Summary}

Around the turn of the 20th century it became clear that several physical phenomena could not be satisfactorily described in the framework of Newtonian physics.
After insights by Lorentz, Poincaré and others, a solution to some of these problems was found and published by Einstein in 1905 and coined special relativity.
In its modern form the changes introduced by Einstein can be seen to arise from to the union of space and time in a single spacetime.
In the absence of gravity or in the weak gravitational limit our universe appears to be well described by Minkowski spacetime, its causal structure and its symmetries.

The changes introduced by Minkwoski spacetime are generalized by the concept of Lorentzian spacetimes.
This class of pseudo-Riemannian manifolds forms the basis of general relativity, introduced by Einstein in 1916, which takes the lessons learned from special relativity to formulate a geometric theory of gravitation.
In Einstein's theory of gravitation the geometrical background of the Universe is itself dynamical and accounts for gravity via Einstein's equation.
The total lack of static background with respect to which length and time measurements can be made, makes general relativity significantly more difficult than Newton's classical theory of gravitation.

Nevertheless, general relativity has made many predictions like gravitational redshift, gravitational lensing and, in a sense, the Big Bang, that have been confirmed with modern observations.
It also forms the basis of modern cosmology which attempts to describe the evolution of the Universe on large scales from the time of the Big Bang until present times and predict its `fate'.

This chapter will give a short introduction to the causal structure of Lorentzian spacetimes, to general relativity and to cosmology in three separate sections.
The aim of the first section (\cref{sec:causality}) is the definition of globally hyperbolic spacetimes and Cauchy surfaces.
Therefore it will define notions such as spacetimes, time functions, the causal past and future of a set, causality conditions, and the splitting of a spacetime into its temporal and spatial part.
In the second section (\cref{sec:general_relativity}) we give an overview of general relativity including the stress-energy tensor, the classical energy conditions, Einstein's equation and its fluid formulation, and the special case of de Sitter spacetime. %TODO:, and asymptotically flat spacetimes.
Topics related to theoretical cosmology will be addressed in the third section (\cref{sec:cosmology}), where we discuss Friedmann--Lemaître--Robertson--Walker spacetimes and perturbations around it.

As we are presenting standard material, we will again mostly refrain from giving proofs or even sketching them.
Excellent references for this chapter are the books \cite{oneill:1983,wald:1984,beem:1996,ellis:2012}:
While \cite{oneill:1983} and in particular \cite{beem:1996} contain detailed discussions of causal structure of Lorentzian spacetimes, \cite{wald:1984} gives an good overall account of general relativity.
The recent book \cite{ellis:2012} is, as the title already suggests, mostly focused on cosmology, and may be taken as a reference for the last section.

%-- Causality -----------------------------------------------------------------%
\section{Causality}
\label{sec:causality}

\IdxRanBegin{causality}

\subsection{Causal structure}
\label{sub:cs_structure}

\IdxRanBegin{caus-struc}

The metric~$g$ of an arbitrary Lorentzian manifold $(M,g)$ distinguishes three regions in the tangent space: a nonzero $v \in T_x M$ is\IdxMain{spacelike}\IdxMain{lightlike}\IdxMain{timelike}\IdxMain{causal}
\begin{align*}
  \text{\emph{spacelike}} \quad \text{if} \quad & g_x(v, v) > 0, \\
  \text{\emph{lightlike} or \emph{null}} \quad \text{if} \quad & g_x(v, v) = 0, \\
  \text{\emph{timelike}}  \quad \text{if} \quad & g_x(v, v) < 0, \\
  \text{\emph{causal}}    \quad \text{if} \quad & g_x(v, v) \leq 0.
\end{align*}
Furthermore, we define the zero vector $v=0$ as spacelike, as \eg\ in \cite{oneill:1983}.

Given a \IdxMain{time-orient}\emph{time-orientation}~$u$ on~$M$, \viz, a smooth unit timelike vector field such that $g(u, u) = -1$, a causal vector $v \in T_x M$ is
\begin{align*}
 \text{\emph{future-directed}} \quad \text{if} \quad & g_x\big(u(x), v\big) < 0, \\
 \text{\emph{past-directed}}   \quad \text{if} \quad & g_x\big(u(x), v\big) > 0.
\end{align*}
One often chooses the time-orientation to be the velocity of a physical fluid.
Not every Lorentzian manifold can be given a time-orientation and we say that a Lorentzian manifold is \emph{time-orientable} is such a vector field exists.

All these notions extend naturally to the cotangent bundle~$T^*\!M$, sections of both $TM$ and $T^*\!M$ and curves.
That is, as seen in \cref{fig:lightcone}, to each point $x \in M$ we can draw a double cone inside~$M$ whose surface is generated by the lightlike curves passing through~$x$; points lying inside the double cone are lightlike to~$x$, while points lying outside are spacelike to~$x$.

A continuous function $t : M \to \RR$ that is strictly increasing along every future-directed causal curve is called a \IdxMain{time-fun}\emph{time function}.
Clearly, if $t$ is a time function, then its gradient yields a time-orientation which is orthogonal to the level surfaces
\begin{equation*}
  t^{-1}(c) = \{ x \in M \mid f(x) = c \}, \quad c \in \RR.
\end{equation*}
These level surfaces give a foliation of the manifold such that each leaf is spacelike and intersected at most once by every causal curve.

\begin{figure}
  \centering
  \includegraphics{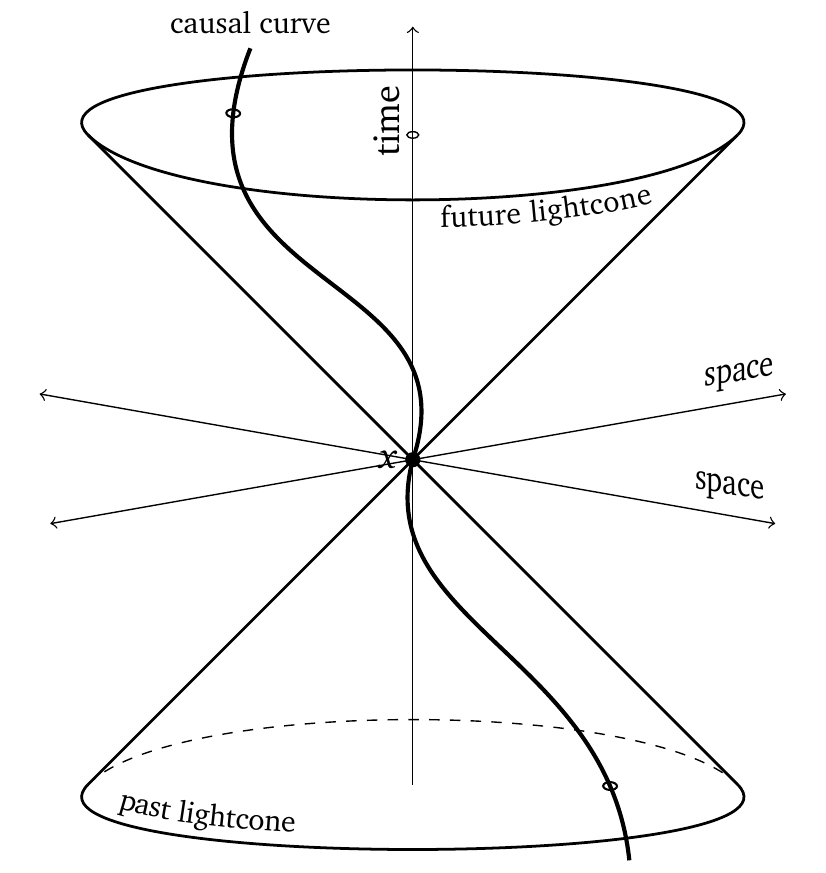}
  \caption{Lightcone of a point in a three-dimensional Minkowski spacetime.}
  \label{fig:lightcone}
\end{figure}

The \IdxMain{chrono-fut}\emph{chronological future}~$I^+(U)$ (\IdxMain{chrono-past}\emph{chronological past}~$I^-(U)$) of a subset $U \subset M$ is defined as the set of points which can be reached from~$U$ by future-directed (past-directed) timelike curves.
Similarly, the \IdxMain{causal-fut}\emph{causal future}~$J^+(U)$ (\IdxMain{causal-past}\emph{causal past}~$J^-(U)$) of a subset $U \subset M$ is defined as the set of points which can be reached from~$U$ by future-directed (past-directed) causal curves; their union $J(U) \defn J^+(U) \cup J^-(U)$ is the \IdxMain{causal-shadow}\emph{causal shadow} of~$U$.
If the set~$U$ consists of only one point $U = \{x\}$, we write $I^\pm(x)$ and $J^\pm(x)$.
Note that $I^\pm(x)$ is always open, whereas $J^\pm(x)$ is not necessarily closed.\footnote{In globally hyperbolic spacetimes (to be defined below) $J^\pm(x)$ is always closed.}
We say that two subsets~$U$ and~$V$ are \IdxMain{causal-sep}\emph{causally separated} (in symbols: $U \mathbin{\raisebox{1pt}{$\bigtimes$}} V$) if $U \cap J(V) = \emptyset$.
These definitions are illustrated in \cref{fig:causality}.
\begin{figure}
  \centering
  \includegraphics{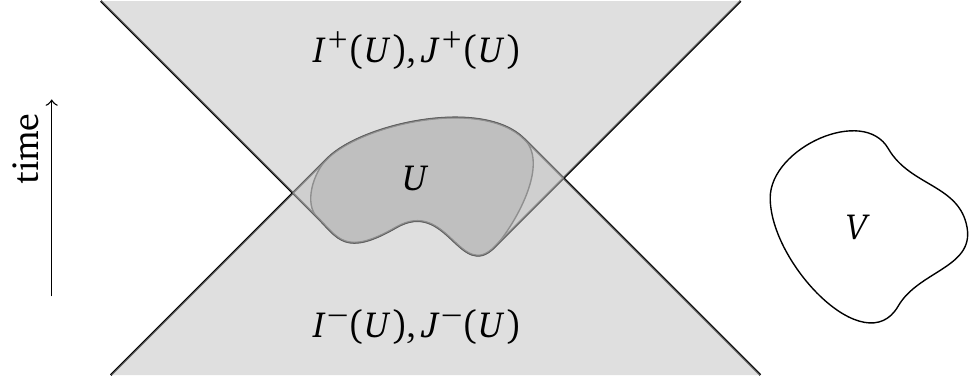}
  \caption{Causal past and future of a set and a causally separated set.}
  \label{fig:causality}
\end{figure}

A \IdxMain{worldline}\emph{world line} or \emph{observer} is a timelike future-directed curve $\tau \mapsto \gamma(\tau)$ such that $g(\dot\gamma, \dot\gamma) = -1$ and the curve parameter~$\tau$ is called the \IdxMain{proper-time}\emph{proper time} of~$\gamma$.
The observer $\gamma$ is said to be freely falling, \ie, moving only under the influence of gravity, if $\gamma$ is a geodesic.

The causal structure of a Lorentzian manifold is physically significant because it is used to imply that two \IdxMain{event}\emph{events}, \ie, two points on the manifold, can only influence each other if one is in the lightcone of the other.
In other words, the Lorentzian causal structure encodes the finiteness of signalling speeds.
This is distinctively different from Newtonian physics, where one event at an absolute time $t_0$ influences all other events at a later time $t > t_0$, no matter the spatial separation.

\IdxRanEnd{caus-struc}

\subsection{Covariant splitting}

\IdxRanBegin{cov-split}

Let $(M,g)$ be a $(1+n)$-dimensional Lorentzian manifold with a time\hyp{}orientation~$u$.
Observe that the integral curves of~$u$ can be understood to define a preferred direction of motion.
It is clear that the integral curves of~$u$ can be parametrized such that they are world lines.
A local frame such that the components of $u$ are
\begin{equation*}
  u^\mu = (1, 0, \dotsc, 0)
\end{equation*}
is called a \IdxMain{comov-frame}\emph{comoving frame}.

Orthogonal to~$u$ at each point are the \IdxMain{rest-space}\emph{rest spaces} of the associated observer.
An induced Riemannian metric tensor for these $n$-spaces is given by the projected tensor
\begin{equation*}
  h \defn g + u^\flat \otimes u^\flat,
\end{equation*}
which has the following properties: $h_{ab} u^b = 0$, $h\indices{_a^c} h\indices{_c^b} = h\indices{_a^b}$, $h\indices{_a^a} = n$.

By means of the projection~$h$ and the time-orientation~$u$, we can decompose any tensor into its temporal and spatial parts with respect to the observer.
In particular, we can define the \IdxMain{temp-cov-deriv}\emph{temporal} and \IdxMain{spat-cov-deriv}\emph{spatial covariant derivatives} of any tensor~$S\indices{^{a\cdots}_{b\cdots}}$ by
\begin{align*}
  \dot S\indices{^{a\cdots}_{b\cdots}} & \defn u^c \nabla_c S\indices{^{a\cdots}_{b\cdots}}, \\
  \wbar{\nabla}_c S\indices{^{a\cdots}_{b\cdots}} & \defn h\indices{^a_p} \dotsm h\indices{_b^q} \dotsm  h\indices{_c^{\smash r}} \nabla_r S\indices{^{p\cdots}_{q\cdots}}.
\end{align*}
Moreover, define the \IdxMain{pstf}\emph{projected symmetric trace-free (PSTF) parts} of a vector field~$v^a$ and a tensor~$S_{ab}$ by
\begin{equation*}
  v^{\langle a \rangle} \defn h\indices{^a_b} v^b,
  \quad
  S_{\langle ab \rangle} \defn \big( h\indices{_{(a}^c} h\indices{_{b)}^d} - \tfrac{1}{n} h_{ab} h^{cd} \big) S_{cd}.
\end{equation*}

The kinematics of the world lines given by~$u$ is characterized by $\nabla_b u_a$ which can be decomposed as
\begin{equation}\label{eq:worldline_kinematics}
  \nabla_b u_a
  = \wbar{\nabla}_b u_a - \dot u_a u_b
  = \omega_{ab} + \Theta_{ab} - \dot u_a u_b
  = \omega_{ab} + \sigma_{ab} + \tfrac{1}{n} \Theta h_{ab} - \dot u_a u_b,
\end{equation}
with
\begin{equation*}
  \omega_{ab} \defn \wbar{\nabla}_{[b} u_{a]},
  \quad
  \sigma_{ab} \defn \wbar{\nabla}_{(a} u_{b)},
  \quad
  \Theta      \defn \Theta\indices{^a_a} = \nabla_a u^a,
  \quad
  \Theta_{ab} \defn \wbar{\nabla}_{\langle a} u_{b \rangle},
\end{equation*}
where $\omega_{ab}$ denotes the \IdxMain{vorticity}\emph{vorticity (twist) tensor}, $\Theta_{ab}$ the \IdxMain{expansion}\emph{expansion tensor}, $\Theta$ the \IdxMain{expansion-scalar}\emph{(volume) expansion scalar}, $\sigma_{ab}$ the \IdxMain{shear}\emph{shear tensor}, and $\dot u_a$ the \IdxMain{acceleration}\emph{acceleration}.
The expansion scalar measures the separation of neighbouring observers and can be used to introduce a length scale $a$ via the definition
\begin{equation*}
  \frac{\dot{a}}{a} \defn \frac{\Theta}{n}.
\end{equation*}
Furthermore, we shall define the magnitudes $\omega$ and $\sigma$ of the vorticity and the shear tensor by
\begin{equation*}
  \omega^2 \defn \tfrac{1}{2} \omega_{ab} \omega^{ab},
  \quad
  \sigma^2 \defn \tfrac{1}{2} \sigma_{ab} \sigma^{ab}.
\end{equation*}

Using the definition of the Ricci tensor, we find $2 \nabla_{[a} \nabla_{b]} u^a u^b = R_{ab} u^a u^b$ which, together with~\eqref{eq:worldline_kinematics}, yields an equation describing the dynamics of world line
\begin{equation}\label{eq:worldline_dynamics}
  - R_{ab} u^a u^b = \dot \Theta + \tfrac{1}{n} \Theta^2 + 2 (\sigma^2 - \omega^2) - \dot u\indices{^a_{;a}},
\end{equation}
where $R_{ab} u^a u^b$ is sometimes called the \IdxMain{ray-scalar}\emph{Raychaudhuri scalar}.

If the vorticity tensor vanishes, then, by Frobenius' theorem, $u$ will be orthogonal to a foliation of the spacetime by $n$-dimensional Euclidean hypersurfaces $\Sigma$ whose metric is given by (the pullback of) $h$.
If the manifold is also simply connected so that its first de Rham cohomology group vanishes, then there exists even a time function.\footnote{A local time function will always exist if the vorticity is vanishing because every manifold is locally contractible.}
Furthermore, in the case of vanishing vorticity, $\wbar{\nabla}$ can be directly identified with the Levi-Civita connection on the spatial slices~$\Sigma$ with metric tensor~$h$.
Therefore it can be used to derive relations between the Riemann curvature~$\wbar{R}_{abcd}$ of $(\Sigma, h)$ and the Riemann curvature~$R_{abcd}$ of the whole manifold~$(M, g)$.
The \IdxMain{gauss-codacci}\emph{Gauss--Codacci equation}~\cite[Chap.~10]{wald:1984} is
\begin{equation}\label{eq:spatial_curvature}
  \wbar{R}_{abcd} = h\indices{_a^{\smash p}} h\indices{_{\smash b}^{\smash q}} h\indices{_c^r} h\indices{_{\smash d}^s} R_{pqrs} - 2 \Theta_{a[c} \Theta_{d]b},
\end{equation}
where we have expressed the extrinsic curvature via by the expansion tensor.
Upon contraction with the projection $h$ the spatial Riemann tensor leads to the projected Ricci tensor and Ricci scalar.

\IdxRanEnd{cov-split}

\subsection{Cauchy surfaces}
\label{sub:cs_cauchy}

Although many results hold in greater generality, for physical reasons we will from now restrict our attention to a subclass of four-dimensional Lorentzian manifolds:

\begin{definition}
  A \IdxMain{spacetime}\emph{spacetime} is a connected, oriented (positively or negatively oriented with respect to the volume form induced by the metric) and time-oriented, four\hyp{}dimensional smooth\footnote{On some occasions we will work with spacetimes whose metric is not smooth. In all these cases we will mention the regularity of the metric explicitly.} Lorentzian manifold $(M, g, \pm, u)$.
  Usually we omit orientation and time-orientation and identify a spacetime with the underlying Lorentzian manifold $(M, g)$; sometimes, when no confusion can arise, we will even drop the metric and say that $M$ is a spacetime.
\end{definition}

Not all spacetimes are of physical significance because they admit features like closed timelike curves (\ie, `time machines') that are usually considered unphysical.
A spacetime that does not have any closed timelike curves (\ie, $x \nin I^+(x)$ for all $x \in M$) is said to satisfy the \IdxMain{chrono-cond}\emph{chronological condition}; no compact spacetime satisfies the chronological condition.
A slightly stronger notion is the \IdxMain{causal-cond}\emph{causality condition}, which forbids the existence of closed \emph{causal} curves.

\begin{figure}
  \centering
  \subbottom[]{
    \includegraphics[page=1]{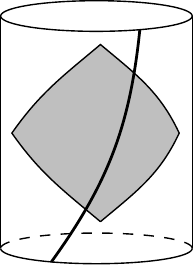}
    \label{fig:causally_convex1}}
  \hspace{1.5cm}
  \subbottom[]{
    \includegraphics[page=2]{fig_einstein_cylinder_1.pdf}
    \label{fig:causally_convex2}}
  \caption{Sets on the two-dimensional Einstein cylinder: (a) is causally convex, while (b) is not, as illustrated by the same causal curve in both figures.} \label{fig:causally_convex}
\end{figure}

Given an open set $U \subset M$, it is called \IdxMain{causal-convex}\emph{causally convex} if the intersection of any causal curve with $U$ is a connected set (possibly the empty set).
That is, a causal curve intersecting $U$ cannot leave the set and enter it again (\cref{fig:causally_convex}).
If for every neighbourhood $V$ of $x \in M$ there is a causally convex neighbourhood $U \subset V$ of~$x$, then the \IdxMain{strong-causal-point}\emph{strong causality condition} holds at~$x$.
A spacetime $(M, g)$ is called \IdxMain{strong-causal-man}\emph{strongly causal} if it is strongly causal at every point.
% Note that the topology of strongly causal spacetimes can be equivalently described by the \emph{Alexandrov topology} which is defined by the basis given by $I^+(x) \cap I^-(y)$ for every pair of points $x, y \in M$.
Finally, if $(M, g)$ is strongly causal and $J^+(x) \cap J^-(y)$ is compact for every pair $x, y \in M$, it is called \IdxMain{glob-hyp}\emph{globally hyperbolic}.

A \IdxMain{cauchy-surf}\emph{Cauchy surface} is a subset $\Sigma \subset M$ which is intersected exactly once by every inextendible timelike curve (and at least once by every inextendible causal curve).
Therefore, the causal shadow of a Cauchy surface is the entire spacetime.
Note that a Cauchy surface may be non-spacelike and non-smooth.
Our ability to pose a Cauchy problem on a spacetime requires the existence of Cauchy surfaces.

The existence of a Cauchy surface imposes strong conditions on the causality of the spacetime.
Given a spacetime, the following statements are equivalent \cite{bernal:2005,oneill:1983}: the spacetime
\begin{enumerate}[(a)]
  \item is globally hyperbolic,
  \item admits a Cauchy surface,
  \item admits a smooth time function compatible with the time-orientation.
\end{enumerate}

Therefore we will be mostly interested in globally hyperbolic spacetimes in the following.
Every globally hyperbolic spacetime $(M, g)$ is diffeomorphic to $\RR \times \Sigma$, where $\Sigma$ is diffeomorphic to a smooth spacelike Cauchy surface of~$M$ \cite{bernal:2003}.
If we denote by~$t$ the time function on~$M$, then the level sets of~$t$ are isometric to~$(\Sigma, g_t)$, where $g_t$ is a Riemannian metric on~$\Sigma$ depending smoothly on~$t$.
In fact, $(M, g)$ is isometric to $\RR \times \Sigma$ with the metric
\begin{equation*}
  - \beta\, \dif t \otimes \dif t + g_t,
\end{equation*}
where $\beta$ is a smooth function on $M$.

\IdxRanEnd{causality}

%-- General relativity --------------------------------------------------------%
\section{General relativity}
\label{sec:general_relativity}

\IdxRanBegin{gen-rel}

Throughout this section let $(M, g, \pm, u)$ be an arbitrary spacetime, unless otherwise specified.

\subsection{The stress-energy tensor}

General relativity describes the interaction of classical matter with the geometrical structure of the Universe.
Giving a precise definition of matter is difficult if not impossible.
Mathematically matter is described in general relativity as a covariant or contravariant symmetric $2$-tensor field~$T_{ab}$ or~$T^{ab}$, the \IdxMain{stress-energy}\emph{energy-momentum} or \emph{stress-energy tensor} which is covariantly conserved
\begin{equation*}
  \nabla^b T_{ab} = 0.
\end{equation*}

The physical content of the stress-energy tensor~$T_{ab}$ becomes clearer once we perform a covariant splitting relative to~$u$ and decompose $T_{ab}$ as
\begin{equation}\label{eq:imperfect_fluid}
  T_{ab} = \rho u_a u_b + p h_{ab} + 2 q_{(a} u_{b)} + \pi_{ab},
\end{equation}
with
\begin{equation}\label{eq:T_ab_components}
  \rho \defn T_{ab} u^a u^b,
  \quad
  p \defn \tfrac{1}{3} T_{ab} h^{ab},
  \quad
  q_a \defn - T_{\langle a \rangle b} u^b,
  \quad
  \pi_{ab} \defn T_{\langle ab \rangle},
\end{equation}
where $\rho$ denotes the \IdxMain{energy-density}\emph{energy density}, $p$ the \IdxMain{pressure}\emph{pressure}, $q_a$ the \IdxMain{momentum-density}\emph{momentum density} (\viz, the dissipation relative to $u^a$) and $\pi_{ab}$ the \IdxMain{aniso-stress}\emph{anisotropic stress}.
By definition of these quantities, the \IdxMain{stress-energy-trace}\emph{trace of the stress-energy tensor} is given by
\begin{equation*}
  T = T\indices{^a_a} = - \rho + 3 p.
\end{equation*}

Instead of working with the stress-energy tensor it is thus possible to work with these four quantities and the equations of state relating them.
For example, in this decomposition the conservation equation $T\indices{^{ab}_{;b}} = 0$ splits into the \IdxMain{energy-conserv}\emph{energy conservation} and the \IdxMain{momentum-conserv}\emph{momentum conservation equation}
\begin{subequations}
  \label{eq:conservation}
  \begin{align}
    \label{eq:energy_conservation}
    \dot \rho + (\rho + p) \Theta + \pi^{ab} \sigma_{ab} + \wbar{\nabla}_a q^a + 2 \dot u_a q^a = 0,\\
    \label{eq:momentum_conservation}
    \wbar{\nabla}_a p + (\rho + p) \dot u_a + \pi_{ab} \dot u^b + \wbar{\nabla}^b \pi_{ab} + \dot q_a + \tfrac{4}{3} \Theta q_a + (\sigma_{ab} + \omega_{ab}) q^b = 0,
  \end{align}
\end{subequations}
which are the familiar equations for an inertial observer on Minkowski spacetime.

If the anisotropic terms in~\eqref{eq:imperfect_fluid} vanish ($q_a = 0 = \pi_{ab}$), the stress-energy takes the \IdxMain{perfect-fluid}\emph{perfect fluid} form
\begin{equation}
  \label{eq:perfect_fluid}
  T_{ab} = \rho u_a u_b + p h_{ab} = (\rho + p) u_a u_b + p g_{ab}
\end{equation}
and the conservation equations~\eqref{eq:conservation} reduce to
\begin{subequations}
  \label{eq:perfect_conservation}
  \begin{align}
    \label{eq:perfect_energy_conservation}
    \dot \rho + (\rho + p) \Theta = 0, \\
    \label{eq:perfect_momentum_conservation}
    \wbar{\nabla}_a p + (\rho + p) \dot u_a = 0.
  \end{align}
\end{subequations}
One further distinguishes between different forms of the general equation of state $p = p(\rho, s)$, where $s$ is the medium's specific entropy.
Namely, if $p = p(\rho)$, one speaks of a \IdxMain{barotropic-fluid}\emph{barotropic fluid}, and if $p = 0$, we have \IdxMain{pressure-free-matter}\emph{pressure-free matter}, also called \emph{`dust'}.

While general relativity imposes no \latin{a priori} constraints on the form of the matter content, many possibilities can be considered unphysical in classical physics.
The most common \IdxMain{energy-cond}\emph{energy conditions} are:
\begin{enumerate}[leftmargin=*,widest=(WEC)]
  \item[(NEC)] \IdxMain{NEC}\emph{null energy condition}
    \begin{equation*}
      T_{a b} v^a v^b \geq 0
      \qquad
      \text{for all lightlike $v^a$},
    \end{equation*}
    \ie, no negative energy densities along any lightray;
  \item[(WEC)] \IdxMain{WEC}\emph{weak energy condition}
    \begin{equation*}
      T_{a b} v^a v^b \geq 0
      \qquad
      \text{for all timelike $v^a$},
    \end{equation*}
    \ie, no observer detects negative energy densities;
  \item[(DEC)] \IdxMain{DEC}\emph{dominant energy condition}
    \begin{equation*}
      T_{a b} v^a w^b \geq 0
      \qquad
      \text{for all future-pointing timelike $v^a, w^b$},
    \end{equation*}
    \ie, the stress-energy flux is causal;
  \item[(SEC)] \IdxMain{SEC}\emph{strong energy condition}
    \begin{equation*}
      T_{a b} v^a v^b - \tfrac{1}{2} T \geq 0
      \qquad
      \text{for all unit timelike $v^a$},
    \end{equation*}
    \ie, gravity is attractive if $\Lambda = 0$, see \eqref{eq:raychaudhuri}.
\end{enumerate}
By continuity, WEC implies NEC and it is also not difficult to see that DEC implies WEC.
Moreover, the SEC \emph{does not} imply the WEC but only the NEC.
The reverse implications are generally not true.
Note that the strong energy condition is too strong for many physically relevant scenarios.
% For a perfect fluid the energy conditions above are equivalently given by
% \begin{enumerate}[leftmargin=*,widest=(WEC)]
%   \item[(NEC)] $\rho + p \geq 0$,
%   \item[(WEC)] $\rho \geq 0$ and $\rho + p \geq 0$,
%   \item[(DEC)] $\rho \geq \abs{p}$,
%   \item[(SEC)] $\rho + p \geq 0$ and $\rho + 3 p \geq 0$.
% \end{enumerate}

A generic feature of quantum field theory is that none of the energy conditions above will hold, even in an averaged sense, because of the Reeh--Schlieder theorem \cite{sanders:2009,haag:1996}.
Instead one finds lower bounds on the averaged energy density, called \emph{quantum energy inequalities}, see \cite{fewster:2006a,fewster:2012c} for a review of the subject.

\subsection{Einstein's equation}

\IdxMain{einstein-eq}\emph{Einstein's equation} with a cosmological constant are
\begin{equation}\label{eq:einstein}
  R_{ab} - \frac{1}{2} R g_{ab} + \Lambda g_{ab} \defn G_{ab} + \Lambda g_{ab} = \frac{8 \uppi \mathrm{G}}{\mathrm{c}^4} T_{ab},
\end{equation}
where $G_{ab}$ is called the \IdxMain{einstein-curv}\emph{Einstein (curvature) tensor} and $\Lambda$ the \IdxMain{cosmo-const}\emph{cosmological constant}.
The constants on the right-hand side are \emph{Newton's gravitational constant}~$\mathrm{G}$ and the \emph{speed of light}~$c$; we shall always choose units such that $8 \uppi \mathrm{G} = c = 1$.
Often we will also work with the trace of~\eqref{eq:einstein}:
\begin{equation}\label{eq:einstein_trace}
  -R + 4 \Lambda = T.
\end{equation}
Sometimes one absorbs the cosmological constant into the stress-energy tensor to emphasize its non-geometric nature.

Note that Einstein's tensor is covariantly conserved and symmetric, \ie,
\begin{equation}
  \nabla^b G_{ab} = 0
  \quad \text{and} \quad
  G_{ab} = G_{(ab)},
\end{equation}
so that the left-hand side of~\eqref{eq:einstein} is consistent with the right-hand side and gives a second-order differential equation in the metric.
In fact, we can derive~\eqref{eq:einstein} from the assumption that the stress-energy tensor of a matter field should be the source of a gravitational potential (the metric tensor) in a second-order differential equation.
Since the stress-energy tensor is conserved and symmetric, \eqref{eq:einstein}~is the only possibility.
% Another derivation, which follows essentially the argument, is via the \emph{Einstein--Hilbert action} with cosmological constant and matter
% \begin{equation*}
%   S_K \defn \int_K (R + \Lambda + L_m) \mu_g,
% \end{equation*}
% where $K$ is a compact subset of~$M$ and $L_m$ is the Lagriangian of a matter field theory.

Combining~\eqref{eq:einstein} with~\eqref{eq:einstein_trace}, Einstein's equations can be recast into the equivalent form $R_{ab} - \Lambda g_{ab} = T_{ab} - \tfrac{1}{2} T g_{ab}$.
Together with the imperfect fluid form~\eqref{eq:imperfect_fluid} of the stress-energy tensor this equation yields
\begin{equation*}
  R_{ab} - \Lambda g_{ab} = \tfrac{1}{2} (\rho + 3 p) u_a u_b + \tfrac{1}{2} (\rho - p) h_{ab} + 2 q_{(a} u_{b)} + \pi_{ab}.
\end{equation*}
Contractions with the time-orientation $u^a$ and the associated projector $h_{ab}$ to the orthogonal surfaces then give the three equation
\begin{subequations}\label{eq:contracted_ricci}
  \begin{align}
    \label{eq:contracted_ricci1}
    R_{ab} u^a u^b & = \tfrac{1}{2} (\rho + 3 p) - \Lambda, \\
    \label{eq:contracted_ricci2}
    R_{bc} h\indices{_a^b} u^c & = - q_a, \\
    \label{eq:contracted_ricci3}
    R_{cd} h\indices{_a^c} h\indices{_b^d} & = \tfrac{1}{2} (\rho - p) h_{ab} + \Lambda h_{ab} + \pi_{ab}.
  \end{align}
\end{subequations}

The Raychaudhuri scalar~\eqref{eq:worldline_dynamics} attains the physical meaning of the active gravitational energy by~\eqref{eq:contracted_ricci1} and thus we can state the \IdxMain{ray-eq}\emph{Raychaudhuri equation} as:
\begin{equation}\label{eq:raychaudhuri}
  \dot \Theta + \tfrac{1}{3} \Theta^2 + 2 (\sigma^2 - \omega^2) - \wbar{\nabla}_a \dot u^a - \dot u_a \dot u^a + \tfrac{1}{2} (\rho + 3 p) - \Lambda = 0.
\end{equation}
Hence we see that expansion, shear and matter (satisfying the strong energy condition) promote gravitational collapse, whereas a positive cosmological constant, vorticity and positive acceleration (due to non-gravitational forces inside the medium) oppose gravitational collapse.
One can also derive differential equations for the shear and the vorticity~\cite{tsagas:2008}.

\subsection{De Sitter spacetime}

A Lorentzian manifold $(M, g)$ is called a \IdxMain{vacuum-sol}\emph{vacuum solutions} to Einstein's equation~\eqref{eq:einstein} if both $T_{ab}$ and $\Lambda$ vanish globally.
That is, such a solution satisfies
\begin{equation*}
  R_{ab} = 0.
\end{equation*}
The vacuum Einstein equation is the most studied special case of the Einstein equation and many important and instructive examples fall into this class.
The most basic solution is \Idx{mink-spacetime}Minkowski spacetime $(M, \eta)$, see~\eqref{eq:minkowski}, which describes a featureless empty universe.
A rotation-symmetric vacuum solution is given by the famous Schwarzschild solution.

If we also allow for a cosmological constant but keep $T_{ab} = 0$, we need to solve
\begin{equation*}
  R_{ab} = \Lambda g_{ab}.
\end{equation*}
The maximally symmetric solutions of this equation fall into three classes depending on the sign of $\Lambda$.
We focus here on the case $\Lambda > 0$ called \IdxMain{de-sitter}\emph{de Sitter spacetime}; for $\Lambda = 0$ one obtains Minkowski spacetime and for $\Lambda < 0$ the so-called \IdxMain{anti-de-sitter}\emph{anti-de Sitter spacetime}.

Four-dimensional de Sitter spacetime is the hyperboloidal submanifold of five\hyp{}dimensional Minkowski spacetime with coordinates $(y^0, y^1, \dotsc, y^4)$ that satisfies the equation
\begin{equation}\label{eq:hyperboloid}
  -(y^0)^2 + (y^1)^2 + \dotsb + (y^4)^2 = H^{-2},
  \quad
  \Lambda = 3 H^2,
\end{equation}
where $H > 0$ is called the \IdxMain{hubble-const}\emph{Hubble constant}; thus de Sitter space is topologically $\RR \times S^3$.
The pullback of the Minkowski metric to this space yields a Lorentzian metric.

A global \IdxMain{de-sitter-chart}coordinate chart with coordinates $(t, \chi, \theta, \varphi)$ can be defined by
\begin{equation*}
  y^0  = H^{-1} \sinh(H t),
  \quad
  y^i  = H^{-1} \cosh(H t) z^i,
\end{equation*}
where $z^i = z^i(\chi, \theta, \varphi)$, $i = 1,2,3,4$, are the usual spherical coordinates on $S^3$ with unit radius.
In these coordinates the induced metric on de Sitter spacetime reads\IdxMain{de-sitter-metric}
\begin{align*}
  g & = - \dif t \otimes \dif t + H^{-2} \cosh(H t) \big( \dif\chi \otimes \dif\chi + \sin^2\!\chi\, ( \dif\theta \otimes \dif\theta + \sin^2\!\theta\, \dif\varphi \otimes \dif\varphi) \big) \\
  & = - \dif t \otimes \dif t + H^{-2} \cosh(H t)\, g_{\mathbb{S}^3},
\end{align*}
where $g_{\mathbb{S}^3}$ denotes the standard metric on the $3$-sphere.
It is clear, that de Sitter spacetime is globally hyperbolic and the constant time hypersurfaces are Cauchy surfaces.

Another coordinate chart with coordinates $(t, x^1, x^2, x^3)$ is given by
\begin{align*}
  y^0 & = H^{-1} \sinh(H t) + \tfrac{1}{2} H \e^{H t} r^2, \\
  y^i & = \e^{H t} x^i, \\
  y^4 & = H^{-1} \cosh(H t) - \tfrac{1}{2} H \e^{H t} r^2,
\end{align*}
where $i = 1,2,3$ and $r^2 = (x^1)^2 + (x^2)^2 + (x^3)^2$.
It covers only the half of de Sitter spacetime which satisfies $y^0 + y^4 \geq 0$ and is called the \emph{cosmological patch} or \IdxMain{de-sitter-cosmo}\emph{cosmological chart}; it is diffeomorphic to four-dimensional Euclidean space.
Within the cosmological chart, the metric reads
\begin{equation*}
  g = - \dif t \otimes \dif t + \e^{2 H t} \delta_{ij}\, \dif x^i \otimes \dif x^j = (H \tau)^{-2} ( - \dif\tau \otimes \dif\tau + \delta_{ij}\, \dif x^i \otimes \dif x^j ),
\end{equation*}
where we define the \IdxMain{conf-time-sitter}\emph{conformal time} $\tau \in (0, \infty]$ via
\begin{equation*}
  t \mapsto \tau(t) = - \int_t^\infty a(t')^{-1}\, \dif t'.
\end{equation*}
Note that this metric is a special case of a Friedmann\hyp{}Lemaître\hyp{}Robertson\hyp{}Walker metric, to be discussed in more generality in \cref{sec:cosmology}.

On de Sitter spacetime, \IdxMain{world-fun-sitter}Synge's world function is known in closed form and is closely related to the geodesic distance on five\hyp{}dimensional Minkowski spacetime.
In fact, since the chord length between two points $x, x'$ on de Sitter space considered as the hyperboloid~\eqref{eq:hyperboloid} is
\begin{equation}\label{eq:deSitter_Z}
  Z(x, x') \defn H^2 \eta_{\alpha\beta} y^\alpha(x) y^\beta(x'),
\end{equation}
Synge's world function on de Sitter spacetime is given by
\begin{equation}\label{eq:synge_desitter}
  \cos\big( H \sqrt{2 \sigma(x, x')} \big) = Z(x, x')
\end{equation}
for $\abs{Z} \leq 1$, \ie, for $x'$ are not timelike to~$x$.
\Cref{eq:synge_desitter} can be analytically continued to timelike separated points $x, x'$, whence we find
\begin{equation*}
  \cosh\big( H \sqrt{-2 \sigma(x, x')} \big) = Z(x, x') \quad \text{for }\abs{Z} > 1.
\end{equation*}
The possible values for the chord length $Z$ are illustrated in a conformal diagram of de Sitter spacetime in \cref{fig:de_sitter_Z}.
In the cosmological chart, the function $Z$ attains the simple form
\begin{equation}\label{eq:deSitter_Z_cosmo}
  Z(x, x') = 1 + \frac{(\tau - \tau')^2 - (\vec{x} - \vec{x}')^2}{2 \tau \tau'} = \frac{\tau^2 + \tau'{}^2 - (\vec{x} - \vec{x}')^2}{2 \tau \tau'},
\end{equation}
where $x = (\tau, \vec{x})$ and $x' = (\tau', \vec{x}')$.
Note that the fraction on the right-hand side is a rescaling of Synge's world function on Minkowski spacetime by $-2 \tau \tau'$.

\begin{figure}
  \centering
  \includegraphics{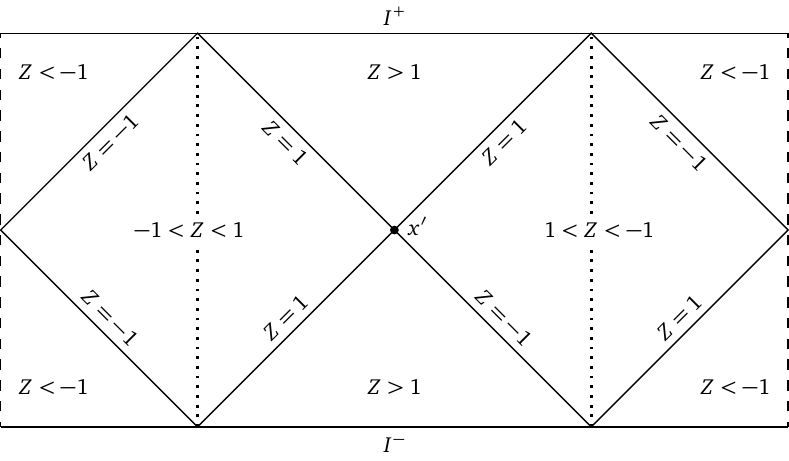}

  \caption{De Sitter spacetime conformally mapped to a cylinder. The cylinder is unwrapped and the left and right dotted edges must be identified. $x'$ is any (fixed) point of de Sitter spacetime and $Z(x, x')$ is shown for all choices of $x$. The dotted line represents $Z=0$. See also \cite[Fig.~2]{allen:1985}.}
  \label{fig:de_sitter_Z}
\end{figure}

\IdxRanEnd{gen-rel}

%-- Cosmology -----------------------------------------------------------------%
\section{Cosmology}
\label{sec:cosmology}

\IdxRanBegin{cosm}

When studying cosmological problems on usually describes the Universe by a homogeneous and isotropic spacetime, \ie, one assumes the Friedmann\hyp{}Lemaître\hyp{}Robertson\hyp{}Walker model.
This standard prescription underpins the so called standard model of cosmology, the $\upLambda$CDM model of cold dark matter with a cosmological constant.
However, from the presence of structure in the Universe (\eg, galaxy clusters, galaxies, stars, \etc) we can directly deduce that the Universe is \emph{neither} homogeneous \emph{nor} isotropic.
Instead it is believed that one can describe the Universe as nearly homogeneous and isotropic on cosmological scales so that observed Universe can be modelled as a perturbation around a FLRW spacetime and homogeneity and isotropy hold in an averaged sense.

The cosmic microwave background (CMB) and the galaxy distribution are often believed to give a direct justification of this idea, but since all observations are along the past light cone and do not measure an instantaneous spatial surface one can only directly observe isotropy.
The link to homogeneity is less clear but if \emph{all} observers measure an isotropic CMB, it can be shown that the spacetime is also homogeneous.
In this respect, the current cosmological model is heavily influenced by the philosophical paradigm in cosmology, the Copernican principle, that we live in no distinguished region of the Universe and that other observers would observe the same.
But even under this additional assumption it is not completely clear that near homogeneity follows since the CMB is not exactly but only nearly isotropic; see \cite{ehlers:1968,ellis:2012,maartens:1995,stoeger:1995} for a discussion of this issue.

Ignoring these shortcomings, we follow the standard approach and assume that we live in an `almost'-FLRW spacetime, \ie, a universe which is correctly described by a perturbation around a FLRW background.
Because such an assignment of a background spacetime to the physical perturbed spacetime is not unique, one has to deal with a gauge problem.

Accordingly, we will begin by introducing the FLRW model and its properties.
Then we will study the general gauge problem and its application in the case of a perturbation around a FLRW spacetime.

\subsection{Friedmann\hyp{}Lemaître\hyp{}Robertson\hyp{}Walker spacetimes}
\label{sub:cos_flrw}

One of the simplest solutions of the Einstein equation is the \IdxMain{flrw-spacetime}\emph{Friedmann\hyp{}Lemaître\hyp{}Robertson\hyp{}Walker (FLRW) solution}, which describes a \emph{homogeneous} and \emph{isotropic}, expanding or contracting universe.
Let $(M, g)$ be a spacetime, which is still to be determined, with a preferred flow (time-orientation)~$u$.
As a consequence of (spatial) isotropy, \ie, the absence of any preferred (spatial) direction, vorticity, shear and acceleration have to vanish:
\begin{equation*}
  \omega_{ab} = 0, \quad \sigma_{ab} = 0, \quad \dot u_a = 0.
\end{equation*}
Also the anisotropic terms of the stress-energy tensor have to vanish, \ie,
\begin{equation*}
  q_a = 0, \quad \pi_{ab} = 0
\end{equation*}
so that the stress-energy tensor takes the perfect fluid form \eqref{eq:perfect_fluid}.
Consequently the FLRW model is completely determined by its energy density~$\rho$, pressure~$p$ and the expansion~$\Theta$.

Since vorticity and acceleration vanish, the spacetime $(M, g)$ is foliated by surfaces $\Sigma$ orthogonal to~$u$, which are required to be homogeneous by assumption and thus $\rho$, $p$ and $\Theta$ are constant on these surfaces.
For the same reason there exists locally a time function~$t$, called \IdxMain{cosm-time}\emph{cosmological time}, that measures proper time, defined up to a constant shift, such that $u_a = - \nabla_a t$.
Henceforth we will always assume the time function exists globally so that the resulting spacetime is stably causal and, in particular, globally hyperbolic.

One can show that the projected Ricci tensor, \ie, the Ricci tensor for the spatial slices $\Sigma$ (\cf~\eqref{eq:spatial_curvature}), simplifies significantly to
\begin{equation*}
  \wbar{R}_{ab}
  = \tfrac{1}{3} \wbar{R} h_{ab}
  = \tfrac{2}{3} (\rho + \Lambda - \tfrac{1}{3} \Theta^2) h_{ab}.
\end{equation*}
We may recast this equation into the more familiar form of the \IdxMain{fst-friedmann}\emph{first Friedmann equation}
\begin{equation}\label{eq:1st_friedmann}
  H^2 = \tfrac{1}{3} \big( \rho + \Lambda - \tfrac{1}{2} \wbar{R} \big).
\end{equation}
with the famous \IdxMain{hubble-fun}\emph{Hubble parameter} or \emph{Hubble function} $H \defn \Theta / 3 = \dot a / a$.
The \IdxMain{snd-friedmann}\emph{second Friedmann equation} is a special case of the Raychaudhuri equation \eqref{eq:raychaudhuri} and reads
\begin{equation}\label{eq:2nd_friedmann}
  \dot H + H^2 + \tfrac{1}{6} (\rho + 3 p) - \tfrac{1}{3} \Lambda = 0.
\end{equation}
These two equations can be complemented with the energy conservation equation for the perfect fluid \eqref{eq:perfect_energy_conservation} to show that $\wbar{R} a^2$ is a constant.

We define $K \defn \wbar{R} a^2 / 3$ and notice that the initial value for the scale factor $a > 0$ is arbitrary so that we can restrict its value to $K = -1, 0, +1$.
The sign of $K$ determines the local geometry of the spatial sections:
$K = -1, 0, +1$ correspond respectively to a hyperbolic, a flat and a elliptic geometry.
The topology of the spatial sections is not completely determined by~$K$ and, in fact, there are many possibilities.
While $K = +1$ implies that the spatial sections are compact, both compact and non-compact spatial section are possible for $K = -1, 0$.

Which of these three distinct values for~$K$ is realized depends on the energy density contained in the universe.
If the energy density takes on the critical value $\rho = \rho_c \defn 3 H^2 - \Lambda$, the spatial surfaces will be flat, while $\rho > \rho_c$ leads to a spherical and $\rho < \rho_c$ leads to a hyperbolical geometry.
Furthermore, according to \eqref{eq:2nd_friedmann}, an accelerating expansion ($\ddot a > 0$) occurs when $\rho + 3 p < 0$ (assuming $\Lambda = 0$), \ie, when the strong energy condition is violated.

A metric tensor realizing the FLRW universe in a local comoving coordinate frame $x^\mu = (t, r, \theta, \phi)$ with respect to~$u^a$ (\ie, $u = \partial_t$) is given (locally) by the \IdxMain{flrw-metric}\emph{FLRW metric}
\begin{equation}\label{eq:flrw_metric_comoving}
  g = - \dif t \otimes \dif t + a(t)^2 \big( \dif r \otimes \dif r + f_K(r)^2\, (\dif\theta \otimes \dif\theta + \sin^2 \theta\, \dif\phi \otimes \dif\phi) \big),
\end{equation}
where
\begin{equation*}
  f_K(r) \defn
  \begin{cases}
    \sin r  & \text{for} \; K = +1  \\
    r       & \text{for} \; K = 0   \\
    \sinh r & \text{for} \; K = -1.
  \end{cases}
\end{equation*}
Therefore, a spacetime is a FLRW spacetime if and only if its metric attains locally the form~\eqref{eq:flrw_metric_comoving} in some coordinate system and the time-orientation is given by~$\dif t$.

Let us define \IdxMain{conf-time-flrw}\emph{conformal time}~$\tau$ via $\dif\tau = \dif t / a$.
That is we set
\begin{equation*}
  \tau(t) = \tau_0 - \int_{t_0}^t \frac{1}{a(t')}\, \dif t'
  \quad \text{or} \quad
  \tau(t) = \int_t^{t_1} \frac{1}{a(t')}\, \dif t' - \tau_1
\end{equation*}
with arbitrary $\tau_0, \tau_1$ and (possibly infinite) $t_0, t_1$ such that the integral converges.
Rewriting \eqref{eq:flrw_metric_comoving} with respect to conformal time, we obtain the alternative metric tensor, the \IdxMain{flrw-conf-metric}\emph{conformal FLRW metric},
\begin{equation}
  \label{eq:flrw_metric_conformal}
  g = a(\tau)^2 \big( - \dif\tau \otimes \dif\tau + \dif r \otimes \dif r + f_K(r)^2\, (\dif\theta \otimes \dif\theta + \sin^2 \theta\, \dif\phi \otimes \dif\phi) \big)
\end{equation}
and thus we notice that a flat FLRW universe is locally conformally isometric to Minkowski space.

Throughout this thesis we mostly work with flat FLRW universes and will fix $M \simeq \RR^4$ in that case.
Therefore we can choose globally Cartesian coordinates for the spatial sections so that we have
\begin{equation*}
  g = - \dif t \otimes \dif t + a(t)^2 \delta_{ij}\, \dif x^i \otimes \dif x^j = a(\tau)^2 \big( - \dif\tau \otimes \dif\tau + \delta_{ij}\, \dif x^i \otimes \dif x^j \big),
\end{equation*}
where the coordinate functions $x^i$ range over the entire real line.

\subsection{Gauge problem}
\label{sub:cos_gauge}
\IdxMain{gauge-prop}

The correspondence of a \IdxMain{background-st}\emph{background spacetime} $(M, g_0)$ to the \IdxMain{physical-st}\emph{physical spacetime} $(M, g)$ is equivalent to the specification of a diffeomorphism $\psi : M \to M$.
Given any other diffeomorphism $\psi'$ between the background spacetime and the physical spacetime, we can construct $\varphi = \psi^{-1} \circ \psi'$.
For a tensor field~$S$ on~$M$ to be \IdxMain{gauge-inv}\emph{gauge invariant} we require that $\pb{\varphi} S = S$.
Then it holds that
\begin{equation*}
  \delta S = S - \pb{\psi} S = S - \pb{\psi'} S
\end{equation*}
and we say that the perturbation~$\delta S$ is a gauge invariant quantity.
Otherwise, the perturbation~$\delta S$ is completely dependent on the mapping~$\psi$ and even if $\psi$ is specified it will not be an observable quantity unless the correspondence~$\psi$ itself has been specified via an observational procedure (\eg\ via an averaging approach, see the discussion in~\cite{ellis:1989}).
Therefore the only possibilities for a tensor field~$\delta S$ to be gauge invariant is that $S$ is a constant scalar, a zero tensor or a product of Kronecker deltas~\cite[Lem.~2.2]{stewart:1974}).

A slightly different picture which clarifies the perturbation aspect of the gauge problem is sometimes more helpful.
Since the local flow of any vector field~$v$ on~$M$ is a diffeomorphism $\psi_\varepsilon : M \to M$, where $\varepsilon$ is contained in a sufficiently small interval around~$0$, we can describe the gauge problem alternatively in terms of vector fields.
Given a (differentiable) tensor field~$S$, we have by definition
\begin{equation*}
  \lieD_v S = \lim_{\varepsilon \to 0} \frac{1}{\varepsilon} (\pb{\psi_\varepsilon} S - S)
\end{equation*}
or, in other words,
\begin{equation*}
  \pb{\psi_\varepsilon} S = S + \varepsilon \lieD_v S + \bigO(\varepsilon^2).
\end{equation*}
The gauge choice is now encoded in~$v$, which is completely arbitrary, and we see that $S$ is gauge invariant to first order if only if $\lieD_v S = 0$ for all~$v$.
For $S$ to be exactly gauge invariant (as discussed above) it must hold that $\lieD^n_v S = 0$ for all~$n$.

In the light of this discussion, two approaches to perturbations of FLRW spacetimes seem expedient:
Since any quantity describing the inhomogeneity or anisotropy of the perturbation of the FLRW spacetime must vanish on the background, it will be gauge-invariant.
This leads to the ``$1+3$ covariant and gauge-invariant'' approach of~\cite{hawking:1966,lyth:1988,ellis:1989}.
The alternative and more commonly used approach due to~\cite{bardeen:1980,kodama:1984} constructs gauge invariant quantities directly from the perturbed metric and stress-energy tensor.

\subsection{Decomposition of tensor fields}
\label{sub:cos_decomposition}
\IdxMain{tensor-decomp}

Before we can discuss metric perturbations, we need to investigate the decomposition of vector and rank-$2$ tensor fields into their `scalar', `vector' and `tensor' parts~\cite{stewart:1990}.

Consider a non-compact,\footnote{The decomposition works also for compact manifold but is non-unique in that case.} boundaryless, orientable $3$-dimensional Riemannian manifold $(\Sigma, \gamma)$ with covariant derivative denoted by a vertical bar (\eg, $\phi_{|i}$).
Using the Hodge decomposition theorem, we can uniquely decompose any sufficiently fast decaying smooth one-form~$B$ as
\begin{equation}
  \label{eq:vector_decomposition}
  B_i = \phi_{|i} + S_i,
\end{equation}
where $\phi$ is a scalar function and $S_i$ is divergence-free.

Any $(0,2)$~tensor field $C$ can be decomposed as
\begin{equation*}
  C_{ij} = \tfrac{1}{3} \delta_{ij} \gamma^{kl} C_{kl} + C_{[ij]} + C_{\langle ij \rangle},
\end{equation*}
\ie, into its trace, its antisymmetric part (equivalent to a vector field via the Hodge operator $\hodge$) and its trace-free symmetric part~$C_{\langle ij \rangle}$.
According to~\cite[Thm.~4.3]{cantor:1981}, as a consequence of the Fredholm alternative and an application of~\eqref{eq:vector_decomposition}, a sufficiently fast decaying $C$ can be further uniquely decomposed so that
\begin{equation}
  \label{eq:tensor_decomposition}
  C_{\langle ij \rangle} = \phi_{|\langle ij \rangle} + S_{(i|j)} + h_{ij},
\end{equation}
where $\phi$ and $S$ are as before and $h$ is a trace- and divergence-free $(0,2)$~tensor field.

\subsection{Metric perturbations}
\label{sub:cos_perturbations}
\IdxMain{metric-pert}

As background spacetime we take $(M, g_0)$ with $g_0$ the flat FLRW metric in conformal time, \ie,
\begin{equation*}
  g_0 = a^2 (- \dif\tau \otimes \dif\tau + \delta_{ij}\, \dif x^i \otimes \dif x^j),
\end{equation*}
which will be used throughout this section to raise and lower indices.
Then we define the \IdxMain{perturbed-flrw-metric}\emph{perturbed FLRW metric}
% NOTE: the notation is slightly different from malik:2009 (C tensor minus \psi)
\begin{equation}\begin{split}
  \label{eq:flrw_metric_perturbed}
  a^{-2} g & = -(1 + 2 \phi)\, \dif\tau \otimes \dif\tau + B_i\, (\dif x^i \otimes \dif\tau + \dif\tau \otimes \dif x^i) \\&\quad + \big( (1 - 2 \psi) \delta_{ij} + 2 C_{ij} \big)\, \dif x^i \otimes \dif x^j
\end{split}\end{equation}
where the scalar fields $\phi$, $\psi$, the 3-vector field~$B$ and the trace-free $3$-tensor field~$C$ are considered `small', \ie, \eqref{eq:flrw_metric_perturbed} should be understand as a $1$-parameter family of metrics $g_\varepsilon$ and each of the perturbation variables as multiplied with a small parameter $\varepsilon$.
However, to avoid cluttering the equations unnecessarily, one usually omits the $\varepsilon$.

The ten degrees of freedom encoded in these four quantities exhibit the full gauge dependence.
Before studying the behaviour of~$g$ under gauge transformations, let us rewrite \eqref{eq:flrw_metric_perturbed} using the decompositions \eqref{eq:vector_decomposition} and \eqref{eq:tensor_decomposition}:
\begin{equation*} % NOTE: malik:2009 includes \psi \delta_{ij} in the definition of C
  B_i = B_{,i} - S_i, \quad C_{ij} = E_{,ij} + F_{(i,j)} + \tfrac{1}{2} h_{ij}
\end{equation*}
with two scalar fields $E, B$, two divergence-free 3-vector fields $S, F$ and a trace-free, transverse 3-tensor field $h$.
Therefore,
\begin{equation}
  \label{eq:flrw_metric_perturbation}
  \begin{split}
    a^{-2} g =
    & -(1 + 2 \phi)\, \dif\tau \otimes \dif\tau + (B_{,i} - S_i)\, (\dif x^i \otimes \dif\tau + \dif\tau \otimes \dif x^i) + {} \\
    & + \big( (1 - 2 \psi) \delta_{ij} + 2 E_{,ij} + 2 F_{(i,j)} + h_{ij} \big)\, \dif x^i \otimes \dif x^j.
  \end{split}
\end{equation}
This decomposition allows us to consider three types of perturbations separately, namely, the \IdxMain{scalar-pert}\emph{scalar perturbations} caused by $\phi, \psi, E, B$, the \IdxMain{vector-pert}\emph{vector perturbations} due to $S, F$ and the \IdxMain{tensor-pert}\emph{tensor perturbations} caused by $h$.
The inverse of the perturbed metric \eqref{eq:flrw_metric_perturbed} up to first order is
\begin{equation*} % NOTE: $(A + B)^{-1} \approx A^{-1} - A^{-1} B A^{-1}$
  a^2 g^{-1} = -(1 - 2 \phi)\, \partial_\tau \otimes \partial_\tau + B^i\, (\partial_i \otimes \partial_\tau + \partial_\tau \otimes \partial_i) + \big( (1 + 2 \psi) \delta^{ij} - 2 C^{ij} \big)\, \partial_i \otimes \partial_j.
\end{equation*}

Let us now determine the transformation behaviour of these perturbation variables by calculating the gauge dependence of $g$ up to linear order, \ie, the dependence of $\lieD_\xi {g_{ab}} = 2 \nabla_{(a} \xi_{b)} = 2 {}^0\nabla_{(a} \xi_{b)} + \bigO(\varepsilon)$ on $\xi^\mu = (\xi^0, \xi^{,i} + \xi^{i})$ with $\xi\indices{^i_{,i}} = 0$, where ${}^0\nabla$ is the covariant derivative on the background spacetime.
Hence we have the transformations:
\begin{align*}
  g_{00} & \to g_{00} - 2 a^2 ( \xi^0{}' + \mathcal{H} \xi^0 ), \\
  g_{0i} & \to g_{0i} + a^2 ( \xi_{,i}' + \xi_i' - \xi\indices{^0_{\smash{,i}}} ), \\
  g_{ij} & \to g_{ij} + 2 a^2 ( \xi_{,(ij)} + \xi_{(i,j)} + \mathcal{H} \delta_{ij} \xi^0 ),
\end{align*}
where we define the \IdxMain{conf-hubble-fun}\emph{conformal Hubble parameter} $\mathcal{H} \defn a H = a'/a$ and a prime denotes a derivative with respect to the conformal time (\eg, $a' = \partial_\tau a$).
It follows that the perturbation variables transform as
\begin{align*}
  \phi & \to \phi + \xi^0{}' + \mathcal{H} \xi^0 &
  \psi & \to \psi - \mathcal{H} \xi^0 \\
  B    & \to B + \xi' - \xi^0 &
  E    & \to E + \xi \\
  S_i  & \to S_i - \xi_i'&
  F_i  & \to F_i + \xi_i
\end{align*}
and $h_{ij} \to h_{ij}$.

Constructing linear combinations of the perturbation variables, we can now construct several first-order gauge invariant quantities.
Two simple (first-order) gauge-invariant functions characterizing the scalar perturbations are the \IdxMain{bardeen}\emph{Bardeen potentials} $\Phi, \Psi$ \cite{bardeen:1980}
\begin{subequations}
  \label{eq:bardeen_potentials}
  \begin{align}
    \label{eq:bardeen_phi}
    \Phi & \defn \phi - \mathcal{H} \sigma - \sigma', \\
    \label{eq:bardeen_psi}
    \Psi & \defn \psi + \mathcal{H} \sigma
  \end{align}
\end{subequations}
in terms of the shear potential $\sigma \defn E' - B$.

\IdxRanEnd{cosm}
\IdxRanEnd{lor-geo}

%% file: analysis.tex
%!TEX root = master.tex

\chapter{Analysis}
\label{cha:analysis}

\section*{Summary}

The contents of the present chapter may be generously subsumed under ``analysis'', hence the title.
We begin with a review of topological vector spaces (\cref{sec:topological_vector_spaces}) starting from the basics of topology and locally convex topological spaces with the particular example of function spaces, and finishing with a discussion duality pairings and tensor products.
The material presented in this section is a excerpt of the results found in the books \cite{kelley:1955,jarchow:1981,treves:1967,schaefer:1971,pietsch:1972}.
As always, although we refer the reader to the books cited above for proofs of the various statements, care has been taken to present the material in a structured way so that no results should appear surprising.

The second section (\cref{sec:algebra}) concerns the theory of ${}^*$-algebras and thus forms the foundation of the algebraic approach to quantum field theory to be discussed in the later chapters.
Here we will discuss the general features of ${}^*$-algebras and $C^*$-algebras, states with an emphasize on the Gel'fand--Naimark--Segal reconstruction theorem, and the Weyl $C^*$-algebra.
For this section we refer the reader to the books \cite{inoue:1998,fragoulopoulou:2005}.
Details on the Weyl algebra can be found in \cite{bar:2007,moretti:2013,slawny:1972,manuceau:1973}.

Functional derivatives have always played an important role in physics but on infinite dimensional spaces, which appear naturally in quantum field theory, they are very subtle.
We introduce two different notions of functional derivatives in the third section (\cref{sec:derivatives}): the directional, or Gâteaux, derivative and the Fréchet derivative.
We will show that the two derivatives are closely related.
Proofs and more information on the directional derivative can be found in \cite{hamilton:1982,neeb:2005}.

In the fourth section (\cref{sec:fixed_point}) we will pick apart the Banach fixed-point theorem and prove several statements on the existence and uniqueness of fixed-points and their properties.
These results will form the basis of the proof existence of solutions to the semiclassical Eisntein equation to be presented in \cref{cha:einstein_solutions}.
Some of the results presented in the fourth section are already shown in~\cite{pinamonti:2015} by Pinamonti and the author.

The theory of distributions plays a fundamental role in quantum field theory: quantum fields can be understood as distributions.
Therefore, we will discuss in detail distributions and microlocal analysis in the fifth section (\cref{sec:microlocal}).
That is, we will review the basic definition of distributions and distributional sections on manifolds, the nuclearity property of the associated function spaces and the Schwartz kernel theorem, the Fourier transform and Schwartz functions and distributions, the wavefront set of distribution and distributional sections, including its behaviour under various operations such as pullbacks, and finally the important propagation of singularities theorem.
Good references for this section are the books by Hörmander \cite{hormander:1985,hormander:1985a,hormander:1990,hormander:1983} and also \cite{treves:1967,strohmaier:2009}.
Several recent results on the properties of spaces of distributions may be found in~\cite{dabrowski:2014}.
An excellent introduction to the wavefront set with several examples is \cite{brouder:2014} by the same author.

In the last section of this chapter (\cref{sec:wave}) we discuss wave equations.
Since bosonic quantum field typically satisfy a wave equation and the Dirac-type equations satisfied by fermionic quantum fields are closely related, an understanding of the solutions of these equations is very important.
With some minor modifications we mostly follow \cite{bar:2012,bar:2007} to introduce the advanced, retarded and causal propagators of normally and pre-normally hyperbolic differential operators and their relation to the Cauchy problem.
Several extensions of the results in \cite{bar:2012,bar:2007}, that are also partially stated here, can be found in~\cite{wrochna:2013,muhlhoff:2011}.

%-- Topological vector spaces -------------------------------------------------%
\section{Topological vector spaces}
\label{sec:topological_vector_spaces}

\IdxRanBegin{topology}

\subsection{Topology}
\label{sub:tvs_topology}

A \IdxMain{top-space}\emph{topological space} is a set~$X$ of points with a notion of neighbourhoods.
More concretely, besides~$X$ it consists of a collection~$\tau$ of subsets, the \IdxMain{open-set}\emph{open sets}, such that
\begin{enumerate}[(a)]
  \item both $\emptyset$ and $X$ are open,
  \item the union of any collection of open sets is open,
  \item any finite intersection of open sets is open.
\end{enumerate}
% \IdxMain{topology}
We call the collection~$\tau$ the \emph{topology} of~$X$; examples are illustrated in \cref{fig:hausdorff}.

There are two topologies that can be defined for every set.
The \IdxMain{discr-top}\emph{discrete topology} of a set contains all its subsets, whereas the \IdxMain{triv-top}\emph{trivial topology} consists only of the empty set and the set itself.

Given two topologies~$\tau$ and~$\tau'$ on the same set, we can compare them:
If $\tau \subset \tau'$, we say that $\tau'$ is \IdxMain{fine-top}\emph{finer} than $\tau$ and that $\tau$ is \IdxMain{coarse-top}\emph{coarser} than $\tau'$.
It follows that for any set the trivial topology and the discrete topology are respectively the coarsest and finest possible topology.

The complements of the open sets are the \IdxMain{closed-set}\emph{closed sets}.
By the de Morgan laws, they have the following properties: both $\emptyset$ and $X$ are close, the intersection of any collection of closed sets is closed and any finite union of closed sets is closed.
It is possible for a set to be both closed and open or neither.
If the only sets in~$X$ that are both open and closed are $\emptyset$ and $X$, then $X$ is \IdxMain{conn-set}\emph{connected}.

The \IdxMain{set-closure}\emph{closure} $\closure U$ of a set $U \subset X$ is the intersection of all closed set that contain $U$.
The subset $U$ is called \IdxMain{set-dense}\emph{dense} in $X$ if its closure is $X$: $\closure U = X$.

A \IdxMain{neighbourhood}\emph{neighbourhood} of a point~$x \in X$ is an open set~$U \in \tau$ that contains~$x$.
If for each pair of distinct points~$x, y$ in a topological space~$X$ there exist disjoint neighbourhoods~$U$ and~$V$ of~$x$ and~$y$, then it is called a \IdxMain{hausdorff}\emph{Hausdorff space} (\cref{fig:hausdorff}).
Obviously, endowing a set with the discrete topology turns it into a Hausdorff space.

\begin{figure}
  \centering
  \includegraphics{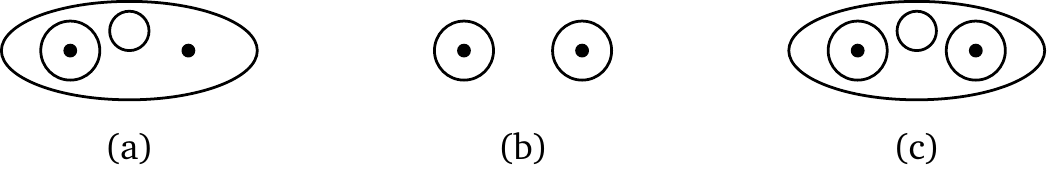}
  \caption{Collections of subsets of two points: (a) is a topology, (b) is not a topology (the empty set and the whole set are missing), (c) is the discrete topology (thus a Hausdorff topology).}
  \label{fig:hausdorff}
\end{figure}

A \IdxMain{top-basis}\emph{basis} (or \emph{base}) of a topological space~$X$ is a collection~$\mathfrak{B}$ of open sets in~$X$ such that every open set in~$X$ can be written as the union of elements of~$\mathfrak{B}$; one can say that the topology of~$X$ is generated by~$\mathfrak{B}$.
If the basis of $X$ is countable, we say that $X$ is \IdxMain{snd-count}\emph{second-countable}.
A \IdxMain{top-basis-loc}\emph{local basis} $\mathfrak{B}(x)$ of a point~$x \in X$ is defined as a collection of neighbourhoods of~$x$ such that every neighbourhood of~$x$ is a superset of an element of the local basis.
The union of all local bases is a basis.
If every point has a countable basis, we say that $X$ is \IdxMain{fst-count}\emph{first-countable}.
Clearly, a second-countable space is first-countable but the implication cannot be reversed.

We say that a map $f : X \to Y$ between two topological spaces~$X, Y$ is \IdxMain{continuous}\emph{continuous} if for all~$x \in X$ and all neighbourhoods~$V$ of~$f(x)$ there is a neighbourhood~$U$ of~$x$ such that $f(U) \subset V$.
The \IdxMain{cont-fun-space}\emph{space of all continuous maps} between~$X$ and~$Y$ is denoted $C(X, Y)$ or $C^0(X, Y)$ and by $C(X) = C^0(X)$ if $Y = \RR$.
A bijective map between two topological spaces~$X, Y$ is a \IdxMain{homeomorph}\emph{homeomorphism} is both $f$ and $f^{-1}$ are continuous.
$X$ and $Y$ are then called \emph{homeomorphic}, \ie, they are topologically equivalent.

Let $X$ be a set and $Y_i, i \in I$ a family of topological spaces with topologies~$\tau_i$.
Given maps $f_i : X \to Y_i$, the \IdxMain{init-top}\emph{initial topology} on~$X$ is the \emph{coarsest} topology such that the $f_i$ are continuous.
It is generated by the finite intersections of $\{ f_i^{-1}(U) \mid U \in \tau_i \}$.
Examples of the initial topology are the \IdxMain{sub-top}\emph{subspace topology}, \ie, the topology induced on a subspace~$X \subset Y$ by the inclusion map $\iota : X \hookrightarrow Y$, and the \IdxMain{prod-top}\emph{product topology}, \ie, the topology induced on a product space $X = \prod_{i \in I} Y_i$ by the projections $\pi_i : X \to Y_i$.
Conversely, given maps $f_i : Y_i \to X$, the \IdxMain{fin-top}\emph{final topology} on~$X$ is the \emph{finest} topology such that each $f_i$ is continuous.
It is given as
\begin{equation*}
  \tau = \big\{ U \subset X \;\big|\; f_i^{-1}(U) \in \tau_i\; \forall i \in I \big\}.
\end{equation*}
Important examples are the \IdxMain{quot-top}\emph{quotient topology} on a quotient space $X = Y / \sim$ with the map given by the canonical projection $Y \to Y / \sim$ and the \IdxMain{dsum-top}\emph{direct sum topology} on the direct sum $X = \sum_{i \in I} Y_i$ given by the canonical injections $Y_i \to X$.

Another application for the initial topology is the topology induced by a pseudometric:
A \IdxMain{pseudometric}\emph{pseudometric} on a set~$X$ is a map~$\mapsto d(\cdot\,,\cdot) : X \times X \to \RR$ such that for all $x, y, z \in X$:
\begin{enumerate}[(a)]
  \item $d(x, x) = 0$,
  \item $d(x, y) = d(y, x)$ (symmetry),
  \item $d(x, z) \leq d(x, y) + d(y, x)$ (triangle inequality);
\end{enumerate}
the set $X$ together with $d$ is a \IdxMain{pseudom-space}\emph{pseudometric space}.
If the pseudometric satisfies $d(x, y) > 0$ (positivity) for all $x \neq y$, then it is a \IdxMain{metric}\emph{metric} and $X$ is a \IdxMain{metric-space}\emph{metric space}.
The pseudometric on~$X$ induces the initial topology on~$X$ which is generated by the open balls around each point:
\begin{equation*}
  \mathfrak{B}(y) = \{ x \in X \mid d(x, y) < r \}.
\end{equation*}
A map $f : X \to Y$ between two pseudometric spaces with pseudometrics $d_X, d_Y$ is called an \IdxMain{isometry}\emph{isometry} if $d_X(x, y) = d_Y(f(x), f(y))$ for all $x, y \in X$.

Given a topological space~$X$, a sequence of points~$(x_n)$ in~$X$ \IdxMain{converge}\emph{converges} to a point~$x \in X$ if every neighbourhood~$U$ of~$x$ contains all but finitely many elements of the sequence.
If $X$ is a pseudometric space, it is called \IdxMain{complete}\emph{complete} if every Cauchy sequence with respect to its pseudometric converges to a point in~$X$.
Every pseudometric space~$X$ that is not complete can be completed.
Namely, the \IdxMain{completion}\emph{completion} of~$X$ is the (essentially unique) complete pseudometric space $\overline{X}$ of which $X$ is a dense isometric subspace.

A topological (sub)space $X$ is called \IdxMain{compact}\emph{compact} if each of its \emph{open covers}, \viz, a collection of subsets whose union contains $X$ as a subset, has a finite subcollection that also contains $X$.
If $X$ is compact, then every subset $Y \subset X$ is also compact in the subspace topology.
The space $X$ is \IdxMain{locally-compact}\emph{locally compact} if every point in $X$ has a compact neighbourhood.
Furthermore, $X$ is called \IdxMain{sigma-compact}\emph{$\sigma$-compact} if it is the union of countably many compact subsets.
Any compact space is locally compact and $\sigma$-compact; the converse is, however, false.
One can also show that every second-countable and locally compact Hausdorff space (thus, in particular, every topological manifold) is $\sigma$-compact.
We say that a map $f : X \to Y$ between two topological spaces $X, Y$ is \IdxMain{proper-map}\emph{proper} if the preimage of every compact set in $Y$ is compact in $X$.

A related concept is that of boundedness.
A subset $U \subset X$ of a pseudometric space~$X$ is \IdxMain{bounded}\emph{bounded} if for each $x,y \in U$ there exists a $r$ such that $d(x, y) \leq r$.

\subsection{Locally convex topological vector spaces}
\label{sub:tvs_locally_convex}

A \IdxMain{tvs}\emph{topological vector space} is a $\KK$-vector space~$X$ such that the vector space operations of addition and scalar multiplication are (jointly) continuous with respect to the topology of~$X$.
Since addition is continuous and the topology therefore translation-invariant, the topology of~$X$ is completely determined by the local basis~$\mathfrak{B}(0)$ at the origin.
For the same reason one can show that $X$ is Hausdorff if and only if $\{0\}$ is closed.
Note that in a Hausdorff topological vector space every complete subset is also closed.
The (topological) vector space~$X$ is called \IdxMain{convex-tvs}\emph{convex} if $x, y \in X$ implies $\lambda x + (1-\lambda) y \in X$ for all $\lambda \in [0,1]$.

Let $X, Y$ be topological vector spaces and $W \subset X$ a subset.
A mapping $f : W \to Y$ is \IdxMain{uniform-cont}\emph{uniformly continuous} if to every neighbourhood~$V$ of zero in~$Y$ there exists a zero neighbourhood $U \subset X$ such that for all $x, y \in W$
\begin{equation*}
  x - y \in U \implies f(x) - f(y) \in V.
\end{equation*}
Every uniformly continuous map is already continuous but the converse is not true.
However, if $f$ is linear and also $W$ is a vector subspace, then continuity also implies uniform continuity.
Moreover, if $W$ is a dense subset, then to every uniformly continuous map $f : W \to Y$ there exists a unique continuous map $\wbar{f} : X \to Y$ which extends $f$.
The extension of a linear map from a dense vector subspace is even uniformly continuous and linear.

It is usually desirable to have topological vector spaces with additional structures:
A \IdxMain{seminorm}\emph{seminorm} on a $\KK$-vector space~$X$ is a map~$\norm{\,\cdot\,} : X \to \RR$ such that for all $x, y \in X$ and $\lambda \in \KK$:
\begin{enumerate}[(a)]
  \item $\norm{x} \geq 0$ (positive-semidefiniteness),
  \item $\norm{\lambda x} = \abs{\lambda} \norm{x}$ (absolute homogeneity),
  \item $\norm{x + y} \leq \norm{x} + \norm{y}$ (triangle inequality).
\end{enumerate}
If the seminorm satisfies $\norm{x} > 0$ (positivity) for all $x \neq 0$, then it is a \IdxMain{norm}\emph{norm} and the vector space is called \IdxMain{normed-space}\emph{normed space}.
Note that each \optword{semi}norm induces a (translation-invariant) \optword{pseudo}metric $d(x, y) = \norm{x - y}$.

Analogously to the pseudometrics, given a \IdxMain{family-seminorms}\emph{family of seminorms}~$\{\norm{\,\cdot\,}_i\}_{i \in I}$ on a vector space~$X$, they induce the initial topology on $X$.
More explicitly, the topology is generated by all finite intersections of $\{ x \in X \mid \norm{x}_i < r \}$, the open balls around the origin.\footnote{Conversely, given a basis of the origin it is possible to construct a family of seminorms from its elements.}
If $I$ is countable, we can assume that $I \subset \NN$ and the topology above is the same topology as the one induced by the metric
\begin{equation*}
  d(x, y) = \sum_{k \in I} \frac{1}{2^k} \frac{\norm{x - y}_k}{1 + \norm{x - y}_k},
\end{equation*}
where the factors $2^{-k}$ may be replaced by the coefficients of any convergent series.

We say that a family of seminorms~$\{\norm{\,\cdot\,}_i\}$ on a vector space~$X$ is \IdxMain{seminorm-sep}\emph{separating} if for every nonzero~$x \in X$ there exists an~$i$ such that $\norm{x}_i > 0$.
It is immediate that a vector space with topology induced by a family of seminorms is Hausdorff if and only if the seminorms are separating.

Let $X$ be vector space endowed with a family of seminorms~$\{\norm{\,\cdot\,}_i\}_{i \in I}$. Then we can define the following topological vector spaces in order of generality:
\begin{enumerate}
  \item if the family of seminorms is separating, then $X$ is a \IdxMain{lcs}\emph{locally convex (topological vector) space};
  \item if, in addition, $I$ is countable and $X$ is complete with respect to each of its seminorms, then $X$ is a \IdxMain{frechet}\emph{Fréchet space};
  \item if, in addition, the family of seminorms consists of only one \emph{norm}, then $X$ is a \IdxMain{banach}\emph{Banach space}.
\end{enumerate}

Finally note that on all these spaces above the \IdxMain{hahn-banach}\emph{Hahn--Banach theorem} can be applied.
That is, given a $\KK$-vector space $X$ with a seminorm $\norm{\,\cdot\,}$ and a linear form $f : U \to \KK$ on a vector subspace $U$ such that $\abs{f(x)} \leq \norm{x}$ for all $x \in U$, there exists a (generally non-unique) linear form $\wbar{f} : X \to \KK$ which extends $f$ such that $\abs{\wbar{f}(x)} \leq \norm{x}$.

\subsection{Topologies on function spaces}
\IdxRanBegin{fun-tops}

Important vector spaces are subspaces of the \IdxMain{fun-space}\emph{space of functions}
\begin{equation*}
  Y^X \defn \{ f : X \to Y \}
\end{equation*}
between a set~$X$ and a topological space~$Y$.
$Y^X$ can be equipped with the \IdxMain{pointwise-top}\emph{topology of pointwise convergence}, which is just the product topology with the projections $\pi_x : Y^X \to Y, f \mapsto f(x)$.
In this topology a sequence of functions~$(f_n)$ converges to some~$f$ if and only if each~$f_n(x)$ converges to~$f(x)$ at all~$x$.

If $(Y, d)$ is a pseudometric space, the space of functions can be equipped another topology, the \IdxMain{uniform-top}\emph{topology of uniform convergence}.
In this case, a subspace $Z \subset Y^X$ can be endowed with the pseudometric
\begin{equation*}
  d(f, g) = \sup_{x \in X} d\big(f(x), g(x)\big)
\end{equation*}
for all $f, g \in Z$, which induces a topology for $Z$.
Note, however, that $Z$ with this topology is not really a topological vector space because multiplication will fail to be continuous unless $Z$ is a subset of the bounded functions.
A sequence of functions~$(f_n)$ converges to some~$f$ if and only if for every $\varepsilon > 0$ there exists a $N$ such that $d(f_n, f) < \varepsilon$ for all $n \geq N$.

If $X$ is a topological space, yet another topology on subspaces $Z \subset Y^X$ is the \IdxMain{c-o-top}\emph{compact-open topology}.
For all compact $K \subset X$ and open $U \subset Y$ it is generated by the finite intersections of
\begin{equation*}
  \{ f \in Z \mid f(K) \subset U \},
\end{equation*}
\ie, the set of functions that carry compact subsets into open subsets.
The compact-open topology is finer than the topology of pointwise convergence.
If $(Y, d)$ is a pseudometric space, the compact-open topology is the initial topology induced by the pseudometric on compact subsets, \ie, it is generated by the finite intersections of
\begin{equation*}
  \big\{ f \in Z \mid \sup\nolimits_{x \in K} d\big(f(x), 0\big) < r \big\}.
\end{equation*}
Therefore it is also called the \IdxMain{uni-compacta-top}\emph{topology of uniform convergence in compacta} and, if $X$ is compact, it is the same as the topology of uniform convergence.
It follows that a sequence of functions~$(f_n)$ converges to some~$f$ if and only if for every $\varepsilon > 0$ and compact $K \subset X$ there exists a $N$ such that $d(f_n(x), f(x)) < \varepsilon$ for all $n \geq N$ and $x \in K$.

Again, suppose that $X, Y$ are topological vector spaces and $Z \subset Y^X$.
$Z$ is called \IdxMain{equicontinuous}\emph{equicontinuous} if for every neighbourhood of the origin $U \subset X$ there exists a neighbourhood of the origin $V \subset X$ such that $f(U) \subset V$ for every $f \in Z$.
Equicontinuity for a set of one element is of course the same as continuity.

\IdxRanEnd{fun-tops}

\subsection{Duality}

A \IdxMain{duality}\emph{duality} or \emph{dual pairing} $\innerp{Y}{X}$ is a triple $(X, Y, \innerp{\cdot\,}{\cdot})$ of two vector spaces $X, Y$ and a \emph{non-degenerate} bilinear form $\innerp{\cdot\,}{\cdot} : Y \times X \to \KK$, \ie,
\begin{align*}
  \innerp{y}{x} & = 0 \text{ for all $y \in Y$ implies } x = 0 \\
  \innerp{y}{x} & = 0 \text{ for all $x \in X$ implies } y = 0.
\end{align*}
The standard example of a duality is that between a vector space~$X$ and its algebraic dual~$X^*$, where the pairing is given by the \IdxMain{can-bil-form}\emph{canonical bilinear form}
\begin{equation*}
  \innerp{\cdot\,}{\cdot} : X^* \times X \to \KK, (f, x) \mapsto \innerp{f}{x} \defn f(x).
\end{equation*}
A more important example is that of the \IdxMain{top-dual}\emph{topological dual}~$X' \subset X^*$ of a topological space~$X$, which consists of all continuous linear maps.
Note that the pairing $\innerp{X'}{X}$ is not a proper duality unless $X$ is Hausdorff because the restriction of the canonical bilinear form to $X' \times X$ is non-degenerate if and only if $X$ is Hausdorff.

For each $x \in X$ the map $x \mapsto \innerp{y}{x}$ gives an injective map of~$X$ into $Y^*$ and an analogous construction embeds $Y$ into~$X^*$.
In the following, the identification of~$X$ with a subspace of~$Y^*$ and~$Y$ with a subspace of~$X^*$ will always be tacitly assumed unless otherwise noted.

In particular, $X$ is a subspace of~$\KK^Y$ and can therefore be equipped with pointwise topology.
This locally convex Hausdorff topology is called the \IdxMain{weak-top}\emph{weak topology} on~$X$ with respect to $\innerp{Y}{X}$; statements for weak topology will often be indicated by the adjective ``weakly''.
The weak topology is the coarsest topology such that $x \mapsto \innerp{y}{x}$ is continuous for all $y \in Y$ and one finds $Y = X'$ with respect to the weak topology on~$X$.
Moreover, if $Y$ is locally convex, the seminorms on~$Y$ yield \IdxMain{dual-seminorm}\emph{dual seminorms} on~$X$ given by
\begin{equation*}
  \norm{x}_i \defn \sup \big\{ \abs{\innerp{y}{x}} \;\big|\; y \in Y \text{ with } \norm{y}_i \leq 1 \big\},
\end{equation*}
which also induce the weak topology.

The statements above can also be made with the role of~$X$ and~$Y$ interchanged to introduce the weak topology on~$Y$ with respect to $\innerp{Y}{X}$.
In particular, if $Y = X'$ and $\innerp{\cdot\,}{\cdot}$ the canonical bilinear form, then $X'$ with the weak topology is called the \IdxMain{weak-dual}\emph{weak dual}.
Furthermore, given a subset $Y \subset X^*$, then the induced pairing between~$Y$ and~$X$ is non-degenerate if and only if $Y$ is weakly dense in~$X^*$.
Thus any~$Y$ is weakly complete if and only if $Y = X^*$.

Another (locally convex Hausdorff) topology on $X$ induced by a duality $\innerp{Y}{X}$ with a locally convex space~$Y$ is the \IdxMain{strong-top}\emph{strong topology}, which is the topology of uniform convergence on the bounded subsets of~$Y$; statements for strong topology will be often be indicated by the adjective ``strongly''.
It is induced by the family of seminorms
\begin{equation*}
  \norm{x}_B = \sup_{y \in B}\, \abs{\innerp{y}{x}}
\end{equation*}
for each bounded set $B \subset Y$.
Again, we can interchange the role of~$X$ and~$Y$ and call $X'$ endowed with the strong topology induced by the canonical pairing the \IdxMain{strong-dual}\emph{strong dual}.

If we equip $X$ with the strong topology with respect to $\innerp{Y}{X}$, then the map $x \mapsto \innerp{y}{x}$ will not be continuous for any $y \in Y$.
The finest topology on~$X$ such that this map is continuous is called the \Idx{mackey-top}\emph{Mackey topology} but it will not concern us here any further.

Finally, note that the dual of a \Idx{banach-dual}Banach space is \emph{always} a Banach space, but the dual of a \Idx{frechet-dual}Fréchet space that is not Banach is \emph{never} a Fréchet space.

\subsection{Tensor products on locally convex spaces}
\label{sub:tvs_nuclear}
\IdxRanBegin{tensor-tops}

Given two locally convex topological vector spaces~$X, Y$ there are many different ways to define a family of seminorms for the space~$X \otimes Y$.
Therefore there is no natural topology for~$X \otimes Y$ if $X$ or $Y$ is infinite-dimensional, whence one speaks of different topological tensor products.
The most common topological tensor products are the projective and injective tensor product introduced below.

The \IdxMain{proj-tensor}\emph{projective tensor product topology} equips the algebraic tensor product~$X \otimes Y$ with the finest topology such that $\otimes : X \times Y \to X \otimes Y$ is \emph{jointly} continuous.
That is, it is the final topology defined by the projections $\pi_X : X \to X \otimes Y$, $\pi_Y : Y \to X \otimes Y$.
Equivalently, the topology is induced by the seminorms
\begin{equation*}
  \norm{u}_{i,j} = \inf \Big\{ \sum\nolimits_k \norm{x_k}_i \norm{y_k}_j \;\Big|\; u = \sum\nolimits_k x_k \otimes y_k \Big\}
\end{equation*}
for all $u \in X \otimes Y$ and where the infimum runs over all representations of~$u$.
The resulting locally convex space is usually denoted $X \otimes_\pi Y$ and its completion $X \mathbin{\widehat{\otimes}_\pi} Y$.

A coarser topology is defined by the \IdxMain{inj-tensor}\emph{injective tensor product topology}; it is the finest topology such that $\otimes : X \times Y \to X \otimes Y$ is \emph{separately} continuous.
Let $X', Y'$ be the weak duals of~$X, Y$.
Note that $X \otimes Y$ can be embedded into the space of bilinear \emph{separately} continuous maps $X' \times Y' \to \KK$, denoted $B(X, Y)$, with the topology of uniform convergence $U \times V$ on all equicontinuous sets $U \subset X$ and $V \subset Y$.
That is, the topology is generated for all $U \times V$ and open $I \subset \KK$ by the finite intersections of
\begin{equation*}
  \{ f \in B(X, Y) \mid f(U \times V) \subset I \}.
\end{equation*}
$X \otimes Y$ can now be endowed with the corresponding subspace topology.
Seminorms that induce this topology are given by
\begin{equation*}
  \norm{u}_{i,j} = \sup \big\{ \abs{(f \otimes g)(u)} \;\big|\; f \in X', g \in Y' \text{ such that } \norm{f}_i = \norm{g}_j = 1 \big\}.
\end{equation*}
The space~$X \otimes Y$ equipped with the injective topology is usually denoted $X \otimes_\varepsilon Y$ and its completion is denoted $X \mathbin{\widehat{\otimes}_\varepsilon} Y$.
Observe that $C(X, Y) \simeq C(X) \mathbin{\widehat{\otimes}_\varepsilon} Y$ if $Y$ is complete.

Locally convex spaces on which the injective and projective tensor product agree are called nuclear.
More precisely, we say that a locally convex space~$X$ is \IdxMain{nuclear-space}\emph{nuclear} if
\begin{equation*}
  X \otimes_\varepsilon Y = X \otimes_\pi Y
  \quad \text{or, equivalently,} \quad
  X \mathbin{\widehat{\otimes}_\varepsilon} Y = X \mathbin{\widehat{\otimes}_\pi} Y
\end{equation*}
for every locally convex space\footnote{Actually it is sufficient to check equality for $Y = \ell_1$, see below.}~$Y$ in which case we simply write $X \otimes Y$.
If both $X$ and $Y$ are nuclear, then also $X \otimes Y$ is nuclear.
Moreover, if a \Idx{nuclear-sub}subspace of a nuclear space is nuclear and the \Idx{nuclear-quot}quotient space of a nuclear space by a closed subspace is nuclear.

A more useful characterisation of nuclear spaces is in terms of summable sequences.
% Denote by~$\ell_1[X]$ the $X$-valued \emph{weakly summable sequences}, \ie, the set of sequences $\{x_i\}$ such that $\sum_i \abs{f(x_i)} < \infty$ for all $f \in X'$.
Denote by~$\ell_1(X)$ the $X$-valued \IdxMain{summable}\emph{summable sequences}, \ie, the set of sequences $(x_n)$ in~$X$ such that all unordered partial sums $\sum_{n \in I \subset \NN} x_n$ converge in~$X$.
Further, denote by~$\ell_1\{X\}$ the $X$-valued \IdxMain{abs-summable}\emph{absolutely summable sequences}, \ie, the set of sequences $(x_n)$ in~$X$ such that $\sum_n \norm{x_n}_i < \infty$ for all seminorms $\norm{\,\cdot\,}_i$ of $X$.
Then $X$ is nuclear if and only if
\begin{equation*}
  \ell_1(X) = \ell_1\{X\}
\end{equation*}
and hence, by the above observation, both sides are equal to $\ell_1 \otimes X = \ell_1 \otimes_\varepsilon X = \ell_1 \otimes_\pi X$.
In other words, $X$ is nuclear if and only if every summable sequence in $X$ is already absolutely summable.
Nuclear spaces are therefore very similar to finite-dimensional spaces and, while \emph{every} finite-dimensional locally convex space is nuclear, \emph{no} infinite-dimensional normed space is.

We finish this section by stating what can be called the \IdxMain{abstract-ker-thm}\emph{abstract kernel theorem} for Fréchet spaces:
\begin{equation*}
  (X \otimes Y)' \simeq X' \otimes Y'
  \quad \text{and} \quad
  (X \mathbin{\widehat\otimes} Y)' \simeq X' \mathbin{\widehat\otimes} Y'
\end{equation*}
for every $X, Y$ such that $X$ (or $Y$) is nuclear and where all duals are strong duals.

\IdxRanEnd{tensor-tops}
\IdxRanEnd{topology}

%- TOPOLOGICAL *-ALGEBRAS -----------------------------------------------------%
\section{Topological \texorpdfstring{${}^*$-algebras}{*-algebras}}
\label{sec:algebra}

\IdxRanBegin{star-algebra}

A \emph{topological ${}^*$-algebra} is a topological algebra $\alg{A}$ over $\CC$, \ie, a topological $\CC$-vector space with a separately continuous ring multiplication, together with a continuous \IdxMain{star-invo}\emph{involution} ${}^*$.
That is, there is an automorphism
\begin{equation*}
  {}^* : \alg{A} \to \alg{A},\;\; x \mapsto x^*,
\end{equation*}
which is antilinear and involutive such that
\begin{enumerate}[(a)]
  \item $(a x + b y)^* = \conj{a} x^* + \conj{b} y^*$,
  \item $(x y)^* = y^* x^*$,
  \item $(x^*)^* = x$
\end{enumerate}
for all $x, y \in \alg{A}$ and $a, b \in \CC$.
If, in addition, $\alg{A}$ has a multiplicative unit $\one$, we say that $\alg{A}$ is a \Idx{unital-alg}\emph{unital} ${}^*$-algebra.
Elements $x, y$ of the algebra $\alg{A}$ are called\IdxMain{elem-adjoint}\IdxMain{elem-self-adjoint}\IdxMain{elem-normal}\IdxMain{elem-unitary}
\begin{alignat*}{2}
  \text{\emph{adjoint}}
  & \quad \text{if} \quad &
  & x^* = y,
  \\
  \text{\emph{self-adjoint}}
  & \quad \text{if} \quad &
  & x^* = x,
  \\
  \text{\emph{normal}}
  & \quad \text{if} \quad &
  & x^* x = x x^*,
  \\
  \text{\emph{unitary}}
  & \quad \text{if} \quad &
  & x^* x = \one = x x^*,
\end{alignat*}
where unitarity obviously requires the existence of a unit element.
Note that $\mathbf{1}$ is always self-adjoint.

A ${}^*$-subalgebra $\alg{I} \subset \alg{A}$ is called a left (right) \IdxMain{star-ideal}\emph{${}^*$-ideal} if $y x$ (resp. $x y$) is in $\alg{I}$ for all $y \in \alg{I}$ and $x \in \alg{A}$.
If the subalgebra is both a left and right ${}^*$-ideal, it is just called a (two-sided) ${}^*$-ideal.
It follows that an ideal $\alg{I}$ of $\alg{A}$ is a ${}^*$-ideal if and only if $\alg{I}^* = \alg{I}$.

The homomorphisms that arise between ${}^*$-algebras, called \IdxMain{star-hom}\emph{${}^*$\hyp{}homomorphisms}, are those that preserve in addition to the multiplicative also the involutive structure, \ie, a map $\alpha : \alg{A} \to \alg{B}$ is a ${}^*$\hyp{}homomorphisms if it is an algebra-homomorphism and $\alpha(x^*) = \alpha(x)^*$ for all $x \in \alg{A}$.
If the ${}^*$-algebras are unital, we also demand that ${}^*$-homomorphisms be unit-preserving.

Often one needs a ${}^*$-algebra which also has the structure of a normed vector space.
In the case of ${}^*$-algebras, it makes sense to require the norm to satisfy an additional property:
A norm $\norm{\,\cdot\,} : \alg{A} \to \RR$ is said to be a \IdxMain{cstar-norm}\emph{$C^*$-norm} if
\begin{equation*}
  \norm{x^* x} = \norm{x}^2
\end{equation*}
for all $x \in \alg{A}$.
This differs from some definitions of $C^*$-norms because, in fact, every $C^*$-norm is automatically a ${}^*$-isomorphism and submultiplicative, \ie,
\begin{equation*}
  \norm{x^*} = \norm{x}
  \quad \text{and} \quad
  \norm{x y} \leq \norm{x} \norm{y}
\end{equation*}
for all $x, y \in \alg{A}$ \cite{sebestyen:1979}.

If a ${}^*$-algebra $\alg{A}$ comes equipped with such a $C^*$-norm $\norm{\,\cdot\,}$ that turns $\alg{A}$ into a Banach space, then it is called a \IdxMain{cstar-alg}\emph{$C^*$-algebra}.
In a $C^*$-algebra, ring multiplication and inversion are continuous operations with respect to the norm; the continuity of addition, scalar multiplication and involution are obvious.
The condition for a unital $C^*$-algebra to have a $C^*$-norm imposes such strong conditions on its algebraic structure that the algebra uniquely determines the norm.
Namely,
\begin{equation*}
  \norm{x}^2 = \norm{x^* x} = \sup\, \big\{ \abs{\lambda} \;\big|\; x^* x - \lambda\one \text{ is not invertible} \big\}
\end{equation*}
for every $x \in \alg{A}$.
A ${}^*$-homomorphism $\alpha : \alg{A} \to \alg{B}$ between two unital $C^*$-algebras is thus always norm-decreasing: $\norm{\alpha(x)} \leq \norm{x}$.

\subsection{States}
\label{sub:alg_states}

Given a ${}^*$-algebra~$\alg{A}$, one can consider its algebraic dual, the space of linear functionals on~$\alg{A}$.
A linear functional $\omega \in \alg{A}^*$ on a unital \mbox{${}^*$-al}gebra~$\alg{A}$ is \emph{positive} if it satisfies
\begin{equation*}
  \omega(x^* x) \geq 0
\end{equation*}
for all $x \in \alg{A}$.
If, $\alg{A}$ is unital and $\omega(\one) = 1$, we say that $\omega$ is \emph{normalized}.
A functional $\omega$ that is both positive and normalized is called a \IdxMain{state}\emph{state}.
If the ${}^*$-algebra~$\alg{A}$ comes equipped with a topology, we always assume that $\omega$ is continuous, \ie, we consider the topological dual $\alg{A}'$ instead of the algebraic one; an algebraic state on a $C^*$-algebra is automatically continuous with respect to the $C^*$-norm.

Given positive $\omega \in \alg{A}'$, it follows that for all $x, y \in \alg{A}$
\begin{align*}
  \omega(x^* y) & = \conj{\omega(y^* x)}, \\
  \abs{\omega(x^* y)}^2 & \leq \omega(x^* x) \omega(y^* y),
\end{align*}
where the second line is called the \emph{Cauchy--Schwarz inequality}.
If $\alg{A}$ is unital, the first equation implies that every positive $\omega$ is hermitian: $\omega(x^*) = \conj{\omega(x)}$.

A state $\omega$ is \IdxMain{pure-state}\emph{pure} if every other state $\eta$ that is majorized by it, $\omega(x^* x) \geq \eta(x^* x)$, is of the form $\eta = \lambda \omega$ with $\lambda \in [0,1]$.
Consequently, a pure state cannot be written as the convex sum of two other states.
States that are not pure are called \IdxMain{mixed-state}\emph{mixed}.

The positive linear functionals equips $\alg{A}$ with a degenerate inner product via the antilinear pairing $\innerp{x}{y} = \omega(x^* y)$ which can be turned into a pre-Hilbert space by taking the quotient by the degenerate elements.
This is the essential content of the famous \IdxMain{gns}\emph{Gel'fand--Naimark--Segal construction}, usually abbreviated as \emph{GNS construction}, which we will state after the following definition.

A \IdxMain{star-rep}\emph{${}^*$-representation}~$\pi$ of a ${}^*$-algebra~$\alg{A}$ is a ${}^*$-homomorphism into the $C^*$-algebra of linear operators on a common dense (with respect to the norm $\norm{\,\cdot\,} = \innerp{\cdot\,}{\cdot}^{1/2}$ on the Hilbert space) domain~$\mathcal{D}$ of a Hilbert space~$\mathcal{H}$.
Note that the ${}^*$-representation~$\pi$ is continuous with respect to the \IdxMain{uni-op-top}\emph{uniform operator topology}, \ie, the topology induced by the \IdxMain{op-norm}\emph{operator norm}
\begin{equation*}
  \norm{T}_{\mathrm{op}} = \sup\, \big\{ \norm{T x} \;\big|\; x \in \mathcal{D} \text{ with } \norm{x} \leq 1 \big\}.
\end{equation*}
Moreover, if the domain~$\mathcal{D}$ is complete in the \IdxMain{graph-top}\emph{graph topology} induced by the family of seminorms $\norm{\,\cdot\,}_x = \norm{\pi(x)\,\cdot\,}$, we say the the ${}^*$-representation~$\pi$ is \IdxMain{closed-star-rep}\emph{closed}.

If there exists a vector $\Omega \in \mathcal{H}$ such that $\pi(\alg{A}) \Omega = \{ \pi(x) \Omega \mid x \in \alg{A} \}$ is dense in~$\mathcal{H}$, then the ${}^*$-representation is called \IdxMain{cyclic-star-rep}\emph{cyclic} and $\Omega$ \IdxMain{cyclic-vec}\emph{cylic vector}.
If $\pi(\alg{A}) \Omega$ is even dense in~$\mathcal{D}$ in the graph topology, then $\pi$ and $\Omega$ are called \IdxMain{strong-cyclic-star-rep}\emph{strongly cyclic}.
The state $\omega(x) = \innerp{\pi(x) \Omega}{\Omega}$ defined by a cyclic vector~$\Omega$ of a cyclic ${}^*$-representation~$\pi$ is pure if and only if the only subspaces left invariant by~$\pi(A)$ are the trivial ones.

\begin{theorem}[GNS construction]\label{thm:gns}
  Let $\omega$ be a state\footnote{Actually it is enough for $\omega$ to be a positive; the normalization is not necessary for the theorem to hold.} on a unital topological ${}^*$-algebra~$\alg{A}$.
  Then there exists a closed (weakly continuous) strongly cyclic \mbox{${}^*$-re}presentation $\pi$ of~$\alg{A}$ on a Hilbert space~$\mathcal{H}$ with inner product $\innerp{\cdot\,}{\cdot}$ and strongly cyclic vector~$\Omega$ such that
  \begin{equation*}
    \omega(x^* y) = \innerp{\pi(x) \Omega}{\pi(y) \Omega}
  \end{equation*}
  for all $x, y \in \alg{A}$.
  The representation~$\pi$ is unique up to unitary equivalence.
\end{theorem}

This is a standard theorem and a proof can be found in many places, \eg, in~\cite{powers:1971}.
Working with general ${}^*$-algebras, we have not excluded the case of ${}^*$\hyp{}representations onto unbounded operators.
For that reason it is not possible to uniquely extend the representation to the whole Hilbert space, and hence self-adjoint elements of the algebra might not be represented by self-adjoint operators but only symmetric operators.
These problems do not occur if one applies the GNS construction to $C^*$-algebras.

\subsection{Weyl algebra}
\label{sub:alg_weyl}

Set $V$ to be a $\RR$-vector space and $\sigma : V \times V \to \RR$ an antisymmetric bilinear form (\ie, a pre-symplectic form).\footnote{We follow \cite{manuceau:1973} and will not assume that $\sigma$ is non-degenerate. In fact it is sufficient to assume that $\sigma$ is linear in its first \emph{or} second argument.}
A \IdxMain{weyl-star-alg}\emph{Weyl ${}^*$-algebra}~$\alg{W}$ for $(V, \sigma)$ is a unital involutive algebra generated by (nonzero) \IdxMain{weyl-gen}\emph{Weyl generators}~$W$, \ie, symbols~$W(\cdot)$ labelled by the vectors in~$V$, which satisfy, for all $v, w \in V$, the relations
\begin{enumerate}[(a)]
  \item $W(v) W(w) = \exp\big( \tfrac{\im}{2} \sigma(v, w) \big) W(v + w)$,
  \item $W(v)^* = W(-v)$.
\end{enumerate}
Therefore the Weyl generators also have the following properties:
\begin{enumerate}[resume*]
  \item $W(0) = \one$,
  \item $W(v)^* = W(-v) = W(v)^{-1}$,
  \item $\big( W(v) \big)_{v \in V}$ are linearly independent.
\end{enumerate}
Since $\alg{W}$ is generated by unitaries, every ${}^*$-representation is necessarily by bounded operators.
Moreover, between two Weyl ${}^*$-algebras generated by Weyl generators~$W$ and~$W'$ for $(V, \sigma)$ there exists a unique ${}^*$-isomorphism $\alpha$ completely determined by $\alpha \circ W = W'$.

One can endow~$\alg{W}$ with a $C^*$-norm, the \IdxMain{min-reg-norm}\emph{minimal regular norm}
\begin{equation*}
  \norm{x} = \sup \big\{ \omega(x^* x)^{1/2} \;\big|\; \omega \text{ is a state on } \alg{W} \}
\end{equation*}
for all $x \in \alg{W}$.
If the bilinear form $\sigma$ is non-degenerate, one can show that all $C^*$-norm over~$\alg{W}$ are equal.
We call the completion~$\overline{\alg{W}}$ of a Weyl ${}^*$-algebra~$\alg{W}$, with respect to the minimal regular norm, the \IdxMain{weyl-cstar-alg}\emph{Weyl $C^*$-algebra}.
It is unique up to ${}^*$-isomorphism and, in particular, \IdxMain{simple-alg}\emph{simple}, \viz, it has no non-trivial closed ${}^*$-ideals, if and only if $\sigma$ is non-degenerate \cite{manuceau:1973}.

The map $\RR \ni \lambda \mapsto W(\lambda v)$ is not continuous in~$\overline{\alg{W}}$, because $\norm{W(v) - W(w)} = 2$ for all distinct $v, w \in V$ as a consequence of the spectral radius formula.
A ${}^*$\mbox{-re}presentation~$\pi$ of~$\overline{\alg{W}}$ on a Hilbert space~$\mathcal{H}$ is called \IdxMain{reg-star-rep}\emph{regular} if the one-parameter unitary groups
\begin{equation*}
  \lambda \mapsto (\pi \circ W)(\lambda v), \quad v \in V,
\end{equation*}
are strongly continuous.
If the ${}^*$-representation induced by a state on~$\overline{\alg{W}}$ is regular, we also call the state regular.
Invoking \emph{Stone's theorem} $(\pi \circ W)(\lambda v)$, we can find a family of self-adjoint operators $F(\cdot)$ on~$\mathcal{H}$, labelled by vectors in~$V$, such that
\begin{equation*}
  (\pi \circ W)(\lambda v) = \exp\big(\im \lambda F(v)\big);
\end{equation*}
the map~$F$ is called the \IdxMain{field-op}\emph{field operator} and is generally unbounded.

A \IdxMain{strong-reg-state}\emph{strongly regular state} \cite{bar:2012} is a regular state for which the operators $F(v)$, $v \in V$, have a common dense domain $\mathcal{D} \subset \mathcal{H}$, which is closed under the action of~$F$, and for which $v \mapsto F(v) w$ is continuous for fixed $w \in V$.
For strongly regular states the field operator is linear in its argument and thus a self-adjoint operator-valued distribution.

\IdxRanEnd{star-algebra}

%- DERIVATIVES ----------------------------------------------------------------%
\section{Derivatives}
\label{sec:derivatives}

Various different notions of derivatives on topological vector spaces exist in the literature, see~\cite{averbukh:1968} for a survey and history of the topic.
On infinite-dimensional spaces these derivatives are inequivalent and care must be taken to make precise which derivative is meant.
On Banach spaces there exists the notable example of the Fréchet derivative.
However, many spaces of interest in physics are not normed and so one must work with derivatives on more general spaces.
Below we will define a directional derivative in the sense of Gâteaux and later compare it with the Fréchet derivative on Banach spaces.

\subsection{The directional derivative}
\label{sub:der_directional}

Let $X, Y$ be two topological vector spaces and $U \subset X$ open.
The \IdxMain{dir-deriv}\emph{(directional) derivative} of a function $f : U \to Y$ at $x \in U$ in the direction $h \in X$ is defined as the map $\dirD f: U \times X \to Y$,
% NOTE: the notation $\dirD f(x;h)$ is meant to suggest that one can do differential geometry on topological vector spaces and manifolds over topological vector spaces where $(x;h)$ is then an element in the cotangent bundle.}
\begin{equation}\label{eq:directional}
  \dirD_h f(x) \defn \dirD f(x; h)
  \defn \lim_{\varepsilon \to 0} \frac{1}{\varepsilon} \big( f(x + \varepsilon h) - f(x) \big)
  = \left. \od{}{\varepsilon}\, f(x + \varepsilon h) \right|_{\varepsilon=0},
\end{equation}
if the limit exists.
In particular, if $f$ is a continuous linear function, then its derivative is $\dirD f(x; h) = f(h)$.

Note that the nomenclature here follows that of~\cite{hamilton:1982,neeb:2005} and differs from that in~\cite{pinamonti:2015}, where Pinamonti and the author called the same derivative \emph{Gâteaux derivative}.
The reason for this choice of a more neutral name is that the name ``Gâteaux derivative'' has sometimes been used for slightly different derivatives.
However, all definitions known to the author agree whenever the derivative is both linear and continuous in the direction of the derivative.

It should be clear from the definition of the directional derivative, that the ordinary and partial derivative are special cases of the directional derivative for functions from Euclidean~$\RR^n$ to~$\RR$ or~$\CC$.
Consequently the directional derivative is also closely related to the local form of the covariant derivative given by a connection on a vector bundle.

The function~$f$ is called \emph{differentiable at~$x$} if the limit exists for all $h \in X$ and simply \emph{differentiable} if it is differentiable at every $x \in U$.
Moreover, $f$ is \emph{continuously differentiable} or $\dirC^1$ on~$U$ if the map~$\dirD f$ is continuous (in the induced topology on~$U \times X$).
Higher derivatives may be defined recursively by
\begin{equation*}
  \dirD_{h_n} \dotsm\, \dirD_{h_1} f(x)
  \defn \dirD_{h_1, \dotsc, h_n}^n f(x)
  \defn \lim_{\varepsilon \to 0} \frac{1}{\varepsilon} \big( \dirD_{h_1, \dotsc, h_{n-1}}^{n-1} f(x + \varepsilon h_n) - \dirD_{h_1, \dotsc, h_{n-1}}^{n-1} f(x) \big)
\end{equation*}
and we say that $f$ is $\dirC^n$ if $\dirD f$ is $\dirC^{n-1}$; if $f$ is $\dirC^n$ for all $n \in \NN$, then we say that $f$ is $\dirC^\infty$ or \emph{smooth}.

Continuity of the derivative already implies many other properties if the involved vector spaces are locally convex.
Hence, let $X, Y$ be locally convex spaces and $f : X \supset U \to Y$ be continuously differentiable.
Then it can be shown that \cite{hamilton:1982,neeb:2005}:
\begin{enumerate}[(a)]
  \item
    the \emph{fundamental theorem of calculus}
    \begin{equation*}
      f(x + h) - f(x) = \int_0^1 \dirD f(x + t h; h)\, \dif t
    \end{equation*}
    holds if $x + [0,1] h \subset U$,
  \item $f$ is locally constant if and only if $\dirD f = 0$,
  \item the map $h \mapsto \dirD f(x; h)$ is linear,
  \item $f$ is continuous (not necessarily true if $X, Y$ are not locally convex!),
  \item
    if $f \in \dirC^n$, the map $(h_1, \dotsc, h_n) \mapsto \dirD^n f(x; h_1, \dotsc, h_n)$ is symmetric and multilinear and we can use yet another notation
  \begin{equation*}
    \innerp{\dirD^n f(x)}{h_1 \otimes \dotsm \otimes h_n} \defn \dirD^n f(x; h_1, \dotsc, h_n),
  \end{equation*}
  \item
    if $f \in \dirC^{n+1}$, \emph{Taylor's formula}
    \begin{equation*}\begin{split}
      f(x + h) & = f(x) + \dirD f(x; h) + \dotsb + \frac{1}{n!} \dirD^n f(x; h, \dotsc, h) \\&\quad + \frac{1}{n!} \int_0^1 (1-t)^n \dirD^{n+1} f(x + t h; h, \dotsc, h)
    \end{split}\end{equation*}
    holds for $x + [0,1]h \subset U$.
\end{enumerate}

Moreover, given locally convex $X, Y, Z$, open subsets $U \subset X$, $V \subset Y$ and $\dirC^n$ maps $f : U \to V$, $g : V \to Z$, the \emph{chain rule} holds for the composition $g \circ f$, \ie, also the composition $g \circ f$ is $\dirC^n$ \cite{hamilton:1982,neeb:2005}.

In ordinary calculus one can show that a continuously differentiable function function is locally Lipschitz.
An analogous result holds for the directional derivative on normed spaces (see also \cite{pinamonti:2015}):
\begin{proposition}\label{prop:lipschitz}
  Let $f : X \to Y$ be a continuously differentiable map between the two normed spaces $(X, \norm{\,\cdot\,}_X)$ and $(Y, \norm{\,\cdot\,}_Y)$.
  Then $f$ is \IdxMain{loc-lipschitz}\emph{locally Lipschitz}, \ie, for every convex neighbourhood~$U$ of $x_0 \in X$ there exists a $K \geq 0$ such that for all $x_1, x_2 \in U$
  \begin{equation*}
    \norm{f(x_1) - f(x_2)}_W \leq K\, \norm{x_1 - x_2}_V.
  \end{equation*}
\end{proposition}
\begin{proof}
  Since the derivative $\dirD f(x; h)$ is continuous and linear in $h \in X$, there exists a convex neighbourhood~$U$ of~$x_0$ such that
  \begin{equation*}
    \norm{\dirD f(x; h)}_Y \leq \norm{\dirD f(x)}_{\mathrm{op}} \norm{h}_X \leq K\, \norm{h}_X
  \end{equation*}
  for all $x \in U$.
  As \IdxMain{lipschitz-const}\emph{Lipschitz constant}~$K$ we can choose the supremum of $x \mapsto \norm{\dirD f(x)}_{\mathrm{op}}$ in~$U$.
  By the fundamental theorem of calculus we have for $x_1, x_2 \in U$
  \begin{equation*}
    f(x_1) - f(x_2) = \int_0^1 \dirD f\big( x_2 + t\, (x_1 - x_2); x_1 - x_2 \big)\, \dif t.
  \end{equation*}
  Hence, taking the norm on both sides, the previous equation yields
  \begin{equation*}
    \norm{f(x_1) - f(x_2)}_Y
    \leq \int_0^1 \norm*{\dirD f\big( x_2 + t\, (x_1 - x_2); x_1 - x_2 \big)}_Y\, \dif t
    \leq K\, \norm{x_1 - x_2}_X.
  \end{equation*}
\end{proof}

Later on we will often encounter spaces of differentiable and smooth functions and thus need an appropriate topology on these space:
Let $X, Y$ be a topological vector space such that $Y$ is locally convex and $U \subset X$ open.
We can equip the vector space $\dirC^n(U, Y)$ of all $n$-times continuously differentiable maps between~$X$ and~$Y$ with the seminorms
\begin{equation*}
  \norm{f}_{i, k, K} = \sup_{x \in K}\, \norm{\dirD^k f(x)}_{i,\mathrm{op}}
\end{equation*}
for all $f \in \dirC^n(U, Y)$, every compact $K \subset U$ and $0 \leq k \leq n$.
These seminorms induce an initial topology on $\dirC^n(U, Y)$ turning it into a locally convex space.\Idx{lcs}
This is another example of a compact-open topology or topology of uniform convergence on compacta.
Note that, if $Y$ is a Fréchet space and $U$ is $\sigma$-compact, $\dirC^n(U, Y)$ becomes a Fréchet space.\Idx{frechet}

\subsection{The Fréchet derivative}
\label{sub:der_frechet}

On Banach spaces it is possible to define another directional derivative, the Fréchet derivative.
Given Banach spaces~$X, Y$ and an open subset $U \subset X$, a map $f : U \to Y$ is called \emph{Fréchet differentiable at $x \in U$} if there exists a bounded linear operator $\freD f(x) : X \to Y$, the \IdxMain{frechet-deriv}\emph{Fréchet derivative} of~$f$ at~$x$, such that
\begin{equation}\label{eq:frechet}
  \lim_{\norm{h}_X \to 0} \norm{h}_X^{-1} \big( f(x + h) - f(x) - \freD f(x) h \big) = 0.
\end{equation}
The operator $\freD f(x)$ is unique if it exists.
In analogy to the directional derivatives that we encountered so far, we also write $\freD_h f(x) \defn \freD f(x; h) \defn \freD f(x) h$.
We call $f$ \emph{Fréchet differentiable} if the Fréchet derivative exists for all $x \in U$.
If the Fréchet derivative is continuous in~$x$, then $f$ is \emph{continuously Fréchet differentiable}.
% or $\freC^1$.
% Moreover, for the higher derivatives we say that $f$ is $\freC^n$ if $\freD f$ is $\freC^{n-1}$ and, if $f$ is $\freC^n$ for all $n \in \NN$, then we say that $f$ is $\freC^\infty$.

The Fréchet derivative is closely related to the directional derivative defined above (see also \cite{pinamonti:2015}):
\begin{proposition}\label{prop:derivative_equivalence}
  Let $X, Y$ be Banach spaces, $U \subset X$ open and $f : U \to Y$ a map.
  $f$ is Fréchet differentiable if and only if $f$ is continuously differentiable.
  In that case the two derivatives agree.
\end{proposition}
\begin{proof}
  ``$\Rightarrow$'':
  We can bring \eqref{eq:frechet} into agreement with~\eqref{eq:directional} by replacing $h$ in~\eqref{eq:frechet} by~$t h$, $t \in \RR$, and take the limit $\norm{t h}_V \to 0$ along the ray of~$h$, \ie, by taking $t$ to zero while keeping $h$ fixed.
  Moreover, $\freD f(x)$ is clearly continuous because it is linear and bounded.

  ``$\Leftarrow$'':
  As in proposition~\ref{prop:lipschitz}, since the derivative $\dirD f(\cdot)$ is a continuous linear map, there exists a (convex) neighbourhood~$V$ of~$x$ where it is bounded.
  Using the fundamental theorem of calculus again, we obtain for any $y \in V$ and sufficiently small $h \in X$
  \begin{equation*}
    \norm{f(x + h) - f(x) - \dirD f(y; h)}_Y \leq \sup_{t \in [0,1]} \norm{\dirD f(x + t h) - \dirD f(y)}_{\mathrm{op}} \norm{h}_X.
  \end{equation*}
  In particular this holds for $x = y$ and thus $f$ is Fréchet differentiable at~$x$ with $\freD f(x) = \dirD f(x)$.
\end{proof}

It follows that any statement on continuously differentiable maps also holds for Fréchet differentiable maps.

Fréchet differentiability is a very strong notion of differentiability and many theorems from ordinary calculus can be generalized to the Fréchet derivative but not further to the directional derivative on arbitrary Fréchet spaces.
An example is the inverse function theorem for which holds for the Fréchet derivative on Banach spaces but does not hold on general Fréchet spaces.
On some Fréchet spaces one has instead the Nash--Moser theorem~\cite{hamilton:1982}.

%-- Fixed-point theorems ------------------------------------------------------%
\section{Fixed-point theorems}
\label{sec:fixed_point}
\IdxRanBegin{fixed-point}

Let us start this section by stating the most elementary fixed-point theorem, the \emph{Banach fixed-point theorem}\IdxMain{banach-fixed-point}:

\begin{theorem}[Banach fixed-point theorem]
  Let $f: X \to X$ be a contraction on a (non-empty) complete metric space $X$.
  Then $f$ has a unique fixed-point $x = f(x)$.
  Furthermore, taking an arbitrary initial value $x_0 \in X$, $x$ is the limit of the sequence $(x_n)$ defined by the iterative procedure $x_{n+1} = f(x_n)$.
\end{theorem}

We will not prove this theorem here; the proof is not difficult and can be found in essentially any introductory book on (functional) analysis.
Instead we will dissect, specialize, generalize and finally prove various parts of this theorem separately.

\subsection{Existence and uniqueness}

Let us start with a useful lemma:

\begin{lemma}\label{lem:smooth}
  Let $\;\dotsb \subset V_k \subset V_{k-1} \subset \dotsb \subset V_0$ be a decreasing sequence of sets.
  Suppose there exists a functional~$f$ such that $f : V_k \to V_{k+1}$ for every non-negative $k < n$.
  Any fixed-point $x = f(x)$ in~$V_0$ is already in~$V_n$.
\end{lemma}
\begin{proof}
  Suppose that $x \in V_0$ but $x \nin V_n$ is a fixed-point.
  Then there exists a $k < n$ such that $x \in V_k$ and $x \nin V_{k+1}$.
  Since $x$ is a fixed-point of~$f$, we have that $x = f(x)$, but $f(x) \in V_{k+1}$ by the properties of~$f$.
\end{proof}

This lemma has serval useful consequences.
One example is the following:
The limit of a convergent sequence in a complete metric space is not necessarily as regular as all the elements of the sequence; \latin{a priori} the regularity of the limit is only controlled by the topology induced by the metric.
However, if the limit is the fixed-point of a smoothing map, the situation is much better.

\begin{corollary}\label{cor:smooth}\IdxMain{fixed-point-reg}
  Let $X, Y$ be a topological vector spaces and $U \subset X$ open.
  Further, let $V_k \subset \dirC^k(U, Y)$ for all~$k$ such that $V_k \subset V_{k-1}$.
  Suppose there exists a smoothing functional~$f$ such that $f : V_k \to V_{k+1}$ for every non-negative $k < n$.
  Any fixed-point $x = f(x)$ in~$\dirC^0$ is already in~$\dirC^n$.
  \qed
\end{corollary}

If one is interested only in existence of fixed-points but not their uniqueness, then one can perform a straightforward generalization of the `existence' part of Banach's fixed-point theorem:

\begin{proposition}\label{prop:fixed_point}\IdxMain{fixed-point-ex}
  Let $(X, d)$ be a non-empty complete metric space and $f : X \to X$ a map.
  Assume that there exists a subset~$U \subset X$ such that $f : U \to U$ and $f$ is a contraction on~$U$ with Lipschitz constant $K \in [0,1)$, \ie, for all $y, z \in U$
  \begin{equation*}
    d\big( f(y), f(z) \big) \leq K d(y, z).
  \end{equation*}
  Then there exists a fixed-point $x = F(x)$ in~$X$.
\end{proposition}
\begin{proof}
  Define for an arbitrary $x_0 \in U$ the Picard sequence~$(x_n)$ where $x_{n+1} = f(x_n)$.
  Using the contractivity of~$f$ on~$U$, we get
  \begin{equation*}
    d(x_{n+1}, x_n) \leq K d(x_n, x_{n-1}) \leq K^n d(x_1, x_0).
  \end{equation*}
  One can then easily show that $(x_n)$ is a Cauchy sequence and take the limit $n \to \infty$ in $x_{n+1} = f(x_n)$ to see that there exists a limit~$x = f(x)$ in~$X$.
\end{proof}

This proposition does not guarantee uniqueness of the fixed-point because the mapping is only required to be a contraction on a subset of a complete metric space and the fixed-point is not necessarily contained in this subset.
Nevertheless, uniqueness holds if the mapping is of the form assumed in \cref{lem:smooth}.
Moreover, if the mapping is smoothing as in \cref{cor:smooth} then the unique fixed-point is even $\dirC^n$.

\begin{proposition}\label{prop:unique_fixed_point}\IdxMain{fixed-point-uniq}
  Let $(X, d)$ be a non-empty complete metric space and $(V_k)$ be a decreasing sequence of sets as in \cref{lem:smooth} such that $V_0 \subset X$ is closed.
  Suppose that ${f: V_k \to V_{k+1}}$ for every non-negative $k < n$ such that $f$ is a contraction on~$V_n$.
  Then $f$ has a unique fixed-point~$x = f(x) \in V_n$.
\end{proposition}
\begin{proof}
  The existence of fixed-points in $V_0$ that are contained in $V_n$ follows from \cref{lem:smooth,prop:fixed_point}.
  Assume now that there exist two distinct fixed-points~$x, y$.
  Since $f$ is a contraction on~$V_n$, we have
  \begin{equation*}
    d(x, y) = d\big(f(x), f(y)\big) \leq K d(x, y),
  \end{equation*}
  where $K \in [0, 1)$ is the Lipschitz constant of~$f$, and thus arrive at a contradiction.
\end{proof}

\subsection{A Lipschitz continuity criterion}
\IdxRanBegin{fixed-point-lip}

Next we will see that it is not necessary for a map to be a contraction for it to have fixed points.
In fact it is sufficient for the map to satisfy a certain Lipschitz continuity condition:

\begin{lemma}\label{lem:contraction}
  Let $(X, d)$ be a non-empty complete metric space.
  Suppose there exists $K \in \RR_+$ such that $f: X \to X$ satisfies
  \begin{equation*}
    d\big(f^n(x), f^n(y)\big) \leq \frac{K^n}{n!}\, d(x, y)
  \end{equation*}
  for all $x, y \in X$ and $n \in \NN$.
  Then $f$ has a unique fixed-point.
\end{lemma}
\begin{proof}
  Since $n!$ grows faster than $K^n$, there exists a $N$ such that $f^n$ is a contraction for all $n \geq N$.
  If we set $V_k = f^k(X)$, we can apply \cref{prop:unique_fixed_point} and the thesis follows.
\end{proof}

The special bound assumed in \cref{lem:contraction} is in fact very natural if $f$ is the intgral functional
\begin{equation}\label{eq:integral_functional}
  f: C[a,b] \to C[a,b],
  \quad
  f(x)(t) \defn f_0(t) +\! \int_a^t k(x)(s)\, \dif s,
\end{equation}
where $f_0 \in C[a, b]$ and with integral kernel $k: C[a, b] \to C[a, b]$.
Recall that space of continuous functions $C[a, b]$ in the interval $[a, b]$ can be turned into a Banach space by equipping it with with the \emph{uniform norm}
\begin{equation*}
  \norm{X}_{C[a,b]} \defn \norm{X}_\infty \defn \sup_{t \in [a,b]} \abs{X(t)},
\end{equation*}
where we will use $\norm{X}_{C[a,b]}$ instead of the more common $\norm{X}_\infty$ to emphasize the interval over which the supremum is taken.

\begin{proposition}\label{prop:contraction}
  Let $f$ be of the form \eqref{eq:integral_functional} with $k$ continuously differentiable in $U \subset C[a, b]$ open such that $f$ closes on a closed subset $V \subset U$, \ie, $f(V) \subset V$.
  Then $f$ has a unique fixed-point in~$V$.
\end{proposition}
\begin{proof}
  We can show the statement using \cref{lem:contraction} and an inductive procedure.
  Applying \cref{prop:lipschitz}, we find that $k$ is locally Lipschitz as a functional on~$U$; denote by $L = \sup_{x \in U} \norm{\dirD k}_{\mathrm{op}}$ its Lipschitz constant.
  Using the uniform norm on $C[a,t]$, we thus obtain
  \begin{equation*}
    \norm{f(x) - f(y)}_{C[a,t]}
    \leq \int_a^t \norm{k(x) - k(y)}_{C[a,t]}\, \dif s
    \leq L (t - a) \norm{x - y}_{C[a,b]}.
  \end{equation*}
  Suppose now that
  \begin{equation}\label{eq:Fn-step}
    \norm{f^n(x) - f^n(y)}_{C[a,t]}
    \leq \frac{L^n (t - a)^n}{n!} \norm{x - y}_{C[a,b]}.
  \end{equation}
  holds up to~$n$ and for arbitrary $t \in [a,b]$.
  Then,
  \begin{align*}
    \abs{f^{n+1}(x)(t) - f^{n+1}(y)(t)}
    & \leq \int_a^t \norm{(k \circ f^n)(x) - (k \circ f^n)(y)}_{C[a,s]}\, \dif s \\
    & \leq L \int_a^t \norm{f^n(x) - f^n(y)}_{C[a,s]}\, \dif s \\
    & \leq \frac{L^{n+1}}{n!} \int_a^t (s - a)^n \norm{x - y}_{C[a,b]}\, \dif s \\
    & \leq \frac{L^{n+1} (t - a)^{n+1}}{(n+1)!} \norm{x - y}_{C[a,b]},
  \end{align*}
  which implies that \eqref{eq:Fn-step} holds also for $n+1$, thus concluding the proof.
\end{proof}

\IdxRanEnd{fixed-point-lip}

\subsection{Closed functionals}

The last proposition contains an apparently minor but in fact very strong condition, namely that the functional~$k$ closes within the set~$V$.
In any application of a fixed-point theorem similar to Banach's theorem, the crucial point to check is usually not that the map is a contraction but that it is closed.
In the given case of an integral functional~\eqref{eq:integral_functional}, however, we can always be assured that there exists an interval $I \subset [a,b]$ on which the functional closes \cite{pinamonti:2015}; this interval might be very small.

\begin{proposition}\label{prop:closed}
  Suppose that $k$ is bounded on a set $U \subset C[a,b]$ which also includes a ball~$V$ around~$f_0$ defined as $V = \{ x \mid \norm{x - f_0}_{C[a,b]} < \delta \}$ for some $\delta$.
  Then there exists $t \in (a,b]$ such that $f$ satisfies $f(U) \restriction_{[a,t]} \subset U \restriction_{[a,t]}$.
\end{proposition}
\begin{proof}
  Since $k$ is bounded on~$U$, it clearly satisfies
  \begin{equation*}
    \norm{k(x)}_{C[a,t]}
    \leq \norm{k(x)}_{C[a,b]}
    \leq K
    = \sup_{y \in U}\, \norm{k(y)}_{C[a,b]}
  \end{equation*}
  for all $x \in U$.
  Then, taking the norm of~\eqref{eq:integral_functional} after subtracting $f_0$, one obtains
  \begin{equation*}
    \norm{f(x) - f_0}_{C[a,t]}
    \leq (t - a) \norm{k(x)}_{C[a,t]}
    \leq (t - a) K,
  \end{equation*}
  because $V \subset U$.
  For any~$\delta$ we can always find a $t$ such that $(t - a) K < \delta$ and therefore $f(U) \restriction_{[a,t]} \subset V \restriction_{[a,t]}$.
  The thesis follows because $V \subset U$.
\end{proof}

\IdxRanEnd{fixed-point}

%-- Microlocal analysis -------------------------------------------------------%
\section{Microlocal analysis}
\label{sec:microlocal}
\IdxRanBegin{microlocal}

\subsection{Distributions}
\label{sub:ma_distributions}

We will now define three important function spaces and their topological duals, which will be called spaces of distributions.

To conform with standard notation we denote the \IdxMain{smooth-funs}\emph{space of smooth functions} on an open subset $U \subset \RR^n$ by
\begin{equation*}
  \mathcal{E}(U) \defn \dirC^\infty(U, \CC).
\end{equation*}
As observed in \cref{sub:der_directional}, it is a Fréchet space with the compact-open topology.
The elements of the topological dual~$\mathcal{E}'(U)$ are called \IdxMain{compact-dists}\emph{compactly supported distributions}.

The vector space of \IdxMain{rapidly-decr}\emph{rapidly decreasing} (or \emph{decaying}) \emph{functions} will be denoted by $\mathcal{S}(U)$.
We say that a smooth function $f \in \mathcal{E}(U^n)$ is \emph{rapidly decreasing (decaying)} if
\begin{equation*}
  \norm{f}_{i,n,m} \defn \sup_{x \in U^n} \big( 1 + \abs{x}^n \big) \norm[\big]{\dirD^m f(x)}_i' < \infty
\end{equation*}
for all $i, n, m$.
We equip $\mathcal{S}(U)$ with the topology induced by these seminorms and see that it is a Fréchet space.
The topological dual~$\mathcal{S}'(U)$ is the space of \IdxMain{temp-dists}\emph{tempered distributions} or \emph{Schwartz distributions}.

Another subspace of $\mathcal{E}(U)$ is the space of \IdxMain{test-funs}\emph{test functions}, denoted by
\begin{equation*}
  \mathcal{D}(U) \defn \dirC^\infty_0(U, \CC).
\end{equation*}
We can equip this space with a topology similar but more complicated than that of~$\mathcal{E}(U)$.
If $K \subset U$ is compact, we can endow $\mathcal{D}(K) = \mathcal{E}(K)$ with the subspace topology.
Then, taking a compact exhaustion $K_1 \subset K_2 \subset \dotsb$, $\bigcup_i K_i = U$, the topology on~$\mathcal{D}(U)$ is the initial topology defined by the projections $\pi_i : \mathcal{D}(K_i) \to \mathcal{D}(U)$.\IdxMain{test-funs-top}
This topology is not Fréchet unless $U$ is compact and $\mathcal{E}(U)$ is already Fréchet, in which case $\mathcal{D}(U) = \mathcal{E}(U)$.
The topological dual $\mathcal{D}'(U)$ is the space of \IdxMain{dists}\emph{distributions}.

More generally, we define $\mathcal{F}(U, X)$, with $\mathcal{F} = \mathcal{D}, \mathcal{E}$ or $ \mathcal{S}$, as the spaces of functions with values in a locally convex vector space~$X$ and by $\mathcal{F}'(U, X)$ the associated distribution spaces.
The necessary generalizations to the definitions above are straightforward but note that the resulting function spaces are not Fréchet unless $V$ is already Fréchet.
Moreover, it is possible to define $Y$-valued distributions $\mathcal{F}'(U, X, Y)$, where $Y$ is a locally convex space.

Given a distribution $u \in \mathcal{D}'(U)$, we can \IdxMain{restrict-dist}\emph{restrict} it to a distribution~$u_V$ on any open $V \subset U$ by setting
\begin{equation}
  u_V(f) = u(f)
\end{equation}
for every $f \in \mathcal{D}(V)$.
A distribution is uniquely determined by its restrictions:
If $(U_i)_{i \in \NN}$ is an open cover of~$U$ and $u_i \in \mathcal{D}'(U_i)$ such that $u_i = u_j$ whenever $U_i \cap U_j \neq \emptyset$, then there exists a unique $u \in \mathcal{D}'(U)$ such that $u_i$ is the restriction of~$u$ to~$U_i$ for every~$i$.

The \IdxMain{dist-supp}\emph{support} $\supp u$ of a distribution $u \in \mathcal{D}'(U)$ is the smallest closed set $V \subset U$ such that the restriction of~$u$ to $U \setminus V$ vanishes.
More precisely,
\begin{equation*}
  \supp u = U \setminus \bigcup \{ V \subset U \text{ open} \mid u(f) = 0 \;\forall f \in \mathcal{D}(V) \}.
\end{equation*}
It follows that $u(f) = 0$ if $\supp u \cap \supp f = \emptyset$ and that the space $\mathcal{E}'(U)$ is indeed the space of compactly supported distributions.

\subsection{Distributions on manifolds}
\label{sub:ma_dist_on_man}

The discussion above does not yet encompass the case of distributions on smooth manifolds because manifolds are not vector spaces.
However, manifolds are locally homeomorphic to a vector space -- Euclidean space.

Let $M$ be a smooth manifold, $E \to M$ a smooth vector bundle and $(U_i)_{i \in \NN}$ an open cover of $M$ such that $(U_i, \varphi_i)$ are coordinate charts and $(U_i, \psi_i)$ local trivializations.
We define again the \IdxMain{smooth-sects}\emph{space of smooth sections}
\begin{equation*}
  \mathcal{E}(M, E) \defn \Gamma^\infty(E)
\end{equation*}
as the space of functions $f : M \to E$ such that $\psi_i \circ f \circ \varphi_i^{-1}$ is smooth for each $i$, \ie, we require
\begin{equation*}
  \psi_i \circ f \circ \varphi_i^{-1} \in \mathcal{E}\big( \varphi_i(U_i), \psi_i(E_{U_i}) \big).
\end{equation*}
A locally convex topology that turns $\mathcal{E}(M, E)$ into a Fréchet spaces is the initial topology induced by the product topology on the right-hand side of the injection
\begin{equation*}
  \iota : \mathcal{E}(M, E) \to \prod_{i \in \NN} \mathcal{E}\big( \varphi_i(U_i), \psi_i(E_{U_i}) \big);
\end{equation*}
the topology is independent of the choice of the cover $(U_i)_{i \in \NN}$.
The topological dual of $\mathcal{E}(M, E)$ is the \IdxMain{compact-dist-sects}\emph{space of compactly supported distributional sections}~$\mathcal{E}'(M, E)$.

The \IdxMain{test-sects}\emph{space of compactly supported smooth sections}, the \emph{test sections}, is denoted
\begin{equation*}
  \mathcal{D}(M, E) \defn \Gamma_0^\infty(E).
\end{equation*}
Analogously to the vector space case we define an initial topology on $\mathcal{D}(M, E)$ induced by that on $\mathcal{E}(K_i, E)$, where $(K_i)$ form a compact exhaustion of $M$; whence the space of test sections becomes a Fréchet space.
The \IdxMain{dist-sects}\emph{space of distributional sections} $\mathcal{D}'(M, E)$ is the topological dual of~$\mathcal{D}(M, E)$.

The \IdxMain{restrict-dist-sect}restriction of distribution generalizes to distributions on manifolds in the obvious way:
Given a manifold $M$ and an open subset $U \subset M$, every distribution $u \in \mathcal{D}'(M, E)$ can be restricted to a distribution $u_U \in \mathcal{D}'(U, E_U)$ by setting
\begin{equation*}
  u_U(f) = u(f)
\end{equation*}
for all $f \in \mathcal{D}(U, E_U)$.
Also on a manifold a distribution is completely supported by its restrictions.

It does not make sense to define a notion of rapidly decaying sections or tempered distributional sections on manifolds.
It is also clear, that the concept of distributions can be further extended to objects such as Fréchet manifolds in very much the same way as above for smooth manifolds.

\subsection{Nuclearity and the Schwartz kernel theorem}
\label{sub:ma_kernel}

All the function spaces but none of the distribution spaces defined in the previous two sections are Fréchet.
However, all the function and distribution spaces (with either the weak or strong topology) are nuclear.
For this reason, we have the isomorphisms
\begin{align*}
  \mathcal{D}(U, X) & \simeq \mathcal{D}(U) \mathbin{\widehat{\otimes}} X,
  &
  \mathcal{D}'(U, \KK, X) & \simeq \mathcal{D}'(U) \mathbin{\widehat{\otimes}} X, \\
  \mathcal{E}(U, X) & \simeq \mathcal{E}(U) \mathbin{\widehat{\otimes}} X,
  &
  \mathcal{E}'(U, \KK, X) & \simeq \mathcal{E}'(U) \mathbin{\widehat{\otimes}} X, \\
  \mathcal{S}(\RR^m, X) & \simeq \mathcal{E}(\RR^m) \mathbin{\widehat{\otimes}} X,
  &
  \mathcal{S}'(\RR^m, \KK, X) & \simeq \mathcal{S}'(\RR^m) \mathbin{\widehat{\otimes}} X
\end{align*}
for a complete locally convex topological vector space~$X$ and an open set $U \subset \RR^m$.
As another consequence we can specialize the abstract kernel theorem (\cf\ \cref{sub:tvs_nuclear}) to these function spaces under which circumstances it is called the \IdxMain{schwartz-ker-thm}\emph{Schwartz kernel theorem}.
One finds the following isomorphism (open subsets $U \subset \RR^m$ and $V \subset \RR^n$):
% \begin{equation*}\begin{array}{rcl}
%   \mathcal{E}'(U \times V) \simeq & \mathcal{E}'(U) \mathbin{\widehat{\otimes}} \mathcal{E}'(V) & \simeq L\big( \mathcal{E}(U), \mathcal{E}'(V) \big), \\
%   \mathcal{D}'(U \times V) \simeq & \mathcal{D}'(U) \mathbin{\widehat{\otimes}} \mathcal{D}'(V) & \simeq L\big( \mathcal{D}(U), \mathcal{D}'(V) \big), \\
%   \mathcal{S}'(\RR^{m+n}) \simeq & \mathcal{S}'(\RR^m) \mathbin{\widehat{\otimes}} \mathcal{S}'(\RR^n) & \simeq L\big( \mathcal{E}(\RR^m), \mathcal{E}'(\RR^n) \big),
% \end{array}\end{equation*}
\begin{align*}
  \mathcal{E}'(U \times V) & \simeq \mathcal{E}'(U) \mathbin{\widehat{\otimes}} \mathcal{E}'(V) \simeq L\big( \mathcal{E}(U), \mathcal{E}'(V) \big), \\
  \mathcal{D}'(U \times V) & \simeq \mathcal{D}'(U) \mathbin{\widehat{\otimes}} \mathcal{D}'(V) \simeq L\big( \mathcal{D}(U), \mathcal{D}'(V) \big), \\
  \mathcal{S}'(\RR^{m+n}) & \simeq \mathcal{S}'(\RR^m) \mathbin{\widehat{\otimes}} \mathcal{S}'(\RR^n) \simeq L\big( \mathcal{E}(\RR^m), \mathcal{E}'(\RR^n) \big),
\end{align*}
where $L(X, Y)$ denotes the space of continuous linear maps between topological vector spaces~$X$ and~$Y$ with the topology of uniform convergence.
Analogous isomorphisms (at least for $\mathcal{E}$ and $\mathcal{D}$) exist for both sets of isomorphisms also for functions and distributions on manifolds.

As a consequence of these isomorphisms, there exists for every distribution $K \in \mathcal{D}'(U \times V)$ a unique linear operator $\mathrm{K} : \mathcal{D}(U) \to \mathcal{D}'(V)$ and, conversely, to every linear operator $\mathrm{K}$ a unique distribution.
Let $f \in \mathcal{D}(U)$ and $g \in \mathcal{D}(V)$ be test functions.
Formally we can write
\begin{equation*}
  K(f \otimes g) = \int_{U \times V} K(x, y) f(x) g(x)\, \dif^m x\, \dif^n y
\end{equation*}
for the distribution with \IdxMain{dist-ker}\emph{distributional kernel} $K(x, y)$ and
\begin{equation*}
  (\mathrm{K} f)(y) = \int_U K(x, y) f(x)\, \dif^m x,
  \quad
  ({}^t \mathrm{K} g)(x) = \int_V K(x, y) g(y)\, \dif^n y,
\end{equation*}
for the associated operator and its transpose.

The operator $\mathrm{K}$ is called \IdxMain{semiregular-ker}\emph{semiregular} if it continuously\footnote{Continuity is meant with respect to the usual topology of $\mathcal{E}(V)$ and not the subspace topology of $\mathcal{D}'(V)$.} maps $\mathcal{D}(U)$ into $\mathcal{E}(V)$ and, analogously, the transpose ${}^t\mathrm{K}$ is called semiregular if it continuously maps $\mathcal{D}(V)$ into $\mathcal{E}(U)$.
In the case that ${}^t\mathrm{K}$ is semiregular, we can uniquely extend $\mathrm{K}$ to an operator acting on compactly supported distributions $\mathcal{E}'(U)$ by duality:
\begin{equation*}
  (\mathrm{K} u)(g) = ({}^t\mathrm{K} g)(u)
\end{equation*}
for all $u \in \mathcal{E}(U)$ and $g \in \mathcal{D}'(V)$.

If both $\mathrm{K}$ and ${}^t\mathrm{K}$ are semiregular, we say that $\mathrm{K}$ \IdxMain{regular-ker}\emph{regular}.
Moreover, it is called \IdxMain{prop-supp-ker}\emph{properly supported} if the projections from $\supp K \subset U \times V$ onto each factor are proper maps.
A properly supported operator~$\mathrm{K}$ maps $\mathcal{D}(U)$ to $\mathcal{E}'(V)$ and can therefore be extended to an operator $\mathcal{D}(U) \to \mathcal{E}'(V)$.
Since linear differential operators are properly supported and regular, they can be uniquely extended to distributions and they can also be composed.

\subsection{Fourier transformation and convolution}
\label{sub:ma_fourier}

Let us denote by $L^p(\RR^n)$ the \IdxMain{lp-space}\emph{$L^p$ spaces} of functions on $\RR^n$ with values in $\CC$.
That is, $L^p(\RR^n)$ is the space of functions for which the Lebesgue integrals
\begin{equation*}
  \norm{f}_p \defn \left( \int_{\RR^n} \abs{f(x)}^p\, \dif^n x \right)^{1/p}
\end{equation*}
exist and where we identify functions which are equal almost everywhere so that the $L^p$ spaces become Banach spaces.
$L^1$ functions are called \IdxMain{lebesgue-int}\emph{Lebesgue integrable}, while $L^2$ are called \IdxMain{square-int}\emph{square-integrable}.

On the space of Lebesgue integrable functions $L^1(\RR^n)$, the \IdxMain{fourier}\emph{Fourier transform} is defined as the automorphism
\begin{equation*}
  \mathcal{F} : f(x) \mapsto \mathcal{F}(f)(\xi) \defn (2 \uppi)^{-n}\! \int_{\RR^n} f(x) \e^{-\im x \cdot \xi}\, \dif^n x,
\end{equation*}
where $\,\cdot\,$ denotes the Euclidean dot product.
The Fourier transform satisfies
\begin{equation*}
  \mathcal{F}^2(f)(x) = (2 \uppi)^n f(-x)
\end{equation*}
and the \IdxMain{invfourier}\emph{inverse Fourier transform} is therefore given by
\begin{equation*}
  \mathcal{F}^{-1}(f)(x) \defn \int_{\RR^n} f(\xi) \e^{\im x \cdot \xi}\, \dif^n \xi.
\end{equation*}
When no confusion can arise, we usually denote the Fourier transform of a function~$f$ by~$\what{f}$ instead of $\mathcal{F}(f)$.

By the Riemann--Lebesgue lemma, it is clear that $\what{f}(\xi) \to 0$ as $\abs{\xi} \to \infty$.
In fact, the Fourier transform is a linear isomorphism from the subspace of rapidly decaying functions $\mathcal{S}(\RR^n)$ into itself.
Since the space of rapidly decreasing functions is stable under differentiation and multiplication by polynomials, one finds for $f \in \mathcal{S}(\RR^n)$
\begin{equation}\label{eq:fourier_diff_mult}
  \mathcal{F}(\partial_\mu f)(\xi) = \xi_\mu \what{f}(\xi)
  \quad \text{and} \quad
  \mathcal{F}(x_\mu f)(\xi) = \partial_\mu \what{f}(\xi).
\end{equation}
Moreover, given also $g \in \mathcal{S}(\RR^n)$, the \IdxMain{parseval-id}\emph{Plancherel--Parseval identities} are
\begin{align}
  \int_{\RR^n} f(x) \conj{g}(x)\, \dif^n x & = (2 \uppi)^{-n}\! \int_{\RR^n} \what{f}(\xi) \conj{\what{g}}(\xi)\, \dif^n \xi, \label{eq:plancherel} \\
  \int_{\RR^n} \abs{f(x)}^2\, \dif^n x & = (2 \uppi)^{-n}\! \int_{\RR^n} \abs{\what{f}(\xi)}^2\, \dif^n \xi. \notag
\end{align}
As a consequence, the Fourier transform can be extended to an isomorphism of $L^2(\RR^n)$ into itself.

The Plancherel--Parseval formula \eqref{eq:plancherel} guides us to extend the Fourier transformation $\mathcal{F}$ further to the space of tempered distributions $\mathcal{S}'(\RR^n)$ by
\begin{equation*}
  \langle \what{u}, f \rangle \defn \langle u, \what{f} \rangle
\end{equation*}
for all $u \in \mathcal{S}'(\RR^n)$ and $f \in \mathcal{S}(\RR^n)$, \ie, it is the transpose of the Fourier transformation on rapidly decreasing functions.
It follows that $\mathcal{F}$ is a linear isomorphism from $\mathcal{S}'(\RR^n)$ (with the weak topology) into itself and analogues of the relations~\eqref{eq:fourier_diff_mult} hold also for tempered distributions $u \in \mathcal{S}'(\RR^n)$:
\begin{equation*}
  \mathcal{F}(\partial_\mu u)(\xi) = \xi_\mu \what{u}(\xi)
  \quad \text{and} \quad
  \mathcal{F}(x_\mu u)(\xi) = - \partial_\mu \what{u}(\xi).
\end{equation*}

If we restrict to the space of compactly supported distributions, the
Fourier transform of $u \in \mathcal{E}'(\RR^n)$ is equivalently given as the smooth function
\begin{equation*}
  \what{u}(\xi) = \langle u, f : x \mapsto \e^{- \im x \cdot \xi} \rangle.
\end{equation*}
The Fourier transform $\what{u}$ can be directly extended to $\CC^n$ as an entire analytic function.

The \IdxMain{convolution}\emph{convolution} of two Lebesgue integrable functions $f, g \in L^1(\RR^n)$ is defined as
\begin{equation*}
  (f \conv g)(x) \defn \int_{\RR^n} f(y) g(x - y)\, \dif^n y.
\end{equation*}
The product thus defined is dual to the usual product with respect to Fourier transformation.
To wit, the identities
\begin{equation*}
  \mathcal{F}(f \conv g) = \what{f}\, \what{g}
  \quad \text{and} \quad
  \mathcal{F}(f\, g) = (2 \uppi)^{-n} (\what{f} \conv \what{g})
\end{equation*}
hold and are the result of the \IdxMain{convol-thm}\emph{convolution theorem}.
Note that for distributions $u \in \mathcal{S}'(\RR^n)$ and $v \in \mathcal{E}'(\RR^n)$, the convolution $u \conv v$ is a well-defined tempered distribution and its Fourier transform satisfies $\mathcal{F}(u \conv v) = \what{u}\, \what{v}$ as in the convolution theorem.
If also the product $u v$ is well-defined as a (tempered) distribution, \cf\ \cref{sub:ma_pullback}), then other statement of the convolution theorem holds and $\mathcal{F}(u v) = (2 \uppi)^{-n} (\what{u} \conv \what{v})$

\subsection{Singularities and the wavefront set}
\label{sub:ma_wavefront}

Every locally Lebesgue integrable function $u \in L^1_{\mathrm{loc}}(\RR^n)$ can be identified with a distribution in $\mathcal{D}'(\RR^n)$, denoted by the same symbol, via
\begin{equation*}
  u(f) = \langle u, f \rangle = \int_{\RR^n} u(x) f(x)\, \dif^n x
\end{equation*}
for all $f \in \mathcal{D}(\RR^n)$.
We say that a distribution $u \in \mathcal{D}'(\RR^n)$ is \IdxMain{dist-smooth}\emph{smooth}, if it is induced from a smooth function via this duality pairing.\IdxMain{dist-duality}
More specifically, every smooth function corresponds to a distribution in this way and in fact $\mathcal{D}(\RR^n)$ is isomorphic to a dense subset of $\mathcal{D}'(\RR^n)$.
Therefore this pairing uniquely extends to a pairing between $\mathcal{D}(\RR^n)$ and $\mathcal{D}'(\RR^n)$.
Using the Plancherel--Parseval identity~\eqref{eq:plancherel}, it can be written explicitly as
\begin{equation}\label{eq:dist_pairing}
  u(f) = (2 \uppi)^{-n}\! \int_{\RR^n} \wwhat{\chi u}(\xi) \what{f}(-\xi)\, \dif^n\xi,
\end{equation}
where $u \in \mathcal{D}'(\RR^n)$ and $\chi \in \mathcal{D}(\RR)$ such that $\chi = 1$ on a compact neighbourhood of the support of $f \in \mathcal{D}(\RR^n)$.
This pairing may be considered the motivation of the wavefront set to be defined below.

The \IdxMain{sing-supp}\emph{singular support} $\singsupp u$ of a distribution $u \in \mathcal{D}'(\RR^n)$ is then defined as the complement of the union of all open sets on which $u$ is smooth in the sense of the pairing above.
In other words, it is the smallest closed subset $U \subset \RR^n$ such that $u_{\RR^n \setminus U} \in \mathcal{E}(\RR^n \setminus U)$.

The Fourier transform, introduced in the previous section, can be used to give a condition on the smoothness of a compactly supported distribution $u \in \mathcal{E}'(\RR^n)$.
Namely, $u$ is smooth if and only if for each $n \in \NN_0$ there exists a constant $C_n$ such that
\begin{equation*}
  \abs{\what u(\xi)} \leq C_n (1 + \abs{\xi})^{-n}
\end{equation*}
for all $\xi \in \RR^n$.

Checking this condition for certain $\xi$, a \IdxMain{reg-dir}\emph{regular direction} of a compactly supported distribution $u \in \mathcal{E}'(\RR^n)$ is a vector $\xi \in \RR^n \setminus \{0\}$ such that there exists an open conical\footnote{A \emph{cone} in~$\RR^n$ is a subset $\Gamma \subset \RR^n$ such that $\lambda \Gamma = \Gamma$ for all $\lambda > 0$.} neighbourhood~$\Gamma$ of~$\xi$ and such that
\begin{equation*}
  \sup_{\zeta \in \Gamma}\, (1 + \abs{\zeta})^n \abs{\what{u}(\zeta)} < \infty
\end{equation*}
for all $n \in \NN_0$.
Conversely, a~$\xi$ is called a \IdxMain{sing-dir}\emph{singular direction} of~$u$ if it is not a regular direction.
The (closed) \IdxMain{sing-dir-set}\emph{set of all singular directions} of~$u$ is
\begin{equation*}
  \Sigma(u) \defn \big\{ \xi \in \RR^n \setminus \{0\} \;\big|\; \xi \text{ is not a regular direction of } u \big\},
\end{equation*}
\ie, the complement of all regular directions.

We can localize the notion of singular directions and say that $\xi$ is a singular direction of $u \in \mathcal{D}'(U)$ at $x \in U$, where $U \subset \RR^n$ is open, if there exists a $n \in \NN_0$ such that
\begin{equation*}
  \sup_{\zeta \in \Gamma}\, (1 + \abs{\zeta})^n \abs{\mathcal{F}(\chi u)(\zeta)}
\end{equation*}
is not bounded for all $\chi \in \mathcal{D}(U)$ localized at~$x$ (\ie, $\chi(x) \neq 0$).
That is, the \IdxMain{sing-dir-set-point}\emph{set of singular directions at~$x$} is the closed set
\begin{equation*}
  \Sigma_x(u) \defn \bigcap_\chi \Sigma(\chi u),
\end{equation*}
where the intersection is over all $\chi \in \mathcal{D}(U)$ such that $\chi(x) \neq 0$.

This leads to the definition of the \IdxMain{wf-set}\emph{wavefront set} as the set of the singular directions at all points:
\begin{equation*}
  \WF(u) \defn \big\{ (x; \xi) \in U \times (\RR^n \setminus \{0\}) \;\big|\; \xi \in \Sigma_x(u) \big\}.
\end{equation*}
Thus the wavefront set is a refinement of the notion of singular support.
Moreover, it can be used as a practical tool for calculating the singular support because $\singsupp u$ is the projection of~$\WF(u)$ onto the first component.
The wavefront set has the following properties:
\begin{enumerate}[(a)]
  \item $\WF(\chi u) \subset \WF(u)$,
  \item $\WF(u + v) \subset \WF(u) \cup \WF(v)$,
  \item $\WF(\mathrm{P} u) \subset \WF(u)$
\end{enumerate}
for all distributions $u, v \in \mathcal{D}'(U)$, localizing functions $\chi \in \mathcal{D}(U)$ and linear differential operators~$\mathrm{P}$ (with smooth coefficients).

\subsection{Wavefront set in cones}
\label{sub:wf_cones}

Let $U \subset \RR^n$ be open and $\Gamma \subset U \times (\RR^n \setminus \{0\})$ a closed cone, where we have extended the definition of a cone to sets for which the projection to the second component at each point is a cone.
We define \IdxMain{dist-cone}distributions with wavefront set contained in the cone~$\Gamma$ as
\begin{equation*}
  \mathcal{D}'_\Gamma(U) \defn \big\{ u \in \mathcal{D}'(U) \;\big|\; \WF(u) \subset \Gamma \big\},
\end{equation*}
which is not empty for any cone~$\Gamma$.
The \IdxMain{normal-top}\emph{normal topology}\footnote{The normal topology \cite{dabrowski:2014} is finer than the often emplyed \Idx{hormander-top}\emph{Hörmander topology} for these spaces. Nuclearity also holds for the Hörmander topology but not completeness.} turns $\mathcal{D}'_\Gamma(U)$ into a complete nuclear space~\cite{dabrowski:2014}.
It is induced by the seminorms
\begin{equation*}
  \norm{u}_B = \sup_{f \in B}\, \abs{u(f)}
  \quad \text{and} \quad
  \norm{u}_{n,V,\chi} = \sup_{\xi \in V}\, (1 + \abs{\xi})^n \abs{\mathcal{F}(\chi u)(\xi)},
\end{equation*}
for all bounded sets $B \subset \mathcal{D}(U)$, $n \in \NN_0$, localizing functions $\chi \in \mathcal{D}(U)$ and closed cones $V \subset \RR^n \setminus \{0\}$ such that $\supp(\chi) \times V \subset \Gamma$.

Given a closed cone $\Gamma$ as above, define the open cone
\begin{equation*}
  \Lambda = (\Gamma')^c \defn \{ (x; \xi) \in U \times (\RR^n \setminus \{0\}) \;\big|\; (x; -\xi) \nin \Gamma \big\}
\end{equation*}
as the complement of the reflection of $\Gamma$ and
\begin{equation*}
  \mathcal{E}'_\Lambda(\RR^n) \defn \big\{ v \in \mathcal{E}'(\RR^n) \;\big|\; \WF(v) \subset \Lambda \big\}
\end{equation*}
as the \IdxMain{comp-dist-cone}space of compactly supported distributions with wavefront set contained in~$\Lambda$.
Then one can find an analogue to the pairing~\eqref{eq:dist_pairing} for all $u \in \mathcal{D}'_\Gamma(\RR^n)$ and $v \in \mathcal{E}'_\Lambda(\RR^n)$ given by~\cite{dabrowski:2014}
\begin{equation*}
  \langle u, v \rangle \defn (2 \uppi)^{-n} \int_{\RR^n} \wwhat{\chi u}(\xi) \what{v}(-\xi)\, \dif^n\xi,
\end{equation*}
for any $\chi \in \mathcal{D}(\RR)$ such that $\chi = 1$ on a compact neighbourhood of~$\supp v$.
$\mathcal{E}'_\Lambda(\RR^n)$ thus becomes the topological dual of $\mathcal{D}'_\Gamma(\RR^n)$ (with the normal topology) and (equipped with the strong topology) it is also nuclear but not complete unless $\Lambda$ is also closed~\cite{dabrowski:2014}.

\subsection{Pullback of distributions}
\label{sub:ma_pullback}
\IdxMain{pb-dist}

Let $U, V$ be open subsets of $\RR^n$ and $\iota : U \to V$ a diffeomorphism.
The pullback $\pb{\iota} u$ of a distribution $u \in \mathcal{D}'(V)$ is (uniquely) defined for every $f \in \mathcal{D}(U)$ as the transpose of the pushforward (up to the Jacobian determinant)
\begin{equation*}
  \langle \pb{\iota} u, f \rangle = \langle u, \pf{\iota} f\, \abs{\det \dirD \iota} \rangle
\end{equation*}
or, equivalently, as the continuous extension of the pullback on smooth function.
Consequently, for any closed cone $\Gamma \subset V \times (\RR^n \setminus \{0\})$, one obtains
\begin{equation}\label{eq:D_pullback}
  \pb{\iota} \mathcal{D}'_\Gamma(V) = \mathcal{D}'_{\pb{\iota} \Gamma}(U),
  \quad
  \pb{\iota} \Gamma \defn \big\{ \big( x; T^*_x \iota (\xi) \big) \;\big|\; \big( \iota(x); \xi \big) \in \Gamma \big\}
\end{equation}
and hence $\WF(\pb{\iota} u) = \pb{\iota} \WF(u)$.

Trying to generalize this result to cases where $\iota : U \to V$ is not a diffeomorphism but an embedding of an open subset of~$\RR^n$ into an open subset of~$\RR^m$ can fail if there are $(x; \xi)$ such that $T^*_x \iota (\xi) = 0$.
It follows that a distribution $u \in \mathcal{D}'(V)$ can only be pulled back to a distribution $\pb{\iota} u$ if $\WF(u) \cap N = \emptyset$, where
\begin{equation*}
  N = \big\{ \big( \iota(x); \xi \big) \in V \times \RR^m \;\big|\; x \in U, T^*_x \iota (\xi) = 0 \big\}
\end{equation*}
is the set of conormals of~$\iota$.

Given two distributions $u \in \mathcal{D}'(U)$ and $v \in \mathcal{D}'(V)$, where $U \subset \RR^n$ and $V \subset \RR^m$ are open, the tensor product
\begin{equation*}
  u \otimes v : f \otimes h \mapsto u(f) v(h)
\end{equation*}
is a distribution in $\mathcal{D}'(U \times V) \simeq \mathcal{D}'(U) \mathbin{\widehat{\otimes}} \mathcal{D}'(V)$ via Schwartz's kernel theorem.
One can show that its wavefront set satisfies
\begin{equation*}\begin{split}
  \WF(u \otimes v) \subset \big( \WF(u) \times \WF(v) \big) & \cup \big( (\supp u \times \{0\}) \times \WF(v) \big) \\ & \cup \big( \WF(u) \times (\supp v \times \{0\}) \big).
\end{split}\end{equation*}

It is possible to pullback the tensor product $u \otimes v$ of two distributions over the same space (\ie, $U = V$) with the diagonal map
\begin{equation*}
  \Delta : U \times U \to U, (x, x) \mapsto x
\end{equation*}
if $\WF(u \otimes v) \cap N_\Delta = \emptyset$, where $N_\Delta$ is the set of conormals with respect to the map~$\Delta$, which gives the (unique) product $u v$ of the two distributions.\IdxMain{mult-dist}
This requirement of the wavefront set implies that it is possible to multiply two distributions if and only if
\begin{equation}\label{eq:WF_prod_restriction}
  (x, \xi) \in \WF(u) \implies (x, -\xi) \nin \WF(v)
\end{equation}
and then wavefront set of the product is bounded by
\begin{equation}\begin{split}\label{eq:WF_prof}
  \WF(u v) & \subset \big\{ (x; \xi + \zeta) \;\big|\; (x; \xi) \in \WF(u), (x; \zeta) \in \WF(v) \big\} \\&\quad \cup \WF(u) \cup \WF(v).
\end{split}\end{equation}
Note that for $u, v$ that do not satisfy \eqref{eq:WF_prod_restriction}, the singular directions would add up to zero in the first term on the right-hand side of~\eqref{eq:WF_prof}.

\subsection{Wavefront set of distributional sections}

The wavefront set can be extended to distributions on vector-valued functions component\hyp{}wise, \ie, using $\mathcal{D}'(U, \KK^m) \simeq \mathcal{D}'(U) \otimes \KK^m \simeq \mathcal{D}'(U)^{\oplus m}$.
Namely, one defines for $u \in \mathcal{D}'(U, \KK^m)$
\begin{equation*}
  \WF(u) \defn \bigcup_{i=1}^{m} \WF(u_i),
\end{equation*}
where $u_i \in \mathcal{D}'(U)$ are the components of~$u$.
This definition is invariant under a change of basis because such a change only implies a multiplication of~$(u_i)$ by a matrix with smooth components.

Moreover, the wavefront set being a local concept, it generalizes to manifolds and distributional sections in a coordinate neighbourhood via local trivializations.
However, to be meaningful, it needs to transform covariantly under diffeomorphisms.

Let $(U_i)_{i \in \NN}$ be an open cover of a smooth $n$-manifold~$M$ such that $(U_i, \varphi_i)$ are coordinate charts and $(U_i, \psi_i)$ are local trivializations of the vector bundle $E \to M$.
Given a distribution $u \in \mathcal{D}'(M, E)$ with restrictions~$u_i$ to~$U_i$, the wavefront set for every $u_i$ given by
\begin{equation*}
  \WF(u_i) \defn \big\{ \big( x; T^*_x \varphi (\xi) \big) \in U_i \times (\RR^n \setminus \{0\}) \;\big|\; \big( \varphi(x); \xi \big) \in \WF\big( \psi_i \circ u_i \circ \varphi_i^{-1} \big) \big\}.
\end{equation*}
and transforms as a conical subset of the cotangent bundles~$T^*\!U_i$ as seen by~\eqref{eq:D_pullback}.
In particular, $\WF(u_i) \cap T^*\!(U_i \cap U_j) = \WF(u_j) \cap T^*\!(U_i \cap U_j)$ for all $i, j$.

The \IdxMain{wf-set-sect}\emph{wavefront set of distributional sections}~$u$ is then defined as the union of all $\WF(u_i)$.
In other words, it is the set of points
\begin{equation*}
  (x; \xi) \in \dot{T}^*\!M \defn T^*\!M \setminus \{ (y; 0) \in T^*\!M \},
\end{equation*}
the cotangent bundle with the zero section removed, such that $(x; \xi) \in \WF(u_U)$, where $U$ is a coordinate and trivialization neighbourhood of $E$.

\subsection{Some distributions and their wavefront set}

For any $f \in \mathcal{E}(\RR)$, \IdxMain{dirac-delta}\emph{Dirac's $\delta$-distribution} is
\begin{equation*}
  \delta(f) = f(0)
\end{equation*}
and it follows that $\delta$ has support only at the origin.
There it does not decay in any direction because $\what{\delta} = 1$ so that $\WF(\delta) = \{0\} \times (\RR \setminus \{0\})$.
Consequently powers of the $\delta$-distribution cannot be defined.

The Dirac $\delta$-distribution can be decomposed into two distributions
\begin{equation*}
  \delta_\pm(f) \defn \lim_{\varepsilon \to 0^+} \int_\RR \frac{f(x)}{x \pm \im\varepsilon}\, \dif x,
\end{equation*}
for all $f \in \mathcal{S}(\RR)$, such that $-2\uppi\im\delta = \delta_+ + \delta_-$, where $\WF(\delta_\pm) = \{0\} \times \RR_\pm$.
Now, powers of either $\delta_\pm$ are well-defined but the distribution $\delta_+ \delta_-$ does not exist.

The wavefront set of $\delta_+$ (and analogously that of $\delta_-$) can be calculated as follows:
Using the residue theorem, the Fourier transform of $1/(x + \im\varepsilon)$ for $\varepsilon > 0$ is\footnote{$\Heaviside$ denotes the Heaviside step-function.}
\begin{equation*}
  \int_\RR \frac{e^{- \im x \xi}}{x + \im\varepsilon}\, \dif x
  = - 2 \uppi \im \Heaviside(\xi)\, \e^{-\xi \varepsilon}.
\end{equation*}
Taking the limit $\varepsilon \to 0^+$, this gives $\what{\delta}_+(\xi) = -2\uppi\im \Heaviside(\xi)$.
Then, applying the convolution theorem, one obtains the Fourier transform of $\chi \delta_+$ for all $\chi \in \mathcal{D}(\RR)$ as
\begin{equation*}
  \mathcal{F}(\chi \delta_+)
  = \frac{1}{2 \uppi} \big( \what{\chi} * \what{\delta}_+ \big)
  = - \im \int_{-\infty}^{\xi} \what{f}(k)\, \dif k.
\end{equation*}
Since this decays rapidly as $\xi \to -\infty$ and does not decay as $\xi \to \infty$, we get the expected wavefront set.

Related to the diagonal map $\Delta : (x, x) \mapsto x$, we can define for all $f \in \mathcal{D}(\RR^2)$ a \IdxMain{diag-dist}\emph{diagonal distribution}
\begin{equation*}
  \Delta(f) \defn \int_{\RR} f(x, x)\, \dif x.
\end{equation*}
It is clear that the wavefront set of $\Delta$ is
\begin{equation*}
  \WF(\Delta) = N_\Delta = \big\{ (x, x; \xi, -\xi) \in \RR^4 \setminus \{0\} \big\}.
\end{equation*}
Given instead two functions $f_1, f_2 \in \mathcal{D}(\RR)$, we can write
\begin{equation*}
  \Delta(f_1 \otimes f_2) = \int_{\RR} (f_1 \conv \delta)(x) f_2(x)\, \dif x
\end{equation*}
Splitting the $\delta$-distribution into its positive and negative frequency components as above, we can therefore define
\begin{equation*}
  \Delta_\pm(f) \defn \lim_{\varepsilon \to 0^+} \int_{\RR} \frac{f(x, y)}{y - x \pm \im \varepsilon}\, \dif x\, \dif y,
\end{equation*}
which in the case $f = f_1 \otimes f_2$ can be written as
\begin{equation*}
  \Delta_\pm(f_1 \otimes f_2) = \int_{\RR} (f_1 \conv \delta_{\pm})(x) f_2(x)\, \dif x.
\end{equation*}
It is not a difficult exercise to show that \cite[Exmpl.~1.4]{siemssen:2011}
\begin{equation*}
  \WF(\Delta_\pm) = \big\{ (x, x; \xi, -\xi) \in \RR^4 \setminus \{0\} \;\big|\; \pm\xi > 0 \big\}.
\end{equation*}
Moreover, using the Plancherel--Parseval identities it is possible to show that $\Delta_\pm$ is well-defined for all $f_1, f_2 \in L^2(\RR)$.

\subsection{Propagation of singularities}
\label{sub:ma_propagation}

In \cref{sub:ma_wavefront} we already noticed that $\WF(\mathrm{P} u) \subset \WF(u)$.
That is, knowing the wavefront set of the distribution $u$, we can deduce information about the wavefront set of $\mathrm{P} u$, where $\mathrm{P}$ is a differential operator.
The theorem on the \IdxMain{prop-sings}\emph{propagation of singularities} gives us information in the opposite direction.
Namely, $\WF(\mathrm{P} u)$ and the form of $\mathrm{P}$, tell us a lot about $\WF(u)$.

Let $\mathrm{P} : \mathcal{E}(M, E) \to \mathcal{E}(M, E)$ be a differential operator acting on sections of a vector bundle $E \to M$.
Its \IdxMain{char-set}\emph{characteristic set} is the cone
\begin{equation*}
  \charset \mathrm{P} = \big\{ (x; \xi) \in \dot{T}^*\!M \;\big|\; \det \sigma_{\mathrm{P}}(x, \xi) = 0 \big\}
\end{equation*}
on which the principal symbol~$\sigma_{\mathrm{P}}$ of~$\mathrm{P}$ cannot be inverted.\footnote{If the principal symbol of~$\mathrm{P}$ is invertible, we say that $\mathrm{P}$ is \emph{elliptic}.
Example of elliptic operators are the Laplace operator and the Cauchy--Riemann operator.}
An integral curve of $\sigma_{\mathrm{P}}$ in $\charset \mathrm{P}$ is called a \IdxMain{bichar-strip}\emph{bicharacteristic strip}, its projection onto $M$ a \IdxMain{bichar}\emph{bicharacteristic}.

\begin{theorem}[Propagation of singularities]
  Suppose that $\mathrm{P}$ is a differential operator with real homogeneous principal symbol such that no complete bicharacterstic stays in a compact set of~$M$ (\ie, $\mathrm{P}$ is of real principal type) and let $u, f \in \mathcal{D}'(M, E)$ such that $\mathrm{P} u = f$.
  Then
  \begin{equation*}\label{eq:elliptic_regularity}
    \WF(u) \subset \charset \mathrm{P} \cup \WF(f)
  \end{equation*}
  and, if $(x; \xi) \in \WF(u) \setminus \WF(f)$, it follows that $(x'; \xi') \in \WF(u)$ for all $(x'; \xi')$ on the bicharacteristic strip passing through $(x; \xi)$.
\end{theorem}

\IdxRanEnd{microlocal}

%-- Wave equations ------------------------------------------------------------%
\section{Wave equations}
\label{sec:wave}

Both classical and quantum fields usually satisfy an equation of motion given by a \IdxMain{wave-eq}\emph{wave equation}
\begin{equation}\label{eq:wave_equation}
  \mathrm{P} u = f,
\end{equation}
where $\mathrm{P}$ is a normally hyperbolic differential operator, $u$ is the field and $f$ and an external source.
On globally hyperbolic manifolds the wave equation can be solved, \ie, the Cauchy problem for~\eqref{eq:wave_equation} is well-posed.

\subsection{Retarded and advanced propagators}

Let $(M, g)$ be a spacetime and $\mathrm{P} : \mathcal{E}(M, E) \to \mathcal{E}(M, E)$ a differential operator on sections of a vector bundle $E \to M$.
A linear operator $\mathrm{G}_\vee : \mathcal{D}(M, E) \to \mathcal{E}(M, E)$ such that for all $f \in \mathcal{D}(M, E)$
\begin{equation*}
  \mathrm{P} \mathrm{G}_\vee f = f
  \quad \text{and} \quad
  \mathrm{G}_\vee \mathrm{P} f = f,
\end{equation*}
\ie, $\mathrm{G}_\vee$ is a left- and right-inverse of~$\mathrm{P}$, and
\begin{equation*}
  \supp(\mathrm{G}_\vee f) \subset J^+(\supp f)
\end{equation*}
is called a \IdxMain{ret-prop}\emph{retarded propagator} or \emph{retarded Green's operator} for~$\mathrm{P}$.
Similarly, a linear operator $\mathrm{G}_\wedge$, which is a two-sided inverse of~$\mathrm{P}$ and satisfies
\begin{equation*}
  \supp(\mathrm{G}_\wedge f) \subset J^-(\supp f)
\end{equation*}
for all test sections~$f$, is called a \IdxMain{adv-prop}\emph{advanced propagator} or \emph{advanced Green's operator}.\footnote{Note that our definition of the support of the retarded and advanced propagators is exactly opposite to that in \cite{bar:2007,bar:2012} and also~\cite{dimock:1980}.}
We say that $\mathrm{P}$ is \IdxMain{green-hyp}\emph{Green-hyperbolic} if it admits unique retarded and advanced propagators when restricted to a globally hyperbolic region.

Given a linear differential operator~$\mathrm{Q}$ such that $\mathrm{P} \circ \mathrm{Q} = \mathrm{Q} \circ \mathrm{P}$, \ie, $\mathrm{Q}$ commutes with $\mathrm{P}$, then it also commutes with the propagators of $\mathrm{P}$.
That is, one finds
\begin{equation*}
  \mathrm{G}_\vee \mathrm{Q} f = \mathrm{Q} \mathrm{G}_\vee f
  \quad \text{and} \quad
  \mathrm{G}_\wedge \mathrm{Q} f = \mathrm{Q} \mathrm{G}_\wedge f
\end{equation*}
for all $f \in \mathcal{E}(M, E)$.

If $\mathrm{P}$ is Green-hyperbolic, then the transpose operator ${}^t\mathrm{P}$ on sections of the dual bundle $E^*$ is also Green-hyperbolic; we denote its propagators by~$\mathrm{G}^t_\vee$ and~$\mathrm{G}^t_\wedge$.
They are closely related to the propagators of~$\mathrm{P}$ and one finds
\begin{equation*}
  \mathrm{G}^{\vphantom{t}}_\vee = {}^t(\mathrm{G}^t_\wedge)
  \quad \text{and} \quad
  \mathrm{G}^{\vphantom{t}}_\wedge = {}^t(\mathrm{G}^t_\vee).
\end{equation*}

Since the propagators are regular, they can be uniquely extended to operators $\mathcal{E}'(M, E) \to \mathcal{D}'(M, E)$.
Although the propagators are not properly supported, they can also be defined for some non-compactly supported sections.
The geometry of $(M, g)$ enables us to define further types of `compact' support:
We say that a (distributional) section~$u$ is future or past compact if there exists a Cauchy surface~$\Sigma$ such that
\begin{equation*}
  \supp u \subset J^+(\Sigma)
  \quad \text{or} \quad
  \supp u \subset J^-(\Sigma),
\end{equation*}
respectively.
Denote by the subscripts `fc' and `pc' the subsets of (distributional) sections of future and past compact support.
Via the transpose propagators $\mathrm{G}^t_\vee, \mathrm{G}^t_\wedge$, we can then uniquely extend the retarded propagator to $\mathcal{D}'_{\mathrm{fc}}(M, E) \to \mathcal{D}'(M, E)$ and the advanced propagator to $\mathcal{D}'_{\mathrm{pc}}(M, E) \to \mathcal{D}'(M, E)$.

Let $E$ be endowed with a bundle metric $(\cdot\,,\cdot)$.
The \IdxMain{formal-adj}\emph{formal adjoint}~$\mathrm{P}^*$ of~$\mathrm{P}$ with respect to $(\cdot\,,\cdot)$ is given by
\begin{equation*}
  \int_M (\mathrm{P} f, h)\, \mu_g = \int_M (f, \mathrm{P}^* h)\, \mu_g
\end{equation*}
for all $f, h \in \mathcal{E}(M, E)$ such that $\supp f \cap \supp h$ is compact.
If $\mathrm{P}^* = \mathrm{P}$, the operator is called \IdxMain{formal-self-adj}\emph{formally self-adjoint}.
In that case, it follows from the last paragraph that
\begin{equation*}
  \int_M (\mathrm{G}_\vee f, h)\, \mu_g = \int_M (f, \mathrm{G}_\wedge h)\, \mu_g.
\end{equation*}

As indicated above, wave operators on globally hyperbolic manifolds play an important role and, in fact, they are particularly well-behaved \cite{friedlander:1975,bar:2007}:

\begin{theorem}
  Any normally hyperbolic operator~$\mathrm{P}$ on a globally hyperbolic manifold admits unique retarded~$\mathrm{G}_\vee$ and advanced propagators~$\mathrm{G}_\wedge$.
\end{theorem}

It is not difficult to extend this result to pre-normally hyperbolic operators on globally hyperbolic spacetimes.
Namely, given pre-normally hyperbolic operators~$\mathrm{P}$ and~$\mathrm{Q}$ such that $\mathrm{P} \circ \mathrm{Q}$ is normally hyperbolic, $\mathrm{P}$ possesses unique retarded and advanced propagators
\begin{equation*}
  \widetilde{\mathrm{G}}_\vee = Q \circ \mathrm{G}_\vee
  \quad \text{and} \quad
  \widetilde{\mathrm{G}}_\wedge = Q \circ \mathrm{G}_\wedge,
\end{equation*}
where $\mathrm{G}_\vee, \mathrm{G}_\wedge$ are the propagators for the composite operator $\mathrm{P} \circ \mathrm{Q}$.

\subsection{Causal propagator}
\label{sub:causal_propagator}

The \IdxMain{causal-prop}\emph{causal propagator} is defined as the difference of the retarded and advanced propagator
\begin{equation*}
  \mathrm{G} \defn \mathrm{G}_\vee - \mathrm{G}_\wedge.
\end{equation*}
From the support properties of the retarded and advanced propagator it is clear that $\supp(\mathrm{G} f)\linebreak[1] = J(\supp f)$ for all $f \in \mathcal{D}(M, E)$.
In \cref{sub:kg_eq_comm} we will see that the causal propagator, or rather the associated distribution via Schwartz's kernel theorem, may also be called the \IdxMain{comm-dist}\emph{commutator distribution} or \emph{Pauli--Jordan distribution}.

By the regularity of the retarded and advanced propagators, it is clear that $\mathrm{G}$ extends to an operator $\mathcal{E}'(M, E) \to \mathcal{D}'(M, E)$.
Noting the support property of $\mathrm{G}$, this statement can be strengthened to extend the causal propagator to $\mathcal{D}'_{\mathrm{tc}}(M, E) \to \mathcal{D}'(M, E)$.
Here we have denoted by a subscript `tc' the space of (distributional) sections of timelike compact support, \ie, the sections~$u$ such that
\begin{equation*}
  \supp u \subset J^+(\Sigma_1) \cap J^-(\Sigma_2)
\end{equation*}
for two Cauchy surfaces $\Sigma_1, \Sigma_2$.

Every smooth and spacelike compact solution of the homogeneous differential equation $\mathrm{P} u = 0$ propagating on a globally hyperbolic spacetime $(M, g)$ with Green-hyperbolic operator~$\mathrm{P}$ can be obtained by applying $\mathrm{G}$ to a test section~$f$.
% \begin{equation*}
%   \mathrm{P} u = 0
%   \quad \Longrightarrow \quad
%   u = \mathrm{G} f \;\text{for some}\; f \in \mathcal{D}(M, E).
% \end{equation*}
In fact, if we denote by $\mathcal{E}_{\mathrm{sc}}(M, E)$ the smooth sections of~$E$ with  spacelike compact support, then we find the exact sequence
\begin{equation*}
  \{0\} \longrightarrow \mathcal{D}(M, E) \overset{\mathrm{P}}{\longrightarrow} \mathcal{D}(M, E) \overset{\mathrm{G}}{\longrightarrow} \mathcal{E}_{\mathrm{sc}}(M, E) \overset{\mathrm{P}}{\longrightarrow} \mathcal{E}_{\mathrm{sc}}(M, E).
\end{equation*}
This sequence also entails the fact that the kernel of~$\mathrm{G}$ is given by $\mathrm{P} \mathcal{D}(M, E)$.
In other words, $f - f' = \mathrm{P} h$ for some $f, f', h \in \mathcal{D}(M, E)$ implies that $\mathrm{G} f = \mathrm{G} f'$.

Closely related to the existence of a causal propagator is the question whether the \Idx{cauchy-prop}\emph{Cauchy problem} is well-posed.
The Cauchy problem for the wave equation $\mathrm{P} u = 0$ on a globally hyperbolic manifold $(M, g)$ is the following:

\textit{Given a Cauchy surface~${\iota : \Sigma \to M}$ with normal vector field~$n$, does there exist a unique section $u \in \mathcal{E}(M, E)$ such that
\begin{equation*}\begin{cases}
  \mathrm{P} u = 0, \\
  \pb{\iota} u = u_0, \\
  \pb{\iota} \nabla_{\!n} u = u_1
\end{cases}\end{equation*}
and the solution~$u$ depends continuously on the data $u_0, u_1 \in \mathcal{E}(\Sigma, \pb{\iota} E)$.}

This question can be answered in the positive for normally hyperbolic operators~$\mathrm{P}$ on globally hyperbolic spacetimes.
With the appropriate modifications, the Cauchy problem can also be formulated for pre-normally hyperbolic operators.
Also in that case Cauchy problem is well-posed~\cite{wrochna:2013}.
For general Green-hyperbolic operators the Cauchy problem is more complicated and it is not obvious whether the Cauchy problem is well-posed.

%% file: combinatorics.tex
%!TEX root = master.tex

\chapter{Enumerative combinatorics}
\label{cha:combinatorics}

\IdxRanBegin{combinat}

\section*{Summary}

In this chapter we discuss the results obtained by Fewster and the author in~\cite{fewster:2014a} on the enumeration of the run structures of permutations.
Some of the results stated here will can be applied in the study of the moment problem in quantum field theory and the connection will be discussed briefly in \cref{sec:applications}.

The first section (\cref{sec:permutations}) gives a summary of the elementary definitions for (linear) permutations and circular permutations.
Then, the subsections of the second section (\cref{sub:run_atomic,sub:run_circular,sub:run_linear}) deal, respectively, with the enumeration of the run structure of atomic, circular and linear permutations.
Using a suitable decomposition, this is accomplished in each case by reducing the enumeration problem to that for atomic permutations.
In the third section (\cref{sec:valleys}) we apply and extend the methods developed in the preceeding sections to enumerate the valleys of permutations, thereby reproducing a result of Kitaev~\cite{kitaev:2007}.
Finally, in the last section (\cref{sec:applications}), we discuss the original motivation of the work~\cite{fewster:2014a} and other possible applications.

\section{Permutations}
\label{sec:permutations}

Let us adopt the following notation for integer intervals: $[a \lddots b] \defn [a, b] \cap \NN = \{a, a+1, \dotsc, b\}$ with the special case $[n] \defn [1 \lddots n]$.

\subsection{Linear permutations}
\label{sub:perm_linear}

Given a set~$S$, a \IdxMain{lin-perm}\emph{(linear) permutation} of~$S$ is a bijection $\sigma : S \to S$.
In the \IdxMain{two-line}\emph{two-line notation} of the permutation of a finite set is written as
\begin{equation*}
  \sigma = \left( \begin{array}{cccc}
    a & b & c & \cdots \\
    \sigma(a) & \sigma(b) & \sigma(c) & \cdots
  \end{array} \right),
\end{equation*}
where $a, b, c, \dotsc \in S$.
It is clear that the order of elements in the first line is irrelevant as long as the second line is ordered accordingly.

The set of all bijection on~$S$ forms the \IdxMain{lin-perm-group}\emph{(linear) permutation group}~$\perms_S$ of~$S$; the group operation is the composition~$\,\circ\,$ of functions.
There are $n!$~permutations in~$\perms_S$ if $S$ is a set of $n$~elements.
In the special case that $S = [n]$, one writes $\perms_n \defn \perms_{[n]}$.

Given a (strict) total order on a finite set~$S$, \ie, a binary relation $\,<\,$ that is transitive and trichotomous, the first line will always be ordered in the natural order.
Since every finite ordered set~$S$ is isomorphic to a subset $[n]$ of the natural numbers with the standard ordering, this identification will tacitly be assumed henceforth.
Therefore the first line can be disposed of and one can use instead the \IdxMain{one-line}\emph{one-line notation}
\begin{equation*}
  \sigma = \left( \begin{array}{cccc}
    \sigma(1) & \sigma(2) & \sigma(3) & \cdots
  \end{array} \right).
\end{equation*}
We see that a permutation is equivalent to a change of the linear order of the set $S$.
Further condensing the one-line notation, the permutations of a finite ordered set can be identified with words
\begin{equation*}
  \sigma = \sigma_1 \sigma_2 \sigma_3 \,\dotsm,
\end{equation*}
where the shorthand $\sigma_i = \sigma(i)$, $i \in S$, was used.

\subsection{Circular permutations}
\label{sub:perm_circular}

Instead of considering different orderings of a set along a line, one can study different arrangements of the elements of the set on an oriented circle (turning the circle over produces in general a different permutation).

Let $S$ be a set with a distinguished element~$e$.
The \IdxMain{circ-perm}\emph{circular permutations} of~$S$ are the bijections $\sigma : S \to S$ that preserve $e$, \ie, $\sigma(e) = e$.
The circular permutations of~$S$ also form a group, the \IdxMain{circ-perm-group}\emph{circular permutation group}~$\cperms_S$; if $S = [n]$, define $\cperms_n \defn \cperms_{[n]}$.
Clearly, $\cperms_S$ is a subgroup of~$\perms_S$ and its cardinality is $(n-1)!$ if that of~$S$ is~$n$.

A \IdxMain{cyclic-ord}\emph{cyclic order} is a ternary relation $[\cdot\,, \cdot\,, \cdot]$ on a set~$S$ is a set of triples $T \subset S^{\times 3}$ that satisfies
\begin{enumerate}[(a)]
  \item $[a, b, c] \in T$ implies $[b, c, a] \in T$ (cyclicity),
  \item $[a, b, c] \in T$ implies $[c, b, a] \nin T$ (asymmetry),
  \item $[a, b, c], [a, c, d] \in T$ implies $[a, b, d] \in T$ (transitivity),
  \item $a, b, c$ mutually distinct implies either $[a, b, c] \in T$ or $[c, b, a] \in T$ (totality).
\end{enumerate}
Every (strict) total order $\,<\,$ induces a cyclic order by setting $[a, b, c] \in T$ if and only if $a < b < c$ or $b < c < a$ or $c < b < a$.
Conversely, every cyclic order induces different possible linear orders.
Namely, setting $a < b$ if and only if $[a, b, e]$ for fixed $e \in S$ yields a total order on $S \setminus \{e\}$ which can be extended to a linear order on~$S$ by defining $e$ as either the minimal or maximal element of the set.
Consequently, the natural choice for the distinguished element of a finite ordered set~$S$ in the construction above is the minimal or maximal element of~$S$; we will always choose the minimal element.

The different notations for linear permutations generalize straightforwardly to circular permutations.
Given a circular permutation~$\sigma$ of a finite ordered set~$S$, write
\begin{equation*}
  \sigma
  = \left( \begin{array}{cccc}
    1 & \sigma(2) & \sigma(3) & \cdots
  \end{array} \right).
\end{equation*}
To distinguish circular permutations more clearly from linear ones and to highlight the circular symmetry, we modify the word-notation in the case of a circular permutation~$\sigma$ to
\begin{equation*}
  \sigma = \dot{1} \sigma_2 \sigma_3 \,\dotsm\, \dot{\sigma}_n,
\end{equation*}
in analogy with the notation for repeating decimals when representing rational numbers.
Moreover, for convenience we define $\sigma(n + 1) \defn \sigma(1)$ for all circular permutations $\sigma$ of $n$-element sets.

\subsection{Atomic permutations}
\label{sub:perm_atomic}
\IdxMain{atomic-perm}

Let us introduce a special subgroup of~$\perms_S$ for a finite ordered set~$S$ with minimal element~$e$ maximal element~$m$.

\begin{definition}
  Define the \IdxMain{rise-atomic}\emph{rising atomic permutations}\footnote{The rationale for this naming should become clear later when we see that arbitrary permutations can be decomposed into atomic permutations, but no further.} $\aperms^+_S \subset \perms_S$ as those permutations that satisfy $\sigma(e) = e$ and $\sigma(m) = m$ for all $\sigma \in \aperms^+_S$.
  The \IdxMain{fall-atomic}\emph{falling atomic permutations} $\sigma \in \aperms^-_S$ are the reversed rising atomic permutations, \ie, $\sigma(e) = m$ and $\sigma(m) = e$.
\end{definition}

Naturally, the cardinality of $\aperms^\pm_S$ is $(n-2)!$.
If $S = [n]$, we write $\aperms^+_n$ ($\aperms^-_n$) and see that it is the set of permutations of the form $1\,\dotsm\,n$ ($n\,\dotsm\,1$).

Let us discuss the significance of the atomic permutations.
We say that a permutation $\sigma \in \perms_S$ of~$S$ \emph{contains an atomic permutation} $\pi \in \aperms_T$, $T \subset S$, if $\pi$ can be considered a subword of~$\sigma$.
The atomic permutation~$\pi$ in~$\sigma$ is called \emph{inextendible} if $\sigma$ contains no other atomic permutation $\pi' \in \aperms_{T'}$, $T' \subset S$, such that $T \subsetneq T'$.

In particular, any permutation $\sigma \in \perms_S$ of~$S$ with $\abs{S} \geq 2$ contains an inextendible atomic permutation $\pi \in \perms_T$ of a subset $T \subset S$ that contains both the smallest and the largest element of~$S$.
That is, if $S = [n]$ and we consider $\sigma$ as a word, it contains a subword~$\pi$ of the form $1 \,\dotsm\, n$ or $n \,\dotsm\, 1$.
The permutation $\pi$ will be called the \IdxMain{princ-atom}\emph{principal atom} of~$\sigma$.

\begin{proposition}\label{prop:decomposition}
  Any permutation $\sigma \in \perms_S$ of a finite set $S \subset \NN$ can be uniquely decomposed into a tuple $(\pi^1, \dotsc, \pi^k)$ of inextendible atomic permutations $\pi^i \in \aperms_{T_i}$, $T_i \subset S$ (non-empty) such that $\pi^i_{\abs{T_i}} = \pi^{i+1}_1$ for all $i < k$ and $\cup_i T_i = S$.
  We call $\pi^i$ the \IdxMain{atoms}\emph{atoms} of $\sigma$.
\end{proposition}
\begin{proof}
  \emph{Existence:}
  It is clear that any permutation of a set of $1$~or $2$~elements is an atomic permutation.
  Suppose, for some $n\ge 3$, that all permutations of $n-1$ elements or less can be decomposed into inextendible atomic permutations.
  Without loss of generality, we show that any non-atomic permutation $\sigma \in \perms_n$ also has a decomposition into inextendible atomic permutations.
  Regarding $\sigma$ as a word, we can write $\sigma = \alpha \cdot n \cdot \omega$, where $\alpha$ and $\omega$ are non-empty subwords.
  Notice that the permutations $\alpha \cdot n$ and $n \cdot \omega$ have a unique decomposition by assumption.
  Since an atomic permutation begins or ends with the largest element, we find that a decomposition of~$\sigma$ into inextendible atomic permutations is given by the combination of the decompositions of $\alpha \cdot n$ and $n \cdot \omega$.

  \emph{Uniqueness:} This is clear from the definition of inextendibility.
\end{proof}

Because of this property, the atomic permutations will prove to be very useful.

\subsection{Mountaineering}

Given a (linear or circular) permutation~$\sigma$ of an ordered set~$S$ of cardinality~$n$, a position $i < n$ is a \IdxMain{perm-desc}\emph{descent} of~$\sigma$ if $\sigma(i) > \sigma(i+1)$.
Any $i < n$ of a permutation~$\sigma$ that is not a descent is called an \IdxMain{perm-asc}\emph{ascent} of~$\sigma$.
For example, the permutation \perm{52364178} has the descents $1, 4, 5$ and the ascents $2, 3, 6, 7$, whereas the circular permutation \cperm{14536782} has the descents $3, 7, 8$ and the ascents $1, 2, 4, 5, 6$.

All the descents of a (linear or circular) permutation $\sigma$ can be collected in the \IdxMain{desc-set}\emph{descent set}
\begin{equation*}
  D(\sigma) \defn \{ i \mid i \text{ is a descent of } \sigma \}.
\end{equation*}
It is an elementary exercise in enumerative combinatorics to count the number of linear permutations of~$[n]$ whose descent set is given by a fixed $S \subseteq [n-1]$.
Let $S = \{s_1, s_2, \dotsc, s_k \}$ be an ordered subset of~$[n-1]$, then~\cite[Thm.~1.4]{bona:2012}
\begin{equation*}
  \beta(S)
  \defn \big|\{ \sigma \in \perms_n \mid D(\sigma) = S \}\big|
  = \sum_{T \subseteq S} (-1)^{|S-T|} \binom{n}{s_1, s_2 - s_1, s_3 - s_2, \dotsc, n - s_k}.
\end{equation*}
This result can also be adapated to circular permutations.

Related to the notions of ascents and descents are the concepts of \IdxMain{peak}\emph{peaks} and \IdxMain{valley}\emph{valleys}.
A peak occurs at position $i \in [2 \lddots n-1]$ of a linear permutation~$\sigma$ if $\sigma(i-1) < \sigma(i) > \sigma(i+1)$, whereas a valley occurs in the opposite situation $\sigma(i-1) > \sigma(i) < \sigma(i+1)$.
Again, this notion can be generalized to circular permutations, where, additionally, $1$ is always a valley and $n$ is a peak if and only if $\sigma(n) > \sigma(n-1)$.
In the example above, $4$ is a peak $2, 6$ are valleys of \perm{52364178}, whereas $3, 7$ are peaks and $1$ is a valley for \cperm{14536782}, see also \cref{fig:run_graphs}.

\section{Run structures}

\begin{definition}
  A \IdxMain{run}\emph{run}~$r$ of a (linear or circular) permutation $\sigma$ is an interval $[i \lddots j]$ such that $\sigma(i) \gtrless \sigma(i+1) \gtrless \dotsb \gtrless \sigma(j)$ is a monotone sequence, either increasing or decreasing, and so that it cannot be extended in either direction; its length is defined to be $j-i$.
  If $\sigma$ is a permutation of an $n$-element set, the collection of the lengths of all runs gives a partition~$p$ of $n-1$ (linear permutations) or $n$ (circular permutations).
  The partition $p$ is called the \IdxMain{run-struct}\emph{run structure} of~$\sigma$.
\end{definition}

It follows that a run starts and ends at peaks, valleys or at the outermost elements of a permutation.
For example, the permutation \perm{52364178} has runs $[1 \lddots 2]$, $[2 \lddots 4]$, $[4 \lddots 6]$, $[6 \lddots 8]$ with lengths $1, 2, 2, 2$, whereas the circular permutation \cperm{14536782} has runs $[1 \lddots 3]$, $[3 \lddots 4]$, $[4 \lddots 7]$, $[7 \lddots 9]$, of lengths $2, 1, 3, 2$.
Representing these runs by their image under the permutation, they are more transparently written as \perm{52}, \perm{236}, \perm{641}, \perm{178} and \perm{145}, \perm{53}, \perm{3678}, \perm{821} respectively.
The runs of permutations can also be neatly represented as directed graphs as shown in \cref{fig:run_graphs}.
In these graphs the peaks and valleys correspond to double sinks and double sources.
\begin{figure}
  \centering
  \includegraphics{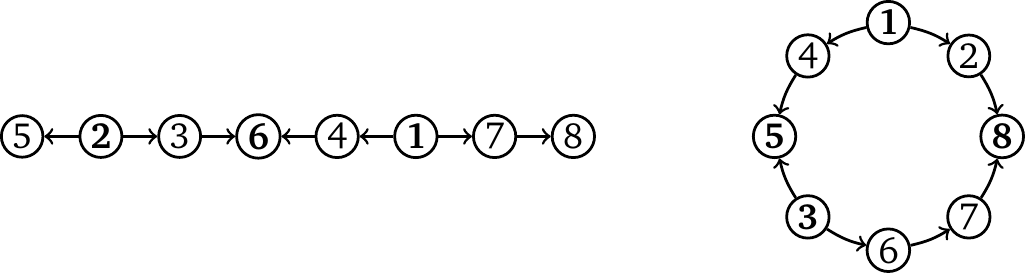}
  \caption{The two directed graphs representing the runs of the linear permutation \protect\perm{52364178} (left) and the circular permutation \protect\cperm{14536782} (right). Peaks and valleys are indicated by boldface numbers.}
  \label{fig:run_graphs}
\end{figure}

Motivated by a problem in mathematical physics \cite{fewster:2012b} (see also \cref{sec:applications}), we are interested in the following issue, which we have not found discussed in the literature.
By definition, the run structure associates each permutation $\sigma \in \cperms_n$ with a partition~$p$ of~$n$.
For example, \cperm{14536782} and \cperm{13452786} both correspond to the same partition $1 + 2 + 2 + 3$ of $8$.
Our interest is in the inverse problem: given a partition $p$ of $n$, we ask for the number~$Z_{\cperms}(p)$ of circular permutations whose run structure is given by~$p$.
One may consider similar questions for other classes of permutations, with slight changes; for example, note that the run structure of a permutation $\sigma \in \perms_n$ is a partition of $n-1$.

To put the research in~\cite{fewster:2014a} in perspective with the existing literature on the enumerative cominatorics of permutations, a short remark is in order:
The enumeration of permutations according to their run structure was already discussed by André~\cite{andre:1895} for alternating permutations, \ie, permutations that alternate between ascents and descents.
In~\cite{brown:2007} the enumeration of linear permutations according to the order and length of their runs was studied, so obtaining a map to compositions, rather than partitions.
In contrast to this approach, the method discussed in~\cite{fewster:2014a} was designed to facilitate computation; for the application in~\cite{fewster:2012b} calculations were taken up to $65$~runs using exact integer arithmetic in Maple\texttrademark~\cite{maple:16}.

\subsection{Atomic permutations}
\label{sub:run_atomic}
\IdxRanBegin{run-atom}

We now begin the enumeration of atomic permutations according to their run structure.
That is, for every partition~$p$ of ${n-1}$ we aim to find the number $Z_{\aperms}(p)$ of atomic permutations~$\aperms^\pm_n$ of length~$n$.

Observe that any $\sigma \in \aperms^+_n$ can be extended to a permutation in $\aperms^+_{n+1}$ by replacing $n$ with $n+1$ and reinserting $n$ in any position after the first and before the last.
Thus, \perm{13425} can be extended to \perm{153426}, \perm{135426}, \perm{134526} or \perm{134256}.
Every permutation in $\aperms^+_{n+1}$ arises in this way, as can be seen by reversing the procedure.
The effect on the run lengths can be described as follows.

\begin{description}[leftmargin=*]
  \item[Case 1:]\phantomsection\label{item:case_1}
    The length of one of the runs can be increased by one by inserting $n$ either at
    \begin{enumerate}[leftmargin=*]
      \item
        the end of an increasing run if it does not end in $n+1$, thereby increasing its length (\eg, \perm{13425} $\to$ \perm{134526})
        \\[.5em]
        \includegraphics[page=1]{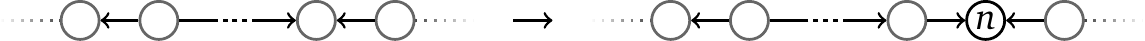}
      \item
        the penultimate position of an increasing run, thereby increasing its own length if it ends in $n+1$ (\eg, \perm{13425} $\to$ \perm{135426}) or increasing the length of the following decreasing run otherwise (\eg, \perm{13425} $\to$ \perm{134256})
        \\[.5em]
        \includegraphics[page=2]{fig_run_lengths.pdf}
    \end{enumerate}
  \item[Case 2:]\phantomsection\label{item:case_2}
    Any run of length $i+j \geq 2$ becomes three run of lengths $1$, $i$ and $j$ if we insert $n$ either after
    \begin{enumerate}[leftmargin=*]
      \item
        $i$ elements of an increasing run (\eg, \perm{13425} $\to$ \perm{153426} exemplifies $i=1$, $j=1$)
        \\[.5em]\hspace*{-10pt}
        \includegraphics[page=3]{fig_run_lengths.pdf}
      \item
        $i+1$ elements of a decreasing run (\eg, \perm{14325} $\to$ \perm{143526} for $i=1$, $j=1$)
        \\[.5em]\hspace*{-10pt}
        \includegraphics[page=4]{fig_run_lengths.pdf}
    \end{enumerate}
\end{description}
An analogous argument can be made for the falling atomic permutations $\aperms^-_n$.

Notice that every partition of a positive integer~$n$ can be represented by monomials in the ring of polynomials\footnote{If one wants to encode also the order of the run (\eg, to obtain a map from permutations of length $n$ to the compositions of $n$), one can exchange the polynomial ring with a noncommutative ring. Alternatively, if one wants to encode the direction of a run, one could study instead the ring $\ZZ[x_1, y_1, x_2, y_2, \dotsc]$, where $x_i$ denotes an increasing run of length $i$ and $y_j$ encodes a decreasing run of length $j$.} $\ZZ[x_1, x_2, \dotsc, x_n]$.
Namely, we can express a partition $p = p_1 + p_2 + \dotsb + p_k$ as $x_{p_1} x_{p_2} \,\cdots\, x_{p_k}$ (for example, the partition $1 + 2 + 2 + 3$ of~$8$ is written as $x_1 x_2^2 x_3$).

Now, let $p$ be a partition and $X$ the corresponding monomial.
To this permutation there correspond $Z_{\aperms}(p)$ permutations in $\aperms^\pm_n$ which can be extended to permutations in $\aperms^\pm_{n+1}$ in the manner described above.
Introducing the (formally defined) differential operator
\begin{equation}\label{eq:diff}
  \mathcal{D}   \defn \mathcal{D}_0 + \mathcal{D}_+
  \quad\text{with}\quad
  \mathcal{D}_0 \defn \sum_{i=1}^\infty x_{i+1} \pd{}{x_i},
  \;\;
  \mathcal{D}_+ \defn \sum_{i,j \,\geq\, 1} x_1 x_i x_j \pd{}{x_{i+j}},
\end{equation}
we can describe this extension in terms of the action of~$\mathcal{D}$ on~ $X$.
We say that $\mathcal{D}_0$ is the \emph{degree-preserving} part of $\mathcal{D}$; it represents the \hyperref[item:case_1]{case 1} of increasing the length of a run: the differentiation $\partial/\partial x_i$ removes one of the runs of length~$i$ and replaces it by a run of length~$i+1$, keeping account of the number of ways in which this can be done. Similarly,
\hyperref[item:case_2]{case 2} of splitting a run into $3$~parts is represented by the \emph{degree-increasing} part $\mathcal{D_+}$.
For example, each of the $7$~atomic permutations corresponding to the partition $1 + 1 + 3$ can be extended as
\begin{equation*}
   \mathcal{D} x_1^2 x_3 = 2 x_1 x_2 x_3 + x_1^2 x_4 + x_1^4 x_2,
\end{equation*}
\ie, each can be extended to two atomic permutations corresponding to the partitions $1 + 2 + 3$, one corresponding to $1 + 1 + 4$ and one to $1 + 1 + 1 + 1 + 2$.

Therefore, starting from the trivial partition $1$ of $1$, represented as $x_1$, we can construct a recurrence relation for polynomials $A_n = A_n(x_1, x_2, \dotsc, x_n)$ which, at every step $n \geq 1$, encode the number of atomic permutations $Z_{\aperms}(p)$ of length $n+1$ with run structure given by a partition~$p$ of~$n$ as the coefficients of the corresponding monomial in $A_n$.
The polynomial~$A_n$, accordingly defined by
\begin{equation}\label{eq:A_rel_Z}
  A_n = \sum_{p \vdash n} Z_{\aperms}(p) \prod_{i=1}^n x_i^{p(i)},
\end{equation}
where the sum is over all partitions $p$ of $n$ and $p(i)$ denotes the multiplicity of $i$ in the partition $p$, can thus be computed from the recurrence relation
\begin{subequations}\label{eq:A_recurrence}
  \begin{align}
    A_1 & \defn x_1, \\
    A_n & \defn \mathcal{D} A_{n-1}, \qquad (n \geq 2).
  \end{align}
\end{subequations}
We say that the polynomials~$A_n$ enumerate the run structure of the atomic permutations.

We summarize these results in the following proposition:
\begin{proposition}\label{prop:Z_A}
  The number $Z_{\aperms}(p)$ of rising or falling atomic permutations of length $n-1$ corresponding to a given run structure (\ie, a partition $p$ of $n$), is determined by the polynomial $A_n$ via~\eqref{eq:A_rel_Z}.
  The polynomials $A_n$ satisfy the recurrence relation~\eqref{eq:A_recurrence}.
\end{proposition}

Note that atomic permutations always contain an odd number of runs and thus $Z_{\aperms}(p)$ is zero for even partitions $p$.

It will prove useful to combine all generating functions $A_n$ into the formal series
\begin{equation*}
  \mathcal{A}(\lambda)
  \defn \sum_{n=0}^\infty A_{n+1} \frac{\lambda^n}{n!}
  = \sum_{n=0}^\infty \mathcal{D}^n A_1 \frac{\lambda^n}{n!},
\end{equation*}
which can be expressed compactly as the exponential
\begin{equation*}
  \mathcal{A}(\lambda) = \exp(\lambda \mathcal{D}) A_1.
\end{equation*}

The first few $A_n$ are given by
\begin{align*}
  A_2 & = x_2 \\
  A_3 & = x_3 + x_1^3 \\
  A_4 & = x_4 + 5 x_2 x_1^2 \\
  A_5 & = x_5 + 7 x_3 x_1^2 + 11 x_2^2 x_1 + 5 x_1^5 \\
  A_6 & = x_6 + 9 x_4 x_1^2 + 11 x_2^3 + 38 x_3 x_2 x_1 + 61 x_2 x_1^4.
\end{align*}
from which we can read off that there is $1$ permutation in $\aperms^\pm_6$ corresponding to the trivial partition $5 = 5$, $7$ corresponding to the partition $5 = 1 + 1 + 3$, $11$ corresponding to $5 = 1 + 2 + 2$ and $5$ corresponding to $5 = 1 + 1 + 1 + 1 + 1$.
As a check, we note that $1 + 7 + 11 + 5 = 24$, which is the total number of elements of $\aperms^\pm_6$; similarly, the coefficients in the expression for $A_6$ sum to $120$, the cardinality of $\aperms^\pm_7$. A direct check that the coefficients in $A_n$ sum to $(n-1)!$ for all $n$ will be given in the last paragraph of \cref{sec:valleys}.

The first degree term $A^{(1)}_n$ of $A_n$ is $x_n$ as can be seen by a trivial induction using  $A^{(1)}_n = \mathcal{D_0} A^{(1)}_{n-1}$, which follows from the recurrence relation~\eqref{eq:A_recurrence}.
Therefore $Z_{\aperms}(n) = 1$.

For $A_n^{(k)}$ with $k > 1$ also the effect of $\mathcal{D}_+$ has to be taken into account, complicating things considerably.
Nevertheless, the general procedure is clear: once $A_m^{(k-2)}$ is known for all $m < n$, $A_n^{(k)}$ can be obtained as
\begin{equation*}
  A_n^{(k)} = \mathcal{D}_0 A_{n-1}^{(k)} + \mathcal{D}_+ A_{n-1}^{(k-2)} = \sum_{m=k-1}^{n-1} \mathcal{D}_0^{n-m-1} \mathcal{D}_+ A_m^{(k-2)}.
\end{equation*}
Here one can make use of the following relation.
Applying $\mathcal{D}_0$ repeatedly to any monomial $x_{i_1} x_{i_2} \,\dotsm\, x_{i_k}$ of degree $k$ yields, as a consequence of the Leibniz rule,
\begin{equation}\label{eq:repeated_D}
  \mathcal{D}_0^n x_{i_1} x_{i_2} \,\dotsm\, x_{i_k} = \sum_{\substack{j_1, j_2, \dotsc, j_k \,\geq\, 0 \\ j_1 + j_2 + \dotsb + j_k = n}} \binom{n}{j_1, j_2, \dotsc, j_k}\, x_{i_1 + j_1} x_{i_2 + j_2} \,\dotsm\, x_{i_k + j_k}.
\end{equation}

This observation provides the means to determine the third degree term $A_n^{(3)}$.
Applying $\mathcal{D}_+$ to any $A_m^{(1)} = x_m$ with $m \geq 2$ produces $x_1 x_p x_q$ with $p+q = m$ and $p, q \geq 1$.
Moreover, the repeated action of $\mathcal{D}_0$ on $x_1 x_p x_q$ is described by~\eqref{eq:repeated_D} and thus
\begin{equation*}
  A_n^{(3)} = \sum_{\substack{p, q, r, s, t \,\geq\, 0 \\ 1 + p + q + r + s + t = n}} \binom{n - p - q - 1}{r, s, t}\, x_{1+r} x_{p+s} x_{q+t}.
\end{equation*}
After some algebra this yields
\begin{proposition}
  The third degree term $A_n^{(3)}$ of the polynomial $A_n, n \geq 3,$ is given by
  \begin{equation}\label{eq:third_degree}
    A_n^{(3)}
    = \sum_{\substack{i, j, k \,\geq\, 1 \\ i + j + k = n}} \sum_{q=1}^{k} \frac{n-q-1}{n-q-j}\, \binom{n - q - 2}{i - 1, j - 1, k - q}\, x_i x_j x_k.
  \end{equation}
\end{proposition}

The equation~\eqref{eq:third_degree} for the third degree term $A_n^{(3)}$ can be rewritten into a formula for $Z_{\aperms}(p_1 + p_2 + p_3)$, \ie, the number of permutations of $[n+1]$ that start with $1$, end with $n+1$ and have three runs of lengths $p_1, p_2, p_3$, by changing the first sum to a sum over $i, j, k \in \{ p_1, p_2, p_3 \}$.
In particular, this gives rise to three integer series for the special cases
\begin{equation*}
  Z_{\aperms}(n+n+n), \quad Z_{\aperms}(1+n+n), \quad Z_{\aperms}(1+1+n),
\end{equation*}
with $n \in \NN$.

The first series
\begin{align*}
  Z_{\aperms}(n+n+n)
  & = \sum_{q=1}^n \frac{3n-q-1}{2n-q}\, \binom{3n-q-2}{n-1,n-1,n-q}\\
  & = 1, 11, 181, 3499, 73501, 1623467, \dotsc \qquad (n \geq 1)
\end{align*}
gives the number of atomic permutations with three runs of equal length $n$.
It does not appear to be known in the literature nor can it be found in the OEIS \cite{oeis} and the existence of closed form expression is currently unkown.
For the second series, however, a simple closed form can be found:
\begin{align*}
  Z_{\aperms}(1+n+n)
  & = \sum_{q=1}^n \bigg( \binom{2n-q}{n-1} + \binom{2n-q-1}{n-1} \bigg) + \frac{1}{2}\, \binom{2n}{n}\\
  & = 2\, \binom{2n}{n} - 1
    = 11, 39, 139, 503, 1847, \dotsc, \qquad (n \geq 2)
\end{align*}
is the number of atomic permutations in $\aperms^\pm_{2n+2}$ with two runs of length $n$. One may understand this directly: there are $\binom{2n}{n}$
permutations in which the length $1$ run is between the others and $\binom{2n}{n}-1$ in which it is either first or last.
The third series, $Z_{\aperms}(1+1+n)$, \ie, the number of atomic permutations in $\aperms^\pm_{n+3}$ with two runs of length $1$, is given by the odd numbers bigger than $3$:
\begin{equation*}
  Z_{\aperms}(1+1+n) = 2n + 1 = 5, 7, 9, 11, 13, 15, \dotsc, \qquad (n \geq 2).
\end{equation*}

Observe that terms of the form $x_1^n$ in $A_n$ encode alternating permutations,
which were already investigated by André in the 1880's \cite{andre:1881}.
As a consequence of his results, we find that the alternating atomic permutations are enumerated by the \IdxMain{sec-numbers}\emph{secant numbers} $S_n$, the coefficients of the Maclaurin series of $\sec x = S_0 + S_1 x^2/2! + S_2 x^4/4! + \dotsb$,
\begin{equation*}
  Z_{\aperms}\bigg(\sum_{i=1}^{2n+1} 1\bigg) = S_n = 1, 1, 5, 61, 1385, 50521, \dotsc  \quad \text{($n \geq 0$, OEIS series \oeis{A000364})}.
\end{equation*}
This is due to the fact that all alternating atomic permutations of $[2n]$ can be understood as the reverse alternating permutations of $[2 \lddots 2n-1]$ with a prepended $1$ and an appended $2n$.
Moreover, since any $x_1^{2n+1}$ can only be produced through an application of $\mathcal{D}$ on $x_2 x_1^{2(n-1)}$, we also have $Z_{\aperms}\big(2 + \sum_{i=1}^{2(n-1)} 1\big) = S_n$.

\IdxRanEnd{run-atom}

\subsection{Circular permutations}
\label{sub:run_circular}
\IdxRanBegin{run-circ}

The methods developed in the last section to enumerate atomic permutations can also be applied to find the number of circular permutations~$Z_{\cperms}(p)$ with a given run structure~$p$.
Indeed, any circular permutation in $\cperms_{n-1}$ can be extended to a permutation in $\cperms_{n}$ by inserting $n$ at any position after the first (\eg, \cperm{14532} can be extended to \cperm{164532}, \cperm{146532}, \cperm{145632}, \cperm{145362} or \cperm{145326}).
As in the case of atomic permutations, this extension either increases the length of a run or splits a run into three runs.
Namely, we can increase the length of one run by inserting $n$ at the end or the penultimate position of an increasing run or we can split a run of length $i+j \geq 2$ into three runs of lengths $i, j$ and $1$ by inserting $n$ after $i$ elements of an increasing run or after $i+1$ elements of a decreasing run.

We introduce polynomials $C_n$ representing the run structures of all elements of $\cperms_n$, by analogy with the polynomials $A_n$ in the previous section:
\begin{equation}\label{eq:C_rel_Z}
  C_n = \sum_{p \vdash n} Z_\cperms(p) \prod_{i=1}^n x_i^{p(i)}
\end{equation}
and we say that the polynomials~$C_n$ enumerate the run structure of the circular permutations.
In the last paragraph we saw that we can use the differential operator $\mathcal{D}$ introduced in~\eqref{eq:diff} to find a recurrence relation similar to~\eqref{eq:A_recurrence}.
Namely,
\begin{subequations}\label{eq:C_recurrence}
  \begin{align}
    C_2 & \defn x_1^2, \\
    C_n & \defn \mathcal{D} C_{n-1}, \qquad (n \geq 3)
  \end{align}
\end{subequations}
giving in particular
\begin{align*}
  C_3 & = 2 x_2 x_1 \\
  C_4 & = 2 x_2^2 + 2 x_3 x_1 + 2 x_1^4 \\
  C_5 & = 2 x_4 x_1 + 6 x_3 x_2 + 16 x_1^3 x_2 \\
  C_6 & = 2 x_5 x_1 + 8 x_4 x_2 + 6 x_3^2 + 62 x_1^2 x_2^2 + 26 x_1^3 x_3 + 16 x_1^6
\end{align*}
from which we can read off that there are $2$ permutations in $\cperms_5$ corresponding to $5 = 4 + 1$, $6$ corresponding to the partition $5 = 3 + 2$ and $16$ corresponding to $5 = 2 + 1 + 1 + 1$.
As a check, we note that $6 + 16 + 2 = 24$, which is the total number of elements of $\cperms_5$; similarly, the coefficients in the expression for $C_6$ sum to $120$, the cardinality of $\cperms_6$.
More on this can be found in the last paragraph of \cref{sec:valleys}.

In summary, we have a result analogous to \cref{prop:Z_A}:
\begin{proposition}\label{prop:Z_C}
  The number $Z_\cperms(p)$ of circular permutations of length $n$ corresponding to a given run structure $p$ is determined by the polynomial $C_n$ via~\eqref{eq:C_rel_Z}.
  The polynomials~$C_n$ satisfy the recurrence relation~\eqref{eq:C_recurrence}.
\end{proposition}

Note that circular permutations, exactly opposite to atomic permutations, always contain an even number of runs and thus $Z_\cperms(p)$ is zero for odd partitions~$p$.

The enumeration of circular and atomic permutations is closely related.
In fact, introducing a generating function~$\mathcal{C}$ as the formal series
\begin{equation*}
  \mathcal{C}(\lambda)
  \defn \sum_{n=0}^\infty C_{n+2} \frac{\lambda^n}{n!}
  = \sum_{n=0}^\infty \mathcal{D}^n C_2 \frac{\lambda^n}{n!}
  = \exp(\lambda \mathcal{D}) C_2,
\end{equation*}
one can show the following:
\begin{proposition}\label{prop:square}
  The formal power series $\mathcal{C}$ is the square of a formal series $\mathcal{A}$; namely,
  \begin{equation}\label{eq:square}
    \mathcal{C}(\lambda) = \mathcal{A}(\lambda)^2 = \big( \exp(\lambda \mathcal{D}) A_1 \big)^2,
  \end{equation}
  where $A_1 \defn x_1$.
\end{proposition}
\begin{proof}
  This may be seen in various ways, but the most convenient is to study the first-order partial differential equation (in infinitely many variables)
  \begin{equation}\label{eq:pde}
    \frac{\partial\mathcal{C}}{\partial \lambda} - \mathcal{D}\mathcal{C} = 0,
    \quad
    \mathcal{C}(0) = C_2
  \end{equation}
  satisfied by~$\mathcal{C}$.

  We can now apply the method of characteristics to this problem.
  Since it has no inhomogeneous part, the p.d.e.~\eqref{eq:pde} asserts that $\mathcal{C}$ is constant along its characteristics.
  So, given $\lambda$ and $x_1, x_2, \dotsc$, let $\chi_1(\mu), \chi_2(\mu), \dotsc$ be solutions to the characteristic equations with $\chi_r(\lambda) = x_r$, \ie, $\chi_1(\mu), \chi_2(\mu), \dotsc$ are the characteristic curves which emanate from the point $(\lambda, x_1, x_2, \ldots)$.
  Then,
  \begin{equation*}
    \mathcal{C}(\lambda)|_{x_\bullet} = \mathcal{C}(0)|_{\chi_\bullet(0)} = C_2 \big( \chi_1(0) \big) = \chi_1(0)^2.
  \end{equation*}
  Applying the same reasoning again to $\mathcal{A}$, which obeys the same p.d.e.\ as $\mathcal{C}$ but with initial condition $\mathcal{A}(0) = A_1$,
  \begin{equation*}
    \mathcal{A}(\lambda)|_{x_\bullet} = \mathcal{A}(0)|_{\chi_\bullet(0)} = A_1 \big( \chi_1(0) \big) = \chi_1(0).
  \end{equation*}
  Therefore, \cref{prop:square} follows by patching these two equations together.
\end{proof}

As a consequence also the polynomials $A_n$ and $C_n$ are related via
\begin{equation}\label{eq:relation}
  C_n = \sum_{m=1}^{n-1} \binom{n-2}{m-1}\, A_m A_{n-m}.
\end{equation}
It then follows that the second degree part of $C_n$ is given by
\begin{equation*}
  C^{(2)}_n = \sum_{m=1}^{n-1} \binom{n-2}{m-1}\, x_m x_{n-m}
\end{equation*}
and, applying \eqref{eq:third_degree}, that the fourth degree part can be written as
\begin{equation*}
  C_n^{(4)} = \sum_{\substack{i, j, k, l \,\geq\, 1 \\ i + j + k + l = m}} \sum_{q=1}^{k} 2\, \frac{n-l-q-1}{n-l-q-j}\, \binom{n-2}{n-l-1} \binom{n - l - q - 2}{i - 1, j - 1, k - q}\, x_i x_j x_k x_l.
\end{equation*}

Similar to the atomic permutations, we find that the alternating circular permutations satisfy (\cf\ \cite[\S41]{andre:1895})
\begin{equation*}
  Z_{\cperms}\bigg(\sum_{i=1}^{2n} 1\bigg) = T_n = 1, 2, 16, 272, 7936, 353792, \dotsc \quad \text{($n \geq 1$, OEIS series \oeis{A000182})}
\end{equation*}
and also $Z_{\cperms}\big(2 + \sum_{i=1}^{2n-3} 1\big) = T_n$, where $T_n$ are the \IdxMain{tan-numbers}\emph{tangent numbers}, the coefficients of the Maclaurin series of $\tan x = T_1 x_1 + T_2 x_3/3! + T_3 x_5/5! + \dotsb$.
Furthermore, from~\eqref{eq:relation} we find the relation
\begin{equation*}
  T_{n+1} = \sum_{m=0}^n \binom{2n}{2m}\, S_m S_{n-m},
\end{equation*}
which can be traced back to $\tan' x = \sec^2 x$.

To conclude this section, we note that the argument of \cref{prop:square} proves rather more: namely, that $\exp(\lambda \mathcal{D})$ defines a ring homomorphism from the polynomial ring $\CC[x_1, x_2, \dotsc]$ to the ring of formal power series $\CC[[x_1, x_2, \dotsc]]$.
This observation can be used to accelerate computations: for example, the fact that $A_3 = x_3 +x_1^3$ implies that
\begin{equation*}
\mathcal{A}''(\lambda) = \mathcal{A}(\lambda)^3 + \exp(\lambda\mathcal{D})x_3,
\end{equation*}
which reduces computation of $A_{n+3}=\mathcal{D}^{n+2}x_1$ to the computation of $\mathcal{D}^n x_3$. Once $\mathcal{A}$ is obtained, we may of course determine $\mathcal{C}$ by squaring.

\IdxRanEnd{run-circ}

\subsection{Linear permutations}
\label{sub:run_linear}
\IdxRanBegin{run-lin}

In the last section we studied the run structures of circular permutations $\cperms_n$ and discovered that their run structures can be enumerated by the polynomials $A_n$.
One might ask, what the underlying reason for this is.
Circular permutations of $[n]$ have the same run structure as the linear permutations of the multiset $\{1,1,2,\dotsc,n\}$ which begin and end with $1$.
These permutations can then be split into two atomic permutations at the occurrence of their maximal element.
For example, the circular permutation \cperm{14532} can be split into the two atomic permutations \perm{145} of $\{1,4,5\}$ and \perm{5321} of $\{1,2,3,5\}$.
This also gives us the basis of a combinatorial argument for the fact that $\mathcal{C} = \mathcal{A}^2$.
Similarly it is in principle possible to encode the run structures of any subset of permutations using the polynomials $A_n$.
The goal of this section is to show how this may be accomplished for $\perms_S$ for any $S \subset \NN$.

As in \cref{sub:run_atomic,sub:run_circular}, we want to find polynomials
\begin{equation*}
  L_n = \sum_{p \vdash n} Z_\perms(p) \prod_{i=1}^n x_i^{p(i)}
\end{equation*}
that enumerate the run structure of the permutations~$\perms_{n+1}$.
This may be achieved in a two step procedure.
Since every permutation has a unique decomposition into inextendible atomic permutations, we can enumerate the set of permutations according to this decomposition.
The enumeration of permutations by their run structure follows because the enumeration of atomic permutations has already been achieved in \cref{sub:run_atomic}.

The key to our procedure is to understand the factorisation of the run structure into those of atomic permutations. Considering $\sigma \in \perms_n$ as a word, we can write it as the concatenation $\sigma = \alpha \cdot \pi \cdot \omega$, where $\pi$ is the principal atom of~$\sigma$ (see \cref{sub:perm_atomic}) and $\alpha, \omega$ are (possibly empty) subwords of~$\sigma$.
Since the decomposition of~$\sigma$ into its atoms also decomposes its run structure, the complete runs of~$\sigma$ are determined by the runs of $\alpha \cdot 1$, $\pi$ and $n \cdot \omega$ if $\pi$ is rising, or of $\alpha \cdot n$, $\pi$ and $1 \cdot \omega$ if $\pi$ is falling.

Let $S_\omega$ be the set of letters in $\omega$ and define $\rho: S_\omega \to S_\omega$ to be the involution mapping the $i$'th smallest element of $S_\omega$ to the $i$'th largest, for all $1 \leq i \leq |S_\omega|$.
Then the run structure of $n \cdot \omega$ is identical to that of $1 \cdot \rho(\omega)$, where $\rho(\omega)$ is obtained by applying $\rho$ letterwise to $\omega$.
Furthermore, in the case $\pi = 1 \,\dotsm\, n$, the combined run structures of $\alpha \cdot 1$ and $n \cdot \omega$ are precisely the run structure of $\alpha \cdot 1 \cdot \rho(\omega)$, while, if $\pi = n \,\dotsm\, 1$, the combined run structures of $\alpha \cdot n$ and $1 \cdot \omega$ precisely form the run structure of $\alpha\cdot n \cdot \rho(\omega)$.
We refer to $\alpha\cdot 1\cdot\rho(\omega)$ or $\alpha\cdot n\cdot\rho(\omega)$ as the \emph{residual permutation}.

Summarising, the run structure of~$\sigma$ may be partitioned into that of~$\pi$ and either $\alpha \cdot 1 \cdot \rho(\omega)$ or $\alpha \cdot n \cdot \rho(\omega)$; accordingly, the monomial for~$\sigma$ factorises into that for the principal atom~$\pi$ and that for the residual permutation.
Therefore, the polynomial enumerating linear permutations by run structure can be given in terms of the those enumerating atomic permutations of the same or shorter length and of linear permutations of strictly shorter length.

This argument can be used to give a recursion relation for~$L_n$, which enumerates permutations of $[n+1]$ by their run structure.
Taking into account that the principal atom consists of $m+1$ letters, where $1 \leq m\leq n$, of which $m-1$ may be chosen freely from the set $[2 \lddots n]$, and that it might be rising or falling, and that the residual permutation may be any linear permutation on a set of cardinality $n-m+1$, we obtain the recursion relation
\begin{equation*}
  L_n = 2 \sum_{m=1}^{n} \binom{n-1}{m-1} A_m L_{n-m},
  \qquad
  L_0=1.
\end{equation*}
Passing to the generating function,
\begin{equation*}
 \mathcal{L}(\lambda) \defn \sum_{n=0}^\infty L_n \frac{\lambda^n}{n!},
\end{equation*}
we may deduce that
\begin{equation}\label{eq:linear_de}
  \frac{\partial\mathcal{L}}{\partial\lambda} = 2\mathcal{A}(\lambda)\mathcal{L}(\lambda).
\end{equation}
Our main result in this section is:
\begin{proposition}\label{prop:linear}
  The run structure of all permutations in $\perms_{n+1}$ is enumerated by
  \begin{equation}\label{eq:linear}
    L_n = \sum_{p \vdash n} \frac{2^{\abs{p}}}{\ord p}\, \binom{n}{p} \prod_{i=1}^{\abs{p}} A_{p_i},
    \quad
    L_0 = 1,
  \end{equation}
  where the sum is over all partitions $p = p_1 + p_2 + \dotsb$ of~$n$, $\abs{p}$ is the number of parts of partition, $\ord p$ is the symmetry order of the parts of~$p$ (\eg, for $p = 1+1+2+3+3$ we have $\ord p = 2! 2!$) and $\binom{n}{p}$ is the multinomial with respect to the parts of~$p$.
  The generating function for the~$L_n$ is
  \begin{equation}\label{eq:linear_gen}
    \mathcal{L}(\lambda) \defn \sum_{n=0}^\infty L_n \frac{\lambda^n}{n!}
    = \exp \left( 2 \int_0^\lambda \mathcal{A}(\mu)\, \dif\mu\right).
  \end{equation}
\end{proposition}
\begin{proof}
  \Cref{eq:linear_gen} follows immediately from~\eqref{eq:linear_de}, as $\mathcal{L}(0)=1$, whereupon Faà di Bruno's formula~\cite[Eq.~(1.4.13)]{olver:2010} yields~\eqref{eq:linear}.
\end{proof}

To conclude this section, we remark that the first few $L_n$ are given by
\begin{align*}
  L_1 & = 2 A_1 \\
  L_2 & = 4 A_1^2 + 2 A_2 \\
  L_3 & = 8 A_1^3 + 12 A_1 A_2 + 2 A_3 \\
  L_4 & = 16 A_1^4 + 48 A_1^2 A_2 + 12 A_2^2 + 16 A_1 A_3 + 2 A_4 \\
  L_5 & = 32 A_1^5 + 160 A_1^3 A_2 + 120 A_1 A_2^2 + 80 A_1^2 A_3 + 40 A_2 A_3 + 20 A_1 A_4 + 2 A_5\\
  L_6 &= 64 A_1^6 + 480 A_1^4 A_2 + 320 A_1^3 A_3 + 720 A_1^2 A_2^2 + 120 A_1^2 A_4 + 480 A_1 A_2 A_3 + 120 A_2^3 \\&\quad + 24 A_1 A_5 +
60 A_2 A_4 + 40 A_3^2 + 2A_6.
\intertext{Expanding the $A_k$ and writing the $L_n$ instead in terms of $x_i$, we obtain from these}
  L_1 & = 2 x_1 \\
  L_2 & = 4 x_1^2 + 2 x_2 \\
  L_3 & = 10 x_1^3 + 12 x_1 x_2 + 2 x_3 \\
  L_4 & = 32 x_1^4 + 58 x_1^2 x_2 + 12 x_2^2 + 16 x_1 x_3 + 2 x_4 \\
  L_5 & = 122 x_1^5 + 300 x_1^3 x_2 + 142 x_1 x_2^2 + 94 x_1^2 x_3 + 40 x_2 x_3 + 20 x_1 x_4 + 2 x_5\\
  L_6 & = 544 x_1^6 + 1682 x_1^4 x_2 + 568x_1^3 x_3 + 1284 x_1^2 x_2^2 + 138 x_1^2 x_4 + 556 x_1 x_2 x_3 + 142 x_2^3 \\&\quad + 24 x_1 x_5 + 60 x_2 x_4 + 40 x_3^2 + 2x_6,
\end{align*}
which show no obvious structure, thereby making \cref{prop:linear} that much more remarkable.

\IdxRanBegin{run-lin}

\section{Enumeration of valleys}
\label{sec:valleys}
\IdxRanBegin{valley}

Instead of enumerating permutations by their run structure, we can count the number of valleys of a given (circular) permutation.
Taken together, the terms $C_n$ involving a product of $2k$ of the $x_i$ relate precisely to the circular permutations $\cperms_n$ with $k$ valleys.
Since any circular permutation in $\cperms_n$ can be understood as a permutation of $[3 \lddots n+1]$ with a prepended $1$ and an appended $2$ (\cf\ beginning of \cref{sub:run_linear}), $C_n$ may also be used to enumerate the valleys of ordinary permutations of $[n-1]$.
Namely, terms of~$C_{n+1}$ with a product of $2(k+1)$ variables~$x_i$ relate to the permutations of~$\perms_n$ with $k$~valleys (\ie, terms of $L_{n+1}$ which are a product of~$2k$ of the~$x_i$).

Let $V(n, k)$ count the number of permutations of $n$~elements with $k$~valleys.
Then we see that the generating function for $V(n, k)$ for each fixed $n \geq 1$ is
\begin{equation*}
  K_n(\kappa) \defn \sum_{k=1}^n \kappa^k V(n, k)
  = \frac{1}{\kappa} C_{n+1}(\sqrt{\kappa}, \dotsc, \sqrt{\kappa})
\end{equation*}
and we define $K_0(\kappa) \defn 1$.
The first few $K_n$ are
{\savebox\strutbox{$\vphantom{\kappa^2}$}
\begin{align*}
  K_1(\kappa) & = 1 \\
  K_2(\kappa) & = 2 \\
  K_3(\kappa) & = 4 + 2 \kappa \\
  K_4(\kappa) & = 8 + 16 \kappa \\
  K_5(\kappa) & = 16 + 88 \kappa + 16 \kappa^2 \\
  K_6(\kappa) & = 32 + 416 \kappa + 272 \kappa^2,
\end{align*}}%
which coincide with the results in~\cite{rieper:2000}.
In particular, the constants are clearly the powers of $2$, the coefficients of $\kappa$ give the sequence \oeis{A000431} of the OEIS~\cite{oeis} and the coefficients of $\kappa^2$ are given by the sequence \oeis{A000487}.
Likewise, the coefficients of $\kappa^3$ may be checked against the sequence \oeis{A000517}.
In fact, the same polynomials appear in André's work, in which he obtained a generating function
closely related to \eqref{eq:kitaev} below; see \cite[\S158]{andre:1895} (his final formula contains a number of sign errors, and is given in a form in which all quantities are real for $\kappa$ near $0$; there is also an offset, because his polynomial $A_n(\kappa)$ is our $K_{n-1}(\kappa)$).

\begin{proposition}
  The bivariate generating function, \ie, the generating function for arbitrary $n$, is
  \begin{equation*}
    \mathcal{K}(\nu, \kappa) = \sum_{n=0}^\infty K_n(\kappa) \frac{\nu^n}{n!} = 1 + \frac{1}{\kappa} \int_0^\nu \mathcal{C}(\mu)|_{x_1 = x_2 = \dotsb = \sqrt{\kappa}}\; \dif\mu
  \end{equation*}
  and is given in closed form by
  \begin{equation}\label{eq:kitaev}
    \mathcal{K}(\nu, \kappa) = 1 - \frac{1}{\kappa} + \frac{\sqrt{\kappa - 1}}{\kappa} \tan\big( \nu \sqrt{\kappa - 1} + \arctan(1/\sqrt{\kappa - 1}) \big).
  \end{equation}
\end{proposition}

This result was found by Kitaev~\cite{kitaev:2007} and in the remainder of this section we will show how it may be derived from the recurrence relation~\eqref{eq:C_recurrence} of $C_n$.

To this end, we first note that $C_{n+1}$ satisfies the useful scaling relation
\begin{equation*}
  \lambda^{n+1} C_{n+1}(x_1, x_2, \dotsc, x_n) = C_{n+1}(\lambda x_1, \lambda^2 x_2, \dotsc, \lambda^n x_n).
\end{equation*}
Setting $x_i = x / \lambda = \sqrt{\kappa}$ for all $i$, this implies
\begin{equation*}
  \lambda^{n+1} C_{n+1}(\sqrt{\kappa}, \dotsc, \sqrt{\kappa}) = C_{n+1}(x, \lambda x, \dotsc, \lambda^{n-1} x)
\end{equation*}
and we find, by inserting the recurrence relations~\eqref{eq:C_recurrence} and applying the chain rule, that with this choice of variables
\begin{equation*}
  \frac{1}{x^2} C_{n+1}(x, \lambda x, \dotsc, \lambda^{n-1} x) = \lambda x \pd{C_n}{x} + x^2 \pd{C_n}{\lambda} + 2 \lambda C_n.
\end{equation*}
Hence, in turn, $K_n(\kappa) = \kappa^{-1} C_{n+1}(\sqrt{\kappa}, \dotsc, \sqrt{\kappa})$ satisfies the recurrence relation
\begin{equation}\label{eq:Kn}
  K_n(\kappa) = 2 \kappa (1 - \kappa) K_{n-1}'(\kappa) + \big( 2 + (n - 2) \kappa \big) K_{n-1}(\kappa)
\end{equation}
for $n \geq 2$.
For the bivariate generating function $\mathcal{K}$ this, together with $K_0 = K_1 = 1$, implies the p.d.e.
\begin{equation*}
  (1 - \nu \kappa) \pd{\mathcal{K}}{\nu} + 2 \kappa (\kappa - 1) \pd{\mathcal{K}}{\kappa} + (\kappa - 2) \mathcal{K} = \kappa - 1,
\end{equation*}
which is to be solved subject to the initial condition $\mathcal{K}(0, \kappa) = 1$.

The above equation may be solved as follows: first, we note that there is a particular integral $1 - 1/\kappa$, so it remains to solve the homogeneous equation.
In turn, using an integrating factor, the latter may be rewritten as
\begin{equation}\label{eq:homog}
  (1 - \nu \kappa) \pd{}{\nu} \frac{\kappa \mathcal{K}}{\sqrt{\kappa - 1}} + 2 \kappa (\kappa - 1) \pd{}{\kappa} \frac{\kappa \mathcal{K}}{\sqrt{\kappa - 1}} = 0,
\end{equation}
for which the characteristics obey
\begin{equation*}
  \od{\nu}{\kappa} = \frac{1 - \nu \kappa}{2 \kappa (\kappa - 1)}.
\end{equation*}
Solving this equation, we find that
\begin{equation*}
  \nu \sqrt{\kappa - 1} + \arctan \frac{1}{\sqrt{\kappa - 1}} = \text{const}
\end{equation*}
along characteristics; as~\eqref{eq:homog} asserts that $\kappa \mathcal{K}/\sqrt{\kappa - 1}$ is constant on characteristics, this gives
\begin{equation*}
  \mathcal{K}(\nu, \kappa) = 1 - \frac{1}{\kappa} + \frac{\sqrt{\kappa - 1}}{\kappa} f\big( \nu \sqrt{\kappa - 1} + \arctan(1/\sqrt{\kappa - 1}) \big)
\end{equation*}
for some function $f$.
Imposing the condition $\mathcal{K}(0, \kappa) = 1$, it is plain that $f = \tan$, and we recover Kitaev's generating function~\eqref{eq:kitaev}.

To close this section, we note that \eqref{eq:Kn} has the consequence that $K_n(1) = n K_{n-1}(1)$ for all $n \geq 2$ and hence that $C_{n+1}(1,\dotsc,1) = K_n(1) = n!$ for such $n$, and indeed all $n\geq 1$, because $C_2(1, 1) = K_1(1) = 1$. The generating function obeys
\begin{equation*}
  \mathcal{C}(\lambda)|_{x_\bullet=1}
  = \sum_{n=0}^\infty (n+1)! \frac{\lambda^n}{n!}
  = (1-\lambda)^{-2}
\end{equation*}
for all non-negative $\lambda < 1$ from which it also follows that
\begin{equation}\label{eq:geom_A}
  \mathcal{A}(\lambda)|_{x_\bullet=1} = (1-\lambda)^{-1}
\end{equation}
(as $A_1(1) = 1$, we must take the \emph{positive} square root) and hence $A_{n}|_{x_\bullet=1} = (n-1)!$ for all $n \geq 1$.
This gives a consistency check on our results: the coefficients in the expression for $A_n$ sum to $(n-1)!$, the cardinality of $\aperms^\pm_{n+1}$, while those in~$C_n$ sum to the cardinality of~$\cperms_n$.
Furthermore, inserting \eqref{eq:geom_A} into the generating function~$\mathcal{L}(\lambda)$ in \eqref{eq:linear_gen}, we find
\begin{equation*}
  \mathcal{L}(\lambda)|_{x_\bullet=1}
  = \sum_{n=0}^\infty L_n(1,\dotsc,1) \frac{\lambda^n}{n!}
  = (1-\lambda)^{-2},
\end{equation*}
and thus $L_{n+1}(1,\dotsc,1) = n!$, which is the cardinality of $\perms_n$.

\IdxRanEnd{valley}

\section{Applications}
\label{sec:applications}

The original motivation for this work arose in quantum field theory, in computations related to the probability distribution of measurement outcomes for quantities such as averaged energy densities~\cite{fewster:2012b}.
One actually computes the cumulants $\kappa_n$ ($n \in \NN$) of the distribution: $\kappa_1 = 0$, while for each $n \geq 2$, $\kappa_n$ is given as a sum indexed by circular permutations $\sigma$ of $[n]$ such that $\sigma(1) = 1$ and $\sigma(2) < \sigma(n)$, in which each permutation contributes a term that is a multiplicative function of its run structure:
\begin{equation*}
  \kappa_n = \sum_{\sigma} \Phi(\sigma)
\end{equation*}
where $\Phi(\sigma)$ is a product over the runs of $\sigma$, with each run of length $r$ contributing a factor $y_r$.
Owing to the restriction $\sigma(2) < \sigma(n)$, precisely half of the circular permutations are admitted, and so $\kappa_n=\frac{1}{2}C_n(y_1,y_2,\ldots,y_n)$.
Thus the cumulant generating function is
\begin{align*}
  W(\lambda)
  &\defn \sum_{n=2}^\infty \kappa_n \frac{\lambda^n}{n!} = \frac{1}{2} \int_0^\lambda \dif\mu\, (\lambda - \mu) \mathcal{C}(\mu)|_{x_\bullet=y_\bullet} \\
  &= \frac{1}{2} \int_0^\lambda \dif\mu\, (\lambda - \mu) \mathcal{A}(\mu)|_{x_\bullet=y_\bullet}^2
\end{align*}
and the moment generating function is $\exp W(\lambda)$ in the usual way.
This expression makes sense a formal power series, but also as a convergent series within an appropriate radius of convergence.
The values of $y_n$ depend on the physical quantity involved and the way it is averaged.
In one case of interest
\begin{align*}
  y_n
  &= 8^n \int_{(\RR^+)^{\times n}} \dif k_1\, \dif k_2\, \cdots\, \dif k_n\, k_1 k_2 \,\cdots\, k_n  \exp \left[ -k_1-\left(\sum_{i=1}^{n-1} |k_{i+1}-k_i|\right)-k_{n} \right] \\
  &= 2^n \sum_{r_{n-1}=0}^{2} \sum_{r_{n-2}=0}^{2+r_{n-1}} \cdots \sum_{r_1=0}^{2+r_2} \prod_{k=1}^{n-1} (1+r_k) \\
  &= 2, 24, 568, 20256, 966592, \dotsc \qquad (n \geq 1)
\end{align*}
(the sums of products must be interpreted as an overall factor of unity in the case $n=1$).
Numerical investigation leads to a remarkable identity
\begin{equation*}
  \mathcal{A}(\lambda)|_{x_\bullet=y_\bullet} = \frac{2}{1-12\lambda} \qquad \textrm{(conjectured)}
\end{equation*}
with exact agreement for all terms so far computed (checked up to $n = 65$).
We do not have a proof for this statement, but the conjecture seems fairly secure.
For example, we have shown above that $A_5= x_5 + 7 x_3 x_1^2 + 11 x_2^2 x_1 + 5 x_1^5$; substituting for $x_n$ the values of $y_n$ obtained above, we find $A_5= 995328$ which coincides with the fourth order coefficient in the expansion
\begin{equation*}
  \frac{2}{1 - 12 \lambda}=  2 + 24 \lambda + 576 \frac{\lambda^2}{2!} + 20736 \frac{\lambda^3}{3!} + 995328 \frac{\lambda^4}{4!} + \bigO(\lambda^5).
\end{equation*}
In~\cite{fewster:2012b}, this conjecture was used to deduce
\begin{equation*}
  \exp\big(W(\lambda)\big) = \e^{-\lambda/6}(1-12\lambda)^{-1/72} \qquad \textrm{(conjectured)},
\end{equation*}
which is the moment generating function of a shifted Gamma distribution.
The other generating functions of interest, with these values for the $x_k$ are
\begin{equation*}
  \mathcal{C}(\lambda)|_{x_\bullet=y_\bullet} = \frac{4}{(1-12\lambda)^2},
  \qquad
  \mathcal{L}(\lambda)|_{x_\bullet=y_\bullet} = (1-12\lambda)^{-1/3}
  \qquad \textrm{(conjectured)}.
\end{equation*}
For example, we have $C_5 =  2 x_4 x_1 + 6 x_3 x_2 + 16 x_1^3 x_2 = 165888$ and $L_5 = 122 x_1^5 + 300 x_1^3 x_2 + 142 x_1 x_2^2 + 94 x_1^2 x_3 + 40 x_2 x_3 + 20 x_1 x_4 + 2 x_5 =3727360$, to be compared with the terms of order $\lambda^3$ and $\lambda^5$, respectively, in the expansions
\begin{align*}
  \frac{4}{(1 - 12 \lambda)^2} &= 4 + 96 \lambda + 3456 \frac{\lambda^2}{2!} + 165888 \frac{\lambda^3}{3!} + 995328 \frac{\lambda^4}{4!} + \bigO(\lambda^5), \\
  (1 - 12 \lambda)^{-1/3} &= 1 + 4 \lambda + 64 \frac{\lambda^2}{2!} + 1792 \frac{\lambda^3}{3!} + 71680 \frac{\lambda^4}{4!} + 3727360 \frac{\lambda^5}{5!} + \bigO(\lambda^6).
\end{align*}

A natural question is whether there are other sequences that can be substituted for the $x_k$ to produce generating functions with simple closed forms.
To close, we give three further examples, with the corresponding generating functions computed.
The first has already been encountered in \cref{sec:valleys} and corresponds to the case $x_k = 1$ for all $k \in \NN$.

The second utilizes the alternating Catalan numbers: setting
\begin{equation*}
  x_{2k+1} = \frac{(-1)^k}{k+1} \binom{2k}{k}, \quad(k \geq 0),
  \qquad
  x_{2k} = 0, \quad (k \geq 1)
\end{equation*}
and thus $A_{2k} = 0$, we obtain, again experimentally,
\begin{equation*}
  \mathcal{C}(\lambda) = \mathcal{A}(\lambda) = 1,
  \quad
  \mathcal{L}(\lambda) = \e^{2\lambda}
  \qquad \textrm{(conjectured)}
\end{equation*}
with exact agreement checked up to permutations of length $n = 65$.
For example, one sees easily that with $x_1 = 1$, $x_3 = -1$, $x_5 = 2$ and $x_2 = x_4 = 0$, the expressions $A_k$ and $C_k$ given in \cref{sub:run_atomic,sub:run_circular} vanish for $2 \leq k \leq 6$,
and have $A_1 = C_1 = 1$, likewise, $L_k = 2^k$ for $1 \leq k \leq 6$.

Third, André's classical result on alternating permutations (\cf last and penultimate paragraph of \cref{sub:run_atomic,sub:run_circular} respectively) gives the following: setting
\begin{equation*}
  x_1 = 1
  \quad \text{and} \quad
  x_{k} = 0, \quad (k \geq 2)
\end{equation*}
we have, using \eqref{eq:square} and \eqref{eq:linear_gen},
\begin{equation*}
  \mathcal{A}(\lambda) = \sec\lambda,
  \quad
  \mathcal{C}(\lambda)=\sec^2\lambda,
  \quad
  \mathcal{L}(\lambda) = (\sec\lambda + \tan\lambda)^2.
\end{equation*}
It seems highly likely to us that many other examples can be extracted from the structures we have described.

Moreover, we remark that it is possible to implement a merge-type sorting algorithm, called \emph{natural merge sort} \cite[Chap.~5.2.4]{knuth:1998}, based upon splitting permutations of an ordered set $S$ into its runs, which are ordered (alternatingly in ascending and descending order) sequences $S_i \subset S$.
Repeatedly merging these subsequences, one ultimately obtains an ordered sequence.
For example, first, we split the permutation \perm{542368719} into \perm{542}, \perm{368}, \perm{71} and \perm{9}.
Then, we reverse every second sequence (depending on whether the first or the second sequence is in ascending order): \perm{542} $\mapsto$ \perm{245} and \perm{71} $\mapsto$ \perm{17}.
Depending on the implementation of the merging in the following step, this `reversal' step can be avoided.
Last, we merge similarly to the standard merge sort: \perm{245} $\vee$ \perm{368} $\mapsto$ \perm{234568}, \perm{17} $\vee$ \perm{9} $\mapsto$ \perm{179} and finally \perm{234568} $\vee$ \perm{179} $\mapsto$ \perm{123456789}.
Natural merge sort is a fast sorting algorithm for data with preexisting order.
Using the methods developed above to enumerate permutations by their run structure, it is in principle possible to give average (instead of best- and worst-case) complexity estimates for such an algorithm.

\IdxRanEnd{combinat}

%% file: locally_covariant.tex
%!TEX root = master.tex

\chapter{Locally covariant quantum field theory}
\label{cha:locally_covariant}

\section*{Summary}

In this chapter we will discuss the framework of locally covariant quantum field theory.
In its present form it was introduced in~\cite{brunetti:2003} but many of its central ideas can already be found in earlier publications.
It may be understood as a generalization of the Haag--Kastler axioms~\cite{haag:1964,haag:1996} to curved spacetimes but it also differs in some subtle points because the Haag--Kastler axioms are `more global' (see, for example, \cite{benini:2014,benini:2014a} on the problem of gauge theories in locally covariant QFT).
A generalization of the Haag--Kastler was performed by Dimock whose work~\cite{dimock:1980,dimock:1982,dimock:1992} can be understood as the foundation of modern algebraic QFT on curved spacetime.
Building on the work of Dimock, the paradigm of locally covariant quantum field theory should be seen a culmination of work done by Brunetti, Fredenhagen, Hollands, Kay, Verch, Wald and others on QFT on curved spacetimes, in particular renormalization, \cite{brunetti:1996,brunetti:2000,hollands:2001,hollands:2002,verch:2001,kay:1997} around the turn of the millennium after the discovery of the microlocal spectrum condition~\cite{radzikowski:1996,radzikowski:1996a,brunetti:1996}.

After discussing some general considerations leading to algebraic and locally covariant quantum field theory in the first section (\cref{sec:qft_philosphy}), we will introduce the  the general framework of locally covariant quantum field theory in the second section (\cref{sec:lcqft_framework}).
More details on the locally covariant framework may be found in the original publication \cite{brunetti:2003} or also \cite{sanders:2008,fewster:2012a} among many others.
This is followed by an abstract study of the Borchers--Uhlmann algebra, the commutation relations and the field equation, which will lead to the field algebra, and the Weyl algebra in the third section (\cref{sec:klein_gordon}), we will discuss two free bosonic fields in the locally covariant framework: the scalar field and the Proca field.

\section{General considerations}
\label{sec:qft_philosphy}

Quantum field theory is a very complex subject which cannot easily be defined.
This is partly due to the fact that quantum field theory is not so much a physical theory but rather a language to formulate theories and models.
However, a more important reason is that quantum field theory, even many decades after its inception, has no clear interpretation, \eg, it is often not clear what the physical objects are.
Nevertheless, it can be considered one of the most successful scientific discoveries ever conceived and some predictions made by quantum field theory have been tested with astonishing precision.
For example, the anomalous magnetic moment of the electron has been measured in agreement with theoretical predictions in the parts in a trillion range, see \cite{hanneke:2008,aoyama:2012}.

The relation of QFT to other theories is schematically depicted in the diagram \cref{fig:qft}.
In particular, QFT may be considered as a Lorentz invariant quantum mechanics in the case of infinitely many degrees of freedom.
One can also argue that a consistent reconciliation of quantum mechanics with special relativity (in particular locality) leads invariably to QFT, \ie, fields and an infinite number of degrees of freedom are necessary, see \cite{malament:1996} and also \cite{busch:1999}.

\begin{figure}
  \centering
  \includegraphics{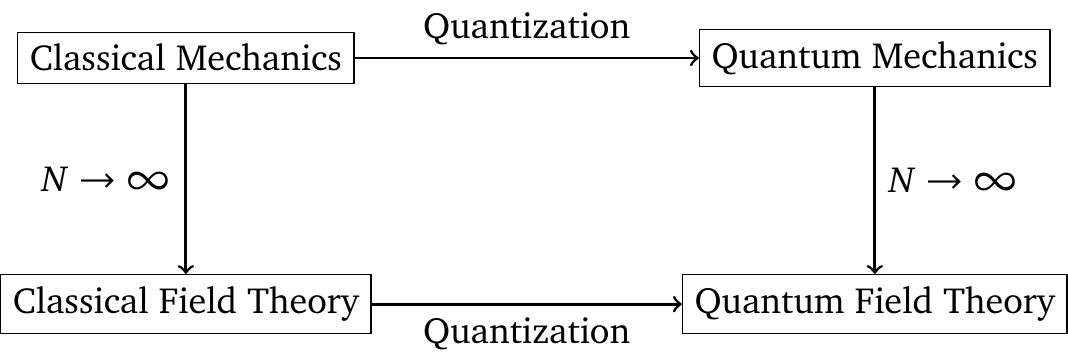}
  \caption{\label{fig:qft} The heuristic relation of quantum field theory with classical mechanics, classical field theory and quantum mechanics.}
\end{figure}

\subsection{Lagrangian QFT}
\IdxMain{lagr-qft}

Quantum field theory is usually formulated in the relatively heuristic approach of the Lagrangian formalism, where, starting from a classical Lagrangian, one imposes the canonical commutation relations between the quantized position and momentum variables.
In analogy to the quantum mechanical harmonic oscillator these yield creation and annihilation `operators' on an abstract representing Hilbert space.
Combining the creation and annihilation operator, one furthermore defines the quantum field.
A specific Hilbert space is then selected by requiring that the annihilation operator annihilates a particular vector in the Hilbert space, the vacuum, so that one obtains the Fock space representation.

Apart from not being mathematically rigorous, the Lagrangian formalism has several conceptual drawbacks.
First, it neglects \latin{a priori} the inequivalent irreducible representations of the canonical commutation relations (as a consequence of the failure of the Stone--von Neumann theorem in infinite dimensions) and instead selects the convenient Fock space representation.
However, it is not obvious what is physical and whether the inequivalent representations are simply mathematical artefacts or physically relevant.
Indeed, the existence of superselection sectors shows that the presence of inequivalent representations is certainly \emph{not} irrelevant.
A closely related issue is described by Haag's theorem \cite[Chap.~II.1.1]{haag:1996} which implies that the standard Fock space representation of the free theory is inequivalent to the that of the interacting theory.

Second, the fundamental entities in the Lagrangian formalism are `operators' at a point and thus neither mathematically not physically meaningful.
Physically, because it would require an infinite amount of energy to localize a field a point.
Mathematically, because a field at a point is not an operator on a Hilbert space but only an operator-valued distribution.
Instead one should consider field which have no sharp localization but are smeared out over a region of spacetime.
That is, the fundamental entities are quantum fields smeared with compactly supported test functions.

Third, the Lagrangian formalism contains global operators, like the particle number operator, which are not operationally meaningful because they cannot be reproduced by measurements in a bounded region of spacetime.
In fact, due to the Reeh--Schlieder theorem, also local number operators cannot exist \cite[Thm.~5.3.2]{haag:1996}.
In this light the common interpretation of QFT in terms of localized particles is very problematic.
On curved spacetimes or for accelerated observers the situation is even more problematic: as shown by the Unruh effect, the particle interpretation appears to be context dependent.

\subsection{Algebraic QFT and locality}
\IdxMain{alg-qft}

In the algebraic approach to quantum field theory, developed by Haag and collaborators, the problems indicated above are addressed in a conceptually simple way:
Rather than taking as observables operators on a Hilbert space, in the algebraic approach one discards the concrete representation of the operators and considers only the algebraic relations satisfied between the operators.
Indeed, the relations between the observables already contain a large part of the physical content of the theory.

The central pillar of algebraic quantum field theory is \IdxMain{princ-locality}\emph{locality}, better described by the German word ``\emph{Nahwirkungsprinzip}''.
Locality means that causally unrelated events do not influence each other and it is implemented in the following way:
To every spacetime region~$U$ we can associate a local algebra of observables $\mathcal{A}(U)$ which can be measured within~$U$.
Consequently, we demand that map $U \mapsto \mathcal{A}(U)$ forms an inductive system, \ie, it satisfies the \IdxMain{isotony-net}\emph{isotony} condition
\begin{equation*}
  U \subset V \implies \mathcal{A}(U) \subset \mathcal{A}(V)
\end{equation*}
or at least that there is an injective homomorphism $\mathcal{A}(U) \to \mathcal{A}(V)$; the correspondence $U \mapsto \mathcal{A}(U)$ for all $U$ is called the \IdxMain{net-algebras}\emph{net of local algebras}.
Further, we require that the local algebras of causally separated, causally convex regions (anti-)commute\IdxMain{princ-causal}
\begin{equation*}
  U \mathbin{\raisebox{1pt}{$\bigtimes$}} V \implies [\mathcal{A}(U), \mathcal{A}(V)] = \{0\}.
\end{equation*}
In the next section it will become clear how these conditions can be consistently imposed on different spacetimes.

Following the choice of words of Haag \cite{haag:1996}, in any concrete case, the \emph{smeared} quantum fields may be seen as a way to `coordinatize' the local algebras.
That is, they provide a map from test functions supported a spacetime region, to the local algebra supported in that region.
From this point of view it seems clear that different quantum fields can lead to equivalent algebras.
The notion of quantum fields might become clearer within the categorical framework to be introduced the next section.

\section{Framework}
\label{sec:lcqft_framework}

\subsection{Background structure}

The physical Universe appears to be well-modelled by a connected, oriented and time-oriented, four-dimensional Lorentzian manifold $(M, g, \pm, u)$, \ie, a spacetime as defined in \cref{sec:causality}.
Moreover, for every possible observer to carry out experiments in a finite region of spacetime one has to require that $J^+(x) \cap J^-(y)$ is compact for all $x, y \in M$.
If we further require that no closed causal curves exist so that time travel is impossible, we have collected all necessary and sufficient conditions for a globally hyperbolic spacetime.
Accordingly we consider as physical spacetimes all globally hyperbolic spacetimes.
Of course we restrict ourselves to globally hyperbolic spacetimes also for technical reasons.
In particular, only on globally hyperbolic spacetimes do we have a good understanding of the Cauchy problem for the wave equation.

To implement simultaneously \IdxMain{princ-cov}\emph{covariance} and the principle of \IdxMain{princ-locality}\emph{locality}, \ie, an observer can conduct experiments in a globally hyperbolic subregion of the Universe and may remain ignorant about processes in the complement of that region, we consider as morphisms between globally hyperbolic spacetimes $M, N$ the isometric embeddings $\psi : M \to N$ that are orientation and time-orientation preserving and whose image~$\psi(M)$ is a causally convex region of~$N$.
We call these morphisms \IdxMain{hyp-emb}\emph{hyperbolic embeddings}; an example and a counter-example are shown in \cref{fig:hyperbolic_embedding}.
If the image of the hyperbolic embedding~$\psi$ contains a neighbourhood of a Cauchy surface in~$N$, we say that it is \IdxMain{cauchy-emb}\emph{Cauchy}.

\begin{figure}
  \centering
  \includegraphics[page=1]{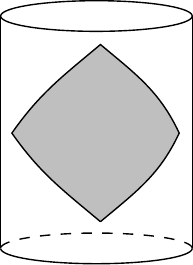}
  \includegraphics[page=2]{fig_einstein_cylinder_2.pdf}
  \includegraphics[page=3]{fig_einstein_cylinder_2.pdf}
  \caption{Example of an hyperbolic embedding and a non-hyperbolic embedding; see also \cref{fig:causally_convex}.} \label{fig:hyperbolic_embedding}
\end{figure}

The set of globally hyperbolic spacetimes as objects with the hyperbolic embeddings as morphisms forms a category denoted~\Loc.\IdxMain{loc-cat}
This category was introduced in~\cite{brunetti:2003} and it is arguably the most fundamental but, as already mentioned in~\cite{brunetti:2003}, not the only possible choice to describe local theories.
In fact, it has been altered in the literature in various ways
\begin{itemize}
  \item to accommodate more background structure by adding to the triple of manifold, metric and time-orientation, which is each object, additional elements like spin-structure \cite{sanders:2010a}, affine, principle and vector bundles \cite{benini:2014,benini:2014a,benini:2014b} or external currents \cite{fewster:2014};
  \item to account for additional symmetry by allowing for more general morphisms like conformal isometries \cite{pinamonti:2009,siemssen:2011,dappiaggi:2013};
  \item to allow for the formulation of theories that are sensitive to the topology of the manifold, \eg, by restricting the set of objects to manifolds that have certain de Rahm cohomology groups \cite{siemssen:2011,dappiaggi:2013}.
\end{itemize}
More recently, it was suggested by Fewster~\cite{fewster:preperation} to consider as objects triples which instead of a metric have a global (co)frame; the morphisms are changed accordingly.
The resulting category is larger and encompasses the original setting via a forgetful functor but the additional structure allows for an interesting discussion of the spin-statistics theorem.

\subsection{Observables}

A theory in this categorical framework is a covariant functor from~\Loc\ into a category whose objects describe physical systems and whose morphisms encode embeddings of physical systems.
In quantum field theory (on curved spacetimes) in the algebraic formulation, physical systems are modelled by ${}^*$-algebras or $C^*$-algebras.
Denote by~\Alg\IdxMain{alg-cat} the category whose objects are unital topological ${}^*$-algebras with morphisms given by the unit-preserving ${}^*$-monomorphisms and call a covariant functor $\func{A} : \Loc \to \Alg$ realizing such an algebra on each background in~\Loc\ a \IdxMain{loc-cov}\emph{locally covariant theory}.
The algebra~$\func{A}(M)$ thus associated to each spacetime $M \in \Loc$ is often called the \IdxMain{alg-obs}\emph{algebra of observables} although in many cases it may contain elements that are not actually physically accessible.

Given two hyperbolic embeddings $\psi_i : M_i \to N$, $i = 1,2$, such that the images~$\psi_i(M_i)$ are causally disjoint in~$N$ (\cf\ \cref{fig:causally_disjoint}), we say that $\func{A}$ is \IdxMain{princ-causal}\emph{causal} if
\begin{equation*}
  [\func{A}(\psi_1) \func{A}(M_1), \func{A}(\psi_2) \func{A}(M_2)] = \{0\}.
\end{equation*}
Causality of~$\func{A}$ is closely related to its tensorial structure as discussed in~\cite{brunetti:2014}.

\begin{figure}
  \centering
  \subbottom[]{
    \includegraphics[page=1]{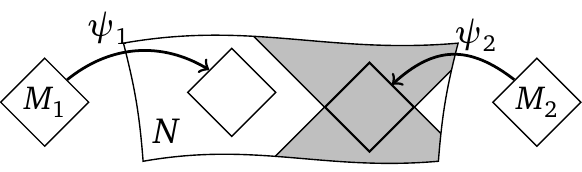}
    \label{fig:causally_disjoint}
  }
  \hspace{1cm}
  \subbottom[]{
    \includegraphics[page=2]{fig_cauchy_embedding.pdf}
    \label{fig:cauchy_embedding}
  }
  \caption{An illustration of (a) two causally disjoint embeddings and (b) a Cauchy embedding.}
\end{figure}

The theory $\func{A}$ obeys the \IdxMain{timeslice}\emph{timeslice axiom} if
\begin{equation*}
  \func{A}(\psi) \func{A}(M) = \func{A}(N)
\end{equation*}
for all Cauchy embeddings $\psi : M \to N$ (\cf\ \cref{fig:cauchy_embedding}).
The timeslice axiom is a prerequisite for the \IdxMain{rel-cauchy-evol}\emph{relative Cauchy evolution}, which describes the response of the physical system to a perturbation of the background structure.

More concretely, let $(M, g, \pm, u)$ and $(M[h], g+h, \pm, u_h)$ be globally hyperbolic spacetimes such that $h$ is a compactly supported symmetric tensor field and $u_h$ is the unique time-orientation that agrees with $u$ outside the support of~$h$.\footnote{Note that a possible `perturbation' is always $M[0]$ but there exists also a neighbourhood $U$ of $0$ in the set of compactly supported symmetric covariant two-tensor fields (with the test function topology) such that $(M[h], g+h)$ is globally hyperbolic for all $h \in U$, see \cite[Thm.~7.2]{beem:1996}.}
Consequently there exist neighbourhoods around two Cauchy surfaces in~$M[h]$, one in the past of $h$ and the other in the future.
We can then find Cauchy morphisms~$\iota^\pm$ and~$\iota[h]^\pm$ from spacetimes $M^\pm \in \obj{Loc}$ into~$M$ and~$M[h]$ as shown in \cref{fig:rce}.
Together these Cauchy morphisms make up the (${}^*$-algebra) homomorphism that is the relative Cauchy evolution map
\begin{equation*}
  \rce[h] \defn \func{A}(\iota^-) \circ \func{A}(\iota[h]^-)^{-1} \circ \func{A}(\iota[h]^+) \circ \func{A}(\iota^+)^{-1}.
\end{equation*}

\begin{figure}
  \centering
  \includegraphics{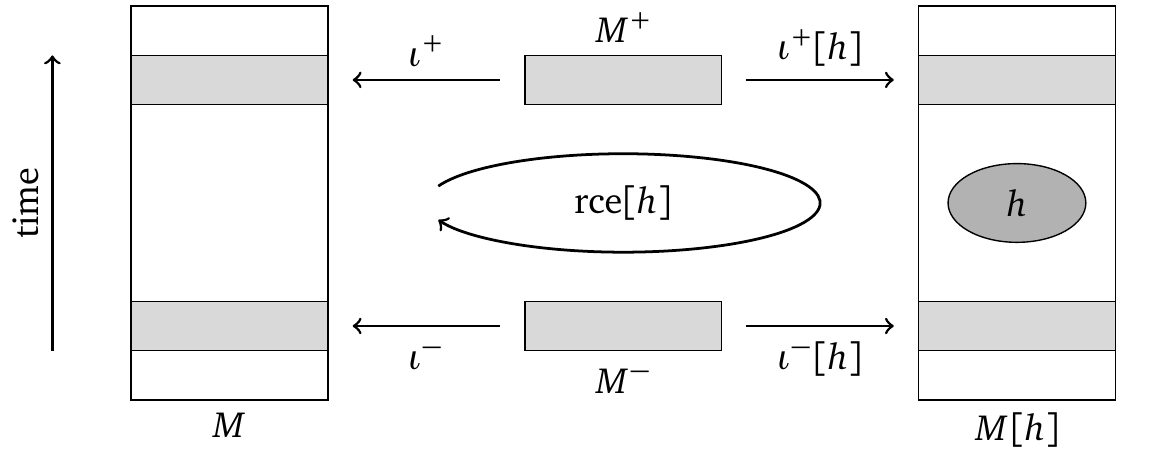}
  \caption{Illustration of the morphisms in the relative Cauchy evolution with the unperturbed background~$M$ on the left and the perturbed background~$M[h]$ on the right.}
  \label{fig:rce}
\end{figure}

It was shown in~\cite{brunetti:2003} that (in an appropriate topology, see \cite{fewster:2012e,fewster:2014} for details),
\begin{equation*}
  - 2\im\, \od{}{\varepsilon}\! \rce[\varepsilon h] A \bigg|_{\varepsilon = 0} = [T(h), A]
\end{equation*}
for any $A \in \func{A}(M)$ and where $T(h) \in \func{A}(M)$ is symmetric and conserved.
Since $T$ is both symmetric and conserved it may be interpreted as a stress-energy tensor \cite{brunetti:2003,fewster:2012,fewster:2012e} and, in fact, in concrete models this interpretation is valid \cite{brunetti:2003,sanders:2010a,benini:2011,fewster:2014b}.

\section{Generalized Klein--Gordon fields}
\label{sec:klein_gordon}

The Klein--Gordon field is usually the first field to be discussed when studying quantum field theory.
We will be no different although we will perform some straightforward generalizations.
Namely, we will quantize fields on natural vector bundles that satisfy a normally hyperbolic equation of motion.
The results below are a generalization of those obtained in~\cite{brunetti:2003} for the scalar Klein--Gordon field.
In principle, further generalizations of the definitions and statements presented below are possible.
For example, one can replace compactly supported $p$-forms by compactly supported sections of `natural' vector bundles, \ie, vector bundles that, like the (co)tangent bundle, are functorially constructed from the geometric structure of the manifold.
However, all these generalizations yield little insight and obfuscate some constructions.
Moreover, the requirements imposed by the usual locally covariant framework make it difficult to find examples that do not just use a standard tensor bundle tensorized with a vector bundle that is independent of the geometry of the spacetime.

% A general framework for the quantization of linear bosonic fields as locally covariant theories can be found in~\cite{hack:2013a}; the hyperbolic case, \ie, without a gauge problem, is a straightforward generalization of the quantization of the scalar field and can be found in~\cite{siemssen:2011}.
% The scalar field with external sources was recently studied in~\cite{fewster:2014}.

% The scalar field is the simplest field but already shows the main features of this framework.
% It was also the first field to be investigated when the framework was introduced \cite{brunetti:2003}.
% With minimal modifications we extend the results obtained for the scalar field to the Proca field.
% Although the Proca field will be of no interest for the remainder of this work, we included it for pedagogical reasons as a precursor to the study of the Maxwell field in the following chapter.

\subsection{Borcher--Uhlmann algebra}
\label{sub:kg_borchers_uhlmann}

For every globally hyperbolic spacetime~$M$, let $\func{D}$ be a covariant functor from \Loc\ into the category of closed nuclear locally convex $\CC$-vector spaces such that $\func{D}(M) \subset \Omega_0^p(M, \CC)$ (for fixed $p$ and with the subspace topology) and $\func{D}(\psi) = \pf{\psi}$.

On each spacetime~$M$ we can define a straightforward generalization of the \IdxMain{borchers-uhlmann}\emph{Borchers--Uhlmann algebra}~\cite{borchers:1962,uhlmann:1962,borchers:1972} as the unital topological ${}^*$-algebra
\begin{equation*}
  \func{U}(M) \defn \bigoplus_{n \in \NN_0} \func{D}(M)^{\widehat{\otimes} n}
\end{equation*}
with, \ie, the set of tuples $(f_n)_{n \in \NN_0}$ with $f_n \in \func{D}(M)^{\widehat{\otimes} n}$ such that only a finite number of~$f_n$ is nonzero, together with
\begin{enumerate}[(a)]
  \item
    addition and scalar multiplication is component-wise,
  \item
    multiplication given by the canonical isomorphic embedding
    \begin{equation*}
      \func{D}(M)^{\widehat{\otimes} m} \otimes \func{D}(M)^{\widehat{\otimes} n} \longrightarrow \func{D}(M)^{\widehat{\otimes} (m+n)}
    \end{equation*}
    and extends (anti-)linearly to all of~$\func{U}(M)$ via the canonical embeddings of these spaces into~$\func{U}(M)$,
  \item
    a ${}^*$-operation that acts on~$(f_n) \in \func{U}(M)$ as $(f_n)^* = (f_n^*)$ and
    \begin{equation*}
      f_n^*(x_1, \dotsc, x_n) = \conj{f_n}(x_n, \dotsc, x_1),
    \end{equation*}
  \item
    a topology given by the direct sum topology of the test function topology on each $\func{D}(M)^{\widehat{\otimes} n}$.\footnote{With this topology the algebra is complete and nuclear, and the algebra product is separately continuous.}
\end{enumerate}

Assigning to each globally hyperbolic spacetime the Borchers--Uhlmann algebra~$\func{U}(M)$, we obtain a covariant functor $\func{U} : \Loc \to \Alg$ that maps each object to the algebra and each morphism to the ${}^*$-algebra morphism generated by the natural pushforward~$\pf{\psi}$, \ie,
\begin{equation*}
  \func{U}(\psi) (f_n) = (\pf{\psi} f_n)
  \quad \text{with} \quad
  \pf{\psi} f_n = f_n \circ (\psi^{-1})^{\otimes n}
\end{equation*}
for all $(f_n) \in \func{U}(M)$.

Considering $\func{D}(M)$ as a topological ${}^*$-algebra where the involution is complex conjugation, we can consider it as a functor from~\Loc\ to~\Alg.
The natural transformation $\Phi : \func{D} \naturalto \func{U}$, which for each $M \in \obj{Loc}$ is the canonical map
\begin{equation*}
  \Phi_M: \func{D}(M) \to \func{U}(M),
  \quad
  f \mapsto (0, f, 0, \dotsc),
\end{equation*}
is called the \Idx{quantum-field}\emph{(locally covariant) quantum field} associated to~$\func{U}$.
Observe that every element in~$\func{U}(M)$ is a limit of sums and products of~$\Phi_M$ applied to test functions because $\bigoplus_n \func{D}(M)^{\otimes n}$ is dense in~$\func{U}(M)$.

\subsection{Field equation and commutator}
\label{sub:kg_eq_comm}

The Borchers--Uhlmann algebra carries no dynamical information and may therefore also be called the \IdxMain{off-field-alg}\emph{off-shell field algebra}.
In particular, the theory~$\func{U}$ is neither causal nor does it satisfy the timeslice axiom.
To obtain a causal theory that satisfies the timeslice axiom, we need to implement a \IdxMain{field-eq}\emph{field equation} (an \emph{equation of motion}) that induces a Cauchy evolution of the algebra elements and a commutator that `separates' causally disjoint algebra elements.

Therefore, we assign now to every globally hyperbolic spacetime~$M$ a natural, formally self-adjoint, Green-hyperbolic operator $\mathrm{P}_M : \Omega^p(M, \CC) \to \Omega^p(M, \CC)$ such that
\begin{equation*}
  \pf{\psi} \circ \mathrm{P}_M = \mathrm{P}_N \circ \pf{\psi}
\end{equation*}
for ever hyperbolic embedding $\psi : M \to N$.
Moreover, we define a functor $\func{D}$ as in the previous section.

As discussed in \cref{sec:wave}, associated to the Green-hyperbolic operator~$\mathrm{P}_M$, there exists on each globally hyperbolic manifold a unique causal propagator $\mathrm{G}_M$.

The causal propagator is also called the \Idx{comm-dist}commutator distribution or Pauli--Jordan distribution (see also \cref{sub:causal_propagator}) because it facilitates the definition of a \IdxMain{commutator}\emph{commutator} on the algebra~$\func{U}(M)$.
Namely, let $f, h \in \func{D}(M)$ then we define on~$\func{U}(M)$
\begin{equation}\label{eq:Phi_commutator}
  [\Phi_M(f), \Phi_M(h)] \defn \big( \im G_M(f \otimes h), 0, \dotsc \big).
\end{equation}
Due to the support properties of the causal propagator, this is exactly the right choice if one wants to implement Einstein causality.

The commutator extends to arbitrary elements of~$\func{U}(M)$ in the following way:
First, we notice that the commutator ought to satisfy the Leibniz rule.
Therefore it may be seen as a map
\begin{equation*}
  \func{D}(M)^{\otimes n} \otimes \func{D}(M)^{\otimes m} \longrightarrow \func{D}(M)^{\otimes (n+m-2)}
\end{equation*}
for $n,m \geq 1$, which can be extended (anti-)linearly to $\bigoplus_n \func{D}(M)^{\otimes n}$, a dense subalgebra of~$\func{U}(M)$.
Finally, we can extend the resulting commutator continuously to the whole algebra~$\func{U}(M)$.
Thereby the algebra becomes a Lie algebra.

The commutator is of immense physical importance.
Foremost, it implements causality and manifests the principle of locality.
Moreover, if the centre of the algebra of observables with respect to the commutator is non-trivial, the algebra contains unobservables and one cannot justify calling it `algebra of observables'.
Nevertheless, non-trivial centres in the `algebra of observables' can lead to important non-local observable effects under spacetime embeddings \cite{sanders:2014}.

Note that the commutator~\eqref{eq:Phi_commutator} defined on~$\func{U}(M)$ is degenerate if $\func{D}(M) \cap \ker \mathrm{G}_M$ is non-trivial and thus it leads to an algebra with a non-trivial centre.
This problem will be addressed in the following section, where we introduce the so-called field algebra.

\subsection{Field algebra}
\label{sub:kg_field_algebra}

Then, taking the wave operator and the commutator, we can define the \IdxMain{on-field-alg}\emph{(on shell) field algebra}~$\func{F}(M)$ as the unital topological ${}^*$-algebra given for every $M \in \obj{Loc}$ by the quotient
\begin{equation*}
  \func{F}(M) \defn \func{U}(M) / \func{I}(M),
\end{equation*}
where $\func{I}(M)$ is the completion of the ${}^*$-ideal finitely generated for all $f, h \in \func{D}(M)$ by
\begin{enumerate}[(a)]
  \item
    the wave equation
    \begin{equation*}
      \Phi_M(\mathrm{P}_M f) \sim \Phi_M(0)
    \end{equation*}
  \item
    the commutator relation
    \begin{equation*}
      \Phi_M(f) \Phi_M(h) - \Phi_M(h) \Phi_M(f) \sim [\Phi_M(f), \Phi_M(h)]
    \end{equation*}
\end{enumerate}
The topology of~$\func{F}(M)$ is the quotient topology with respect to~$\func{U}(M)$.

Like for the Borchers--Uhlmann algebra, $\func{F} : \Loc \to \Alg$ defines a functor, where $\func{F}(M)$ is the field algebra and $\func{F}(\psi)$ the ${}^*$-algebra homomorphism
\begin{equation*}
  \func{F}(\psi) [F] =  [\psi_* F]
\end{equation*}
on all $[F] \in \func{F}(M)$, which is naturally induced from~$\func{U}(\psi)$ via the canonical projection $[\,\cdot\,] : \func{U}(M) \to \func{F}(M)$.
That these assignments give indeed a covariant functor, relies on the naturality of all involved operators.
In particular,
\begin{equation*}
  \pf{\psi} (\mathrm{P}_M f) = \mathrm{P}_N (\pf{\psi} f)
  \quad \text{and} \quad
  G_M(f \otimes h) = G_N(\pf{\psi} f \otimes \pf{\psi} h)
\end{equation*}
for all $\psi : M \to N$ and $f, h \in \func{D}(M)$.
Note that the field algebra is a Lie algebra, where the bracket is simply
\begin{equation*}
  \big[ [F], [H] \big] \defn [F H - H F] = \big[ [F, H] \big]
\end{equation*}
for all $[F], [H] \in \func{F}(M)$; the centre is trivial by construction.
Moreover, we can construct a \Idx{quantum-field}quantum field~$\phi$ for~$\func{F}$ as the natural transformation $\phi : \func{D} \naturalto \func{F}$, which is given for each $M \in \obj{Loc}$ by
\begin{equation*}
  \phi_M \defn [\,\cdot\,] \circ \Phi_M.
\end{equation*}
That is, $\phi_M$ is the map $f \mapsto [(0, f, 0, \dotsc)]$ for all test functions $f \in \func{D}(M)$.

The following is a standard result, see \eg\ \cite[Chap.~3.1]{sanders:2008}:

\begin{proposition}\label{prop:causality_time_slice}
  The locally covariant theory~$\func{F}$, given by the field algebra, satisfies both causality and the timeslice axiom.
\end{proposition}
\begin{proof}
  Since $\func{F}(M)$ is the completion of the algebra generated by~$\phi_M$, causality follows immediately from the support properties of the causal propagator~$\mathrm{G}_M$.

  For the timeslice axiom we only need to show that the algebra in the whole spacetime can be reconstructed from the algebra in a causally convex neighbourhood $N \subset M$ of a Cauchy surface; $N$ may be considered as a spacetime in $\Loc$\ with a Cauchy embedding into~$M$.
  Set $\chi \in C^\infty(M)$ such that $\chi = 1$ on $J^+(N) \setminus N$ and and $\chi = 0$ on $J^-(N) \setminus N$.
  For every $f \in \func{D}(M)$ there exists $f' \in \func{D}(M)$ given by
  \begin{equation*}
    f' = \mathrm{P}_M (\chi \mathrm{G}_M f)
  \end{equation*}
  such that $\supp f' \subset N$ and
  \begin{equation*}
    f - f' = \mathrm{P}_M \big( (1 - \chi) \mathrm{G}_{\vee,M} f + \chi \mathrm{G}_{\wedge,M} f \big) \in \ker \mathrm{G}_M \cap \func{D}(M)
  \end{equation*}
  The statement follows again because $\func{F}(M)$ is the completion of the algebra generated by~$\phi_M$.
\end{proof}

If $\omega$ is a state on the Borchers--Uhlmann algebra $\func{U}(M)$ for some spacetime~$M$, it induces a state for the field algebra $\func{F}(M)$ if it also satisfies the commutation relation\Idx{commutator}
\begin{equation*}
  \omega(F H) - \omega(H F) = \omega([F, H])
\end{equation*}
and the equation of motion
\begin{equation*}
  \omega\big( (\id \otimes \dotsm \otimes \mathrm{P}_M \otimes \dotsm \otimes \id) F \big) = 0
\end{equation*}
for all $F, H \in \func{U}(M)$.
In that case, the state on $\func{F}(M)$ is defined by the pushforward of~$\omega$ by $[\,\cdot\,]$.
Conversely, a state on $\func{F}(M)$ always induces a state on $\func{U}(M)$ by the pullback via $[\,\cdot\,]$.

\subsection{Weyl algebra}
\label{sub:kg_weyl_algebra}

The disadvantage of the Borchers--Uhlmann algebra and the field algebra is that they are only ${}^*$-algebras and not $C^*$-algebras and hence cannot generally be represented by an algebra of bounded operators.
However, heuristically speaking, we can exponentiate the field algebra to produce a $C^*$-algebra, the Weyl $C^*$-algebra introduced in \cref{sub:alg_weyl}.

In this section, set
\begin{equation*}
  \func{D}(M) = \Omega^p_0(M, \CC) / \ker \mathrm{G}_M.
\end{equation*}
$\func{D}$ extends to a functor from~$\Loc$ into the category of closed nuclear locally convex $\CC$ vector spaces because $\mathrm{G}_M$ is continuous (so that $\ker \mathrm{G}_M$ is closed) and transforms covariantly under hyperbolic embeddings (\cf\ the previous section).

Then define on every spacetime $M \in \Loc$ the Weyl $C^*$-algebra~$\overline{\func{W}}(M)$ obtained from the Weyl operators $W_M : \func{D}(M) \to \overline{\func{W}}(M)$ with commutator relation\Idx{commutator}
\begin{equation*}
  W_M([f]) W_M([h]) = \exp\big( \tfrac{\im}{2} G_M([f] \otimes [h]) \big) W_M([f + h])
\end{equation*}
for every $[f], [h] \in \func{D}(M)$ and where we denoted also by~$G_M$ the well-defined pushforward of $G_M$ to $\func{D}(M)$ via the canonical quotient map.
Again, it may be shown that $\overline{\func{W}}$ extends to a functor from~$\Loc$ to $\Alg$ by the covariance of the commutator distribution.\footnote{To be precise, it is also necessary to show that the minimal regular norm behaves in a locally covariant way. This follows from the Hahn--Banach theorem.}
The proof that $\overline{\func{W}}$ satisfies both causality and the timeslice axiom is very similar to~\cref{prop:causality_time_slice} and will not be repeated.

\subsection{Scalar Klein--Gordon field}
\label{sub:kg_scalar}

The (free, scalar) \IdxMain{klein-gordon-eq}\emph{Klein--Gordon equation} is
\begin{equation}\label{eq:klein_gordon}
  (\mathord\Box + \xi R + m^2) \varphi = 0,
\end{equation}
where $\varphi \in \mathcal{E}(M)$ is the classical \IdxMain{klein-gordon-field}\emph{Klein--Gordon field} and the parameters $\xi$ and $m \geq 0$ are the \IdxMain{curv-coupl}\emph{curvature coupling} and the \emph{mass}.

One distinguishes in particular between two different curvature couplings: \IdxMain{min-curv-coupl}\emph{minimal coupling} if $\xi = 0$ and \IdxMain{conf-curv-coupl}\emph{conformal coupling} if $\xi = 1/6$.
The reason for naming $\xi = 1/6$ conformal coupling is that, in the massless case $m=0$, \eqref{eq:klein_gordon} is invariant under conformal isometries, see \eg\ \cite{wald:1984,dappiaggi:2006}.
Namely, given a conformal embedding $\psi : M \to N$ with $\pb{\psi} h = \Omega^2 g$, one finds
\begin{equation*}
  \pb{\psi} \left( \mathord\Box + \tfrac{1}{6} R \right) \varphi = \Omega^3 \left( \mathord\Box + \tfrac{1}{6} R \right) \Omega^{-1} \pb{\psi} \varphi,
\end{equation*}
where $\varphi \in \mathcal{E}(M)$.
That is, if $\varphi$ solves the massless conformally coupled Klein--Gordon equation on~$(N, h)$, then $\Omega^{-1} \pb{\psi} \varphi$ solves the massless conformally coupled Klein--Gordon equation on~$(M, g)$.

The field algebra for the Klein--Gordon field can be constructed exactly as outlined above, where we choose
\begin{equation*}
  \mathrm{P}_M = \Box + \xi R + m^2
  \quad \text{and} \quad
  \func{D}_\varphi(M) = \mathcal{D}(M, \CC),
\end{equation*}
and denote the resulting functor~$\func{F}_\varphi$ and the quantum field~$\what{\varphi}$.
The Weyl algebra may be constructed in a similar way.
In case of conformal coupling, the Klein--Gordon field can also be quantized as a \emph{conformally} locally covariant theory~\cite{pinamonti:2009,siemssen:2011}.

\subsection{Proca field}
\label{sub:kg_proca}

The \IdxMain{proca-eq}field equation for the classical \IdxMain{proca-field}\emph{Proca field}~$Z \in \Omega^1(M)$ of mass $m > 0$ is
\begin{equation}\label{eq:proca}
  (\codif \dif + m^2) Z = 0
\end{equation}
and one can almost immediately see that it is not normally hyperbolic.
However, applying the codifferential to this equation we find that $\codif Z = 0$ so that \eqref{eq:proca} is equivalent to
\begin{equation*}
  (\mathord\Box + m^2) Z = 0
\end{equation*}
with the constraint $\codif Z = 0$.

There are two equivalent approaches to quantize this constrained system.
Both rely on the same fact that the exterior derivative~$\dif$ and the codifferential~$\codif$ commute with $(\Box + m^2)$ and thus also with its causal propagator~$\mathrm{G}_M$ \cite{fewster:2003,pfenning:2009}.

For first approach \cite{fewster:2003,dappiaggi:2011b} we notice that \eqref{eq:proca} is pre-normally hyperbolic so that we can construct the causal propagator
\begin{equation*}
  \widetilde{\mathrm{G}}_M = \mathrm{G}_M \circ \big( m^{-2} \dif \codif + \id \big).
\end{equation*}
Accordingly we can perform the construction of the field algebra discussed in \cref{sub:kg_field_algebra} with this propagator.
More precisely, we set
\begin{equation*}
  \mathrm{P}_M = \codif \dif + m^2
  \quad \text{and} \quad
  \func{D}_Z(M) = \Omega^1_0(M, \CC)
\end{equation*}
and follow through with the construction of the field algebra, where we denote the resulting locally covariant theory~$\func{F}_Z$ and the corresponding quantum field~$\what{Z}$.

The second possibility, which is related to the framework developed in~\cite{hack:2013a}, is to consider
\begin{equation*}
  \mathrm{P}_M = \Box + m^2
  \quad \text{and} \quad
  \func{D}'_Z(M) = \{ f \in \Omega^1_0(M, \CC) \mid \codif f = 0 \}
\end{equation*}
Since every section $f \in \func{D}'_Z(M)$ is coclosed, $\mathrm{G}_M f$ solves the Proca equation~\eqref{eq:proca}.
We can then perform the usual construction for the field algebra and denote the corresponding locally covariant theory~$\func{F}'_Z$.

\begin{proposition}
  The locally covariant theories~$\func{F}_Z$ and~$\func{F}'_Z$ are equivalent.
  That is, $\func{F}_Z(M) \simeq \func{F}'_Z(M)$ for every globally hyperbolic~$M$.
\end{proposition}
\begin{proof}
  For every $f \in \func{D}_Z(M)$
  \begin{equation*}
    m^2 \what{Z}(f) = \what{Z} \big(m^2 f - (\codif \dif + m^2) f \big) = -\what{Z} (\codif \dif f)
  \end{equation*}
  and thus $\what{Z} \big( \func{D}_Z(M) \big) = \what{Z} \big( \func{D}'_Z(M) \big)$.
  Then it is easy to check that
  \begin{equation*}
    (\codif \dif + m^2) f = (\Box + m^2) f
    \quad \text{and} \quad
    \widetilde{\mathrm{G}}_M f = \mathrm{G}_M f
  \end{equation*}
  for all $f \in \func{D}'_Z(M)$.
\end{proof}

Similarly, two apparently different Weyl algebras can be constructed for the Proca field and shown to be equivalent.

%% file: states.tex
%!TEX root = master.tex

\chapter{Quantum states}
\label{cha:states}

\IdxRanBegin{state}

\section*{Summary}

In \cref{sub:alg_states} we already introduced some general features of states on ${}^*$-algebras.
In this section we will discuss features important or specific to quantum field theory.

We begin our discussion with the introduction of the $n$-point distributions (\cref{sub:n-point}) associated to (some) states of the algebras defined above: the Borchers--Uhlmann algebra, the field algebra and the Weyl algebra.
Of particular importance are states which satisfy the microlocal spectrum condition to be defined in \cref{sub:microlocal}.
States which satisfy this constraint on the wavefront set are the so-called Hadamard states and their singular part is given by the Hadamard parametrix (\cref{sub:hadamard_parametrix}).

After we introduced these general notions, we will discuss the construction of quantum states on particular spacetimes.
Due to their importance in cosmology and their relative simplicity, we discuss adiabatic and Hadamard states on cosmological spacetimes in \cref{sec:states_cosmology}.

\section{Preliminaries}

Let $(M, g)$ be a globally hyperbolic spacetime and let us consider, as in \ref{sec:klein_gordon}, $p$-form fields.
It is important to notice that none of the results here are fundamentally restricted to the assumption of $p$-form field and can be easily generalized.

\subsection{\texorpdfstring{$n$}{n}-Point distributions}
\label{sub:n-point}

Let $\func{U}(M)$ be the Borchers--Uhlmann algebra of a quantum field theory on~$M$ which is built on a test function space~$\func{D}(M) \subset \Omega_0^p(M)$.
Since
\begin{equation*}
  \func{U}(M) = \bigoplus_{n \in \NN_0} \func{D}(M)^{\mathbin{\widehat{\otimes}} n},
\end{equation*}
the topological dual is of the form
\begin{equation*}
  \func{U}'(M) = \prod_{n \in \NN_0} \func{D}'(M)^{\mathbin{\widehat{\otimes}} n},
\end{equation*}
where we have used the kernel theorem.
In other words, whereas any element of $\func{U}(M)$ can be understood as a polynomial, $\func{U}'(M)$ also contains power series.
It follows that the any state~$\omega$ on~$\func{U}(M)$ is uniquely defined by a family $(\omega_n)_{n \in \NN}$ of \IdxMain{n-point}\emph{$n$-point distributions} (also called $n$-point functions or Wightman functions) $\omega_n \in \func{D}'(M)^{\mathbin{\widehat{\otimes}} n}$.
If we denote by~$\Phi_M$ the quantum field associated to~$\func{U}(M)$, then the $n$-point distributions satisfy
\begin{equation*}
  \omega_n(f_1 \otimes \dotsm \otimes f_n) = \omega\big(\Phi_M(f_1) \,\dotsm\, \Phi_M(f_n)\big)
\end{equation*}
for all $f_1, \dotsc, f_n \in \func{D}(M)$ and $n \in \NN$.

The \IdxMain{trunc-n-point}\emph{connected} or \emph{truncated $n$-point distributions}~$\omega^T_n$ of a state~$\omega$ are defined by the relation
\begin{equation}\label{eq:truncated}
  \omega_n(x_1, \dotsc, x_n) = \sum_{P \in \mathcal{P}_{\!n}} \prod_{r \in P} \omega^T_{\abs{r}}(x_{r(1)}, \dotsc, x_{r(\abs{r})}),
\end{equation}
where $\mathcal{P}_{\!n}$ denotes the (ordered) partitions of the set $\{ 1, \dotsc, n \}$.
Therefore, they can be calculated recursively from the $n$-point distributions.
The first two truncated $n$-point distributions are
\begin{align*}
  \omega^T_1(x_1) & = \omega_1(x_1), \\
  \omega^T_2(x_1, x_2) & = \omega_2(x_1, x_2) - \omega_1(x_1) \omega_1(x_2)
\intertext{and a general recursive formula is given by}
  \omega^T_n(x_1, \dotsc, x_n) & = \omega_n(x_1, \dotsc, x_n) - \sum_{\substack{P \in \mathcal{P}_{\!n}\\\abs{P} > 1}} \prod_{r \in P} \omega^T_{\abs{r}}(x_{r(1)}, \dotsc, x_{r(\abs{r})}).
\end{align*}

Thanks to the close relation of the Borchers--Uhlmann algebra and the field algebra, see the last paragraph of \cref{sub:kg_field_algebra}, the space of states of the field algebra is related to a subspace of the space of states for the Borchers--Uhlmann algebra and a state~$\omega$ on the field algebra also has associated $n$-point distributions.
These $n$-point distributions naturally satisfy the commutation relations\Idx{commutator}
\begin{equation*}\begin{split}
  \omega_n(x_1, \dotsc, x_n) & = \omega_n(x_1, \dotsc, x_{i+1}, x_i, \dotsc, x_n) \\&\quad + \omega_{n-2}(x_1, \dotsc, \what{x}_i, \what{x}_{i+1}, \dotsc, x_n) G_M(x_i, x_{i+1}),
\end{split}\end{equation*}
where the hats denote omitted points, and are weak solutions of the equations of motion\Idx{field-eq}
\begin{equation*}
  \mathrm{P}_M(x_i) \omega_n(x_1, \dotsc, x_i, \dotsc, x_n) = 0
\end{equation*}
for all $i \in [1 \lddots n]$.
Therefore, if we denote by~$\Phi_M$ the quantum field associated to $\func{F}(M)$, the $n$-point distributions satisfy
\begin{equation*}
  \omega_n(f_1 \otimes \dotsm \otimes f_n) = \omega\big(\Phi_M(f_1) \dotsm \Phi_M(f_n)\big)
\end{equation*}
independently of the chosen representatives $f_1, \dotsc, f_n \in \func{D}(M)$ of $[f_1], \dotsc, [f_n]$.

The definition of $n$-point distributions for a state on the Weyl algebra $\func{W}(M)$ is slightly more involved.
$n$-Point distributions in the algebraic sense only exist for strongly regular states as defined in \cref{sub:alg_weyl}, see also \cite{bar:2012}.
In this case they are defined by the relation
\begin{equation*}
  \omega_n(f_1 \otimes \dotsm \otimes f_n) = (-\im)^n \frac{\partial^n}{\partial t_1 \dotsm \partial t_n}\, \omega\big(W_M(t_1 [f_1]) \dotsm W_M(t_n [f_n])\big) \bigg|_{t_\bullet = 0}\!\!.
\end{equation*}
Clearly, the $n$-point distributions of a state on the Weyl algebra satisfy the commutation relation and the equation of motion.

A state~$\omega$ is called \IdxMain{quasi-free}\emph{quasi-free} or \emph{Gaussian} if all its truncated $n$-point distributions vanish for $n \neq 2$, whence it is completely determined by its two-point distribution~$\omega_2$.
For a quasi-free state, all odd $n$-point distributions vanish and all even $n$-point distributions are given by
\begin{equation*}
  \omega_n(x_1, \dotsc, x_n)
  = \sum_{\sigma} \omega_2(x_{\sigma(1)}, x_{\sigma(2)}) \,\dotsm\, \omega_2(x_{\sigma(n-1)}, x_{\sigma(n)}),
\end{equation*}
where the sum is over all ordered pairings, \ie, over all permutations $\sigma \in \mathfrak{S}_n$ such that $\sigma(1) < \sigma(3) < \dotsb < \sigma(n-1)$ and $\sigma(1) < \sigma(2), \dotsc, \sigma(n-1) < \sigma(n)$.

\subsection{Microlocal spectrum condition}
\label{sub:microlocal}

A quasi-free state~$\omega$ satisfies the \IdxMain{micro-spec-cond}\emph{microlocal spectrum condition} \cite{radzikowski:1996,brunetti:1996,sahlmann:2001} if
\begin{equation}\label{eq:spectrum_condition}
  \WF(\omega_2) \subset \big\{ (x, x'; \xi, -\xi') \in \dot{T}^*(M \times M) \;\big|\; (x; \xi) \sim (x'; \xi') \text{ and } \xi \rhd 0 \big\},
\end{equation}
or, in words, the wavefront set of $\omega_2$ is contained in the set of $(x, x'; \xi, -\xi') \in \dot{T}^*(M \times M)$ such that $x, x'$ are connected by a lightlike geodesic $\gamma$ with cotangent $\xi$ at $x$ and $\xi'$ is the parallel transport of $\xi$ to $x'$ along $\gamma$ (in symbols: $(x; \xi) \sim (x'; \xi')$) and $\xi$ is future directed (in symbols: $\xi \rhd 0$).
That is, $(x, x'; \xi, -\xi')$ is contained in the wavefront set if $(x; \xi)$ and $(x'; \xi')$ lie on the same future-directed bicharacteristic strip generated by~$\sigma(\xi) = -g(\xi, \xi)$.

The microlocal spectrum condition can also be generalized to states that are not quasi-free~\cite{brunetti:1996,sanders:2010}.
States that satisfy the microlocal spectrum condition are called \IdxMain{hadamard-state}\emph{(generalized) Hadamard states}.

Let $\mathrm{P} = \Box + B$, where $B$ is a scalar function, the potential, and $G(x, x')$ its causal propagator.
If the kernel of the two-point distribution~$\omega_2$ satisfies the commutator relation (weakly)
\begin{equation*}
  \omega_2(x, x') - \omega_2(x', x) = \im G(x, x')
\end{equation*}
then equality of sets holds in~\eqref{eq:spectrum_condition}.
If, moreover, the two-point distribution is a parametrix of~$\mathrm{P}_M$, \ie a weak bisolution up to smooth terms, then it attains the \IdxMain{hadamard-form}\emph{local\footnote{Also a \emph{global} Hadamard form can be formulated \cite{kay:1991}, but since the discovery of the microlocal spectrum condition this global form has lost its importance. In fact, it was shown in~\cite{radzikowski:1996a} that a state that is everywhere locally of Hadamard form is also globally a Hadamard state.} Hadamard form} in a geodesically convex neighbourhood $U \subset M$
\begin{equation}\begin{split}\label{eq:hadamard_form}
  \omega_2(x, x')
  & = \lim_{\varepsilon \to 0^+} \frac{1}{8 \uppi^2} \bigg( \frac{u(x, x')}{\sigma_\varepsilon(x, x')} + v(x, x') \ln \frac{\sigma_{\varepsilon}(x, x')}{\lambda^2} + w(x, x') \bigg) \\
  & = \mathcal{H}(x, x') + w(x, x'),
\end{split}\end{equation}
where we take the \emph{weak} limit, $x, x' \in U$, $\lambda \in \RR$ is arbitrary and the detailed form of the coefficients $v, w \in \Gamma(\bigwedge^p(TM) \boxtimes \bigwedge^p(TM))$ will be discussed in the next section.
Above we used a `vectorized' van Vleck--Morette determinant
\begin{equation*}
  u(x, x') \defn \Delta^{1/2}(x, x') g^{[a_1^{\vphantom{\prime}}| b_1'}(x, x') \,\dotsm\, g^{|a_p^{\vphantom{\prime}}] b_p'}(x, x'),
\end{equation*}
which is antisymmetrized in the indices~$a_i$ and parallel transported along the geodesic connecting $x$ and $x'$, and the regularized world function
\begin{equation*}
  \sigma_{\varepsilon}(x, x') \defn \sigma(x, x') + \im \varepsilon\, \big(t(x) - t(x')\big) + \tfrac{1}{2} \varepsilon^2
\end{equation*}
with $t$ a smooth time function on $(M, g)$ compatible with the time-orientation.

\subsection{The Hadamard parametrix}
\label{sub:hadamard_parametrix}

The coefficient functions
\begin{equation}\label{eq:hadamard_v_series}
  v(x, x') = \frac{1}{\lambda^2} \sum_{k = 0}^\infty v_k(x, x') \bigg( \frac{\sigma(x, x')}{\lambda^2} \bigg)^k
\end{equation}
in~\eqref{eq:hadamard_form} are called \IdxMain{hadamard-coeff}\emph{Hadamard coefficients} and are another example of bitensors.
Note that the expansion above is an \emph{asymptotic} expansion in terms of the world function $\sigma$ and cannot be expected to converge unless the spacetime is analytic.
Although it can be turned into a convergent series by replacing the series~\eqref{eq:hadamard_v_series} by \cite[Chap.~2.5]{bar:2007}
\begin{equation*}
  \sum_{k = 0}^n v_k(x, x') \sigma(x, x')^k + \!\sum_{k = n+1}^\infty\! v_k(x, x') \chi\big( \alpha_k^{-1} \sigma(x, x') \big)^k
\end{equation*}
for any $n \in \NN$ and some sequence~$(\alpha_k), \alpha_k \in (0,1]$, where $\chi \in C_0^\infty(\RR)$ is $1$ in a neighbourhood of~$0$ (note that we omitted the scale~$\lambda$), this will not concern us any further because we will only ever need a finite number of terms.\footnote{If we inserted this modification into~\eqref{eq:hadamard_form}, the two-point distribution would not be an exact (weak) solution of~$\mathrm{P}$ any more, but only up to a smooth biscalar, \ie, it would only be a parametrix.}
Therefore, we also define the truncated local Hadamard parametrix
\begin{equation*}
  \mathcal{H}_n(x, x') \defn \lim_{\varepsilon \to 0^+} \frac{1}{8 \uppi^2} \bigg( \frac{u(x, x')}{\sigma_\varepsilon(x, x')} + \sum_{k = 0}^n v_k(x, x') \sigma(x, x')^k \ln \frac{\sigma_{\varepsilon}(x, x')}{\lambda^2} \bigg).
\end{equation*}
One can show that there always exists a $n \in \NN_0$ such that
\begin{equation*}
  \lim_{x' \to x} \mathcal{D} \big( \Hadamard(x, x') - \Hadamard_n(x, x') \big) = 0
\end{equation*}
for all differential operators~$\mathcal{D}$ and $n$ depends on the order of~$\mathcal{D}$.

The coefficients~$v_k$ can be recursively calculated by (formally) applying $\mathrm{P}$ to~$\mathcal{H}$;
One then finds the so-called \IdxMain{hadamard-recurs}\emph{Hadamard recursion relations} (\cf\ \cite{ottewill:2011,decanini:2006,fewster:2008})
\begin{subequations}\label{eq:v_rec_relation}\begin{align}
  \lambda^2 \mathrm{P} u & = ( 2 \nabla_\sigma + \sigma\indices{^a_a} - 2 ) v_0, \\
  \lambda^2 \mathrm{P} v_{k-1} & = ( 2 \nabla_\sigma + \sigma\indices{^a_a} + 2k - 2 ) k v_k,
\end{align}\end{subequations}
where we have used the transport operators defined in~\eqref{eq:transport_operators}.
It can be shown that the Hadamard coefficients are symmetric in their arguments~\cite{moretti:2000,friedlander:1975}.
Together with the first term in~\eqref{eq:hadamard_form} the Hadamard coefficients make up the \IdxMain{hadamard-para}\emph{Hadamard parametrix}~$\mathcal{H}(x, x')$, which is therefore completely determined by the differential operator~$\mathrm{P}$ and the geometry of the spacetime.

The covariant expansion of the Hadamard coefficients can be efficiently calculated using the Avramidi method described in \cref{sub:bi_avramidi} using the transport equations~\eqref{eq:v_rec_relation}.
If one is only interested in the coincidence limits, one can directly take the limit in~\eqref{eq:v_rec_relation} to find (omitting necessary Kronecker deltas originating from coincidence limits of parallel propagators)
\begin{align*}
  [v_0] & = \frac{1}{2} [\mathrm{P} u] = \frac{1}{2} \bigg( B - \frac{1}{6} R \bigg), \\
  [v_k] & = \frac{1}{2 (1 - k) k} [\mathrm{P} v_{k-1}], \quad k > 1.
\end{align*}
Note that $[v_0]$ vanishes for a conformally coupled massless scalar field.
The coincidence limit of~$v_1$ cannot be found in this easy way and must be calculated directly.
After a lengthy calculation using pencil and paper or a (fast) calculation using a tensor algebra software, one obtains (again omitting Kronecker deltas)
\begin{equation}\label{eq:hadamard_v1}
  8 [v_1] = B^2 + \tfrac{1}{3} \Box B - \tfrac{1}{3} R B + \tfrac{1}{36} R^2 - \tfrac{1}{90} R_{ab} R^{ab} + \tfrac{1}{90} R_{abcd} R^{abcd} - \tfrac{1}{15} \Box R.
\end{equation}

Different from $v(x, x')$, the symmetric bitensor $w(x, x')$ is not directly determined by the geometry or a differential operator.
Instead the term $w(x, x')$ reflects the freedom in the choice of the state.
Writing the asymptotic expansion
\begin{equation*}
  w(x, x') = \frac{1}{\lambda^2} \sum_{k = 0}^\infty w_k(x, x') \bigg( \frac{\sigma(x, x')}{\lambda^2} \bigg)^k,
\end{equation*}
we notice that the freedom to choose a state is completely encoded in the first coefficient $w_0$ and the remaining coefficients obey the recursion relation \cite{decanini:2006,fewster:2008}
\begin{equation*}\begin{split}
  \lambda^2 \mathrm{P} w_k & = 2 (k+1) \nabla_\sigma w_{k+1} + 2 k (k+1) w_{k+1} + (k+1) w_{k+1} \sigma\indices{^a_a} \\&\quad + 2 \nabla_\sigma v_{k+1} - 2 (2k+1) v_{k+1} + v_{k+1} \sigma\indices{^a_a}.
\end{split}\end{equation*}
A common choice is to set $w_0 = 0$ as in~\cite{wald:1978}.
In any case, $w_0$ must be chosen such that $w$ is symmetric.

\section{Construction of states on cosmological spacetimes}
\label{sec:states_cosmology}

Explicit examples of quantum states are known only on a small class of highly symmetric spacetimes.
Below we will first discuss the so-called Bunch--Davies state \cite{bunch:1978,schomblond:1976,allen:1985}, which can be considered the vacuum state of de Sitter spacetime.
Then we study a construction of states on FLRW spacetimes due to Parker~\cite{parker:1969} called \Idx{adiabatic-state}adiabatic states.
Although adiabatic states are in general not Hadamard, indeed only adiabatic states of infinite order satisfy the microlocal spectrum condition \cite{junker:2002}, they can be considered approximate Hadamard states and have proven to be very useful thanks to their relatively straightforward construction.
Since we will make extensive use of adiabatic states when we discuss the semiclassical Einstein equation on cosmological backgrounds in \cref{cha:einstein_solutions}, they will be treated in some detail below.
Moreover, we will introduce the states of low energy by Olbermann~\cite{olbermann:2007}, which are constructed via a careful Bogoliubov transformation from adiabatic states.

\subsection{Bunch--Davies state}
\label{sub:bunch_davies}

A distinguished Hadamard state for the (massive) scalar field on de Sitter spacetime is the \IdxMain{bunch-davies}\emph{Bunch--Davies state} \cite{bunch:1978,schomblond:1976,allen:1985}.
It is the unique pure, quasi-free Hadamard state invariant under the symmetries of de Sitter spacetime.
Note that equations of motion for the scalar field on de Sitter spacetime are
\begin{equation*}
  \Box \varphi + (12 \xi H^2 + m^2) \varphi = 0,
\end{equation*}
where $H$ is the Hubble constant, $m$ the mass of the scalar field and $\xi$ the curvature coupling, \cf\ \cref{sub:kg_scalar}.
Therefore the curvature coupling~$\xi$ acts like a mass and we set $M^2 = 12 \xi H^2 + m^2$.
The Bunch--Davies state is also a Hadamard state in the limit $M = 0$ but in that case it fails to be invariant under the symmetries of de Sitter spacetime \cite{allen:1985,allen:1987}; below we assume $M > 0$.

Using the function $Z$ defined in~\eqref{eq:deSitter_Z}, the Bunch--Davies state is the quasifree state given by the two-point distribution\footnote{More precisely, one should replace $Z$ by $Z + \im \varepsilon (t(x) - t(x'))$, where $t$ is a smooth time function, and take the limit $\varepsilon \to 0^+$. Note that  ${}_2F_1$ has a branch cut from $1$ to $\infty$.}
\begin{equation}\label{eq:bunch_davies}
  \omega_2(x, x') \defn \omega_2\big(Z(x, x')\big) \defn \frac{M^2 - 2H^2}{8 \uppi \cos(\uppi \nu)}\, {}_2F_1\big( \nu_+, \nu_-; 2; \tfrac{1}{2} (1+Z) \big),
\end{equation}
where ${}_2F_1$ is the analytically continued hypergeometric function and with
\begin{equation*}
  \nu_\pm \defn \frac{3}{2} \pm \nu
  \quad \text{and} \quad
  \nu \defn \sqrt{\frac{9}{4} - \frac{M^2}{H^2}}.
\end{equation*}
We can rewrite~\eqref{eq:bunch_davies} into a form which exhibits the Hadamard nature of the state more clearly.
In fact, using well-known transforms~\cite[Eq.~(15.8.10)]{olver:2010} of the hypergeometric function~${}_2F_1$, one can show
\begin{align*}
  \omega_2(Z) & = \frac{H^2}{8 \uppi^2} (1-Z)^{-1} + \frac{M^2 - 2 H^2}{8 \uppi^2} \Big( \wtilde{v}\big( \tfrac{1}{2}(1-Z) \big) \ln\big( \tfrac{1}{2}(1-Z) \big) + \wtilde{w}\big( \tfrac{1}{2}(1-Z) \big) \Big) \\
\intertext{for $\abs{Z} < 1$ (spacelike separated points) with}
  \wtilde{v}(z) & = {}_2F_1( \nu_+, \nu_-; 2; z), \\
  \wtilde{w}(z) & = \sum_{k=0}^{\infty} \frac{(\nu_+)_k (\nu_-)_k}{k! (k+1)!} \big( \digamma(\nu_+ + k) + \digamma(\nu_- + k) - \digamma(k+1) - \digamma(k+2) \big)\, z^k,
\end{align*}
where $\digamma$ is the digamma-function.

In the cosmological chart of de Sitter spacetime the function~$Z$ attains the simple form~\eqref{eq:deSitter_Z_cosmo} and the spatial sections are flat.
Therefore, one can represent $\omega_2$ as a spatial Fourier transformation with respect to $\vec{x} - \vec{x}'$.
Indeed, using known integrals of (modified) Bessel functions \cite[\S6.672]{gradshteyn:2007}, a lengthy calculation shows \cite{schomblond:1976} (omitting again the $\varepsilon$-prescription)
\begin{equation}\label{eq:bunch_davies_mode}\begin{split}
  \omega_2(x, x') = \frac{H^2 (\tau \tau')^{3/2}}{32 \uppi^2} \int_{\RR^3} \e^{-\uppi \Im\nu} H^{\smash{\mbox{\tiny $(\!1\!)$}}}_\nu(- k \tau) H^{\smash{\mbox{\tiny $(\!2\!)$}}}_{\wbar\nu}(- k \tau')\, \e^{\im \vec{k} \cdot (\vec{x} - \vec{x}')}\, \dif\vec{k},
\end{split}\end{equation}
where $x = (\tau, \vec{x})$ and $x' = (\tau', \vec{x}')$ in the conformal coordinates and $H^{\smash{\mbox{\tiny $(\!1\!)$}}}, H^{\smash{\mbox{\tiny $(\!2\!)$}}}$ are the Hankel functions of first and second kind.

\subsection{Homogeneous and isotropic states}
\label{sub:homogeneous_isotropic_state}

It is usually reasonable to restrict ones attention to states that respect the symmetry of the background spacetime.
Under this assumption, a state on a FLRW spacetime should be both homogeneous and isotropic.
That is, if the state is also quasifree, its two-point distribution needs to satisfy
\begin{equation*}
  \omega_2(x, x') = \omega_2(t, t', \vec{x} - \vec{x}') = \omega_2(\tau, \tau', \vec{x} - \vec{x}'),
\end{equation*}
where $x = (t, \vec{x}) = (\tau, \vec{x})$ and $x' = (t', \vec{x}') = (\tau', \vec{x}')$ with respect to cosmological or conformal time.

Under a certain relatively weak continuity assumption on the two-point distribution (such that they may be represented as bounded operators on a certain Hilbert space and the Riesz representation theorem can be used \cite{luders:1990}), it was shown in \cite{luders:1990,schlemmer:2010} that every \IdxMain{hom-iso-state}\emph{quasifree, homogeneous and isotropic state} for the scalar field is of the form\footnote{Note that we omit here and below the necessary $\varepsilon$-regularization of the integral, where one multiplies the integrand with $\e^{-\varepsilon k}$ and considers the weak limit $\varepsilon \to 0^+$.}
\begin{equation}\begin{split}\label{eq:homogeneous_isotropic_state}
  \omega_2(x, x') = \frac{1}{(2 \uppi)^3 a(\tau) a(\tau')} \int_{\RR^3} & \big( \Xi(k) S_k(\tau) \conj{S}{}_k(\tau') \\&\!\! + (\Xi(k) + 1) \conj{S}{}_{k}(\tau) S_k(\tau') \big)\, \e^{\im \vec{k} \cdot (\vec{x} - \vec{x}')}\, \dif \vec{k},
\end{split}\end{equation}
where $k = \abs{\vec{k}}$ and $\Xi(k)$ is a non-negative (almost everywhere), polynomially bounded function in $L^1(\RR_0^+)$; for pure states $\Xi = 0$.
Moreover, the modes $S_k$ are required to satisfy the \IdxMain{mode-equation}\emph{mode equation} of motion
% \begin{equation}\label{eq:mode_equations}
%   \left( \pd[2]{}{\tau} + k^2 + \left( \xi - \frac{1}{6} \right) a(\tau)^2 R + a(\tau)^2 m^2 \right) S_k^{\phantom{'}}(\tau) = 0,
%   \qquad
%   \conj{S}{}_k^{\phantom{'}} S_k' - \conj{S}{}_k' S_k^{\phantom{'}} = \im,
% \end{equation}
\begin{equation}\label{eq:mode_equation}
  (\partial_\tau^2 + \omega_k^2) S_k^{\phantom{'}}(\tau) = 0,
  \quad
  \omega_k^2 \defn k^2 + \big( \xi - \tfrac{1}{6} \big) a(\tau)^2 R + a(\tau)^2 m^2,
\end{equation}
and the Wronski-determinant condition\footnote{Imposing the Wronski-determinant condition guarantees that the imaginary part of the two-point distribution is given by half the commutator distribution. It is sufficient to impose this condition at one instance in time.}
\begin{equation}\label{eq:wronski}
  \conj{S}{}_k^{\phantom{'}} S_k' - \conj{S}{}_k' S_k^{\phantom{'}} = \im,
\end{equation}
where both $S_k^{\phantom{'}}$ and $S_k'$ are polynomially bounded in $k$.
States constructed in this manner are in general \emph{not} of Hadamard type.

Two important examples of pure Hadamard states expressible in the mode form above are the \IdxMain{mink-state}\emph{Minkowski vacuum state} on Minkowski spacetime
\begin{equation*}
  \frac{1}{(2 \uppi)^3} \int_{\RR^3} \frac{1}{2 E(k)} \e^{- \im E(k) (t - t')} \e^{\im \vec{k} \cdot (\vec{x} - \vec{x}')}\, \dif\vec{k},
\end{equation*}
with $E(k)^2 = k^2 + m^2$, and the Bunch--Davies state on the cosmological patch of de Sitter spacetime \eqref{eq:bunch_davies_mode}.
Interesting examples of non-pure states are the \emph{approximate KMS states} at inverse temperature $\beta$ for the conformally coupled scalar field \cite{dappiaggi:2011d}
\begin{equation*}
  \frac{1}{(2 \uppi)^3 a(\tau) a(\tau')} \int_{\RR^3} \left( \frac{S_k(\tau) \conj{S}{}_k(\tau')}{\e^{\beta k_F} - 1} + \frac{\conj{S}{}_{k}(\tau) S_k(\tau')}{1 - \e^{- \beta k_F}} \right)\, \e^{\im \vec{k} \cdot (\vec{x} - \vec{x}')}\, \dif \vec{k}
\end{equation*}
with $k_F = \sqrt{k^2 + a(\tau_F)^2 m^2}$ for some `freeze-out' time $\tau_F$.
These states are KMS states if the spacetime admits a global timelike Killing vector field which is a symmetry of the state; they are Hadamard states if the pure state specified by the modes $S_k$ is already a Hadamard state \cite{dappiaggi:2011d}.

Given fixed reference modes~$\chi_k$ that satisfy~\eqref{eq:mode_equation} and~\eqref{eq:wronski}, all other possible mode solutions~$S_k$ can be constructed via a \IdxMain{bogoliubov}\emph{Bogoliubov transformation}, \ie,
\begin{equation*}
  S_k = A(k) \chi_k + B(k) \conj{\chi}_k
  \qquad \text{with} \qquad
  \abs{A(k)}^2 - \abs{B(k)}^2 = 1,
\end{equation*}
where $A(k)$ and $B(k)$ are such that $k \mapsto S_k^{\phantom{'}}$ and $k \mapsto S_k'$ are (essentially) polynomially bounded, measurable functions.
Note that changing $S_k$ by a phase does not affect the state and therefore $B(k)$ can always be chosen to be real.
The choice of $A$ and $B$ thus corresponds to two degrees of freedom, \eg, the phase of $A$ and the modulus of $B$.
If the reference modes~$\chi_k$ specify a pure Hadamard state, one can show that the modes~$S_k$ with the mixing~$\Xi(k)$ specify a Hadamard state as well if and only if (in addition to the conditions above) $k^n B(k)$ and $k^n \Xi(k)$ are in~$L^1(\RR_0^+)$ for all~$n \in \NN$ and $\Arg A - \Arg B$ is measurable~\cite{pinamonti:2011,zschoche:2014}.
An important example of a Bogoliubov transformation of a Hadamard state that (clearly) does \emph{not} give a Hadamard state are the $\alpha$-vacua associated to the Bunch--Davies state for which $A = \sinh \alpha$ and $B = \cosh \alpha$ with $\alpha \in \RR$.

\subsection{Adiabatic states}
\label{sub:adiabatic}
\IdxRanBegin{adiabatic-state}

Any solution of~\eqref{eq:mode_equation} and~\eqref{eq:wronski} is of the form\footnote{The lower limit in the integration is arbitrary as it gives a constant phase.}
\begin{equation}\label{eq:mode_form}
  S_k(\tau) = \frac{\rho_k(\tau)}{\sqrt{2}}\, \e^{\im \theta_k(\tau)}
  \quad \text{with} \quad
  \theta_k(\tau) = \int^{\tau} \rho_k(\eta)^{-2}\, \dif\eta,
\end{equation}
where $\rho_k$ satisfies the differential equation
\begin{equation}\label{eq:rho_diff_eq}
  \rho_k'' = \big( \rho_k^{-4} - \omega_k^2 \big) \rho_k^{\vphantom{'}} = \big( \theta_k'{}^2 - \omega_k^2 \big) \rho_k^{\vphantom{'}}.
\end{equation}
Modes~$W_k = \sigma_k \e^{\im \psi_k} / \sqrt{2}$ of the form~\eqref{eq:mode_form} with arbitrary~$\sigma_k$ that do \emph{not} satisfy the differential equation~\eqref{eq:rho_diff_eq} can be used to specify initial values for solutions~$S_k$ of the mode equation~\eqref{eq:mode_equation}, \ie,
\begin{equation*}
  S_k^{\phantom{'}}(\tau_0) = W_k^{\phantom{'}}(\tau_0),
  \quad
  S_k'(\tau_0) = W_k'(\tau_0)
\end{equation*}
at some initial time~$\tau_0$.
Using an idea of Parker~\cite{parker:1969}, it can be shown that the $W_k$ yield the solution~$S_k$ via a Bogoliubov-like transformation
\begin{equation}\label{eq:W_to_S}
  S_k(\tau) = A(\tau) W_k(\tau) + B(\tau) \conj{W}_k(\tau)
\end{equation}
with time-dependent coefficients given by
\begin{subequations}\label{eq:mode_coefficients}\begin{align}
    A(\tau) & = 1 - \im\! \int_{\tau_0}^\tau\! G(\eta) \left( A(\eta) + B(\eta)\, \e^{-2\im \psi_k(\eta)} \right) \dif\eta, \\
    B(\tau) & = -\!\! \int_{\tau_0}^\tau\! A'(\eta)\, \e^{2\im \psi_k(\eta)}\, \dif\eta
\end{align}\end{subequations}
and $2 G \defn \sigma_k^{-2} - \omega_k^2 \sigma^2 - \sigma_k^{\phantom{'}} \sigma_k''$, where we have suppressed the $k$-dependence of~$A, B$ and~$G$ in all four equations above.
Applying arguments from the analysis of Volterra integrals, it can be shown that ${1-A}$, $A'$, $B$ and~$B'$ have the same large~$k$ behaviour as~$G$ (\cf\ \cite{luders:1990,olbermann:2007}).
As a consequence, the modes~$S_k$ and their derivatives have (almost) the same asymptotic behaviour as the modes~$W_k$ if~$\sigma_k$ looks asymptotically like~$\omega_k^{-1/2}$:

\begin{proposition}
  Suppose $\partial_\tau^l \sigma_k = \partial_\tau^l\big(\omega_k^{-1/2}\big) + \bigO(k^{-9/2})$ for all~$l = 0, 1, 2$ such that $G = \bigO(k^{-m})$ for some~$m \geq 3$.
  The modulus~$\rho_k$ of the modes~$S_k$ satisfies
  \begin{equation*}
    \rho_k = \sigma_k + \bigO\big(k^{-1/2-m}\big),
    \quad
    \partial_\tau^n \rho_k = \partial_\tau^n \sigma_k + \bigO\big(k^{3/2-m}\big)
  \end{equation*}
  for all~$n = 0, 1, 2$.
\end{proposition}
\begin{proof}
  First note that the assumed bounds of~$\sigma_k$ and~$\sigma_k''$ yield
  \begin{equation*}
    2 G = \sigma_k^{-2} - \omega_k^2 \sigma^2 - \sigma_k^{\phantom{'}} \sigma_k'' = \bigO\big(k^{-3}\big),
  \end{equation*}
  which is consistent with the assumption on~$G$.
  Further, recall that ${1-A}$, $A'$, $B$ and~$B'$ are~$\bigO(k^{-m})$, too.
  We then derive from~\eqref{eq:W_to_S} that
  \begin{equation*}
    \rho_k^2 = \abs{A + B\, \e^{- 2\im \psi_k}}^2 \sigma_k^2 = \sigma_k^2 + \bigO\big(k^{-1-m}\big)
  \end{equation*}
  and thus $\rho_k = \sigma_k + \bigO(k^{-1/2-n})$.
  Next we use that~$\rho_k$ satisfies the differential equation~\eqref{eq:rho_diff_eq} to find
  \begin{equation*}
    \rho_k'' - \sigma_k'' = \big( \rho_k^{-4} - \omega_k^2 \big) \rho_k - \big( \sigma_k^{-4} - \omega_k^2 \big) \sigma_k + 2 G \sigma_k^{-1} = \bigO\big(k^{3/2-m}\big),
  \end{equation*}
  whereby we obtain the estimate for~$\rho_k''$ and, after an integration in time, also that for~$\rho_k'$.
\end{proof}

In summary, the asymptotic behaviour of the initial values given by~$\omega_k$ fixes the asymptotic behaviour of the solutions~$\omega_k$.

We can now construct \emph{adiabatic states} as in~\cite{parker:1969,luders:1990} by specifying appropriate initial values for~\eqref{eq:mode_equation} respectively~\eqref{eq:rho_diff_eq}:
Making a WKB-like Ansatz, one finds the adiabatic modes of Parker~\cite{parker:1969}.
Namely, the adiabatic modes~$W_k^{\vphantom{\smash{\mbox{\tiny $(\!n\!)$}}}} = W_k^{\smash{\mbox{\tiny $(\!n\!)$}}}$ of order~$n$ are modes of the form~\eqref{eq:mode_form} with~$\sigma_k^{\vphantom{\smash{\mbox{\tiny $(\!n\!)$}}}} = \sigma_k^{\smash{\mbox{\tiny $(\!n\!)$}}}$, given iteratively via\footnote{The notation used here can be transformed into the usual one by setting $\sigma_k^{\smash{\mbox{\tiny $(\!n\!)$}}} = (\Omega_k^{\smash{\mbox{\tiny $(\!n\!)$}}})^{-1/2}$.}
\begin{equation*}
  (\sigma_k^{\smash{\mbox{\tiny $(\!n + 1\!)$}}})^{-4} \defn \omega_k^2 + \frac{\sigma_k^{\smash{\mbox{\tiny $(\!n\!)$}}}{}{''\!\!\!}}{\sigma_k^{\smash{\mbox{\tiny $(\!n\!)$}}}}
  \quad \text{with} \quad
  (\sigma_k^{\smash{\mbox{\tiny $(\!0\!)$}}})^{-4} \defn \omega_k^2.
\end{equation*}
The adiabatic modes~$W_k^{\smash{\mbox{\tiny $(\!n\!)$}}}$ are then used to specify initial values for the mode equation~\eqref{eq:mode_equation}, \eg, by solving the integral equations~\eqref{eq:mode_coefficients}.

A useful result on the asymptotic behaviour of the adiabatic modes is stated in \cite[Lem.~3.2]{luders:1990}.
Using the fact that~$\omega_k' = \bigO(k^{-1})$, one can easily improve this lemma to obtain for all~$n \in \NN_0$ and~$m \in \NN$
\begin{equation}\label{eq:adiabatic_asymptotics}\begin{aligned}
  \sigma_k^{\smash{\mbox{\tiny $(\!n\!)$}}} & = \bigO\big(k^{-1/2}\big), &
  \varepsilon_k^{\smash{\mbox{\tiny $(\!n\!)$}}} & = \bigO\big(k^{-2(n+1)}\big), \\
  \pd[m]{}{\tau}\, \sigma_k^{\smash{\mbox{\tiny $(\!n\!)$}}} & = \bigO\big(k^{-5/2}\big),\qquad &
  \pd[m]{}{\tau}\, \varepsilon_k^{\smash{\mbox{\tiny $(\!n\!)$}}} & = \bigO\big(k^{-2(n+1)}\big)
\end{aligned}\end{equation}
as~$k \to \infty$ and where~$(\sigma_k^{\smash{\mbox{\tiny $(\!n\!)$}}})^{-4} = (\sigma_k^{\smash{\mbox{\tiny $(\!n - 1\!)$}}})^{-4} (1 + \varepsilon_k^{\smash{\mbox{\tiny $(\!n\!)$}}})$.
The asymptotic behaviour of the coefficients~\eqref{eq:mode_coefficients} for adiabatic states was analyzed in~\cite{luders:1990,olbermann:2007}.
It can be found, using the improved bounds~\eqref{eq:adiabatic_asymptotics}, that they satisfy
\begin{equation*}\begin{aligned}
  1 - A(k, \tau) & = \bigO\big(k^{-2n-3}\big), &
  B(k, \tau) & = \bigO\big(k^{-2n-3}\big), \\
  A'(k, \tau) & = \bigO\big(k^{-2n-3}\big),\qquad &
  B'(k, \tau) & = \bigO\big(k^{-2n-3}\big)
\end{aligned}\end{equation*}
as $k \to \infty$.

These results can be seen as a starting point to show the relation between adiabatic states of a certain order and Hadamard states.
Indeed, one can show \cite{junker:2002} that adiabatic states of infinite order are Hadamard states and that adiabatic states of a finite order satisfy a microlocal spectrum condition on a certain \emph{Sobolev wavefront set}\footnote{Sobolev wavefront sets are very similar to the usual wavefront set, but instead of using smooth functions at the foundation of the definition, one uses functions from a Sobolev set of a certain order.} of the two-point distribution.

\subsection{An adiabatic state for conformal coupling}
\label{sub:adiabatic_zero}
\IdxMain{adiabatic-zero}

Let study the construction of the (adiabatic) states already considered in \cite{pinamonti:2011,pinamonti:2015}, see also \cite{anderson:1985,anderson:1986}.
For a (massive) conformally coupled scalar field ($\xi = 1/6$) the initial values for an adiabatic state of order zero can be taken to be
\begin{equation*}
  \chi_k^{\phantom{'}}(\tau_0) = \frac{1}{\sqrt{2 k_0}}\, \e^{\im k_0 \tau_0},
  \qquad
  \chi_k'(\tau_0) = \frac{\im k_0}{\sqrt{2 k_0}}\, \e^{\im k_0 \tau_0},
\end{equation*}
with $k_0^2 \defn (\Omega_k^{\smash{\mbox{\tiny $(\!0\!)$}}})^2 = k^2 + a(\tau_0)^2 m^2$.
Note that these initial values are essentially a conformal transformation of the modes of the Minkowski vacuum.

It is possible solve the mode equation~\eqref{eq:mode_equation} with these initial values perturbatively.
For this purpose, define the potential $V(\tau) = m^2 (a(\tau)^2 - a(\tau_0)^2)$ and make the recursive Ansatz $\chi_k(\tau) = \sum_n \chi_k^{\smash{\mbox{\tiny $(\!n\!)$}}}(\tau)$ with the recurrence relation
\begin{equation}\label{eq:adiabatic_recc}
  \chi_k^{\smash{\mbox{\tiny $(\!n\!)$}}}{}''(\tau) + k_0^2 \chi_k^{\smash{\mbox{\tiny $(\!n\!)$}}}(\tau) = - V(\tau) \chi_k^{\smash{\mbox{\tiny $(\!n-1\!)$}}}(\tau)
\end{equation}
with initial condition
\begin{equation}\label{eq:initial_adiabatic_zero}
  \chi_k^{\smash{\mbox{\tiny $(\!0\!)$}}}(\tau) \defn \frac{1}{\sqrt{2 k_0}}\, \e^{\im k_0 \tau}.
\end{equation}
The mode equation~\eqref{eq:mode_equation} is then solved as described in the proof of the following proposition:

\begin{proposition}\label{prop:adiabatic_zero}
  The recurrence relation~\eqref{eq:adiabatic_recc} is solved iteratively (for $\tau > \tau_0$) by
  \begin{equation}\label{eq:adiabatic_recc_int}
    \chi_k^{\smash{\mbox{\tiny $(\!n\!)$}}}(\tau) = \int_{\tau_0}^\tau \frac{\sin\big( k_0 (\eta - \tau) \big)}{k_0} V(\eta) \chi_k^{\smash{\mbox{\tiny $(\!n-1\!)$}}}(\eta)\, \dif\eta
  \end{equation}
  and the sum $\chi_k(\tau) = \sum_n \chi_k^{\smash{\mbox{\tiny $(\!n\!)$}}}(\tau)$ converges.
\end{proposition}
\begin{proof}
  Consider for each~$k$ the retarded propagator of $\partial_\tau^2 + k_0^2$ given by
  \begin{equation*}
    \Delta_{\mathrm{ret},k}(f)(\tau_0, \tau) = \int_{\tau_0}^\tau \frac{\sin\big( k_0 (\tau - \eta) \big)}{k_0} f(\eta)\, \dif\eta, \quad \tau > \tau_0,
  \end{equation*}
  for all~$f \in C^0(I)$, where $I$ is the domain of the conformal time.
  Applying $\Delta_{\mathrm{ret},k}$ to~\eqref{eq:adiabatic_recc}, it can be solved as
  \begin{equation*}
    \chi_k^{\smash{\mbox{\tiny $(\!n\!)$}}}(\tau) = \int_{\tau_0}^\tau \frac{\sin\big( k_0 (\eta - \tau) \big)}{k_0} V(\eta) \chi_k^{\smash{\mbox{\tiny $(\!n-1\!)$}}}(\eta)\, \dif\eta
  \end{equation*}
  for $\tau > \tau_0$.

  The straightforward estimates
  \begin{equation*}
    \abs{\chi_k^{\smash{\mbox{\tiny $(\!n\!)$}}}} \leq \frac{m^2}{k_0} \int_{\tau_0}^\tau \abs[\big]{ V(\eta) \chi_k^{\smash{\mbox{\tiny $(\!n-1\!)$}}}(\eta) }\, \dif\eta,
    \quad
    \abs{\chi_k^{\smash{\mbox{\tiny $(\!n\!)$}}}} \leq m^2 \int_{\tau_0}^\tau (\tau - \eta)\, \abs[\big]{ V(\eta) \chi_k^{\smash{\mbox{\tiny $(\!n-1\!)$}}}(\eta) }\, \dif\eta
  \end{equation*}
  can be iterated using the initial value $\chi_k^{\smash{\mbox{\tiny $(\!0\!)$}}} = (2 k_0)^{-1/2}$ and the standard `trick' of extending the integration of a symmetric function over a time-ordered domain to an integration over a symmetric domain by diving through the appropriate factorial (\cf\ \cite[Prop.~4.4]{pinamonti:2011}).
  Combining the two estimates, this gives
  \begin{equation}\label{eq:estimate_chi_n}
    \abs{\chi_k^{\smash{\mbox{\tiny $(\!n\!)$}}}} \leq \frac{1}{\sqrt{2 k_0}\, n!}\, \bigg( \frac{m^2}{k_0} \int_{\tau_0}^\tau \abs[\big]{V(\eta)}\, \dif\eta \bigg)^l \bigg( m^2 \int_{\tau_0}^\tau (\tau - \eta)\, \abs[\big]{V(\eta)}\, \dif\eta \bigg)^{n-l}
  \end{equation}
  for any $0 \leq l \leq n$.
  Therefore the sum $\chi_k(\tau) = \sum_n \chi_k^{\smash{\mbox{\tiny $(\!n\!)$}}}(\tau)$ converges absolutely.
\end{proof}

Equivalent results can be found in~\cite[Sect.~2.1]{pinamonti:2015}, \cite[Prop.~4.4]{pinamonti:2011} and also \cite{anderson:1985}.
It is clear, that the recurrence relation can be solved in the same for $\tau < \tau_0$ by applying the advanced propagator.

\begin{remark}
  The partial modes~$\chi_k^{\smash{\mbox{\tiny $(\!n\!)$}}}$ can be computed as in \cref{prop:adiabatic_zero} even if the metric (equivalently the scale factor) is not smooth.
  If the scale factor is $C^0$, the resulting mode~$\chi_k$ will be at least $C^2$.
  This relies crucially on the fact that the curvature, which is not well-defined if $a$ is not at least $C^2$, does not enter the mode equation~\eqref{eq:mode_equation}.
\end{remark}

\subsection{States of low energy}
\label{sub:low_energy}
\IdxRanBegin{low-energy}

Let us \emph{define} the \IdxMain{unreg-energy-mode}\emph{(unregularized) energy density per mode}~$S_k$ as
\begin{equation}\label{eq:energy_mode}\begin{split}
  \hat\rho(S_k, \conj{S}_k) & \defn \frac{1}{2 a^4} \Big( S_k' \conj{S}{}_k' + (6 \xi - 1) a H (S_k^{\phantom'} \conj{S}_k^{\phantom'})' \\&\quad + \big( k^2 + a^2 m^2 - (6 \xi - 1) a^2 H^2 \big) S_k^{\phantom'} \conj{S}_k^{\phantom'} \Big).
\end{split}\end{equation}
For now, we will not interpret this quantity in any way and leave its derivation to \cref{sec:semiclassical_friedmann}.

Given reference modes~$\chi_k$ and Bogoliubov coefficients $A, B$, the energy density per mode $S_k = A \chi_k + B \conj{\chi}_k$ is related to the energy density per reference mode~$\chi_k$ by
\begin{equation}\label{eq:energy_density_difference}
  \tfrac{1}{2} \big( \hat\rho(S_k, \conj{S}_k) - \hat\rho(\chi_k, \conj{\chi}_k) \big) = \abs{B}^2 \hat\rho(\chi_k, \conj{\chi}{}_k) + \Re\big(A \conj{B}\, \hat\rho(\chi_k, \chi_k) \big)
\end{equation}
One can now attempt to minimize the energy density per mode by varying the Bogoliubov coefficients and we find that:

\begin{proposition}\label{prop:instant_miminum}
  The energy density per mode at a fixed instance of time is minimal if and only if the Bogoliubov coefficients are given by (up to unitary equivalence)\footnote{Recall that, without loss of generality, we can always choose $B$ to be real and positive such that $A$ is completely determined by its phase.}
  \begin{subequations}\label{eq:minimal_bogoliubov}\begin{align}
    \Arg A(k) & = \uppi - \Arg\big( \hat\rho(\chi_k, \chi_k) \big), \\
    B(k) & = \left( \frac{\hat\rho(\chi_k, \conj{\chi}{}_k)}{2 \sqrt{\hat\rho(\chi_k, \conj{\chi}{}_k){}^2 - \abs{\hat\rho(\chi_k, \chi_k)}{}^2}} - \frac{1}{2} \right)^{1/2} \label{eq:minimal_bogoliubov2}
  \end{align}\end{subequations}
  and $\hat\rho(\chi_k, \conj{\chi}{}_k)^2 > \abs{\hat\rho(\chi_k, \chi_k)}{}^2$.
  The inequality is satisfied for all~$k > 0$ if and only if
  \begin{equation}\label{eq:minimal_conditions}
    \sqrt{1 + \frac{4 m^2}{H^2}} \pm 1 \geq \pm 12 \xi
    \qquad \text{or} \qquad
    H = 0,
  \end{equation}
  \ie, in particular whenever $0 \leq \xi \leq 1/6$.
\end{proposition}
\begin{proof}
  For fixed $B > 0$, \eqref{eq:energy_density_difference} is minimized by $\Arg A = \uppi - \Arg\big( \hat\rho(\chi_k, \chi_k) \big)$ so that the second summand is maximally negative.
  Therefore, minimizing \eqref{eq:energy_density_difference} is equivalent to finding the minima of
  \begin{equation*}
    B^2 \hat\rho(\chi_k, \conj{\chi}{}_k) + \Re\big(A B \hat\rho(\chi_k, \chi_k) \big) = B^2 \hat\rho(\chi_k, \conj{\chi}{}_k) - \sqrt{1 + B^2} B \abs{\hat\rho(\chi_k, \chi_k)}.
  \end{equation*}
  Differentiating this expression by~$B$, one finds that an extremum exists on the positive real axis only if $\hat\rho(\chi_k, \conj{\chi}{}_k)^2 > \abs{\hat\rho(\chi_k, \chi_k)}{}^2$ and that its locus is given by~\eqref{eq:minimal_bogoliubov2}; it is easy to see that this is indeed an minimum.

  Inserting the definition \eqref{eq:energy_mode} into the condition $\hat\rho(\chi_k, \conj{\chi}{}_k)^2 > \abs{\hat\rho(\chi_k, \chi_k)}{}^2$, we find that it is equivalent to
  \begin{equation*}
    k^2 + a^2 \big( m^2 + 6(1-6\xi) \xi H^2 \big) > 0.
  \end{equation*}
  If this conditions is to hold for all $k > 0$, then \eqref{eq:minimal_conditions} must be satisfied.
\end{proof}

Instead of trying to minimize the energy density per mode at an instant, states of low energy are constructed by minimizing the smeared energy density per mode.
That is, by minimizing
\begin{equation}\label{eq:smeared_energy_density_difference}\begin{split}
  & \frac{1}{2} \int_I f(\tau)^2 \big( \hat\rho(S_k, \conj{S}_k) - \hat\rho(\chi_k, \conj{\chi}{}_k) \big)\, \dif\tau \\&\qquad = \int_I f(\tau)^2 \Big( \abs{B}^2 \hat\rho(\chi_k, \conj{\chi}{}_k) + \Re\big(A \conj{B}\, \hat\rho(\chi_k, \chi_k) \big) \Big)\, \dif\tau
\end{split}\end{equation}
for a fixed smearing function $f \in C_0^\infty(I)$, where $I \subset \RR$ is the domain of the conformal time coordinate.

This minimization was performed for minimally coupled scalar fields in~\cite{olbermann:2007} to find the so-called \IdxMain{low-energy}\emph{states of low energy}.
It can be shown that the states of low energy satisfy the microlocal spectrum condition and thus they are Hadamard states.
The arguments presented in~\cite{olbermann:2007} can be straightforwardly repeated for the conformally coupled scalar field to find states of low energy, which are Hadamard states too.
In both cases the Bogoliubov coefficients are given as in~\eqref{eq:minimal_bogoliubov} with the replacements
\begin{equation*}
  \hat\rho(\chi_k, \conj{\chi}{}_k) \to \int_I f(\tau)^2 \hat\rho(\chi_k, \conj{\chi}{}_k)\, \dif\tau
  \quad \text{and} \quad
  \hat\rho(\chi_k, \chi_k) \to \int_I f(\tau)^2 \hat\rho(\chi_k, \chi_k)\, \dif\tau.
\end{equation*}

There are good reasons to believe the following:

\begin{conjecture}
  States of low energy for arbitrary smearing function, mass and scale factor exist only in the curvature coupling range $0 \leq \xi \leq 1/6$.
  For all such $\xi$ the state satisfies the microlocal spectrum condition.
\end{conjecture}

The crucial point in proving this conjecture is to show that \eqref{eq:smeared_energy_density_difference} has a minimum.
This can be shown similar to~\cite{olbermann:2007} for the minimally coupled case and the conformally coupled case.
For other values of~$\xi$ the proof is more difficult, but by continuity it is clear from \cref{prop:instant_miminum} that for some smearing functions states of low energy exist in the interval $\xi \in [0, 1/6]$ but not outside that range.
Namely, if $(f_n)$ is a sequence of functions such that $f_n^2$ converges weakly to the delta distribution, then there exists $N$ such that a state of low energy exists for all $f_m$, $m \geq N$ because a minimum exists for $f^2 = \delta$ by \cref{prop:instant_miminum}.
Once existence is show, one can expect that the state satisfies the microlocal spectrum condition using proofs analogous to those in~\cite{olbermann:2007}.

\IdxRanEnd{low-energy}
\IdxRanEnd{adiabatic-state}

\section{Holographic construction of Hadamard states}
\label{sec:holographic}

In the absence of a global timelike Killing field on a generic globally hyperbolic spacetime it is difficult to find physically well-motivated quantum states.
Therefore, in recent years, a lot of effort was put into the construction of proper Hadamard states on non-trivial spacetimes.
A promising method is the `holographic' construction of Hadamard states on characteristic surfaces introduced in \cite{dappiaggi:2006,moretti:2006a,moretti:2008}.
The holographic method has been applied to construct Hadamard states for the conformally coupled, massless scalar field \cite{dappiaggi:2006,moretti:2006a,moretti:2008}, the Weyl (massless Dirac) field \cite{hack:2010,dappiaggi:2009}, the vector potential \cite{siemssen:2011,dappiaggi:2013} and linearized gravity \cite{benini:2014c} on asymptotically flat spacetimes and cosmological backgrounds \cite{dappiaggi:2009b,dappiaggi:2009a}.
It was also used to construct local Hadamard states in lightcones in~\cite{dappiaggi:2011c}.

Forgetting for a moment the application of the bulk-to-boundary construction to asymptotically flat spacetimes and limiting ourselves to the scalar field, it may be roughly sketched as follows (see also \cite{gerard:2014b}).
Let $(M,g)$ be a globally manifold with a distinguished point~$p$ such that the future lightcone of $p$ satisfies some technical conditions.
The it is possible to construct on the lightcone (without the tip and as a manifold on its own) a positive bidistribution $\lambda$ on all functions on the lightcone that are compactly supported to the future and falls off sufficiently fast to the past, such that the antisymmetric part of $\lambda$ agrees with the pullback of the commutator distribution on the whole spacetime, and the wavefront set of $\lambda$ is of positive frequency with respect to the future-directed lightlike geodesics through~$p$.
This bidistribution has all the necessary properties to define a state for a quantum field theory on the lightcone.
Moreover, taking any compactly supported function in the interior of the lightcone of~$p$, it can be mapped to a function on the lightcone using the advanced propagator and a pullback such that the resulting function on the lightcone is compact towards the future and has good fall-off properties towards the past of the lightcone.
This way one obtains the so-called bulk-to-boundary (projection) map.
Pulling back all functions in the interior of the lightcone to the boundary of the lightcone using this map, one thus finds a state for the scalar field restricted to the interior of the lightcone.
Applying the propagation of singularities theorem it is possible to show that the resulting state satisfies the microlocal spectrum condition.

In a second step, on may construct Hadamard states for conformally invariant scalar field on asymptotically flat spacetimes with globally hyperbolic unphysical spacetimes.
First, one notices that boundary of the conformal completion of an asymptotically flat spacetime in the unphysical spacetimes satisfies all the necessary technical conditions.
Then one can compose the bulk-to-boundary map inside the unphysical spacetime with the (non-unique) conformal transformation associated with the asymptotically flat spacetime, to find a state for the conformally invariant scalar field.
Since conformal transformation leave lightlike geodesics invariant, also this state is of Hadamard form.

In \cite{siemssen:2011,dappiaggi:2013} this construction was generalized to the electromagnetic vector potential.
Also in this more complicated case a bulk-to-boundary construction of Hadamard states was be found, but it involves careful use of the gauge freedom of the vector potential to construct a \emph{positive} state.
Otherwise, the construction remains largely unchanged.

%% file: einstein.tex
%!TEX root = master.tex

\chapter{The semiclassical Einstein equation}
\label{cha:einstein}

\IdxRanBegin{semiclassical-einstein-eq}

\section{Introduction}

The equation
\begin{equation}\label{eq:semiclassical_einstein}
  G_{ab} + \Lambda g_{ab} = \omega(\norder{T_{ab}})
\end{equation}
is called the \IdxMain{semiclassical-einstein-eq}\emph{semiclassical Einstein equation}.\footnote{Remember that we chose units such that $8 \uppi \mathrm{G} = c = 1$.}
It is obtained from the ordinary Einstein equation by replacing the classical stress-energy tensor with the (normal ordered) expectation value of the stress-energy tensor of a quantum field in a suitable quantum state~$\omega$.
Many developments in quantum field theory on curved spacetimes were driven by problems related to the quantum stress-energy tensor.
See also the monographs \cite{birrell:1984,hack:2010,wald:1994,fulling:1989} for an overview of the subject.

The semiclassical Einstein equation is usually understood as an equation that describes physics midway between the classical regime covered by the Einstein equation~\eqref{eq:einstein} and a full-fledged, but still elusive, quantum gravity.
Namely, in the semiclassical Einstein equation one takes into account that the `matter' content of the universe is fundamentally of quantum nature as described by quantum field theory on curved spacetimes, whereas the background structure which is the spacetime is treated on a classical level and is not separately quantized.

On the right-hand side one usually considers only Hadamard states.
The reason for restricting to Hadamard states is that only for Hadamard states the higher moments
\begin{equation*}
  \omega(\norder{T_{ab}(x)} \norder{T_{ab}(x)})
  \quad\text{\etc}
\end{equation*}
can be defined.
This is due to the fact that the $n$-point distributions of a state are distributions and thus they cannot simply be multiplied (\cf\ \cref{sub:ma_pullback}).
Since the two-point distributions of Hadamard states satisfy the microlocal spectrum condition, their wavefront set is contained inside a convex cone in $\dot{T}^*(M \times M)$ and hence powers of the $n$-point distribution are well-defined distributions.
This will be discussed in more detail in \cref{cha:fluctuations}.

\section{The stress-energy tensor}
\IdxRanBegin{stress-energy}

While the left-hand side of the semiclassical Einstein equation remains unchanged with respect to the ordinary Einstein equation, the right-hand side changes quite dramatically.
Namely, the classical stress-energy tensor $T_{a b}$ is replaced by a the expectation value of a quantum observable $\norder{T_{a b}}$ in a certain state $\omega$.
For this expression to be mathematically consistent, we need to require the conservation of the quantum stress-energy tensor, \ie, $\nabla^a \norder{T_{a b}} = 0$.

\subsection{The stress-energy tensor of the Klein--Gordon field}
\label{sub:T_ab_scalar}

We do not aim at discussing the semiclassical Einstein equation in all possible generality.
Instead we restrict our discussion to the semiclassical Einstein equation sourced by a scalar field.
The classical \IdxMain{stress-energy-scalar}\emph{stress-energy tensor of a Klein--Gordon field} $\varphi$ with equation of motion
\begin{equation*}
  \mathrm{P} \varphi = (\Box + \xi R + m^2) \varphi = 0
\end{equation*}
may be written as \cite{hollands:2005}
\begin{equation}\label{eq:classical_T_ab}\begin{split}
  T_{a b} & \defn \tfrac{1}{2} \nabla_a \nabla_b \varphi^2 + \tfrac{1}{4} g_{a b} \Box \varphi^2 - \varphi \nabla_a \nabla_b \varphi + \tfrac{1}{2} g_{a b} g^{c d} \varphi \nabla_c \nabla_d \varphi \\&\quad + \xi (G_{a b} - \nabla_a \nabla_b - g_{a b} \Box) \varphi^2 - \tfrac{1}{2} g_{a b} m^2 \varphi^2.
\end{split}\end{equation}
It may be obtained by varying the classical action of the scalar field with respect to the metric \cite[App.~E]{wald:1984}.
The \IdxMain{stress-energy-quantum}\emph{quantum stress-energy tensor} is obtained from the classical expression \eqref{eq:classical_T_ab} by replacing products of classical fields by Wick products of quantum fields, \ie,
\begin{equation}\label{eq:quantum_T_ab}\begin{split}
  \norder{T_{a b}} & \defn \tfrac{1}{2} \nabla_a \nabla_b \norder{\what\varphi^2} + \tfrac{1}{4} g_{a b} \Box \norder{\what\varphi^2} - \norder{\what\varphi\, \nabla_a \nabla_b \what\varphi} + \tfrac{1}{2} g_{a b} g^{c d} \norder{\what\varphi\, \nabla_c \nabla_d \what\varphi} \\&\quad + \xi (G_{a b} - \nabla_a \nabla_b - g_{a b} \Box) \norder{\what\varphi^2} - \tfrac{1}{2} g_{a b} m^2 \norder{\what\varphi^2}.
\end{split}\end{equation}
This expression is not obviously conserved as
\begin{equation*}
  \nabla^a \norder{T_{a b}} = - \norder{(\nabla_b \what\varphi) (\mathrm{P} \what\varphi)}
\end{equation*}
is not necessarily vanishing even if the Wick square was a solution of the equations of motion.
Nevertheless, either by a judicious choice in the renormalization freedom of $\norder{\varphi^2}$ and $\norder{\varphi \nabla_a \nabla_b \varphi}$ \cite{hollands:2005} or, equivalently, by a redefinition of the quantum stress-energy tensor \cite{moretti:2003}, a conserved quantum stress-energy tensor can be found; here we follow the approach of Hollands and Wald.
While the renormalization freedom can be used to find a conserved stress-energy tensor, it is not possible to impose the equations of motions on a locally covariant normal ordering prescription \cite{hollands:2005}.

\subsection{Renormalization of the stress-energy tensor}
\label{sub:T_ab_renormalization}
\IdxMain{stress-energy-ren}

The \emph{renormalization freedom} of $\norder{\what\varphi^2}$ and $\norder{\what\varphi\, \nabla_a \nabla_b \what\varphi}$ is spanned by $m^2, R$ and
\begin{equation*}\begin{split}
  & g_{a b} m^4,\; g_{a b} m^2 R,\; m^2 R_{a b},\; \nabla_a \nabla_b R,\; g_{a b} \Box R,\; \Box R_{a b},\; g_{a b} R^2,\; \\& R R_{a b},\; R_{a c} {R^c}_b,\; g_{a b} R_{c d} R^{c d},\; R^{c d} R_{c a d b},\; g_{a b} R_{c d e f} R^{c d e f}.
\end{split}\end{equation*}
We have to split this renormalization freedom into two classes: (a) combinations of terms that are conserved and represent a true renormalization freedom, and (b) combinations of terms that are not conserved and need to be fixed to produce a conserved $\norder{T_{a b}}$.

Denote by $I_{a b}$ and $J_{a b}$ the two conserved curvature tensors of derivative order~$4$:
% NOTE: These are exactly the negative of the tensors computed in wald:1978.
\begin{align*}
  I_{a b} & \defn 2 R R_{a b} - 2 \nabla_a \nabla_b R - \tfrac{1}{2} g_{a b} \left( R^2 + 4 \Box R \right), \\
  J_{a b} & \defn
  2 R^{c d} R_{c a d b} - \nabla_a \nabla_b R - \Box R_{a b} - \tfrac{1}{2} g_{a b} \left( R_{c d} R^{c d} + \Box R \right).
\end{align*}
The following is often stated in the form of a conjecture (\eg, in \cite{wald:1978}):

\begin{proposition}
  $I_{a b}$ and $J_{a b}$ span the whole space of conserved fourth order local curvature tensors.
\end{proposition}
\begin{proof}
  It is an easy task to confirm this statement by a direct computation along the lines of \cite{davies:1977}:
  Taking the linear span of all fourth order curvature tensors
  \begin{equation*}\begin{split}
    & \nabla_a \nabla_b R,\; g_{a b} \Box R,\; \Box R_{a b},\; g_{a b} R^2,\; R R_{a b},\; R_{a c} {R^c}_b, \\
    & g_{a b} R_{c d} R^{c d},\; R^{c d} R_{c a d b},\; g_{a b} R_{c d e f} R^{c d e f},
  \end{split}\end{equation*}
  one can show that any covariantly conserved combination $C_{a b}$ must be of the form
  \begin{equation*}\begin{split}
    C_{a b} & = \alpha_1 \nabla_{a} \nabla_{b} R - \alpha_2 g_{a b} \Box R + 2 (\alpha_1 + \alpha_2) \Box R_{a b} - \tfrac{1}{4} (\alpha_1 + 2 \alpha_2) g_{a b} R^2 \\&\quad + (\alpha_1 + 2 \alpha_2) R R_{a b} + (\alpha_1 + \alpha_2) g_{a b} R_{c d} R^{c d} - 4 (\alpha_1 + \alpha_2) R^{c d} R_{c a d b},
  \end{split}\end{equation*}
  \ie, one obtains (for general metrics) a two-dimensional solution space.
  For $\alpha_1 = -2, \alpha_2 = 2$ and $\alpha_1 = -1, \alpha_2 = 1/2$ we recover the tensors $I_{a b}$ and $J_{a b}$, respectively.
\end{proof}

\begin{remark}
  In conformally flat spacetimes (\eg, FLRW spacetimes) the Weyl tensor vanishes and thus the solution space reduces to one dimension as the two tensors become proportional: $I_{a b} = 3 J_{a b}$.
  On the level of traces this proportionality holds for all metrics, namely, ${I^a}_a = 3 {J^a}_a = - 6 \Box R$.
\end{remark}

We therefore find that the \emph{conserved} renormalization freedom of $\norder{T_{a b}}$ is spanned by $m^4 g_{a b}$, $m^2 G_{a b}$, $I_{a b}$, $J_{a b}$.
The remaining terms renormalization parameters are fixed by the requirement of $\norder{T_{ab}}$ to be conserved.

\subsection{Point-splitting regularization of the stress-energy tensor}
\label{sub:T_ab_pointsplit}

Up to the renormalization freedom, a normal ordering prescription for the stress-energy tensor is given by the \IdxMain{point-split-reg}\emph{Hadamard point-splitting method}.
Given two linear (possibly tensorial) differential operators $\mathcal{D}_1, \mathcal{D}_2$, the Hadamard point-splitting method yields the expectation value of $\norder{(\mathcal{D}_1 \what\varphi)(\mathcal{D}_2 \what\varphi)}$ by seperating points, regularizing and then taking the coincidence limit, that is
\begin{equation*}
  \omega\big(\norder{(\mathcal{D}_1 \what\varphi)(\mathcal{D}_2 \what\varphi)}\big) = \lim_{x' \to x} \mathcal{D}^{\vphantom{\prime}}_1 \mathcal{D}'_2 \big( \omega_2(x, x') - \Hadamard(x, x') \big) = [\mathcal{D}^{\vphantom{\prime}}_1 \mathcal{D}'_2 w],
\end{equation*}
where $\mathcal{D}'_2$ acts on $x'$ and is (implicitly) parallel transported during the limit $x' \to x$.

In the Hadamard point-splitting approach the stress-energy tensor in a state $\omega$ of sufficient regularity is thus calculated as
\begin{equation*}\begin{split}
  \omega(\norder{T_{a b}}) & = \frac{1}{8 \uppi^2} \big( \mathcal{T}_{a b} [w] + \mathcal{T}\indices{_{ab}^{\!cd}} [\nabla_c \nabla_d w] \big) + \frac{1}{4 \uppi^2} [v_1] g_{a b} \\&\quad + c_1 m^4 g_{a b} + c_2 m^2 G_{a b} + c_3 I_{a b} + c_4 J_{a b},
\end{split}\end{equation*}
where $\mathcal{T}_{a b}$ and $\mathcal{T}\indices{_{ab}^{\!cd}}$ are the differential operators acting, respectively, on $\norder{\what\varphi^2}$ and $\norder{\what\varphi\, \nabla_a \nabla_b \what\varphi}$ in~\eqref{eq:quantum_T_ab}:
\begin{align*}
  \mathcal{T}_{ab} & \defn \tfrac{1}{2} \nabla_a \nabla_b + \tfrac{1}{4} g_{a b} \Box\!{} + \xi (G_{a b} - \nabla_a \nabla_b - g_{a b} \Box) - \tfrac{1}{2} g_{a b} m^2, \\
  \mathcal{T}\indices{_{ab}^{\!cd}} & \defn - \delta_a^c \delta_b^d + \tfrac{1}{2} g_{a b} g^{c d}.
\end{align*}
Furthermore, $c_i$ are dimensionless (renormalization) constants fixed once for all spacetimes\footnote{$c_1, c_2$ are due to the renormalization of $\norder{\what\varphi^2}$ and $c_3, c_4$ correspond to the renormalization freedom of $\norder{\what\varphi\,\nabla_a\nabla_b\what\varphi}$} and the addition of the Hadamard coefficient $[v_1]$ (see \eqref{eq:hadamard_v1} for an explicit expression) makes the quantum stress-energy tensor conserved because\Idx{hadamard-coeff}
\begin{equation*}
  \lim_{x' \to x} \nabla'_a \mathrm{P}\, \Hadamard(x, x') = -\frac{1}{4 \uppi^2} \nabla_a [v_1].
\end{equation*}
Observe that $c_1 m^4 g_{a b}$ can be interpreted as a renormalization of the cosmological constant and $c_2 m^2 G_{a b}$ corresponds to a renormalization of Newton's gravitational constant $\mathrm{G}$; the remaining two terms have no classical interpretation.

\subsection{Trace of the stress-energy tensor}
\label{sub:T_ab_trace}
\IdxRanBegin{stress-energy-trace}

Because of its simple form, a first step towards analyzing the stress-energy tensor of a scalar field is often the study its trace, which is given by
\begin{equation}\label{eq:quantum_T_ab_trace}
  \norder{T} \defn g^{a b} \norder{T_{a b}} = -m^2 \norder{\what\varphi^2} + 3 \big( \tfrac{1}{6} - \xi \big) \Box \norder{\what\varphi^2} - \norder{\what\varphi\, \mathrm{P} \what\varphi}.
\end{equation}
It follows that the trace of the stress-energy tensor is calculated via point-splitting as
\begin{equation}\label{eq:T_ab_trace_pointsplit}\begin{split}
  \omega(\norder{T}) & = - \Big( m^2 - 3 \big( \tfrac{1}{6} - \xi \big) \Box\! \Big) \frac{1}{8 \uppi^2} [w] + \frac{1}{4 \uppi^2} [v_1] \\&\quad + 4 c_1 m^4 - c_2 m^2 R - (6 c_3 + 2 c_4) \Box R,
\end{split}\end{equation}
where $c_i$ are the same constants as above and we used\Idx{hadamard-coeff}
\begin{equation*}
  \lim_{x' \to x} \mathrm{P}\, \Hadamard(x, x') = -\frac{3}{4 \uppi^2} [v_1].
\end{equation*}

\Cref{eq:quantum_T_ab_trace,eq:T_ab_trace_pointsplit} clearly show the so-called \IdxMain{trace-anom}\emph{trace anomaly} \cite{wald:1978}.
Namely, because the normally ordered quantum field does not satisfy the equations of motion, the massless, conformally coupled scalar field ($m=0$ and $\xi=1/6$) has non-vanishing trace of the stress-energy tensor although it is conformally invariant.
It is not possible to remove the trace anomaly by a judicious choice of the renormalization constants because $[v_1]$ is not a polynomial of~$m^4$, $m^2 R$ and~$\Box R$.
The trace anomaly is a distinct feature of the quantum theory and does not appear in a classical theory because the classical fields are solutions of the equation of motion.

\IdxRanEnd{stress-energy-trace}
\IdxRanEnd{stress-energy}

\section{The semiclassical Friedmann equations}
\label{sec:semiclassical_friedmann}
\IdxRanBegin{semiclassical-friedmann}

On FLRW spacetimes $(M, g)$ the classical Einstein equation~\eqref{eq:einstein} simplifies significantly to the first and second Friedmann equation~\eqref{eq:1st_friedmann} and~\eqref{eq:2nd_friedmann}.\Idx{fst-friedmann}\Idx{snd-friedmann}
Since the left-hand side remains unchanged when crossing over to the semiclassical Einstein equation, also the semiclassical equations must simplify in an analogue way for every state that satisfies the equation.
Whence one obtains the semiclassical Friedmann equations
% \begin{subequations}\label{eq:semiclassical_friedmann}\begin{align}
%   H^2 + \tfrac{1}{6} \overline{R} & = \tfrac{1}{3} \omega(\norder{\rho}) + \tfrac{1}{3} \Lambda, \\
%   \dot H + H^2 & = - \tfrac{1}{6} \omega(\norder{\rho} + 3 \norder{p}) + \tfrac{1}{3} \Lambda = - \tfrac{1}{6} \omega(\norder{T} + 2 \norder{\rho}) + \tfrac{1}{3} \Lambda,
% \end{align}\end{subequations}
\begin{subequations}\label{eq:semiclassical_friedmann}\begin{align}
  6 H^2 + \overline{R} & = 2 \omega(\norder{\rho}) + 2 \Lambda, \\
  6 (\dot H + H^2) & = - \omega(\norder{\rho} + 3 \norder{p}) + 2 \Lambda = - \omega(\norder{T} + 2 \norder{\rho}) + 2 \Lambda,
\end{align}\end{subequations}
where the \IdxMain{energy-den-quantum}\emph{quantum energy-density}~$\norder{\rho}$ and the \IdxMain{pressure-quantum}\emph{quantum pressure}~$\norder{p}$ are constructed out of the quantum stress-energy tensor just like their classical analogues are obtained from the classical stress-energy tensor.
Henceforth we will restrict again to flat FLRW spacetimes but similar statements can also be made in the case of elliptic and hyperbolic spatial sections.

States that satisfy the semiclassical Einstein equation need to respect the symmetries of the spacetime.
Therefore, any candidate state for a solution of the semiclassical Friedmann equations must be homogeneous and isotropic.
That is, under reasonable assumptions, it must be a state of the form discussed in \cref{sub:homogeneous_isotropic_state}.
Important examples of homogeneous and isotropic states are the adiabatic states (\cref{sub:adiabatic}) and the states of low energy (\cref{sub:low_energy}).

\subsection{Semiclassical Friedmann equations for the scalar field}

For the scalar field, the energy density and pressure are obtained from~\eqref{eq:quantum_T_ab} and they read
\begin{align}
  \label{eq:energy_density_wick}
  \norder{\rho} & = \Big( \big( \tfrac{1}{2} - \xi \big) \partial_t^2 - \big( \tfrac{1}{4} - \xi \big) \Box\!{} + 3 \xi H^2 + \tfrac{1}{2} m^2 \Big) \norder{\what\varphi^2} - \norder{\what\varphi\, \big( \partial_t^2 - \tfrac{1}{2} \Box\! \big) \what\varphi}, \\
  \notag
  3 \norder{p} & = \Big( \big( \tfrac{1}{2} - \xi \big) \partial_t^2 + \big( \tfrac{1}{4} - 2 \xi \big) \Box\!{} - \xi (6 \dot{H} + 9 H^2) - \tfrac{3}{2} m^2 ) \norder{\what\varphi^2} - \norder{\what\varphi\, \big( \partial_t^2 + \tfrac{1}{2} \Box\! \big) \what\varphi}
\end{align}
with respect to cosmological time~$t$.
The expressions for conformal time~$\tau$ are given by the replacement $\partial_t \mapsto a^{-1}\partial_\tau$.
A short calculation shows that the difference $3 \norder{p} - \norder{\rho}$ agrees with~\eqref{eq:quantum_T_ab_trace}.

The expectation values of~$\norder{\rho}$ and~$\norder{p}$ in a state~$\omega$ can again be calculated via Hadamard point-splitting.
For the energy density this approach yields:
\begin{equation}\label{eq:energy_density_pointsplit}\begin{split}
  \omega(\norder{\rho}) & = \mathcal{P} [w] - \lim_{x' \to x} \big[\big( \partial_t^2 - \tfrac{1}{2} \Box\! \big) w \big] - \frac{1}{4 \uppi^2} [v_1] \\&\quad - c_1 m^4 + 3 c_2 m^2 H^2 - 6 (3 c_3 + c_4) (2 H \ddot{H} - \dot{H}^2 + 6 H^2 \dot{H}),
\end{split}\end{equation}
where we have used the differential operator
\begin{equation*}
  \mathcal{P} \defn \big( \tfrac{1}{2} - \xi \big) \partial_t^2 - \big( \tfrac{1}{4} - \xi \big) \Box\!{} + 3 \xi H^2 + \tfrac{1}{2} m^2
\end{equation*}
and the renormalization constants~$c_i$ are again the same as in \cref{sub:T_ab_pointsplit,sub:T_ab_trace}.
The coincidence limit of the Hadamard coefficient $v_1$ on FLRW spacetimes can be obtain from \eqref{eq:hadamard_v1}:\Idx{hadamard-coeff}
\begin{equation}\label{eq:hadamard_v1_flrw}\begin{split}
  2 [v_1] & = \tfrac{1}{4} m^4 - 3 \big( \tfrac{1}{6} - \xi \big) \big( \dot{H} + 2 H^2 \big) m^2 + 9 \big( \tfrac{1}{6} - \xi \big)^2 \big( \dot{H}^2 + 4 \dot{H} H^2 + 4 H^4 \big) \\&\quad - \tfrac{1}{30} \big( \dot{H} H^2 + H^4 \big) + \tfrac{1}{12} \big( \tfrac{1}{5} - \xi \big) \Box R.
\end{split}\end{equation}
The point-split expression for $\omega(\norder{p})$ will not concern us here and is left as an exercise to the reader.

The next step is to replace the Hadamard point-splitting prescription with something a that is slightly more useful under the given circumstances.

\subsection{Adiabatic regularization}
\label{sub:adiabatic_regularization}
\IdxRanBegin{adiabatic-reg}
\IdxRanBegin{point-split-reg}

An effective means of regularizing in momentum space is provided by the \IdxMain{adiabatic-reg}\emph{adiabatic regularization} \cite{parker:1974,bunch:1980}, which is essentially equivalent to the Hadamard point-splitting regularization discussed above.
In the adiabatic regularization prescription one subtracts from a homogeneous and isotropic state $\omega$ the bidistribution specified by the adiabatic modes $W_k^{\smash{\mbox{\tiny $(\!n\!)$}}}$ of sufficiently high order~$n$.
The usefulness of this prescription lies in the fact that the bidistributions constructed out of the adiabatic modes satisfy the microlocal spectrum condition in the Sobolev sense up to a certain order.
Consequently, the two regularization prescriptions can only differ by local curvature tensors because both the Hadamard parametrix and the adiabatic modes are constructed from the local geometric structure of the spacetime.

It is therefore necessary to find the difference between point-splitting and adiabatic regularization
\begin{equation}\label{eq:poinsplit_adiabatic_difference}
  \lim_{x' \to x} \mathcal{D}\, \bigg( \Hadamard_n(x, x') - \frac{1}{(2 \pi)^3 a(\tau) a(\tau')} \!\int_{\RR^3}\! \conj{W}{}_{\!k}^{\smash{\mbox{\tiny $(\!m\!)$}}}(\tau) W{}_{\!k}^{\smash{\mbox{\tiny $(\!m\!)$}}}(\tau')\, \e^{\im \vec{k} \cdot (\vec{x} - \vec{x}')}\, \dif\vec{k} \bigg)
\end{equation}
for various (bi)differential operators~$\mathcal{D}$ on $C^\infty(M \times M)$ and the minimal orders~$n,m$ depending on the order of~$\mathcal{D}$.
Note that we omitted the necessary $\varepsilon$-regularization in the integrand.

It is helpful to note that the (truncated) Hadamard parametrix on FLRW is spatially isotropic and homogeneous and therefore $\Hadamard(\tau, \vec{x}; \tau', \vec{x}') = \Hadamard(\tau, \tau', \abs{\vec{x} - \vec{x}'})$.
This fact can simplify some computations because the coincidence limit can be taken in two steps: first one takes the limit onto the equal time surface $\tau = \tau'$ and then the spatial coincidence limit $\vec{x}' \to \vec{x}$.
Making efficient use of the equation of motion, this has the advantage that we can replace any higher than first order time derivative in~$\mathcal{D}$ by a spatial derivative and one needs to calculate the temporal coincidence limit only with the differential operators $\partial_\tau$, $\partial'_\tau$ and $\partial^{\vphantom{\prime}}_\tau \partial'_\tau$.
A proof of this statement can be found in~\cite[Chap.~5]{schlemmer:2010}.

The computation~\eqref{eq:poinsplit_adiabatic_difference} can be done in a general and efficent manner with a computer algebra system by combining the method of Avramidi to calculate Hadamard coefficients (\cref{sub:bi_avramidi}), the coordinate expansion of the world function (\cref{sub:bi_coordinate_synge}) and analytic Fourier transformation.
See also \cite{eltzner:2011} for a related approach or \cite{schlemmer:2010} for a different method that makes more efficient use of the symmetries of FLRW spacetimes.

The difference between point-splitting and adiabatic regularization for the Wick square can be calculated in this way as
\begin{equation}\label{eq:adiabatic_wick_square}\begin{split}
  \MoveEqLeft \lim_{x' \to x} \left( \Hadamard_0(x, x') - \frac{1}{(2 \uppi)^3 a^2} \int_{\RR^3} \frac{1}{2 \omega_k}\, \e^{\im \vec{k} \cdot (\vec{x} - \vec{x}')}\, \dif\vec{k} \right) \\
  & = \frac{B}{16 \uppi^2} \Big( 1 - 2 \gamma - \ln\big( \tfrac{1}{2} \lambda^2 B \big) \Big) + \frac{R}{288 \uppi^2},
\end{split}\end{equation}
where we used the potential $B = m^2 + (\xi - \tfrac{1}{6}) R$ and $\gamma$ denotes Euler's constant.

The exact form of the mode subtraction performed in the adiabatic regularization is inessential as long as it has the right $\vec{k}$-asymptotics (\cf\ \cite[Chap.~5]{schlemmer:2010}).
For example, instead of subtracting the adiabatic modes of order zero in~\eqref{eq:adiabatic_wick_square}, one can perform the following subtraction (\cf\ \cite{pinamonti:2015,pinamonti:2011} and \cref{sub:wick_square}):
\begin{equation}\label{eq:adiabatic_wick_square_2}\begin{split}
  \MoveEqLeft \lim_{x' \to x} \left( \Hadamard_0(x, x') - \frac{1}{(2 \uppi)^3 a^2} \int_{\RR^3} \bigg( \frac{1}{2 \omega_k(\eta)} - \frac{\omega_k^2 - \omega_k(\eta)^2}{4 \omega_k(\eta)^3}  \bigg)\, \e^{\im \vec{k} \cdot (\vec{x} - \vec{x}')}\, \dif\vec{k} \right) \\
  & = \frac{B(\eta) a(\eta)^2}{16 \uppi^2 a^2} - \frac{B}{16 \uppi^2} \bigg( 2 \gamma + \ln \frac{\lambda^2 a(\eta)^2 B(\eta)}{2 a^2} \bigg) + \frac{R}{288 \uppi^2},
\end{split}\end{equation}
where $\eta$ is an arbitrary instant in conformal time.
But this realization is much more important once one attempts to calculate the adiabatic subtraction with several derivatives such as the energy density, where higher order adiabatic modes need to be subtracted.
Here one generically finds that the adiabatic regularization involves terms with higher derivatives of the metric that do not become singular in the coincidence limit, \ie, terms which do not need to be subtracted to achieve the regularization.

To find an adiabatic subtraction for other components of the stress-energy tensor one can follow the approach of Bunch~\cite{bunch:1980} and take the second order adiabatic modes but discard all (non-singular) terms involving higher than fourth order time-derivatives of the metric.
Comparing this subtraction scheme with the Hadamard point-splitting regularization for the energy density, we find (see also~\cite{hack:2013})
\begin{align}
  \notag\MoveEqLeft \lim_{x' \to x} \Bigg( \mathop{\mathcal{D}_{\rho}} a(\tau) a(\tau') \Hadamard_1(x, x') - \frac{1}{(2 \uppi)^3 a^2} \int_{\RR^3} \e^{\im \vec{k} \cdot (\vec{x} - \vec{x}')} \bigg( \frac{k}{2 a^2} + \frac{m^2 - 6 (\xi - \tfrac{1}{6}) H^2}{4 k} \\\notag&\quad - \frac{m^4 a^2 + 12 (\xi - \tfrac{1}{6}) m^2 a^2 H^2 + 36 (\xi - \tfrac{1}{6})^2 (6 H^2 \dot{H} - \dot{H}^2 + 2 H \ddot{H})}{16 k^3} \bigg)\, \dif\vec{k} \Bigg) \\\notag
  & = - \frac{1}{4 \uppi^2} [v_1] - \frac{H^4}{960 \uppi^2} + \Big(2 - 2 \gamma + \ln \frac{a^2}{2 \lambda^2} \Big) \frac{m^4}{64 \uppi^2} \\\notag&\quad + \bigg( 1 + 18 \big( \xi - \tfrac{1}{6} \big) \Big( 2 - 2\gamma + \ln \frac{a^2}{2 \lambda^2} \Big) \bigg) \frac{m^2 H^2}{96 \uppi^2} + \frac{3 (\xi - \tfrac{1}{6})^2 H^2 R}{8 \uppi^2} \\\label{eq:regularized_energy}&\quad + \bigg( \frac{1}{17280 \uppi^2} - \frac{\xi - \tfrac{1}{6}}{288 \uppi^2} - \frac{(\xi - \tfrac{1}{6})^2}{32 \uppi^2} \Big( 2 - 2 \gamma + \ln \frac{a^2}{2 \lambda^2} \Big) \bigg) {I_0}^0,
\end{align}
where the (bi)differential operator
\begin{equation*}
  2 a^4 \mathcal{D}_{\rho} = a^2 \big( m^2 + (1 - 6 \xi) H^2 \big) \boxtimes \id{} + 2 (6 \xi - 1) a H \partial_\tau \boxtimes \id{} + \partial_\tau \boxtimes \partial_\tau + \sum_{i = 1}^3 \partial_{x^i} \boxtimes \partial_{x^i}
\end{equation*}
can be derived from~\eqref{eq:energy_density_wick}.
In the case of conformal coupling it is in fact sufficient to work with adiabatic modes of order zero for this computation and many of the terms in the above formula drop out.

Observe that both the point-splitting regularization with the truncated Hadamard parametrix and the adiabatic regularization do not depend on arbitrarily high derivatives in the metric.
Consequently, it is possible to perform both regularization schemes in non-smooth spacetimes.
This will be essential in \cref{cha:einstein_solutions}.

\IdxRanEnd{adiabatic-reg}
\IdxRanEnd{point-split-reg}
\IdxRanEnd{semiclassical-friedmann}
\IdxRanEnd{semiclassical-einstein-eq}

%% file: solutions.tex
%!TEX root = master.tex

\chapter[Solutions of the semiclassical Einstein equation][Solutions of the semiclassical Einstein equation]{Solutions of the semiclassical Einstein eq.}
\label{cha:einstein_solutions}

\IdxRanBegin{semiclassical-einstein-sol}

\section*{Introduction}

If one wants to attach any physical meaning to the semiclassical Einstein equation~\eqref{eq:semiclassical_einstein}, it is necessary that solutions of this equation exist and that it possesses a well-posed initial value problem.
It is not difficult to show that solutions do indeed exist at least in two very special scenarios: Minkowski spacetime and de Sitter spacetime.
In both cases the semiclassical Einstein equation (or, alternatively, the semiclassical Friedmann equations) can be solved for a specific choice of the renormalization constants.

Solutions of the semiclassical Friedmann equations were investigated already numerically by Anderson in a series of four articles beginning with the massless conformally coupled scalar field \cite{anderson:1983,anderson:1984} and later also considering the massive field \cite{anderson:1985,anderson:1986}.
Anderson discovered a complex landscape of solutions depending on the choice of renormalization constants and studied in particular solutions which show an asymptotically classical behaviour at late times.

More recently, Pinamonti discussed the local existence of solutions to the semiclassical Friedmann equations in so-called null Big Bang (NBB) spacetimes~\cite{pinamonti:2011}, where initial values are specified on the initial lightlike singularity.

In~\cite{pinamonti:2015} the author and Pinamonti proved for the first time the existence of global solutions to the semiclassical Einstein equation coupled to a massive, conformally coupled scalar field in `non-trivial' spacetimes.
More precisely, it was shown that the semiclassical Friedmann equations can solved simultaneously for the spacetime metric (\ie, the scale factor or the Hubble function) and a quantum state from initial values given at some Cauchy surface.
This was achieved by showing existence and uniqueness of local solutions for given initial values and subsequently extending local solutions to a maximal solution that cannot be extended any further because it exists either eternally or reaches a singularity.
In this chapter, a slightly updated version of the results of~\cite{pinamonti:2015} will be presented and complemented with recent numerical results.

In more generality, solving the semiclassical Einstein equations for a given quantum field means the following:

\textit{Given initial values for a spacetime metric and a quantum state prescribed on a three-dimensional Riemannian manifold $(\Sigma, h)$, do there exist a globally hyperbolic manifold $(M, g)$ of which $(\Sigma, h)$ is a Cauchy surface and a state~$\omega$ (preferably a Hadamard state) such that the semiclassical Einstein equation~\eqref{eq:semiclassical_einstein} is fulfilled?}

In a concrete case this problem can be tackled by selecting a class of globally hyperbolic spacetimes that are foliated by the same topological Cauchy surface and a functional that associates to each spacetime in this class a unique states.
For this approach to succeed, it is expected that the mentioned functional must satisfy some minimal regularity conditions with respect to the metric, \eg, continuous differentiability.

\section{Preliminaries}

\subsection{The traced stress-energy tensor}
\Idx{stress-energy-trace}

Recall from \cref{sub:T_ab_trace} that the expectation value of the traced stress-energy tensor for a massive, conformally coupled scalar field reads
\begin{align}
  \label{eq:T_ab_trace_conformal}
  \omega(\norder{T}) & = - m^2 \omega(\norder{\what\varphi^2}) + \frac{1}{4 \uppi^2} [v_1] - (6 c_3 + 2 c_4) \Box R \\
  \notag
  & = - \frac{m^2}{8 \uppi^2} [w] + \frac{1}{4 \uppi^2} [v_1] + 4 c_1 m^4 - c_2 m^2 R - (6 c_3 + 2 c_4) \Box R,
\end{align}
where the Hadamard coefficient~$[v_1]$ is obtained from~\eqref{eq:hadamard_v1_flrw} with $\xi = 1/6$ as
\begin{equation*}
  [v_1] = -\frac{H^2}{60} (\dot H + H^2) - \frac{1}{720} \Box R + \frac{m^4}{8}.
\end{equation*}
Working with the traced stress-energy tensor~$\norder{T}$ simplifies calculations considerably compared to the energy density~$\norder{\rho}$ given by~\eqref{eq:energy_density_pointsplit}.

In order to find solutions of the semiclassical Friedmann equation with the methods discussed here, it is necessary to fix the renormalization constants according to the following rules:

We will choose $c_3, c_4$ in such a way as to cancel higher order derivatives of the metric coming from $[v_1]$.
Following \cite{wald:1977} and \cite[Chap.~4.6]{wald:1994}, this is necessary because we want to have a well-posed initial value problem for a second-order differential equation.
Removing the $\Box R$ term might not be suitable for describing the physics close to the initial Big Bang singularity.
In the Starobinsky model this term is the single term which is considered to drive a phase of rapid expansion close to the Big Bang, see the original paper of Starobinsky~\cite{starobinsky:1980}, its further development \cite{kofman:1985} and also \cite{hack:2010,hack:2013} for a recent analysis.
However, this is surely suitable to describe the physics in the regime where $\Box H \ll H^4$.

Furthermore, remember that changing $c_1$ corresponds to a renormalization of the cosmological constant $\Lambda$, whereas a change of $c_2$ corresponds to a renormalization of the Newton constant $\mathrm{G}$ (\cf\ \cref{sub:T_ab_pointsplit}).
For this reason we choose $c_1$ in such a way that no contribution proportional to $m^4$ is present in $\omega(\norder{T})$ and we set $c_2$ in order to cancel the terms proportional to $m^2 R$ in $\omega(\norder{T})$.
All in all, we fix the renormalization constants as
\begin{equation*}
  4 c_1 = - \frac{1}{32 \uppi^2},
  \quad
  c_2 = \frac{1}{288 \uppi^2}
  \quad \text{and} \quad
  6 c_3 + 2 c_4 = - \frac{1}{2880 \uppi^2}.
\end{equation*}

\subsection{The semiclassical Friedmann equations}
\Idx{semiclassical-friedmann}

We can rewrite the semiclassical Friedmann equations to make use of the simplicity of the traced stress-energy tensor for the conformally coupled scalar field:
Adding the equations~\eqref{eq:semiclassical_friedmann} (for flat FLRW spacetimes) yields
\begin{equation}\label{eq:semiclassical_friedmann_trace}
  - 6 (\dot H + 2 H^2) = \omega(\norder{T}) - 4 \Lambda.
\end{equation}
Since $\norder{T} = 3\norder{p} - \norder{\rho}$, this equation is equivalent to~\eqref{eq:semiclassical_friedmann} if we also specify an initial value $\rho_0 \defn \omega(\norder{\rho})(\tau_0)$ for the expectation value of the energy density at a time $\tau_0$:
\begin{equation}\label{eq:constraint-einstein}
  3 H_0^2 \defn 3 H(\tau_0)^2 = \rho_0 + \Lambda.
\end{equation}
We call \eqref{eq:constraint-einstein} the \emph{constraint equation} because it relates the initial value $H_0$ for the spacetime geometry with the initial value $\rho_0$ for its matter content; these values cannot be fixed independently.

Inserting \eqref{eq:T_ab_trace_conformal} into~\eqref{eq:semiclassical_friedmann_trace} and solving for $\dot H$ we thus find
\begin{equation}\label{eq:friedmann_differential}
  \dot{H} = \frac{1}{H^2_c - H^2} \big( H^4 - 2 H_c^2 H^2 - \tfrac{15}{2} m^4 + 240 \uppi^2 (m^2 \omega(\norder{\what\varphi^2}) + 4 \Lambda) \big),
\end{equation}
which can integrated in conformal time to give
\begin{equation}\label{eq:friedmann_integral}\begin{split}
  H(\tau) = H_0 + \int_{\tau_0}^\tau & \frac{a(\eta)}{H^2_c - H(\eta)^2} \big( H(\eta)^4 - 2 H_c^2 H(\eta)^2 - \tfrac{15}{2} m^4 \\& + 240 \uppi^2 ( m^2 \omega(\norder{\what\varphi^2})(\eta) + 4 \Lambda) \big)\, \dif\eta,
\end{split}\end{equation}
where $H_c^2 \defn 1440 \uppi^2 / (8 \uppi \mathrm{G}) = 180 \uppi / \mathrm{G}$.
This integral equation will be our principal tool to solve the semiclassical Einstein equation.

% and with the initial condition~$H_0$ fixed (up to a sign) by the constraint equation~\eqref{eq:constraint-einstein}.

\subsection{A choice of states}
\Idx{adiabatic-zero}

As discussed in the introduction, a possible approach to solving the semiclassical Einstein equation is to select a class of candidate spacetimes and then for each of these spacetimes a unique state.
Here we restrict ourselves to the semiclassical Friedmann equations as given by~\eqref{eq:semiclassical_friedmann_trace} and~\eqref{eq:constraint-einstein}, \viz, the candidate spacetimes are flat FLRW spacetimes.
It remains to find a functional that associates to each flat FLRW spacetime a suitable state.

It would be desirable to associate to each spacetime a Hadamard state.
In the literature there are a few concrete examples of such states but unfortunately none of them are suitable for our purposes.
On FLRW spacetimes there is the notable construction of states of low energy discussed in \cref{sub:low_energy}, which are also Hadamard.
But the employed construction is based on a smearing of the modes with respect to an extended function of time and \latin{a priori} we do not even know if a solution of~\eqref{eq:friedmann_integral} exists in the future of the initial Cauchy surface.
The holographic constructions of Hadamard states, discussed in \cref{sec:holographic}, requires that the spacetime has certain asymptotic properties which are not under control for generic FLRW spacetimes.

Moreover, for technical reasons to be discussed later, we also have to consider spacetimes with $C^1$ metrics.
But on spacetimes with non-smooth metrics Hadamard states cannot exist.
Instead we will use the construction of adiabatic states of order zero as presented in~\cref{sub:adiabatic_zero}, which is also applicable to spacetimes with low regularity.
The price we have to pay for working with non-Hadamard states is that the solutions of~\eqref{eq:friedmann_integral} are not be smooth spacetimes.

We recall, that the states constructed in~\cref{sub:adiabatic_zero} are of the form
\begin{equation}\label{eq:2pt_adiabatic_zero}
  \omega_2(x, x') = \frac{1}{(2 \uppi)^3 a(\tau) a(\tau')} \int_{\RR^3} \conj{\chi}{}_{k}(\tau) \chi_k(\tau')\, \e^{\im \vec{k} \cdot (\vec{x} - \vec{x}')}\, \dif \vec{k}
\end{equation}
with modes $\chi_k = \sum_n \chi_k^{\smash{\mbox{\tiny $(\!n\!)$}}}(\tau)$ given by
\begin{equation*}
  \chi_k^{\smash{\mbox{\tiny $(\!0\!)$}}}(\tau) = \frac{1}{\sqrt{2 k_0}}\, \e^{\im k_0 \tau},
  \quad
  \chi_k^{\smash{\mbox{\tiny $(\!n\!)$}}}(\tau) = \int_{\tau_0}^\tau \frac{\sin\big( k_0 (\eta - \tau) \big)}{k_0} V(\eta) \chi_k^{\smash{\mbox{\tiny $(\!n-1\!)$}}}(\eta)\, \dif\eta,
\end{equation*}
where $k_0^2 = k^2 + a(\tau_0)^2 m^2$ and $V(\tau) = m^2 (a(\tau)^2 - a(\tau_0)^2)$.
Note that $\chi_k$ and $\omega_2$ can be defined in this way even if the scale factor~$a$ is only $C^1$.
This may be confirmed by taking a closer look at \cref{prop:adiabatic_zero}.

\subsection{Adiabatic regularization of the Wick square}
\label{sub:wick_square}
\IdxRanBegin{adiabatic-reg}

The integral equation~\eqref{eq:friedmann_integral} does not contain the two-point function but instead only its smooth part $w$ in the coincidence limit, \ie, Hadamard point-splitting has to applied to~\eqref{eq:2pt_adiabatic_zero}.

The equation~\eqref{eq:friedmann_integral} that we seek to solve contains the Wick square $\norder{\what\varphi^2}$ in a state~$\omega$ and thus (on a smooth spacetime) we would need to compute
\begin{equation}\label{eq:wick_square_pointsplit}
  \omega(\norder{\what\varphi^2}) = \lim_{x' \to x} \big( \omega_2(x, x') - \Hadamard(x, x') \big) - 4 c_1 m^4 + c_2 m^2 R.
\end{equation}
Since we are on a FLRW spacetime, we can use the method of adiabatic regularization instead (\cref{sub:adiabatic_regularization}) to perform an equivalent subtraction on the level of modes.
The difference of the two approaches is given in~\eqref{eq:adiabatic_wick_square} or, equivalently,~\eqref{eq:adiabatic_wick_square_2}.
It is useful to show directly that this regularization prescription indeed regularizes the two-point distribution~\eqref{eq:2pt_adiabatic_zero}:

\begin{proposition}\label{prop:regularized_state}
  The regularized two-point distribution
  \begin{equation*}
    \omega_2(\tau, \vec{x} - \vec{x}') - \lim_{\varepsilon \to 0^+} \frac{1}{(2 \uppi)^3 a(\tau)^2} \int_{\RR^3} \left( \frac{1}{2 k_0} - \frac{V(\tau)}{4 k_0^3} \right) \e^{\im \vec{k} \cdot (\vec{x} - \vec{x}')} \e^{-\varepsilon k}\, \dif\vec{k},
  \end{equation*}
  with $\omega_2$ given by~\eqref{eq:2pt_adiabatic_zero}, converges in the coinciding point limit for all continuously differentiable scale factors~$a$.
\end{proposition}

\begin{proof}
  We have to show that
  \begin{equation}\label{eq:regularized_state}
    \lim_{\varepsilon \to 0^+} \int_{\RR^3} \left( \abs{\chi_k}^2 - \frac{1}{2 k_0} + \frac{V}{4 k_0^3} \right) \e^{-\varepsilon k}\, \dif\vec{k}
    = \int_{\RR^3} \left( \abs{\chi_k}^2 - \frac{1}{2 k_0} + \frac{V}{4 k_0^3} \right)\, \dif\vec{k}
  \end{equation}
  is finite.
  To this end we expand the product $\abs{\chi_k}^2$ with $\chi_k = \sum_n \chi_k^n$ as
  \begin{equation*}
    \abs{\chi_k}^2 = \sum_{n=0}^\infty \sum_{l=0}^n \chi_k^l\, \conj\chi{}_k^{n-l}
  \end{equation*}
  in terms of the order~$n$.
  Inserting this expansion into~\eqref{eq:regularized_state}, we can prove the statement order by order:

  \paragraph{0th order}
  Since $\chi_k^0\, \conj\chi{}_k^0 = (2 k_0)^{-1}$, the first term in the subtraction exactly cancels the zeroth order term $\abs{\chi_k^0}^2$ in~\eqref{eq:regularized_state}.

  \paragraph{1st order}
  Using an integration by parts, we can rewrite the first order terms as
  \begin{align*}
    (\chi_k^0\, \conj\chi{}_k^1 + \chi_k^1\, \conj\chi{}_k^0)(\tau)
    & = \frac{1}{k_0^2} \int_{\tau_0}^\tau \sin\big( k_0 (\eta - \tau) \big) \cos\big( k_0 (\eta - \tau) \big) V(\eta)\, \dif\eta \\
    & = \frac{1}{2 k_0^2} \int_{\tau_0}^\tau \sin\big( 2 k_0 (\eta - \tau) \big) V(\eta)\, \dif\eta \\
    & = - \frac{V(\tau)}{4 k_0^3} + \frac{1}{4 k_0^3} \int_{\tau_0}^\tau \cos\big( 2 k_0 (\eta - \tau) \big) V'(\eta)\, \dif\eta.
  \end{align*}
  While the first summand $V(\tau) (4 k_0^3)^{-1}$ in the last line is exactly cancelled by the second term in the subtraction in \eqref{eq:regularized_state}, the second summand yields
  \begin{align}
    \notag
    \MoveEqLeft \int_{\RR^3} \frac{1}{4 k_0^3} \left( \int_{\tau_0}^\tau \cos\big( 2 k_0 (\eta - \tau) \big) V'(\eta)\, \dif\eta \right) \e^{-\varepsilon k}\, \dif\vec{k} \\
    \notag
    & = \uppi \int_0^\infty \frac{k^2}{k_0^3} \left( \int_{\tau_0}^\tau \cos\big( 2 k_0 (\eta - \tau) \big) V'(\eta)\, \dif\eta \right) \e^{-\varepsilon k}\, \dif k \\
    \notag
    & = \uppi \int_{a_0 m}^\infty k_0^{-1} \sqrt{1 - a^2 k_0^{-2}}\, \left( \int_{\tau_0}^\tau \cos\big( 2 k_0 (\eta - \tau) \big) V'(\eta)\, \dif\eta \right) \e^{-\varepsilon k}\, \dif k_0 \\
    \label{eq:1st_order_k_integral}
    & = \uppi \int_{\tau_0}^\tau V'(\eta)\, \left( \int_{a_0 m}^\infty k_0^{-1} \cos\big( 2 k_0 (\eta - \tau) \big)\, \e^{-\varepsilon k}\, \dif k_0 \right)\, \dif\eta - R(\tau)
  \end{align}
  for $\varepsilon > 0$.
  Here $R$ is a finite remainder term since it contains terms in the $k_0$-integration which decay at least like $k_0^{-3}$.
  Notice that, in the last equation of the previous formula, thanks to the positivity of $\varepsilon$ we have switched the order in which the $k_0$- and $\eta$-integration are taken.
  We would like to show that the weak limit $\varepsilon\to 0^+$ can be taken before the $\eta$-integration in~\eqref{eq:1st_order_k_integral}.

  To this end it remains to be shown that the $k_0$-integral in~\eqref{eq:1st_order_k_integral} converges in the limit $\varepsilon \to 0^+$ to an integrable function in $[\tau_0, \tau]$.
  First, note that the exponential integral
  \begin{equation}\label{eq:incomplete_gamma}
    E_1(x) = \Gamma(0, x)
    = \int_1^\infty \frac{\e^{-x t}}{t}\, \dif t
    = \int_0^1 \frac{\e^{-x}}{x - \ln (1 - s)}\, \dif s
  \end{equation}
  converges for $x \neq 0, \Re x \geq 0$.
  To show the identity, we used the substitution
  \begin{equation*}
    t = -x^{-1} \ln (1 - s) + 1
  \end{equation*}
  involving a subtle but inconsequential change of the integration contour in the complex plane if $x$ is complex.
  Then we easily see that
  \begin{equation}\label{eq:gamma_k_integral}
    \lim_{\varepsilon \to 0^+} \int_{a_0 m}^\infty k_0^{-1} \e^{\pm 2 \im k_0 (\eta - \tau) - \varepsilon k_0}\, \dif k_0 = E_1\big( \!\pm 2 \im a_0 m (\eta - \tau) \big)
  \end{equation}
  converges for $\eta \neq \tau$.
  This result is related with \eqref{eq:1st_order_k_integral} via
  \begin{equation*}
    \lim_{\varepsilon \to 0^+} \int_{a_0 m}^\infty k_0^{-1} \cos\big( 2 k_0 (\eta - \tau) \big)\, \e^{-\varepsilon k}\, \dif k_0
    = \lim_{\varepsilon \to 0^+} \int_{a_0 m}^\infty k_0^{-1} \cos\big( 2 k_0 (\eta - \tau) \big)\, \e^{-\varepsilon k_0}\, \dif k_0,
  \end{equation*}
  where we have used the boundedness of $k - (k^2 - a_0^2\, m^2)^{1/2}$ and \eqref{eq:gamma_k_integral}.
  Finally, a bound sufficient to see the $\eta$-integrability of the $k_0$-integral in \eqref{eq:1st_order_k_integral} can be obtained from the identity in \eqref{eq:incomplete_gamma}, namely,
  \begin{equation*}
    \abs*{E_1(\im x)}
    = \abs*{\int_0^1 \big( \im x - \ln (1 - s) \big)^{-1}\, \dif s}
    \leq \int_0^1 \big( x^2 + s^2 \big)^{-1/2}\, \dif s
    = \ln\bigg( \frac{1 + \sqrt{1 + x^2}}{x} \bigg).
  \end{equation*}

  \paragraph{2nd order}
  For the second order we calculate
  \begin{align}
    \notag
    \MoveEqLeft (\chi_k^0\, \conj\chi{}_k^2 + \chi_k^1\, \conj\chi{}_k^1 + \chi_k^2\, \conj\chi{}_k^0)(\tau) \\
    \notag
    \begin{split}
      & = \frac{1}{k_0^3} \int_{\tau_0}^\tau \sin\big( k_0 (\eta - \tau) \big) V(\eta) \\
      &\qquad \times \bigg( \int_{\tau_0}^\eta \sin\big( k_0 (\xi - \eta) \big) V(\xi) \cos\big( k_0 (\xi - \tau) \big)\, \dif\xi \\
      &\qquad\quad + \frac{1}{2} \int_{\tau_0}^\tau \sin\big( k_0 (\xi - \tau) \big) V(\xi) \cos\big( k_0 (\xi - \eta) \big)\, \dif\xi \bigg)\, \dif\eta
    \end{split}\\
    \notag
    & = \frac{1}{k_0^3} \int_{\tau_0}^\tau \sin\big( k_0 (\eta - \tau) \big) V(\eta)\, \left( \int_{\tau_0}^\eta \sin\big( k_0 (2 \xi - \eta - \tau) \big) V(\xi)\, \dif\xi \right)\, \dif\eta \\
    \label{eq:2nd_order_final}
    \begin{split}
      & = \frac{1}{2 k_0^4} \int_{\tau_0}^\tau \sin\big( k_0 (\eta - \tau) \big) V(\eta)\, \bigg( \int_{\tau_0}^\eta \cos\big( k_0 (2 \xi - \eta - \tau) \big) V'(\xi)\, \dif\xi \\
      &\qquad - \cos\big( k_0 (\eta - \tau) \big) V(\eta) \bigg)\, \dif\eta,
    \end{split}
  \end{align}
  where we have used integration by parts in the last equality.
  It is easy to obtain a $\vec{k}$-uniform estimate for the integral above and thus the integrability of the second order follows from $\int_{\RR^3} k_0^{-4}\, \dif\vec{k} < \infty$.

  \paragraph{Higher orders}
  For orders $n > 2$ it is sufficient to use the rough estimate from~\eqref{eq:estimate_chi_n}:
  \begin{align}
    \notag
    \abs*{\sum_{n=3}^\infty \sum_{l=0}^n \chi_k^l\, \conj\chi{}_k^{n-l}}(\tau)
    & \leq \frac{1}{2 k_0} \sum_{n=3}^\infty \frac{2^n}{n!} \left( \frac{1}{k_0} \int_{\tau_0}^\tau \abs[\big]{V(\eta)}\, \dif\eta \right)^3 \left( \int_{\tau_0}^\tau (\tau - \eta)\, \abs[\big]{V(\eta)}\, \dif\eta \right)^{n-3} \\
    \label{eq:higher_order_final}
    & \leq \frac{4}{k_0^4} \left( \int_{\tau_0}^\tau \abs[\big]{V(\eta)}\, \dif\eta \right)^3 \exp\left( 2 \int_{\tau_0}^\tau (\tau - \eta)\, \abs[\big]{V(\eta)}\, \dif\eta \right).
  \end{align}
  As above, the integrability of the higher orders follows from $\int_{\RR^3} k_0^{-4}\, \dif\vec{k} < \infty$.

  Note that none of the estimates above depends on higher derivatives of the scale factor.
  Therefore, combining these partial results, we see that the thesis holds true.
\end{proof}

It follows that we can consistently define the renormalized Wick square of the state given by~\eqref{eq:2pt_adiabatic_zero} at conformal time~$\tau$ for every FLRW spacetime with $C^1$ scale factor
\begin{equation}\label{eq:renormalized_wick_square}\begin{split}
  \omega(\norder{\what\varphi^2}) & = \frac{1}{(2 \uppi)^3 a^2} \int_{\RR^3} \left( \abs{\chi_k}^2 - \frac{1}{2 k_0} + \frac{V}{4 k_0^3} \right)\, \dif\vec{k} \\&\qquad + \frac{m^2}{(4 \uppi)^2} \left( \frac{1}{2} - \left(\frac{a_0}{a(\tau)}\right)^2 + 2 \ln\left( \frac{a_0}{a(\tau)} \right) + 2 \ln\left(\frac{\e^{\gamma} m \lambda}{\sqrt{2}}\right) \right),
\end{split}\end{equation}
which coincides with \eqref{eq:wick_square_pointsplit} for smooth spacetimes.
Moreover, we notice that, as a consequence of the previous proposition, it is possible to obtain global estimates for the renormalized Wick square:

\begin{proposition}\label{prop:bounds}
  The renormalized Wick square is bounded on every $a' \in C[\tau_0,\tau_1]$ with $a > 0$ in $[\tau_0,\tau_1]$ for every $\tau_1$ and with $a(\tau_0) = a_0$, namely,
  \begin{equation*}
    \abs*{\omega(\norder{\what\varphi^2})(\tau_1)} \leq C\left(\sup_{[\tau_0,\tau_1]}{a},\sup_{[\tau_0,\tau_1]}{a'}, (\tau_1-\tau_0),\frac{1}{\inf_{[\tau_0,\tau_1]}a } \right) \;
  \end{equation*}
  where $C$ is a finite increasing function.
\end{proposition}
\begin{proof}
  The proof of this proposition and the explicit value of $C$, can be obtained combining \eqref{eq:wick_square_pointsplit} with \eqref{eq:adiabatic_wick_square_2} and then analyzing the adiabatic subtraction \eqref{eq:regularized_state} order by order as in the proof of the preceding proposition.
\end{proof}

\subsection{Adiabatic regularization of the energy density}

There is another nice feature about the states we have constructed above.
Thanks to the conformal coupling of the scalar field with the curvature, the energy density computed in these states is finite even though these states are (on smooth spacetimes) only adiabatic states of order zero.
This is a crucial feature which permits us to solve the constraint~\eqref{eq:constraint-einstein} as a first step towards solving the semiclassical Einstein equation.

\begin{proposition}
  The energy density~$\rho$ in the state~$\omega$ defined by~\eqref{eq:2pt_adiabatic_zero} at the initial time $\tau = \tau_0$ is finite.
\end{proposition}
\begin{proof}
  Following~\cite{hack:2013}, in order to show that $\rho(\tau_0)$ is finite, we just need to show that the adiabatically regularized expression (see also \eqref{eq:regularized_energy} and the subsequent remark)
  \begin{equation}\label{eq:adiabatic-conformal-energy}
    \int_0^\infty \!\Big(
      \big( \abs{\chi_k'}^2 + (k^2 + m^2 a^2) \abs{\chi_k}^2 \big)
      -
      \big( \abs*{W_k^{\smash{\mbox{\tiny $(\!0\!)$}}}{}'}^2 + (k^2 + m^2 a^2) \abs*{W_k^{\smash{\mbox{\tiny $(\!0\!)$}}}}^2 \big)
    \Big)\, k^2\, \dif k
  \end{equation}
  does not diverge at $\tau = \tau_0$.
  % Here $\chi_k$ are the modes constructed in \cref{sec:adiabatic_zero}, whereas $W_k^{\smash{\mbox{\tiny $(\!0\!)$}}}$ are the adiabatic modes of order zero.
  Evaluating the expression~\eqref{eq:adiabatic-conformal-energy} at $\tau = \tau_0$ gives
  \begin{equation*}
    \frac{m^4}{8} \int_0^\infty \!\!\frac{a_0^2\, (a')^2}{ (k^2 + m^2 a_0^2)^{5/2}}\, k^2\, \dif k < \infty.
  \end{equation*}
\end{proof}
Notice that the previous proposition only guarantees that the energy density is well-defined at the initial time.
Nevertheless, the conservation equation for the stress-energy tensor permits to state that it is well-defined everywhere.

The expression~\eqref{eq:adiabatic-conformal-energy} coincides with the energy density~$\rho$ of the system up to a conformal rescaling and up to the addition of some finite terms.
Thus, since the energy density~$\rho$ is finite in the considered state, the constraint~\eqref{eq:constraint-einstein} holds, provided a suitable choice of $H(\tau_0)$ and $\Lambda$ is made.
We stress that, if we do not want to alter $\Lambda$, the same result can be achieved adding classical radiation to the energy density of the universe in a suitable state.

We would like to conclude this section with a remark.
In adiabatic states of order zero the expectation values of local fields containing derivatives are usually ill-defined.
Despite this, in the case of conformal coupling and for our choice of initial conditions~\eqref{eq:initial_adiabatic_zero}, the energy density turns out to be well-defined.
This is essentially due to the fact that in the massless conformally coupled case the adiabatic modes of order zero are solutions of the mode equation~\eqref{eq:mode_equation} and in that case the obtained state is the well known conformal vacuum. Hence, the adiabatically regularized energy density vanishes.
In the massive case the states constructed above are not very different than the conformal vacuum and, in particular, the energy density remains finite under that perturbation.

\IdxRanEnd{adiabatic-reg}

\section{Local solutions}
\label{sec:einstein_solutions_local}
\IdxRanBegin{semiclassical-einstein-local}

Our aim is to show the existence and uniqueness of local solutions to the semiclassical Friedmann equation.
In particular, according to the discussion in the introduction, we will analyze the uniqueness and existence of solutions of~\eqref{eq:friedmann_integral}.
Similar to the Picard--Lindelöf theorem, we will use the Banach fixed-point theorem to achieve this goal.
Some results on functional derivatives and the Banach fixed-point theorem are collected in \cref{cha:analysis}, in particular \cref{sec:fixed_point}.

Solving \eqref{eq:friedmann_integral} is equivalent to finding fixed-points of the functional~$F$ defined by
\begin{subequations}\label{eq:functional-einstein}\begin{align}
  \begin{split}
    F(H)(\tau) & \defn H_0 + \int_{\tau_0}^\tau \frac{a(\eta)}{H^2_c - H(\eta)^2} \big( H(\eta)^4 - 2 H_c^2 H(\eta)^2 - \tfrac{15}{2} m^4 \\&\qquad\qquad + 240 \uppi^2 ( m^2 \omega(\norder{\what\varphi^2})(H)(\eta) + 4 \Lambda) \big)\, \dif\eta,
  \end{split} \\
   & \defn H_0 + \int_{\tau_0}^\tau f(H)(\eta)\, \dif\eta.
\end{align}\end{subequations}
Since $\omega(\norder{\what\varphi^2})(H)$ is well-defined for continuous Hubble functions (see also \eqref{eq:renormalized_wick_square}), we select for the Banach space of candidate Hubble functions~$H$ the space\footnote{Until fixed, we take both $\tau_0$ and $\tau_1$ as variable and thus consider a family of Banach spaces.} $C[\tau_0, \tau_1]$, $\tau_0 < \tau_1$, equipped with the uniform norm
\begin{equation*}
  \norm{X}_{C[\tau_0,\tau_1]} \defn \norm{X}_\infty \defn \sup_{\tau \in [\tau_0,\tau_1]} \abs{X(\tau)}.
\end{equation*}
However, once $\tau_0$ and the initial condition $a_0 = a(H)(\tau_0) > 0$ are fixed, we find that
\begin{equation}\label{eq:a_of_H}
  a(H)(\tau) = a_0 \left(1 - a_0 \int_{\tau_0}^\tau H(\eta)\, \dif\eta \right)^{-1},
\end{equation}
as a functional of $H$, is not continuous on $C[\tau_0, \tau_1]$.
But we can find an open subset
\begin{equation}\label{eq:set_U}
  \mathcal{U}[\tau_0,\tau_1] \defn \big\{ H \in C[\tau_0, \tau_1] \;\big|\; \norm{H}_{C[\tau_0,\tau_1]} < \min\{ a(\tau_0)^{-1} (\tau_1 - \tau_0)^{-1}, H_c \} \big\}
\end{equation}
on which $a$ and thus also $V = m^2 (a^2 - a_0^2)$ depend smoothly on $H$.
Indeed, we can show the following:

\begin{proposition}\label{prop:gateaux_diffable}
  The functional
  \begin{equation}\label{eq:functional_f}
    f(H) = \frac{a(H)}{H^2_c - H^2} \big( H^4 - 2 H_c^2 H^2 - \tfrac{15}{2} m^4 + 240 \uppi^2 (m^2 \omega(\norder{\what\varphi^2})(H) + 4 \Lambda) \big)
  \end{equation}
  is continuously differentiable on $\mathcal{U}[\tau_0,\tau_1]$ for arbitrary but fixed $\tau_0, \tau_1$ and $a_0 = a(\tau_0)$.
\end{proposition}
\begin{proof}
  Given \eqref{eq:wick_square_pointsplit}, \eqref{eq:adiabatic_wick_square_2}, \cref{prop:regularized_state} and \cref{prop:bounds}, it is enough to show that $a(H)$ and $(H_c^2 - H^2)^{-1}$ are bounded and that $\omega(\norder{\what\varphi^2})(H)(\tau_0)$ is continuously differentiable.
  The former is assured by the condition $\norm{H}_{C[\tau_0,\tau_1]} < \min\{ a_0^{-1} (\tau_1 - \tau_0)^{-1}, H_c \}$ in the definition of $\mathcal{U}[\tau_0,\tau_1]$.
  For the latter it remains to be shown that the renormalized Wick square \eqref{eq:renormalized_wick_square} is continuously differentiable on $\mathcal{U}[\tau_0,\tau_1]$:

  We start by calculating the functional derivative of the scale factor
  \begin{equation*}
    \dirD a(H; \delta H)(\tau) = a(H)(\tau)^2 \int_{\tau_0}^\tau \delta H(\eta)\, \dif\eta.
  \end{equation*}
  The functional derivatives for $a^{-2}$, $\ln a$, $V$ and $V'$ follow easily.
  In particular we note that all these functions are continuously differentiable on $\mathcal{U}[\tau_0,\tau_1]$ because integration is a continuous operation and $a$ depends smoothly on $H$ in $\mathcal{U}[\tau_0,\tau_1]$.
  Therefore it suffices to analyze the differentiability of the integral \eqref{eq:regularized_state} appearing in the regularized two-point distribution.
  Moreover, within $\chi_k$ only the potential $V$ is (smoothly on $\mathcal{U}[\tau_0,\tau_1]$) dependent on $H$, thus simplifying the computations considerably.\footnote{If we were to work in cosmological time as in~\cite{pinamonti:2011}, we would also have to consider the functional dependence of conformal time on the scale factor.}
  Continuing with the regularized two-point distribution order by order as in \cref{prop:regularized_state}, we have:

  \paragraph{1st order}
  Since
  \begin{equation*}\begin{split}
    & \dirD \left( \chi_k^0\, \conj\chi{}_k^1 + \chi_k^1\, \conj\chi{}_k^0 + \frac{V}{4 k_0^3} \right)(H; \delta H)(\tau) \\
    & \qquad = \frac{1}{4 k_0^3} \int_{\tau_0}^\tau \cos\big( 2 k_0 (\eta - \tau) \big)\, \dirD V'(H; \delta H)(\eta)\, \dif\eta,
  \end{split}\end{equation*}
  we can proceed with the proof as in \cref{prop:regularized_state} with $V'$ replaced by $\dirD V'$ and differentiability follows.

  \paragraph{2nd order}
  As above, this part of the proof can be shown by replacing occurrences of $V$ and $V'$ in \eqref{eq:2nd_order_final} of \cref{prop:regularized_state} with $\dirD V$ and $\dirD V'$ respectively.

  \paragraph{Higher orders}
  For orders $n > 2$ we can again use an estimate similar to~\eqref{eq:estimate_chi_n} to obtain a result analogous to~\eqref{eq:higher_order_final}:
  \begin{multline*}
    \abs*{\dirD \left( \sum_{n=3}^\infty \sum_{l=0}^n \chi_k^l\, \conj\chi{}_k^{n-l} \right)(H; \delta H)(\tau)}
    \leq \frac{4}{k_0^4} \left( \int_{\tau_0}^\tau \abs[\big]{\dirD V(H; \delta H)(\eta)}\, \dif\eta \right)
    \\ \times \left( \int_{\tau_0}^\tau \abs[\big]{V(\eta)}\, \dif\eta \right)^2 \exp\left( 2 \int_{\tau_0}^\tau (\tau - \eta)\, \abs[\big]{V(\eta)}\, \dif\eta \right).
  \end{multline*}

  \noindent In this way we can conclude the proof of the present proposition.
\end{proof}

We can now formulate the main theorem of this chapter:
\begin{theorem}\label{thm:main}\Idx{banach-fixed-point}\Idx{fixed-point-lip}
  Let $(a_0, H_0), a_0 > 0, \abs{H_0} < H_c$, be some initial conditions fixed at $\tau_0$ for~\eqref{eq:friedmann_integral}.
  There is a non-empty interval $[\tau_0,\tau_1]$ and a closed subset $U \subset C[\tau_0,\tau_1]$ on which a unique solution to~\eqref{eq:friedmann_integral} exists.
\end{theorem}
\begin{proof}
  In \cref{prop:gateaux_diffable} we showed that $f$ is continuously differentiable on $\mathcal{U}[\tau_0,\tau_1]$ for any $\tau_1$.
  Using \cref{prop:closed}, we can thus find a $\tau_1 > \tau_0$ and a closed subset $U \subset \mathcal{U}[\tau_0,\tau_1]$ such that $F(U) \subset U$.
  It then follows from \cref{prop:contraction} that $F$ has a unique fixed point in $U$.
\end{proof}

Notice that the solution provided by the previous theorem is actually more regular, it is at least differentiable.
Thus the corresponding spacetime is $C^2$ and has well-defined curvature tensors.
The extra regularity is provided by~\eqref{eq:friedmann_integral} and can be easily seen when it is written in its differential form~\eqref{eq:friedmann_differential}.
It might be surprising that the solutions are not smooth, since the procedure to find the solution involves repeated integration, but because the chosen adiabatic state is only guaranteed to be continuous on every spacetime, $H$ is only $C^1$.
Using \cref{cor:smooth}, one can see that a more regular state immediately improves also the regularity of the solution.

\IdxRanEnd{semiclassical-einstein-local}

\section{Global solutions}
\label{sec:einstein_solutions_global}
\IdxRanBegin{semiclassical-einstein-global}

In this section we would like to show that it is always possible to extend a `regular' local solution up to the point where either $H^2$ becomes bigger than $H_c^2$ or $a$ diverges.\footnote{$H^2 = H_c^2$ corresponds to a singularity in the derivative of $H$ in~\eqref{eq:friedmann_differential}.}
To this end we start giving a definition we shall use below.

\begin{definition}\label{def:regular}
  A continuous solution $H_*$ of \eqref{eq:friedmann_integral} in the interval $[\tau_0,\tau_1]$ with initial conditions
  \begin{equation*}
    a(H_*)(\tau_0) = a_0,
    \qquad
    H_*(\tau_0)^2 = H_0^2 = \frac{1}{3} \big( \rho(\tau_0) + \Lambda \big)
  \end{equation*}
  will be called \IdxMain{semiclassical-einstein-regular}\emph{regular}, if no singularity for either $a$, $H_*$ or $H_*'$ is encountered in $[\tau_0,\tau_1]$.
  Namely, $H_*$ must satisfy the following conditions:
  \begin{enumerate}
    \item $\norm*{H_*(\tau)}_{C[\tau_0,\tau_1]} < H_c$,
    \item $a_0 \int_{\tau_0}^{\tau} H_*(\eta)\, \dif\eta < 1$ for every $\tau$ in $[\tau_0,\tau_1]$.
  \end{enumerate}
\end{definition}

We remark that a local solution obtained from \cref{thm:main} is a regular solution.
Henceforth, assume that we have a regular solution $H_*$ as described in the definition.
Notice that condition a) ensures that no singularity in $H_*'$ is met in $[\tau_0,\tau_1]$.
Condition b), on the other hand, ensures that $a$ does not diverge in the interval $[\tau_0,\tau_1]$. Moreover, both a) and b) together imply that $a$ is strictly positive, as can be seen from~\eqref{eq:a_of_H}.

We would like to prove that a regular solution can always be extended in $C[\tau_0,\tau_2]$ for a sufficiently small $\tau_2 - \tau_1 > 0$.
To this end, let us again consider the set
\begin{equation*}
  \mathcal{U}[\tau_1,\tau_2] \defn \Big\{ H \in C[\tau_1, \tau_2] \;\Big|\; \norm{H}_{C[\tau_1,\tau_2]} < \min\big\{ a_1^{-1} (\tau_2 - \tau_1)^{-1}, H_c \big\} \Big\}
\end{equation*}
defined in \eqref{eq:set_U} and where $a_1 \defn a(H_*)(\tau_1)$ is the value assumed by the solution $a(H_*)$ at $\tau_1$.
Now we can give a proposition similar to \cref{prop:gateaux_diffable}, namely:
\begin{proposition}\label{prop:gateaux_diffable_ext}
  Let $H_*$ be a solution of \eqref{eq:friedmann_integral} in $C[\tau_0,\tau_1]$ which is also regular.
  The functional $f(H)$ of~\eqref{eq:functional_f}, when evaluated on regular extensions of $H_*$ in $\mathcal{U}[\tau_1,\tau_2]$, is continuously differentiable for arbitrary $\tau_2 > \tau_1$.
\end{proposition}
\begin{proof}
  The proof of this proposition can be obtained exactly as the proof of \cref{prop:gateaux_diffable}.
  However, the estimates we have obtained in \cref{prop:bounds} and the proof of \cref{prop:gateaux_diffable} cannot be applied straightforwardly because the state $\omega$ depends on the initial time $\tau_0$ and the initial datum $a_0$ through the construction described in \cref{sub:adiabatic_zero}.
  Moreover, the estimates of \cref{prop:bounds} depend on the knowledge of $a$ and $a'$ on the whole interval $[\tau_0,\tau_2]$.
  Luckily enough, we know that the solution $H_*$ is regular in $[\tau_0,\tau_1]$, while we know that the extension restricted to $[\tau_1,\tau_2]$ is in the set $\mathcal{U}[\tau_1,\tau_2]$; thus we just need to use the following estimates
  \begin{align*}
    \norm{a}_{C[\tau_0,\tau_2]} & = \max \big\{ \norm{a}_{C[\tau_1,\tau_2]}, \norm{a}_{C[\tau_0,\tau_1]} \big\}, \\
    \norm{a}_{C[\tau_0,\tau_2]}^{-1} & = \max \big\{ \norm{a}^{-1}_{[\tau_1,\tau_2]}, \norm{a}^{-1}_{C[\tau_0,\tau_1]} \big\}, \\
    \norm{a'}_{C[\tau_0,\tau_2]} & = \max \big\{ \norm{a'}_{C[\tau_1,\tau_2]}, \norm{a'}_{C[\tau_0,\tau_1]} \big\}.
  \end{align*}
  With this in mind, we can again use \cref{prop:bounds} to control the boundedness of $\omega(\norder{\what\varphi^2})$.
  Then, making the replacements $\tau_0 \to \tau_1, \tau_1 \to \tau_2$ and $a_0 \to a_1$ at the \emph{appropriate} places in \cref{prop:gateaux_diffable}, one can see that estimates are not substantially influenced and that thesis still holds for $\mathcal{U}[\tau_1,\tau_2]$.
\end{proof}

Notice that it is always possible to fix $\tau_2$ such that $a_1^{-1} (\tau_2 - \tau_1)^{-1} \geq H_c$, whereby $\mathcal{U}[\tau_1,\tau_2]$ becomes the set of \emph{all} possible regular extensions of $H_*$ in $[\tau_1,\tau_2]$.
This guarantees that any extension in $\mathcal{U}[\tau_1,\tau_2]$ is the unique regular extension.

We are now ready to state the main theorem of the present section which can be proven exactly as \cref{thm:main}.
\begin{theorem}\label{thm:extension}
  Consider a solution $H_*(\tau)$ in $C[\tau_0,\tau_1]$ of~\eqref{eq:friedmann_integral}.
  If the solution is regular in $[\tau_0,\tau_1]$, as defined in \cref{def:regular},
  then it is possible to find a $\tau_2 > \tau_1$ such that, the solution $H_*$ can be extended uniquely to $C[\tau_0,\tau_2]$ and the solution is regular therein.
\end{theorem}
\begin{proof}
  Thanks to \cref{prop:gateaux_diffable_ext}, $f$ is continuously differentiable on all regular extensions of $H_*$ in $\mathcal{U}[\tau_1,\tau_2]$ for any $\tau_2$ such that $a_1^{-1} (\tau_2 - \tau_1)^{-1} \geq H_c$.
  With the remarks of the proof of \cref{prop:gateaux_diffable_ext} we can use \cref{prop:bounds} to estimate the boundedness of $\omega(\norder{\what\varphi^2})$ and apply \cref{prop:closed} to find a $\tau_2 > \tau_1$ and a closed subset $U \subset \mathcal{U}[\tau_1,\tau_2]$ such that $F(U) \subset U$.
  It then follows from \cref{prop:contraction} that $F$ has a unique fixed point in $U$.
\end{proof}

We study now all possible solutions of~\eqref{eq:friedmann_integral} which are defined on intervals of the form $[\tau_0,\tau)$, which are regular on any closed interval contained in their domain and which enjoy the same initial values $a_0 = a(\tau_0), H_0 = H(\tau_0)$.

\begin{proposition}
  A \IdxMain{semiclassical-einstein-maximal}\emph{maximal} solutions exists; it is unique and regular.
\end{proposition}
\begin{proof}
  Let $\mathcal{S} = \{ I_\alpha, H_\alpha \}_{\alpha \in A}$, with $A \subset \NN$ some index set, be the set of all possible regular solutions $H_\alpha$ with domain $I_\alpha$ for the same initial values.
  By the existence of local solutions $\mathcal{S}$ is not empty.
  We then take the union $I = \bigcup_{\alpha \in A} I_\alpha$ and define $H(\tau) = H_\alpha(\tau)$ for $\tau \in I_\alpha$, which is a well-defined regular solution by \cref{prop:gateaux_diffable_ext}.
  Since every $I$ is a superset of every $I_\alpha$, $H$ is the unique maximal regular solution. 
\end{proof}

As for the solution provided by theorem~\ref{thm:main}, also the maximal solution obtained above correspond to a metric with $C^2$ regularity.

\IdxRanEnd{semiclassical-einstein-global}

\section{Numerical solutions}
\label{sec:numerical_solutions}
\IdxRanBegin{semiclassical-einstein-numerical}

The first problem that one encounters when attempting to treat the semiclassical Einstein equation in a numerical fashion, is the construction of states.
Here, the mode equation~\eqref{eq:mode_equation}, which describes an oscillator with a time-dependent `resonance frequency' $\omega_k^2 = k^2 + a^2 m^2$, has to be solved.
Standard numerical solvers, like the Runge--Kutta method, rely on differentiation and their error scales like a (high-order) derivative of the solution.
However, each derivative of an oscillating function increases the amplitude by a power of the frequency, thus ultimately leading to large errors for quickly oscillating differential equations after a short time span.
This problem can be partially counteracted by choosing ever smaller step sizes in time, but eventually one will encounter a computational barrier.
Another possibility is to look for a non-standard approach to solve the mode equation.
Such methods replace differentiation with integration, but, since the numerical integration of highly oscillatory functions is also a non-trivial problem, this is still an active area of research, see for example \cite{engquist:2009,iserles:2002,iserles:2006,levin:1996,levin:1997}.

One might say, that the large frequency behaviour of the modes is of no relevance when solving the semiclassical Einstein equation because it involves the state only \emph{after} regularization, \ie, after the terms that contribute for large $\omega_k$ have been subtracted.
While this response is, to some degree, certainly true, numerical errors in the solution before the regularization and in the subtraction itself can still accumulate.
Therefore, the issue of the high frequency modes and their regularization has to be carefully addressed in a numeric approach.

Although the perturbative construction of the state used in this chapter (see also  \cref{sub:adiabatic_zero}), the functional~\eqref{eq:friedmann_integral} and the use of the Banach fixed-point theorem in the proof of \cref{thm:main} were not developed with a numerical application in mind, there are reasons why they might be useful also for numerics:
The mode solution are found recursively from~\eqref{eq:adiabatic_recc_int}, an integral equation which avoids the differentiation problem discussed above. 
Moreover, as seen in \cref{prop:regularized_state}, only the first two partial modes $\chi_k^{\smash{\mbox{\tiny $(\!0\!)$}}}$ and $\chi_k^{\smash{\mbox{\tiny $(\!1\!)$}}}$ are affected by the regularization in a well-understood way, so that also the numerical difficulty in the regularization can circumvented.
Nevertheless, also this approach is not without its problems as it involves repeated numerical integration of oscillating functions and therefore it is very slow if naïvely implemented as a Riemann sum, because it requires small time steps.
The other feature of the proof of existence that allows a translation to numerics is the use of the Banach fixed-point theorem.
Namely, we can be assured that an iterated recursive application of~\eqref{eq:friedmann_integral} will converge to a solution, even though we do not know how quickly convergence occurs.

\IdxRanEnd{semiclassical-einstein-numerical}

\section{Outlook}
\label{sec:solutions_outlook}

In this chapter we have studied the backreaction of a quantum massive scalar field conformally coupled with gravity to cosmological spacetimes.
We have given initial conditions at finite time $\tau=\tau_0$ and we have shown that a unique maximal solution exists. The maximal solution either lasts forever or until a spacetime singularity is reached.

In order to obtain this result, we have used a state which looks as much as possible like the vacuum at the initial time.
Notice that it is possible to choose other classes of states without significantly altering the results obtained in this chapter.
In particular, if we restrict ourself to Gaussians pure state which are homogeneous and isotropic, their two-point function takes the form
\begin{equation*}\label{eq:two_point_other}
  \widetilde\omega_2(x, y) = \lim_{\varepsilon \to 0^+} \frac{1}{(2 \uppi)^3} \int_{\RR^3} \frac{\conj\xi_k(\tau_x)}{a(\tau_x)} \frac{\xi{}_k(\tau_y)}{a(\tau_y)}\, \e^{\im \vec{k} \cdot (\vec{x} - \vec{y})} \e^{-\varepsilon k}\, \dif\vec{k},
\end{equation*}
where $\xi_k$ are solutions of \eqref{eq:mode_equation} which enjoy the Wronskian condition \eqref{eq:wronski}.
These $\chi_k$ can then be written as a Bogoliubov transformation of the modes $\chi_k$ studied earlier in this chapter, namely,
\begin{equation*}
  \xi_k = A(k) \chi_k + \conj{B}(k) \conj\chi_k
\end{equation*}
for suitable functions $A$ and $B$.
Then, because of the constraint $\abs{A}^2 - \abs{B}^2 = 1$, the difference
\begin{equation*}
  \widetilde\omega(\norder{\what\varphi^2}) - \omega(\norder{\what\varphi^2}) =
  \lim_{\varepsilon \to 0^+} \frac{1}{(2 \uppi)^3} \frac{2}{a^2} \int_{\RR^3} \left(
    \abs{B}^2\chi_k\conj\chi{}_k + \Re{\left(A B\chi_k\chi_k\right)}
  \right) \e^{-\varepsilon k}\, \dif\vec{k}
\end{equation*}
can be easily controlled employing \eqref{eq:estimate_chi_n} if $\abs{B}$ is sufficiently regular (\eg, if $B(k)$ is in $L^2 \cap L^1$).\footnote{A detailed analysis of this problem is present in \cite{zschoche:2014}.}
With this observation it is possible to obtain again all the estimates used in the proofs of \cref{thm:main,thm:extension}.

In the future, it would be desirable to study the semiclassical equations in more general cases, namely for more general fields, abandoning for example the conformal coupling, and for more general background geometries.
The results presented here cannot straightforwardly be extended to fields which are not conformally coupled to curvature or to spacetimes that are not conformally flat because in that case fourth order derivatives of the metric originating in the conformal anomaly cannot be cancelled by a judicious choice of renormalization parameters, \ie, Wald's fifth axiom \cite{wald:1977} cannot be satisfied.
To still solve the semiclassical Einstein equation with methods similar to those presented here,
a deeper analysis of the states is required, in particular, one needs states of higher regularity.
A preliminary study in this direction can be found in a paper of Eltzner and Gottschalk \cite{eltzner:2011}, where the semiclassical Einstein equation on a FLRW background with non-conformally coupled scalar field is discussed.
The case of backgrounds which are only spherically symmetric is interesting from many perspectives.
Its analysis could give new hints on the problem of semiclassical black hole evaporation and confirm the nice two-dimensional results obtained in \cite{ashtekar:2011}.
Finally, the limit of validity of the employed equation needs to be carefully addressed in the future.

\IdxRanEnd{semiclassical-einstein-sol}

%% file: fluctuations.tex
%!TEX root = master.tex

\chapter{Induced semiclassical fluctuations}
\label{cha:fluctuations}

\section*{Introduction}

As described in \cref{cha:einstein}, in semiclassical Einstein gravity one equates a classical quantity, the Einstein tensor, with the expectation value of a quantum observable, the quantum stress-energy tensor, \ie, a quantity with a probabilistic interpretation.
Such a system could make sense only when the fluctuations of the quantum stress-energy tensor can be neglected.
Unfortunately, as also noticed in~\cite{pinamonti:2011}, the variance of quantum unsmeared stress-energy tensor is always divergent even when proper regularization methods are considered.
The situation is slightly better when a smeared stress-energy tensor is analyzed.
In that way, however, the covariance of~\eqref{eq:semiclassical_einstein} gets lost.
A possible way out is to allow for fluctuations also on the left-hand side of~\eqref{eq:semiclassical_einstein}.
This is the point of view we shall assume within this chapter, which is based on an article~\cite{pinamonti:2015a} in collaboration with Pinamonti.

More precisely, we interpret the Einstein tensor as a stochastic field and equate its $n$-point distributions with the symmetrized $n$-point distributions of the quantum stress-energy tensor.
As an application of this (toy) model, we analyze the metric fluctuations induced by a massive, conformally coupled scalar field via the (quantum) stress-energy tensor in the simplest non-trivial spacetime -- de Sitter spacetime.
We find that the potential in a Newtonianly perturbed FLRW spacetime has a almost scale-invariant power spectrum.

These results encourage a comparison with the observation of anisotropies in the cosmic microwave background and their theoretical explanations.
Anisotropies in the angular temperature distribution were predicted by Sachs and Wolfe \cite{sachs:1967} shortly after the discovery of the cosmic microwave background (CMB) by Penzias and Wilson \cite{penzias:1965}.
In their famous paper they discuss what was later coined the \emph{Sachs--Wolfe effect}: The redshift in the microwave radiation caused by fluctuations in the gravitational field and the corresponding matter density fluctuations.
In the standard model of inflationary cosmology the fluctuations imprinted upon the CMB are seeded by quantum fluctuations during inflation \cite{mukhanov:1981,mukhanov:1992}, see also the reviews in~\cite{durrer:2008,ellis:2012}.

The usual computation of the power spectrum of the initial fluctuations produced by single-field inflation can be sketched as follows \cite{bartolo:2004,durrer:2008,ellis:2012}:
First, one introduces a (perturbed) classical scalar field $\varphi + \delta \varphi$, the inflaton field, which is coupled to a (perturbed) expanding spacetime $g + \delta g$.
Then, taking the Einstein equation and the Klein--Gordon equation at first order in the perturbation variables, one constructs an equation of motion for the \emph{Mukhanov--Sasaki variable} $Q = \delta \varphi + \dot\varphi\, H^{-1} \Phi$, where $\Phi$ is the Bardeen potential \cite{bardeen:1980} and $H$ the Hubble constant.
$Q$ is then quantized\footnote{A recent discussion about the quantization of a such system can be found in \cite{eltzner:2013}.} (in the slow-roll approximation) and one chooses as the state of the associated quantum field a Bunch--Davies-like state.
Last, one evaluates the power spectrum $P_Q(k)$ of~$Q$, \ie, the Fourier-transformed two-point distribution of the quantum state, in the super-Hubble regime $k \ll aH$ and obtains an expression of the form\footnote{An alternative definition of the power spectrum is $\mathcal{P}_Q(k) = (2 \uppi)^{-2} k^3 P_Q(k)$.}
\begin{equation}
  \label{eq:powerspectrum_Q}
  P_Q(k) = \frac{A_Q}{k^3}\, \left( \frac{k}{k_0} \right)^{n_s-1},
\end{equation}
where $A_Q$ is the amplitude of the fluctuations, $k_0$ a pivot scale and $n_s$ the spectral index.
Notice the factor of $k^{-3}$ in~\eqref{eq:powerspectrum_Q} which gives the spectrum the `scale-invariant' \IdxMain{harrison-zeldovich}\emph{Harrison--Zel'dovich} form if $n_s = 1$.
Depending on the details of model, $n_s \lesssim 1$ and there is also a possibility for a scale dependence of~$n_s$ -- the `running' of the spectral index $n_s = n_s(k)$.

This result can then be related to the power spectrum of the comoving curvature perturbation~$\mathcal{R}$, which is proportional to~$Q$, and can be compared with observational data.
Assuming adiabatic and Gaussian initial perturbations, the WMAP collaboration finds $n_s = 0.9608 \pm 0.0080$ (at $k_0 = \unit[0.002]{Mpc^{-1}}$) in a model without running spectral index and gravitational waves, excluding a scale-invariant spectrum at~$5\sigma$~\cite{hinshaw:2012}.
Furthermore, the data of WMAP and other experiments can be used to constrain the deviations from a pure Gaussian spectrum, the so called non-Gaussianities, that arise in some inflationary models \cite{bartolo:2004,bennett:2012,maldacena:2003}.

In \cite{parker:2007,agullo:2008,agullo:2009,agullo:2011a} concerns have been raised whether the calculation leading to~\eqref{eq:powerspectrum_Q} and similar calculations are correct: The authors argue that the two-point distribution of the curvature fluctuations has to be regularized and renormalized similarly to what is done in semiclassical gravity.
As a result the power spectrum is changed sufficiently that previously observationally excluded inflation models become realistic again.
On the contrary the authors of \cite{durrer:2009,marozzi:2011} argue that the adiabatic regularization employed in \cite{parker:2007,agullo:2008,agullo:2009,agullo:2011a} is not appropriate for low momentum modes if evaluated at the Horizon crossing and irrelevant for these modes if evaluated at the end of inflation.

A slightly different approach to the calculation of the power spectrum based on \emph{stochastic gravity} can be found in \cite{roura:2000,roura:2008,hu:2008}.
In spirit similar to the approach presented in this chapter, the authors equate fluctuations of the stress-energy tensor with the correlation function of the Bardeen potential.
In the super-Hubble regime they obtain an almost scale-invariant power spectrum.
Moreover, they discuss the equivalence of their stochastic gravity approach with the usual approach of quantizing metric perturbations.

Our approach here is strictly different from the standard one described above.
Instead of quantizing a coupled system of linear inflaton and gravitational perturbations, we aim at extending the semiclassical Einstein equation to describe metric fluctuations via the fluctuations in the stress-energy tensor of a quantum field.

\section{Fluctuations of the Einstein tensor}

Consider now the Einstein tensor as a random field.
Then we could imagine to equate the probability distribution of the Einstein tensor with the probability distribution of the stress-energy tensor.
This suggestion, however, seems largely void without a possibility of actually computing the probability distributions of the stress-energy tensor because, as discussed above, its moments of order larger than one are divergent.

Instead we may approach this idea by equating the hierarchy of $n$-point distributions of the Einstein tensor with that of the stress-energy tensor:
\begin{subequations}
  \label{eq:einstein_prob}
  \begin{align}
    \label{eq:einstein_prob_1}
    \E[\big]{G_{ab}(x_1)} & = \omega\big( \norder{T_{ab}(x_1)} \big), \\
    \label{eq:einstein_prob_2}
    \E[\big]{\delta G_{ab}(x_1)\, \delta G_{c'd'}(x_2)} & = \tfrac{1}{2} \omega\big( \norder{\delta T_{ab}(x_1)}\, \norder{\delta T_{c'd'}(x_2)} + \norder{\delta T_{c'd'}(x_2)}\, \norder{\delta T_{ab}(x_1)} \big), \\
    \intertext{and}
    \label{eq:einstein_prob_n}
    \E{\delta G^{\boxtimes n}} & = \omega\big(\Sym(\norder{\delta T}^{\boxtimes n})\big), \quad n > 1,
  \end{align}
\end{subequations}
where $\omega$ is a Hadamard state and we defined
\begin{equation*}
  \delta G_{ab} \defn G_{ab} - \E{G_{ab}}
  \quad \text{and} \quad
  \norder{\delta T_{ab}} \defn \norder{T_{ab}} - \omega(\norder{T_{ab}}).
\end{equation*}
The symmetrization on the right-hand side is necessary because the classical quantity on left-hand side is invariant under permutation.

We emphasize that we are equating singular objects in~\eqref{eq:einstein_prob_n}.
Having all the $n$-point distributions of the Einstein stochastic tensor, we can easily construct an equation for the moments of the smeared Einstein tensor which equals the moments of a smeared stress-energy tensor by smearing both sides of~\eqref{eq:einstein_prob} with tensor products of a smooth compactly supported function.
This smearing also automatically accounts for the symmetrization in~\eqref{eq:einstein_prob}.

Furthermore we stress that equating moments, obtained smearing both side of~\eqref{eq:einstein_prob}, is not equivalent to equating probability distributions.
Although it is also possible to arrive at a description in terms of moments when coming from a probability distribution, the inverse mapping is not necessarily well-defined.
Successful attempts to construct a probability distribution for smeared stress-energy tensors can be found in~\cite{fewster:2010,fewster:2012b}.

Consider now a quasi-free Hadamard state~$\omega$ of a conformally coupled scalar field~$\varphi$ on a spacetime $(M, \wbar{g})$, the \Idx{background-st}\emph{background spacetime}.
Our aim is to calculate the perturbation of the background spacetime as specified by the
correlation functions on the left-hand side of~\eqref{eq:einstein_prob} due to the fluctuations of the stress-energy in the quantum state~$\omega$ as specified on the right-hand side of~\eqref{eq:einstein_prob}.
In particular we will require that $\omega$ satisfies \eqref{eq:einstein_prob_1} when we identify the Einstein tensor of the background spacetime~$\wwbar{G}_{ab}$ with~$\E{G_{ab}}$ (\cf\ \cref{cha:einstein_solutions} for a discussion of the solutions of the semiclassical Einstein equation in cosmological spacetimes).
Note that by choosing this Ansatz we are completely ignoring any backreaction effects of the fluctuations to the background metric and evaluate the stress-energy tensor on a state specified on the background spacetime.

Later on we consider perturbations of the scalar curvature induced by a `Newtonianly' perturbed FLRW metric.
For this reason it will be sufficient to work with the trace of~\eqref{eq:einstein_prob} (using the background metric) instead of the full equations.
With the definition
\begin{equation*}
  S \defn - \wbar{g}^{ab} G_{ab},
\end{equation*}
such that $R = \E{S}$, the equations \eqref{eq:einstein_prob} simplify to
\begin{subequations}
  \label{eq:einstein_prob_trace}
  \begin{align}
    \label{eq:einstein_prob_trace_1}
    \E[\big]{S(x_1)} & = \frac{m^2}{8 \uppi^2} [w] - \frac{1}{4 \uppi^2} [v_1] + \text{ren. freedom}, \\
    \label{eq:einstein_prob_trace_2}
    \E[\big]{S(x_1)\, S(x_2)} - \E[\big]{S(x_1)}\E[\big]{S(x_2)} & = m^4 \big( \omega_2^2(x_1, x_2) + \omega_2^2(x_2, x_1) \big), \\
    \intertext{and}
    \label{eq:einstein_prob_trace_n}
    \E[\big]{(S - \E{S})^{\boxtimes n}(x_1, \ldots, x_n)} & = 2^n m^{2n} \Sym\bigg( \sum_\Gamma \prod_{i, j}\frac{\omega_2^{\lambda^\Gamma_{ij}}(x_i, x_j)}{\lambda^\Gamma_{ij}!} \bigg),
  \end{align}
\end{subequations}
where the sum is over all directed graphs~$\Gamma$ with~$n$ vertices $1, \dotsc, n$ with two arrows at every vertex directed to a vertex with a larger label.
$\lambda^\Gamma_{ij} \in \{0,1,2\}$ is the number of arrows from~$i$ to~$j$.
If we perform the symmetrization in~\eqref{eq:einstein_prob_trace_n}, we see that the sum is over all acyclical directed graphs with two arrows at every vertex.
For illustration some graphs are shown in \cref{fig:graphs}.

\begin{figure}
  \centering
  \includegraphics{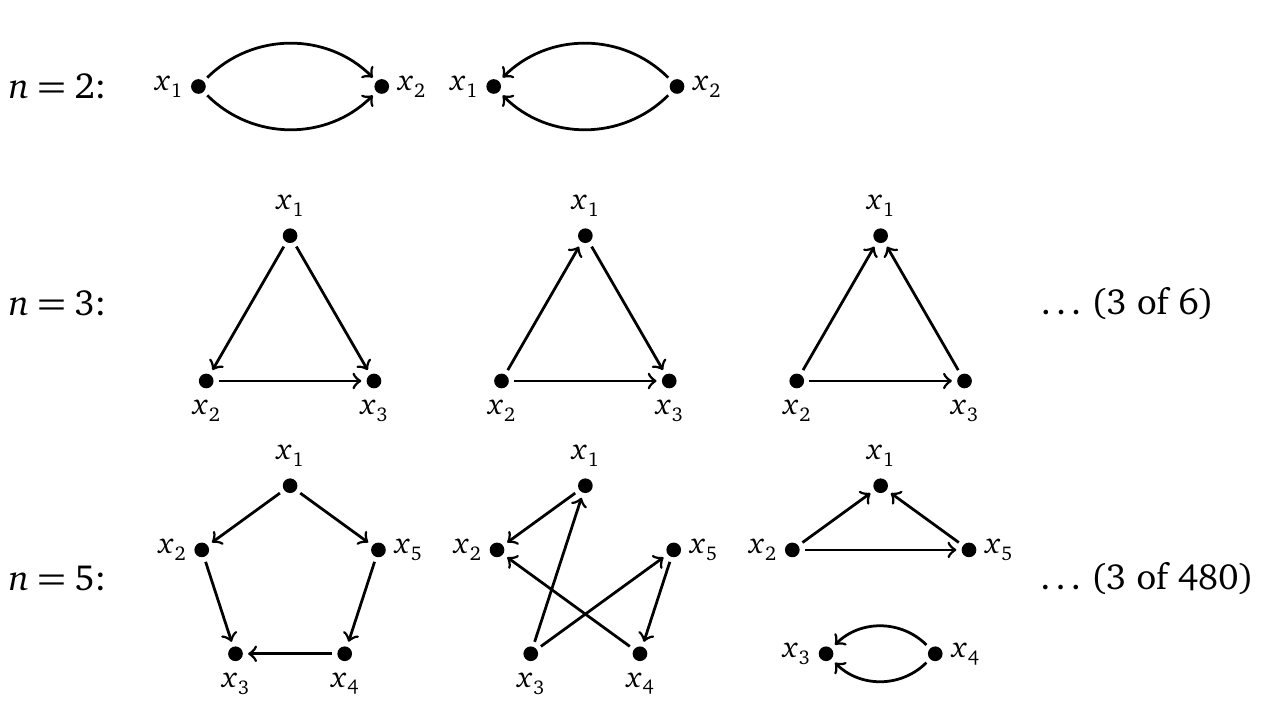}
  \caption{A few graphs illustrating \eqref{eq:einstein_prob_trace_n} for $n=2$, $n=3$ and $n=5$.}
  \label{fig:graphs}
\end{figure}

To obtain~\eqref{eq:einstein_prob_trace_2} and~\eqref{eq:einstein_prob_trace_n}, note that $\norder{\what\varphi^2} - \omega(\norder{\what\varphi^2})$ does not depend on the choice of normal ordering\footnote{Indeed this holds true if we replace $\varphi^2$ with $L \varphi^2$, for any linear operator~$L$.} and thus only \eqref{eq:einstein_prob_trace_1} needs to be renormalized.
Therefore we may choose normal ordering with respect to~$\omega_2$ to see that the combinatorics are equivalent to those in Minkowski space.
Moreover, as $\omega_2$ is a bisolution of the Klein--Gordon equation, the term $\tfrac{1}{3} \norder{\what\varphi\, \mathrm{P} \what\varphi}$ which causes the trace anomaly in~\eqref{eq:einstein_prob_trace_1} (\cf\, \cref{sub:T_ab_trace}) does not contribute to the higher moments.

\section{Fluctuations around a de Sitter spacetime}
\label{sec:fluctuations_around_de_sitter}

We shall now specialize the general discussion presented above to Newtonianly perturbed, exponentially expanding, flat FLRW universes.
That is, the background spacetime $(M, \wbar{g})$ is given in conformal time $\tau < 0$ by the metric tensor
\begin{equation*}
  \label{eq:background_metric}
  \wbar{g} \defn (H \tau)^{-2} (-\dif\tau \otimes \dif\tau + \delta_{ij}\; \dif x^i \otimes \dif x^j)
\end{equation*}
and we consider fluctuations of the scalar curvature derived from metric perturbations of the form
\begin{equation}
  \label{eq:perturbed_metric}
  g \defn (H \tau)^{-2} \big( -(1 + 2 \Phi)\, \dif\tau \otimes \dif\tau + (1 - 2 \Phi)\, \delta_{ij}\, \dif x^i \otimes \dif x^j \big).
\end{equation}

The kind of fluctuations that we consider by choosing \eqref{eq:perturbed_metric} resemble those that are present in single-scalar field inflation in the longitudinal gauge, where there are only `scalar fluctuations' without anisotropic stress (so that the two Bardeen potentials coincide) \cite{ellis:2012,mukhanov:1992}.
Notice that, for classical metric perturbation, these constraints descend from the linearized Einstein equation, however, \latin{a priori} there is no similar constraint in~\eqref{eq:einstein_prob_2}.
Despite these facts, we proceed analyzing the influence of quantum matter on this special kind of metric perturbations and we also refrain from discussing the gauge problem associated to choosing a perturbed spactime; the chosen perturbation potential $\Phi$ is \emph{not} gauge invariant.

We can now calculate the various perturbed curvature tensors and obtain in particular
\begin{equation*}
  S = 12 H^2 (1 - 3 \Phi) + 24 H^2 \tau \partial_\tau \Phi - 6 H^2 \tau^2 \partial_\tau^2 \Phi + 2 H^2 \tau^2 \laplace \Phi + \bigO(\Phi^2)
\end{equation*}
for the trace of the perturbed Einstein tensor, where $\laplace = \partial_{x^1}^2 + \partial_{x^2}^2 + \partial_{x^3}^2$ is the ordinary Laplace operator.
Dropping terms of higher than linear order, this can also be written as
\begin{equation}
  \label{eq:hyp_op}
  S - \E{S} = - 6 H^2 \tau^4 \big( \partial_{\tau} - \tfrac{1}{3} \laplace\! \big) \tau^{-2} \Phi,
\end{equation}
where $\E{S} = 12 H^2$ is nothing but the scalar curvature of the background spacetime.
Notice that, up to a rescaling, the operator on the right-hand side of~\eqref{eq:hyp_op} looks like a wave operator with the characteristic velocity equal to $1/\sqrt{3}$ of the velocity of light.

We can now evaluate the influence of quantum matter fluctuations on the metric fluctuations by inverting the previous hyperbolic operator by means of its retarded fundamental solutions~$\upDelta_{\mathrm{ret}}$ and applying it on both sides of~\eqref{eq:einstein_prob_trace_2} and~\eqref{eq:einstein_prob_trace_n}.
From~\eqref{eq:einstein_prob_trace_2} we can then (formally) obtain the two-point correlation functions of~$\Phi$ (per definition $\E{\Phi} = 0$):
\begin{equation}\label{eq:retarded_applied_to_w}\begin{split}
  \E[\big]{\Phi(x_1)\, \Phi(x_2)} = m^4 \iint_{\RR^8} \upDelta_{\mathrm{ret}}(x_1, y_1) \upDelta_{\mathrm{ret}}(x_2, y_2) & \big( \omega_2^2(y_1, y_2) \\& + \omega_2^2(y_2, y_1) \big)\, \dif^4 y_1\, \dif^4 y_2.
\end{split}\end{equation}
Employing the retarded fundamental solutions in the inversion without adding any solution of~\eqref{eq:hyp_op}, we are implicitly assuming that all the $n$-point distributions of the perturbation potential~$\Phi$ are sourced by quantum fluctuations.
Here we are only interested in evaluating their effect.

\subsection{The squared two-point distribution}

In order to proceed with our analysis, we shall specify the quantum state~$\omega$ for the matter theory.
Following the Ansatz discussed in the preceding section, we choose a quasi-free Hadamard state which satisfies the semiclassical Einstein equation on the background.
In particular, we require that $\omega$ solves \eqref{eq:einstein_prob_trace_1}, namely
\begin{equation*}
  12 H^2 = \frac{m^2}{8 \uppi^2} [w] - \frac{1}{4 \uppi^2} [v_1] + 4 c_1 m^4 - c_2 m^2 R
\end{equation*}
The right-hand side of the previous equation is characterized by three contributions:
The state dependent part $[w]$, the anomaly part $[v_1]$, which takes the simple form (\cf\ Eq.~\eqref{eq:hadamard_v1_flrw})
\begin{equation*}
  [v_1] = -\frac{H^4}{60} + \frac{m^4}{8},
\end{equation*}
and the renormalization freedom $c_1 m^4$ and $c_2 m^2 R$.
Here we set $c_2 = 0$, because we assume the point of view that we have already measured Newton's gravitational constant and do not wish to renormalize it (\cf\ \cref{sub:T_ab_pointsplit}).
That is, we have
\begin{equation*}
  12 H^2 = \frac{m^2}{8 \uppi^2} [w] + \frac{1}{8 \uppi^2} \bigg( \frac{H^4}{15} - \frac{m^4}{2} \bigg) + 4 c_1 m^4.
\end{equation*}
For the semiclassical Einstein equation to hold, we therefore have to require that $[w]$ is a constant.
Then, having fixed $H$ and $m$ (no matter their absolute value), there is always a choice of $c_1$ for which the chosen metric~$\wbar{g}$ and $\omega$ satisfy the semiclassical Einstein equation.
On a de Sitter spacetime these criteria are satisfied by the Bunch--Davis state, \cf\ \cref{sub:bunch_davies}.

In order to evaluate the influence of the quantum matter fluctuations on~$\Phi$ via equation~\eqref{eq:retarded_applied_to_w}, we have to discuss the form of the two-point distribution of the chosen state and its square.
Any Hadamard state on $(M, \wbar{g})$ can be written is equal to the Bunch--Davies state up to smooth terms.
In particular, the two-point distribution $\omega_2$ of every Hadamard state on de Sitter spacetime is of the form
\begin{equation}\label{eq:bunch_davies_expansion}
  \omega_2(x, x') = \lim_{\varepsilon \to 0^+} \frac{H^2}{4 \uppi^2} \frac{\tau \tau'}{(x - x')^2 + 2 \im \varepsilon (\tau - \tau') + \varepsilon^2} + \text{less singular terms},
\end{equation}
where we write $(x - x')^2 \defn -(\tau - \tau')^2 + (\vec{x} - \vec{x}')^2$ and, as always, the limit $\varepsilon \to 0^+$ is a weak limit.
It is no surprise that the leading singularity is conformally related to the two-point distribution of a massless scalar field on Minkowski spacetime, we denote it by $\omega_{\MM}$.
Thus it is also clear that the less singular contributions vanish in the limit of zero mass.

As can be seen in~\eqref{eq:einstein_prob_trace_2}, we need to compute the square of the two-point distribution of the state in question.
For our purposes it will be sufficient to compute the square of the leading singularity in the Hadamard state.\footnote{Note that, in (spatial) momentum space, the leading singularity in~\eqref{eq:bunch_davies_expansion} contributes the smallest inverse power of the momentum~$\vec{k}$; the `less singular terms' correspond to higher inverse powers of~$\vec{k}$. Accordingly, these terms fall off faster for large~$\vec{k}$. This is nothing but the usual relation ship between high momenta and short distances.}
The square of the massless two-point distribution on Minkowski space is
\begin{equation*}
  \omega_{\MM}(x, x')^2 = \lim_{\varepsilon \to 0^+} \left( \frac{1}{4 \uppi^2} \frac{1}{(x - x')^2 + 2 \im \varepsilon (\tau - \tau') + \varepsilon^2} \right)^2.
\end{equation*}
Writing $\omega_{\MM}$ in terms of its spatial Fourier transform, an expression for the spatial Fourier transform of the square of the massless Minkowski vacuum can be obtained as
\begin{equation}
  \label{eq:spectrum_omega2}
  \omega_{\MM}(x, x')^2 = \lim_{\varepsilon \to 0^+} \frac{1}{128 \uppi^5}
  \int_{\RR^3} \e^{\im \vec{k} \cdot (\vec{x} - \vec{x}')}
  \int_{k}^\infty \e^{-\im p (\tau - \tau')} \e^{- \varepsilon p }\, \dif p\, \dif\vec{k}.
\end{equation}
Later on we will use this expression in order to obtain the power spectrum of $\Phi$.

\subsection{Power spectrum of the metric perturbations}

We want to compute the power spectrum $P(\tau, \vec{k})$ of the two-point correlation of~$\Phi$ at the time~$\tau$.
Since both the spacetime and the chosen state are invariant under spatial translation, it can be defined as
\begin{equation*}
  \E[\big]{\Phi(\tau, \vec{x})\, \Phi(\tau, \vec{x}')} \defn \frac{1}{(2 \uppi)^3} \int_{\RR^3} P(\tau,\vec{k})\, \e^{\im \vec{k} \cdot (\vec{x} - \vec{x}')}\, \dif\vec{k}.
\end{equation*}
To obtain $P$, we first need an expression for the retarded operator $\upDelta_{\mathrm{ret}}$ corresponding to~\eqref{eq:hyp_op}:
\begin{align*}
  (\upDelta_{\mathrm{ret}}\,f)(\tau, \vec{x}) & = \frac{1}{(2 \uppi)^3} \int_{\RR^3} \int_{-\infty}^{\tau}
  \what\upDelta_{\mathrm{ret}}(\tau, \tau_1, \vec{k})
  \widehat{f}(\tau_1, \vec{k}) \e^{\im \vec{k} \cdot \vec{x}}\, \dif\tau_1\, \dif\vec{k},
  \quad \text{with}
  \\
  \what\upDelta_{\mathrm{ret}}(\tau, \tau_1, \vec{k}) & \defn
  - \frac{1}{6 H^2} \frac{\tau^2}{\tau_1^4} \frac{\sqrt{3}}{k} \sin\left(k\, (\tau - \tau_1)/\sqrt{3}\right),
\end{align*}
where $f$ is a compactly supported smooth function.
We can then rewrite \eqref{eq:retarded_applied_to_w} in Fourier space to obtain
\begin{equation*}
  P(\tau, \vec{k}) = 2 m^4 \int_{-\infty}^\tau \int_{-\infty}^\tau
  \what\upDelta_{\mathrm{ret}}(\tau, \tau, \vec{k}) \what\upDelta_{\mathrm{ret}}(\tau, \tau', \vec{k})
  \widehat{\omega^2_{BD}}(\tau, \tau', \vec{k})\,
  \dif\tau\, \dif\tau'.
\end{equation*}
Note that the symmetrization of the state is taken care of indirectly by the equal limits of the two integrations.

As discussed above (see \eqref{eq:bunch_davies_expansion} and the following paragraph), we will compute the contribution due to the leading singularity of the Hadamard state:
\begin{equation*}
  P_0(\tau, \vec{k}) \defn 2 H^4 m^4 \int_{-\infty}^\tau \int_{-\infty}^\tau
  \what\upDelta_{\mathrm{ret}}(\tau, \tau, \vec{k})\, \what\upDelta_{\mathrm{ret}}(\tau, \tau', \vec{k})\,
  \tau^2 \tau^{\prime2} \widehat{\omega^2_{\MM}}(\tau_1, \tau', \vec{k})\,
  \dif\tau\, \dif\tau'.
\end{equation*}
We emphasize at this point that, because of the form of~\eqref{eq:spectrum_omega2} and of $\what\upDelta_{\mathrm{ret}}$, no infrared (with respect to~$\vec{k}$) singularity appears in $P_0(\tau,\vec{k})$ at finite $\tau$.
Recall also that the error we are committing, using $P_0(\tau,\vec{k})$ at the place of $P(\tau,\vec{k})$, tends to vanish in the limit of small masses.
Inserting the spectrum of $\omega_{\MM}^2$ obtained in~\eqref{eq:spectrum_omega2} and switching the order in which the integrals are taken (for $\varepsilon > 0$), we can write
\begin{equation}
  \label{eq:P0_with_A}
  P_0(\tau, \vec{k}) =
  \lim_{\varepsilon \to 0^+} \frac{m^4}{16 \uppi^2}
  \int_{k}^\infty
  \frac{1}{k^4}\, \abs*{A\left(\tau, k/\sqrt{3}, p\right)}^2 \e^{-\varepsilon p}\, \dif p,
\end{equation}
where we have introduced the auxiliary function
\begin{equation}
  \label{eq:defA}
  A(\tau, \kappa, p) \defn \int_{-\infty}^\tau
  \frac{\kappa\, \tau^2}{\tau_1^2}
  \sin\big( \kappa\, (\tau - \tau_1) \big)\,
  \e^{-\im p \tau_1}\,
  \dif\tau_1,
\end{equation}
which can also be written in closed form in terms of the generalized exponential integral\footnote{For a definition and various properties of these special functions see \eg\ \cite[Chap.~8]{olver:2010}.} $E_2$ as
\begin{equation*}
  \label{eq:A_with_E2}
  A(\tau, \kappa, p) =
  A(\kappa \tau, p \tau) =
  \frac{\im}{2}\, \kappa \tau\,
  \Big(
    E_2 \big( \im\, (p + \kappa)\, \tau \big)\, \e^{ \im \kappa \tau} -
    E_2 \big( \im\, (p - \kappa)\, \tau \big)\, \e^{-\im \kappa \tau}
  \Big)
\end{equation*}
for $p \geq \kappa > 0$ and by the complex conjugate of this expression if $\kappa > p, \kappa > 0$.
In the following study of the form of the power spectrum $P_0$ the auxiliary function $A$ will be instrumental.

\begin{lemma}
  \label{lem:A}
  For $\abs{p} \neq \kappa > 0$, $A(\tau, \kappa, p)$ has the $\tau$-uniform bound
  \begin{equation}
    \label{eq:bound_A}
    \abs{A} \leq \frac{4\, \kappa^2}{\abs{\kappa^2 - p^2}}.
  \end{equation}
  For large negative times it satisfies the limit
  \begin{equation}
    \label{eq:limit_A}
    \lim_{\tau \to -\infty} \abs{A} = \frac{\kappa^2}{\abs{\kappa^2 - p^2}}.
  \end{equation}
\end{lemma}
\begin{proof}
  Using the fact that
  \begin{equation*}
    \e^{-\im p \tau_1} = \left( \od[2]{}{\tau_1} + \kappa^2 \right) \frac{\e^{-\im p \tau_1}}{\kappa^2 - p^2},
  \end{equation*}
  we can perform two integrations by parts to obtain
  \begin{align*}
    A(\tau, \kappa, p) & = \frac{\kappa^2}{\kappa^2 - p^2} \big( \e^{-\im p \tau} + R(\tau, \kappa, p) \big),
    \quad \text{with}
    \\
    R(\tau, \kappa, p) & \defn \tau^2 \int_{-\infty}^\tau
    \left(
      \frac{4}{\tau_1^3} \cos\big(\kappa\, (\tau - \tau_1)\big)
      +
      \frac{6}{\kappa \tau_1^4} \sin\big(\kappa\, (\tau - \tau_1)\big)
    \right) \e^{-\im p \tau_1}\, \dif\tau_1.
  \end{align*}
  It is now easy to obtain an upper bound for $R$ which is uniform in conformal time, namely $\abs{R} \leq 3$, which then yields the bound \eqref{eq:bound_A}.

  For the second part of the proposition we perform a change of the integration variable to $x = \tau_1 / \tau$:
  \begin{equation*}
    R(\tau, \kappa, p) = - \int_1^{\infty}
    \left(
      \frac{4}{x^3} \cos\big(\kappa \tau\, (1 - x)\big)
      +
      \frac{6}{\kappa\, \tau\, x^4} \sin\big(\kappa \tau\, (1 - x)\big)
    \right) \e^{-\im p \tau x}\, \dif x.
  \end{equation*}
  The contribution proportional to $1/\tau$ in $R$ is bounded by $C(\kappa)/\abs{\tau}$ and thus vanishes in the limit $\tau \to -\infty$.
  Moreover, since $\abs{p} \neq \kappa$ and ${1}/{x^3}$ is $L^1$ on $[1,\infty)$, we can apply the Riemann--Lebesgue lemma and see that this contribution vanishes in the limit $\tau \to -\infty$.
  The remaining part of $\abs{A}$ is $\kappa^2 \abs{\kappa^2 - p^2}^{-1}$, which is independent of $\tau$, and thus the limit \eqref{eq:limit_A} holds true.
\end{proof}

Note that the bound for $A$ obtained above is not optimal.
Numerical integration indicates that $\abs{A}^2$ is monotonically decreasing in $\tau$ and thus bounded by the limit stated in~\eqref{eq:limit_A} (see also \cref{fig:P0_plot}).
Nevertheless, we can use this lemma to derive the following bounds and limits for $P_0$:

\begin{proposition}
  \label{prop:bound_large_limit}
  The leading contribution $P_0$ to the power spectrum of the potential $\Phi$ induced by a conformally coupled massive scalar field in the Bunch--Davis state is bounded by the Harrison--Zel'dovich spectrum uniformly in time, namely
  \begin{equation*}
    \abs[\big]{P_0(\tau, \vec{k})} \leq \frac{16\, C}{\abs{\vec{k}}^3},
    \quad
    C \defn \frac{3 - 2 \sqrt{3}\, \arccoth\!\sqrt{3}}{192 \uppi^2}\, m^4,
  \end{equation*}
  and it tends to the Harrison--Zel'dovich spectrum for $\tau \to -\infty$, \ie,
  \begin{equation*}
    \label{eq:limit_P0}
    \lim_{\tau \to -\infty} P_0(\tau, \vec{k}) = \frac{C}{\abs{\vec{k}}^3}.
  \end{equation*}
\end{proposition}
\begin{proof}
  The proof can be easily obtained using the $\tau$-uniform estimate \eqref{eq:bound_A} obtained in \cref{lem:A} and computing the integral
  \begin{equation*}
    \abs[\big]{P_0(\tau, \vec{k})}
    \leq \frac{m^4}{\uppi^2} \int_k^\infty \left(\frac{1}{3 p^2 - k^2}\right)^2\, \dif p
    = \frac{3 - 2 \sqrt{3}\, \arccoth\!\sqrt{3}}{12 \uppi^2} \frac{m^4}{k^3}.
  \end{equation*}

  Having shown the first part of the proposition, let us now analyze the limit
  \begin{equation*}
    \lim_{\tau \to -\infty} P_0(\tau, \vec{k}) = \frac{m^4}{16 \uppi^2} \int_k^\infty \frac{1}{k^4}
    \lim_{\tau \to -\infty} \abs*{A\left(\tau, k/\sqrt{3}, p\right)}^2\, \dif p,
  \end{equation*}
  where we have taken the $\tau$-limit before the integral and already evaluated the $\varepsilon$-limit because $\abs{A}^2$ is bounded by an integrable function uniformly in time.
  Inserting the limit \eqref{eq:limit_A} from \cref{lem:A}, we can compute the $p$-integral
  \begin{equation*}
    \lim_{\tau \to -\infty} P_0(\tau, \vec{k})
    = \frac{m^4}{16 \uppi^2} \int_k^\infty (3 p^2 - k^2)^{-2}\, \dif p
    = \frac{3 - 2 \sqrt{3}\, \arccoth\!\sqrt{3}}{192 \uppi^2} \frac{m^4}{k^3},
  \end{equation*}
  thus concluding the proof.
\end{proof}

We can complement the results of \cref{prop:bound_large_limit} with the following observation:
\begin{proposition}
  \label{prop:P0_k3}
  The power spectrum $P_0$ has the form
  \begin{equation*}
    P_0(\tau, \vec{k}) = \frac{\mathcal{P}_0(\abs{\vec{k}} \tau)}{\abs{\vec{k}}^3},
  \end{equation*}
  where $\mathcal{P}_0$ is a function of $\abs{\vec{k}} \tau$ only.
\end{proposition}
\begin{proof}
  Noting that $A(\tau, \kappa, p)$ is a function of $\kappa\, \tau$ and $p\, \tau$ only and performing the $\varepsilon$-limit in~\eqref{eq:P0_with_A} inside the integral, this can be seen by the substitution $x = p\, \tau$ in~\eqref{eq:P0_with_A}.
\end{proof}

We would like to improve the estimate of $P_0(\tau, \vec{k})$ for $\tau$ close to zero.
Adhering to our previous strategy, we shall first give a new estimate for $A(\tau, k, p)$:

\begin{lemma}
  The auxiliary function $A(\tau, \kappa, p)$ is bounded by
  \begin{equation*}
    \abs[\big]{A(\tau, \kappa, p)} \leq - 2\, \frac{\kappa^2\, \tau}{\abs{p}}, \quad p \neq 0,\quad \tau<0.
  \end{equation*}
\end{lemma}
\begin{proof}
  Recalling the form of $A$ given in~\eqref{eq:defA} and integrating by parts, where we use that $\e^{-\im p \tau_1} = \im\, p^{-1}\, \partial_{\tau_1} \e^{-\im p \tau_1}$, we find
  \begin{equation*}
    A(\tau, \kappa, p) = \frac{\im \kappa^2 \tau^2}{p} \int_{-\infty}^\tau
    \left(
      \frac{1}{\tau_1^2} \cos\big(\kappa\, (\tau - \tau_1)\big) +
      \frac{2}{\kappa \tau_1^3} \sin\big(\kappa\, (\tau - \tau_1)\big)
    \right) \e^{-\im p \tau_1}\, \dif\tau_1.
  \end{equation*}
  We then take the absolute value and estimate the trigonometric functions, which gives us a bound on $A$, namely
  \begin{equation*}
    \abs[\big]{A(\tau, \kappa, p)}
    \leq \frac{\kappa^2 \tau^2}{\abs{p}} \int_{-\infty}^\tau
    \left( \frac{1}{\tau_1^2} - 2\, \frac{\tau - \tau_1}{\tau_1^3} \right)\, \dif\tau_1
    = - 2\, \frac{\kappa^2 \tau}{\abs{p}}.
  \end{equation*}
\end{proof}

Performing the integration in $p$ analogously to the second part of proposition \eqref{prop:bound_large_limit}, the last lemma immediately leads to a corresponding bound for $P_0$:

\begin{proposition}
  \label{prop:small_limit}
  The leading contribution~$P_0$ of the power spectrum of the potential~$\Phi$ satisfies the inequality
  \begin{equation*}
    \abs[\big]{P_0(\tau, \vec{k})} \leq \frac{m^4}{36 \uppi^2} \frac{\tau^2}{\abs{\vec{k}}}
  \end{equation*}
  and therefore, in particular, $P_0(0, \vec{k}) = 0$.
\end{proposition}

\begin{figure}
  \centering
  \includegraphics{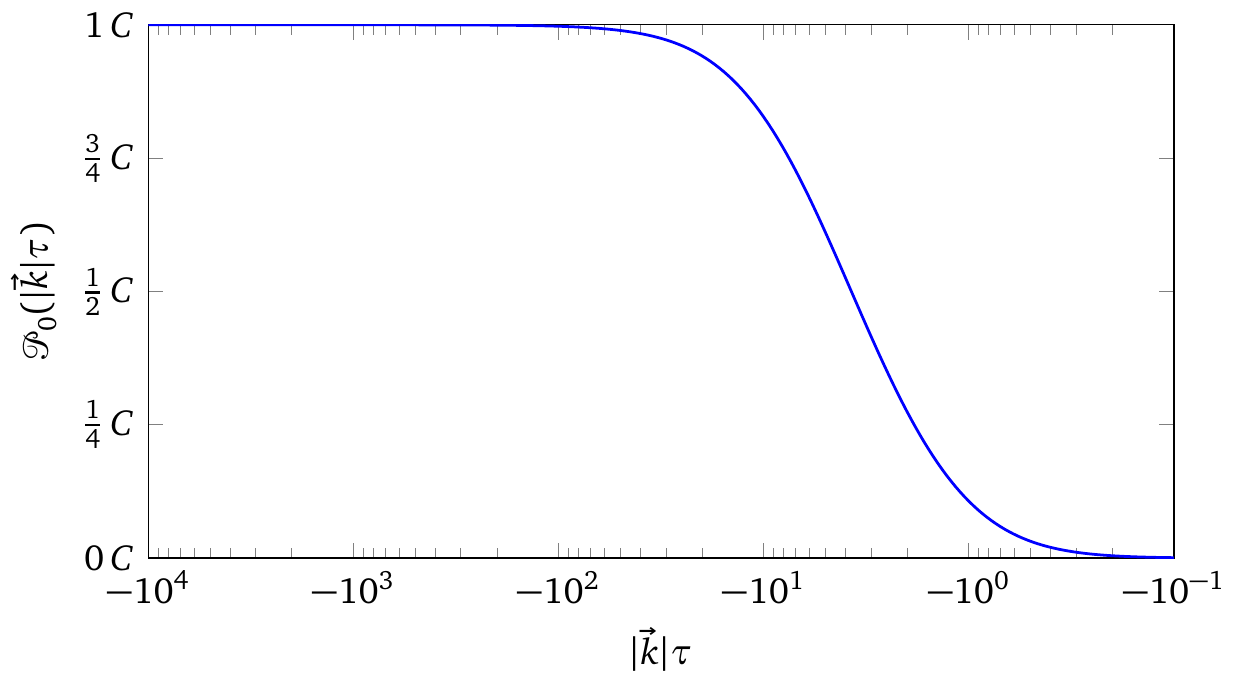}
  \caption{Logarithmic plot of the rescaled power spectrum $\mathcal{P}_0(\abs{\vec{k}} \tau)$, where $C$ is the same proportionality constant as in \cref{prop:bound_large_limit}.}
  \label{fig:P0_plot}
\end{figure}

The rescaled power spectrum $\mathcal{P}_0(\abs{\vec{k}} \tau)$ can be analyzed numerically and a plot is shown in \cref{fig:P0_plot}.
It clearly exhibits the asymptotic behaviour of $P_0$ discussed in \cref{prop:bound_large_limit,prop:small_limit}.
Note that the horizontal axis is logarithmically scaled to highlight the behavior of $\mathcal{P}_0$ for small $\abs{\vec{k}} \tau$, which would be concealed by the fast approach of $\mathcal{P}_0$ to its bound had we used a linear scaling.

In this section we have used the leading singularity\footnote{Recall that considering only the leading singularity in the Bunch--Davis state also corresponds to the limit of vanishing mass.} of the two-point function of the Bunch--Davis state on a de Sitter universe to compute the influence of quantum matter on the power spectrum of the metric perturbation $\Phi$.
We have seen that this results in an almost scale-invariant power spectrum.
We stress that such a singularity is not a special feature of the Bunch--Davis state but is common for every Hadamard state.
Moreover, although our analysis has been done on a de Sitter universe, similar quantum states have been constructed on universes which are asymptotically de Sitter spaces in the past \cite{dappiaggi:2009a,dappiaggi:2009b}.
All these states tend to the Bunch--Davis state for $\tau \to -\infty$ and are of Hadamard form.

%-- NON-GAUSSIANITIES -------------------------------------------------------%

\subsection{Non-Gaussianities of the metric perturbations}

It follows from \eqref{eq:einstein_prob_trace_n} that the $n$-point correlation for $\Phi$ will, in general, not vanish.
Also for odd $n$ they will be different from zero and hence $\Phi$ is not a Gaussian random field.
As a first measure of the non-Gaussianity of $\Phi$ one usually calculates its three-point correlation function or the corresponding bispectrum $B$:
\begin{multline*}
  \E[\big]{\Phi(\tau, \vec{x}_1)\, \Phi(\tau, \vec{x}_2)\, \Phi(\tau, \vec{x}_3)} \defn
  \frac{1}{(2 \uppi)^9} \iiint_{\RR^9}
  \delta(\vec{k}_1 + \vec{k}_2 + \vec{k}_3) B(\tau, \vec{k}_1, \vec{k}_2, \vec{k}_3)
  \\ \times
  \e^{\im\, (\vec{k}_1 \cdot \vec{x}_1 + \vec{k}_2 \cdot \vec{x}_2 + \vec{k}_3 \cdot \vec{x}_3)}\,
  \dif\vec{k}_1\, \dif\vec{k}_2\, \dif\vec{k}_3.
\end{multline*}

Assuming nonzero $\vec{k}_1$, $\vec{k}_2$ and $\vec{k}_3$, we will derive the form of the bispectrum $B$ considering (as above) only the contribution due to the leading singularity of the Bunch--Davis state, which we will denote by $B_0$.
We will follow the same steps that lead us to the calculation of the power spectrum in the previous section.
That is, we apply the retarded propagator $\upDelta_{\mathrm{ret}}$ of~\eqref{eq:hyp_op} as in~\eqref{eq:retarded_applied_to_w} to the right-hand side of~\eqref{eq:einstein_prob_trace_n} for $n=3$ to obtain an equation for $\Phi$ and insert for the two-point distribution the conformally rescaled two-point distribution of the massless Minkowski vacuum.
The result can again be expressed in terms of the auxiliary function $A$ defined in~\eqref{eq:defA}:
\begin{equation}\label{eq:B0_with_A}\begin{split}
  B_0(\tau, \vec{k}_1, \vec{k}_2, \vec{k}_3) & = \lim_{\varepsilon \to 0^+} \frac{m^6}{32 \sqrt{3} \vec{k}_1^2 \vec{k}_2^2 \vec{k}_3^2} \int_{\RR^3} \Bigg( \frac{\e^{-\varepsilon \big(\omega_{\vec{p}}(-\vec{k}_1) + \omega_{\vec{p}}(\vec{k}_3) + \abs{\vec{p}}\big)}}{\omega_{\vec{p}}(-\vec{k}_1) \omega_{\vec{p}}(\vec{k}_3) \abs{\vec{p}}} \\
  &\qquad \times A\big(\tau, \kappa_1, \omega_{\vec{p}}(-\vec{k}_1) + \abs{\vec{p}} \big) A\big(\tau, \kappa_3, -\omega_{\vec{p}}(\vec{k}_3) - \abs{\vec{p}} \big) \vphantom{\Bigg)} \\
  &\qquad \times A\big(\tau, \kappa_2, \omega_{\vec{p}}(\vec{k}_3) - \omega_{\vec{p}}(-\vec{k}_1) \big) + \text{permutations} \Bigg)\, \dif\vec{p},
\end{split}\end{equation}
where $\kappa_i \defn \abs{\vec{k}_i} / \sqrt{3}$, $\omega_{\vec{p}}(\vec{k}) \defn \abs{\vec{k} + \vec{p}}$ and the sum is over all permutations of $1, 2, 3$.

We can apply the same bound on $A$ which has been used in the previous section to bound the power spectrum $P_0$ to produce a bound on the integrand of $B_0$ almost everywhere.\footnote{We cannot bound the integrand of $B_0$ everywhere using \eqref{eq:bound_A} because $\abs{\vec{k}_2} / \sqrt{3} \neq \abs[\big]{\omega_{\vec{p}}(\vec{k}_3) - \omega_{\vec{p}}(-\vec{k}_1)}$ (and permutations) does not hold everywhere.}
Nevertheless, the singularity in the integrand in~\eqref{eq:B0_with_A} is integrable, \ie, $B_0$ is bounded.
As a consequence we can perform the limit $\varepsilon \to 0^+$ inside the integral.

\begin{proposition}
  The leading contribution $B_0$ of the bispectrum of the metric perturbation $\Phi$ has the form
  \begin{equation*}
    B_0(\tau, \vec{k}_1, \vec{k}_2, \vec{k}_3) = \frac{\mathcal{B}_0(k_1 \tau, k_2 \tau, k_3 \tau)}{k_1^2\, k_2^2\, k_3^2},
  \end{equation*}
  where $\mathcal{B}_0$ is a function of $k_1 \tau$, $k_2 \tau$, and $k_3 \tau$ only and $k_i = \abs{\vec{k}_i}$.
\end{proposition}
\begin{proof}
  Analogously to \cref{prop:P0_k3}, we note that after a change of variables $\vec{x} = \tau\, \vec{p} $ the integrand in~\eqref{eq:B0_with_A} is a function of $k_1 \tau$, $k_2 \tau$, and $k_3 \tau$ only.
\end{proof}

To finish our discussion about non-Gaussianities, we notice that, although the employed quantum field is a linear one, we obtained a three-point function for $\Phi$ which is similar to the one obtained by Maldacena \cite{maldacena:2003} who has quantized metric perturbations outside the linear approximation.

\section{Outlook}
\label{sec:fluctuations_outlook}

In this chapter the influence of quantum matter fluctuations on metric perturbations over de Sitter backgrounds were analyzed.
We used techniques proper of quantum field theory on curved spacetime to regularize the stress-energy tensor and to compute its fluctuations.
In particular, we interpreted the perturbations of the curvature tensors as the realization of a stochastic field.
We then obtained the $n$-point distributions of such a stochastic field as induced by the $n$-point distributions of a quantum stress tensor by means of semiclassical Einstein equations.

We also noticed that, while the expectation value of the stress-energy tensor is characterized by renormalization ambiguities, this is no longer the case when fluctuations are considered.
Hence the obtained results are independent on the particular regularization used to define the stress tensor.

In order to keep superficial contact with literature on inflation, we investigated perturbations of the scalar curvature generated by a Newtonian metric perturbation, which is related to the standard Bardeen potentials.
However, the considered model is certainly oversimplified to cover any real situation and is not gauge invariant.

Within this model it was possible to recover an almost-Harrison--Zel'dovich power spectrum for the considered metric perturbation.
Furthermore, the amplitude of such a power spectrum depends on the field mass which is a free parameter in our model and can be fixed independently of $H$.
At the same time, since it does not depend on the Hubble parameter of the background metric, this indicates that it is not a special feature of de Sitter space.
At least close to the initial singularity, the obtained result depends only on the form of the most singular part of the two-point function of the considered Bunch--Davis state.
We thus argue that a similar feature is present in every Hadamard state and for backgrounds which are only asymptotically de Sitter in the past.

Finally we notice that, since the stress-energy tensor is not linear in the field, its probability distribution cannot be of Gaussian nature.
Thus we showed that non-Gaussianities arise naturally in this picture.

%% file: conclusions.tex
%!TEX root = master.tex

\chapter*{Conclusions}
\addcontentsline{toc}{chapter}{Conclusions}
\label{cha:conclusions}

In this thesis aspects of the backreaction of quantum matter fields on the curvature of spacetime were discussed.
The main results in this direction were discussed in \cref{cha:einstein_solutions,cha:fluctuations}: the existence of local and global solutions of the semiclassical Einstein equation on cosmological spacetimes, and the a coupling the fluctuations of the quantum matter field to a Newtonianly perturbed de Sitter spacetime.
Further results presented are the enumerative combinatorics of the run structure of permutations in \cref{cha:combinatorics} with applications to the moment problem of the Wick square and the stress-energy on Minkowski spacetime.

In each case the problems were not treated in all possible generality, mainly due to the difficulty of constructing Hadamard states on general globally hyperbolic spacetimes but also due to other factors.
Nevertheless, we studied the effects of quantum fields on cosmological spacetimes not only because of their relative simplicity but also because of the relevance in cosmology.
Therefore the first avenue is not always the generalization of results to more general spacetimes, but also the better understanding of possible effects on this restricted class of spacetimes.
For example, we already mentioned in \cref{sec:solutions_outlook} that the results are restricted to the conformally coupled scalar field with a certain choice for the renormalization freedom as other choices can lead to equations involving higher than second-order derivatives of the metric and ask for slightly different approach.
However, it would be desirable to understand this problem also for non-conformal coupling and discuss the full dependence of the solutions to the Einstein equation on the renormalization freedom.

In the case of results on the metric fluctuations induced by quantum matter fluctuations as presented in \cref{cha:fluctuations}, we were even more restrictive and the discussion is mostly based on the special case of a Newtonianly perturbed de Sitter spacetime.
While straightforward generalization of this idea to asymptotically de Sitter spacetimes are possible and were already published in~\cite{dappiaggi:2014}, the next step should be to gain a clearer physical and mathematical motivation of the used equations.
A development in this direction is \cite{drago:2014}, but also this work should only be seen as a first step.
In any case, as soon as one attempts to take into account the fluctuations of the stress-energy tensor one is faced with the limitations of the semiclassical Einstein equation and any attempts to generalize them, even if well-motivated, remains speculative in the absence of an accepted theory of quantum gravity.

On one hand, when attempting to study quantum field theory in a mathematically rigorous fashion one sees even clearer the non-uniqueness of many constructions and one is confronted with many choices:
Are all Hadamard states physically sensible? What topology should be chosen for the algebra of quantum fields? What is the appropriate gauge freedom for the electromagnetic potential on non-contractible spacetimes? Is it reasonable to work with an algebra of unbouded operators such as the field algebra or should one always use a $C^*$-algebra?
There are many more questions of this kind and they require further mathematical and physical investigations but also intuition.
Quantum field theory and the quest for a quantum gravity is and will continue to be not only a research effort and playground of physicists but also one of mathematicians.

On the other hand, many aspects of quantum field theory are now conceptually and mathematically very well understood but only a few models have been studied in all their detail.
In particular interacting quantum fields on curved spacetimes have received relatively little attention given that already free fields are a complicated matter.
Investigations of physically interesting interacting models, using perturbative techniques, are largely absent from the literature and deserve more attention.

For these reasons one can expect that the field of quantum field theory (on curved spacetimes) will remain an interesting field of study for many more years to come.

%% file: acknowledgements.tex
%!TEX root = master.tex

\chapter*{Acknowledgements}
\addcontentsline{toc}{chapter}{Acknowledgements}
\label{cha:acknowledgments}

First of all, I would like to thank the University of Genova for providing me with the opportunity to carry out this work.
I am especially indebted to my amazing supervisor Nicola Pinamonti, whom I could always seek for advice and guidance, also on fields unrelated to science; only the collaboration with him made much of this thesis possible.

Second, my thanks go to the University of York for their hospitality during a 6 months stay (September 2013 -- March 2014) to collaborate with Chris Fewster.
In many enlightening discussions with him I could broaden my horizon further and complete parts of this work.

Third, I would like to thank Valter Moretti for agreeing to be my examiner and thereby taking upon him the task to read this long thesis.

Furthermore, I would like to thank my friends and colleagues from Genova, York and the Local Quantum Physics community.
Not only am I very grateful for their friendship, unforgettable shared memories and good times far from home, but also countless inspiring conversations on science and other topics.

Last, but not least, my gratitude also goes to my family, especially my parents, who always supported me, even if the do not understand what I am doing.

%% file: compiled-index.tex
\begin{theindex}
{\bfseries\sffamily *}\nopagebreak

  \item ${}^*$-algebra\cftdotfill{\cftsectiondotsep} \hyperpage{59--62}
    \subitem category of $\sim $\cftdotfill{\cftsectiondotsep} 
		\hyperpage{107}
    \subitem homomorphism between $\sim $\cftdotfill{\cftsectiondotsep} 
		\hyperpage{59}
    \subitem ideal of $\sim $\cftdotfill{\cftsectiondotsep} 
		\hyperpage{59}
    \subitem involution of a $\sim $\cftdotfill{\cftsectiondotsep} 
		\hyperpage{59}
    \subitem representation of a $\sim $\cftdotfill{\cftsectiondotsep} 
		\hyperpage{60}
    \subitem simple $\sim $\cftdotfill{\cftsectiondotsep} 
		\hyperpage{62}
    \subitem unital $\sim $\cftdotfill{\cftsectiondotsep} 
		\hyperpage{59}
    \subitem Weyl $\sim $\cftdotfill{\cftsectiondotsep} \hyperpage{61}
  \item ${}^*$-homomorphism\cftdotfill{\cftsectiondotsep} 
		\hyperpage{59}
  \item ${}^*$-ideal\cftdotfill{\cftsectiondotsep} \hyperpage{59}
  \item ${}^*$-representation\cftdotfill{\cftsectiondotsep} 
		\hyperpage{60}
    \subitem closed $\sim $\cftdotfill{\cftsectiondotsep} 
		\hyperpage{60}
    \subitem cyclic $\sim $\cftdotfill{\cftsectiondotsep} 
		\hyperpage{60}
    \subitem regular $\sim $\cftdotfill{\cftsectiondotsep} 
		\hyperpage{62}
    \subitem strongly cyclic $\sim $\cftdotfill{\cftsectiondotsep} 
		\hyperpage{60}

  \indexspace
{\bfseries\sffamily A}\nopagebreak

  \item absolutely summable sequence\cftdotfill{\cftsectiondotsep} 
		\hyperpage{58}
  \item abstract index notation\cftdotfill{\cftsectiondotsep} 
		\hyperpage{19}
  \item abstract kernel theorem\cftdotfill{\cftsectiondotsep} 
		\hyperpage{58}
  \item acceleration\cftdotfill{\cftsectiondotsep} \hyperpage{38}
  \item adiabatic regularization\cftdotfill{\cftsectiondotsep} 
		\hyperpage{136--137}, \hyperpage{142--146}
  \item adiabatic state\cftdotfill{\cftsectiondotsep} \hyperpage{119}, 
		\hyperpage{122--127}
    \subitem $\sim $ of order zero\cftdotfill{\cftsectiondotsep} 
		\hyperpage{124}, \hyperpage{141}
  \item adjoint element\cftdotfill{\cftsectiondotsep} \hyperpage{59}
  \item advanced propagator\cftdotfill{\cftsectiondotsep} 
		\hyperpage{79}
  \item algebra of observables\cftdotfill{\cftsectiondotsep} 
		\hyperpage{107}
  \item algebraic quantum field theory\cftdotfill{\cftsectiondotsep} 
		\hyperpage{105}
  \item anisotropic stress\cftdotfill{\cftsectiondotsep} \hyperpage{41}
  \item anti-de Sitter spacetime\cftdotfill{\cftsectiondotsep} 
		\hyperpage{43}
  \item antisymmetric tensor product\cftdotfill{\cftsectiondotsep} 
		\hyperpage{14}
  \item ascending permutation\cftdotfill{\cftsectiondotsep} 
		\hyperpage{85}
  \item atlas\cftdotfill{\cftsectiondotsep} \hyperpage{8}
    \subitem maximal $\sim $\cftdotfill{\cftsectiondotsep} 
		\hyperpage{8}
    \subitem smooth $\sim $\cftdotfill{\cftsectiondotsep} \hyperpage{8}
  \item atom of a permutation\cftdotfill{\cftsectiondotsep} 
		\hyperpage{85}
    \subitem principal $\sim $\cftdotfill{\cftsectiondotsep} 
		\hyperpage{85}
  \item atomic permutation\cftdotfill{\cftsectiondotsep} \hyperpage{85}
    \subitem falling $\sim $\cftdotfill{\cftsectiondotsep} 
		\hyperpage{85}
    \subitem rising $\sim $\cftdotfill{\cftsectiondotsep} 
		\hyperpage{85}
    \subitem run structure of $\sim $\cftdotfill{\cftsectiondotsep} 
		\hyperpage{87--91}
  \item auto-parallel curve\cftdotfill{\cftsectiondotsep} 
		\see{geodesic}{title-1}
  \item Avramidi method\cftdotfill{\cftsectiondotsep} \hyperpage{30}
  \item axiom
    \subitem $\sim $ of causality\cftdotfill{\cftsectiondotsep} 
		\hyperpage{105}, \hyperpage{107}
    \subitem $\sim $ of covariance\cftdotfill{\cftsectiondotsep} 
		\hyperpage{106}
    \subitem $\sim $ of locality\cftdotfill{\cftsectiondotsep} 
		\hyperpage{105, 106}
    \subitem timeslice $\sim $\cftdotfill{\cftsectiondotsep} 
		\hyperpage{107}

  \indexspace
{\bfseries\sffamily B}\nopagebreak

  \item background spacetime\cftdotfill{\cftsectiondotsep} 
		\hyperpage{47}, \hyperpage{155}
  \item Banach fixed-point theorem\cftdotfill{\cftsectiondotsep} 
		\hyperpage{65}, \hyperpage{148}
  \item Banach space\cftdotfill{\cftsectiondotsep} \hyperpage{55}
    \subitem dual of a $\sim $\cftdotfill{\cftsectiondotsep} 
		\hyperpage{57}
  \item Bardeen potentials\cftdotfill{\cftsectiondotsep} \hyperpage{50}
  \item barotropic fluid\cftdotfill{\cftsectiondotsep} \hyperpage{41}
  \item base space\cftdotfill{\cftsectiondotsep} \hyperpage{9}
  \item basis, topological\cftdotfill{\cftsectiondotsep} \hyperpage{52}
    \subitem local\cftdotfill{\cftsectiondotsep} \hyperpage{52}
  \item Bianchi identity
    \subitem first $\sim $\cftdotfill{\cftsectiondotsep} \hyperpage{21}
    \subitem second $\sim $\cftdotfill{\cftsectiondotsep} 
		\hyperpage{21}
  \item bicharacteristic\cftdotfill{\cftsectiondotsep} \hyperpage{78}
    \subitem $\sim $ strip\cftdotfill{\cftsectiondotsep} \hyperpage{78}
  \item biscalar\cftdotfill{\cftsectiondotsep} 
		\see{bitensor field}{title-1}
  \item bitensor field\cftdotfill{\cftsectiondotsep} \hyperpage{15}, 
		\hyperpage{27--33}
  \item Bogoliubov transformation\cftdotfill{\cftsectiondotsep} 
		\hyperpage{121}
  \item Borchers--Uhlmann algebra\cftdotfill{\cftsectiondotsep} 
		\hyperpage{108}
  \item bounded pseudometric space\cftdotfill{\cftsectiondotsep} 
		\hyperpage{54}
  \item Bunch--Davies state\cftdotfill{\cftsectiondotsep} 
		\hyperpage{119}
  \item bundle
    \subitem $\sim $ homomorphism\cftdotfill{\cftsectiondotsep} 
		\hyperpage{9}
    \subitem $\sim $ projection\cftdotfill{\cftsectiondotsep} 
		\hyperpage{9}

  \indexspace
{\bfseries\sffamily C}\nopagebreak

  \item $C^*$-algebra\cftdotfill{\cftsectiondotsep} \hyperpage{59}
    \subitem Weyl $\sim $\cftdotfill{\cftsectiondotsep} \hyperpage{61}
  \item $C^*$-norm\cftdotfill{\cftsectiondotsep} \hyperpage{59}
  \item canonical bilinear form\cftdotfill{\cftsectiondotsep} 
		\hyperpage{56}
  \item category
    \subitem of ${}^*$-algebras\cftdotfill{\cftsectiondotsep} 
		\hyperpage{107}
    \subitem of globally hyperbolic spacetimes\cftdotfill{\cftsectiondotsep} 
		\hyperpage{106}
  \item Cauchy embedding\cftdotfill{\cftsectiondotsep} \hyperpage{106}
  \item Cauchy problem\cftdotfill{\cftsectiondotsep} \hyperpage{81}
  \item Cauchy surface\cftdotfill{\cftsectiondotsep} \hyperpage{40}
  \item causal\cftdotfill{\cftsectiondotsep} \hyperpage{36}
    \subitem $\sim $ future\cftdotfill{\cftsectiondotsep} 
		\hyperpage{36}
    \subitem $\sim $ past\cftdotfill{\cftsectiondotsep} \hyperpage{36}
    \subitem $\sim $ shadow\cftdotfill{\cftsectiondotsep} 
		\hyperpage{36}
  \item causal propagator\cftdotfill{\cftsectiondotsep} \hyperpage{80}
  \item causal structure\cftdotfill{\cftsectiondotsep} 
		\hyperpage{36--37}
  \item causality\cftdotfill{\cftsectiondotsep} \hyperpage{36--40}
  \item causality condition\cftdotfill{\cftsectiondotsep} 
		\hyperpage{39}
  \item causally convex\cftdotfill{\cftsectiondotsep} \hyperpage{39}
  \item causally separating\cftdotfill{\cftsectiondotsep} 
		\hyperpage{36}
  \item characteristic set\cftdotfill{\cftsectiondotsep} \hyperpage{78}
  \item chart\cftdotfill{\cftsectiondotsep} \hyperpage{8}
    \subitem cosmological $\sim $ of de Sitter space\cftdotfill{\cftsectiondotsep} 
		\hyperpage{44}
    \subitem $\sim $ of de Sitter spacetime\cftdotfill{\cftsectiondotsep} 
		\hyperpage{43}
    \subitem $\sim $ of Euclidean space\cftdotfill{\cftsectiondotsep} 
		\hyperpage{16}
    \subitem $\sim $ of Minkowski spacetime\cftdotfill{\cftsectiondotsep} 
		\hyperpage{16}
    \subitem smoothly compatible $\sim $\cftdotfill{\cftsectiondotsep} 
		\hyperpage{8}
  \item Christoffel symbol\cftdotfill{\cftsectiondotsep} \hyperpage{20}
  \item chronological
    \subitem $\sim $ future\cftdotfill{\cftsectiondotsep} 
		\hyperpage{36}
    \subitem $\sim $ past\cftdotfill{\cftsectiondotsep} \hyperpage{36}
  \item chronological condition\cftdotfill{\cftsectiondotsep} 
		\hyperpage{39}
  \item circular permutation\cftdotfill{\cftsectiondotsep} 
		\hyperpage{84}
    \subitem $\sim $ group\cftdotfill{\cftsectiondotsep} \hyperpage{84}
    \subitem run structure of $\sim $\cftdotfill{\cftsectiondotsep} 
		\hyperpage{91--93}
  \item closed ${}^*$-representation\cftdotfill{\cftsectiondotsep} 
		\hyperpage{60}
  \item closed differential form\cftdotfill{\cftsectiondotsep} 
		\hyperpage{24}
  \item closed set\cftdotfill{\cftsectiondotsep} \hyperpage{52}
  \item closure of a set\cftdotfill{\cftsectiondotsep} \hyperpage{52}
  \item coarser topology\cftdotfill{\cftsectiondotsep} \hyperpage{52}
  \item coclosed differential form\cftdotfill{\cftsectiondotsep} 
		\hyperpage{25}
  \item codifferantial\cftdotfill{\cftsectiondotsep} \hyperpage{25}
  \item coexact differential form\cftdotfill{\cftsectiondotsep} 
		\hyperpage{25}
  \item coincidence limit\cftdotfill{\cftsectiondotsep} \hyperpage{27}
  \item combinatorics\cftdotfill{\cftsectiondotsep} \hyperpage{83--100}
  \item commutation coefficient\cftdotfill{\cftsectiondotsep} 
		\hyperpage{17}
  \item commutator\cftdotfill{\cftsectiondotsep} \hyperpage{110--112}, 
		\hyperpage{116}
  \item commutator distribution\cftdotfill{\cftsectiondotsep} 
		\hyperpage{80}, \hyperpage{110}
  \item comoving frame\cftdotfill{\cftsectiondotsep} \hyperpage{37}
  \item compact topological space\cftdotfill{\cftsectiondotsep} 
		\hyperpage{54}
  \item compact-open topology\cftdotfill{\cftsectiondotsep} 
		\hyperpage{56}
  \item compactly supported distribution\cftdotfill{\cftsectiondotsep} 
		\hyperpage{69}
    \subitem wavefront set in a cone\cftdotfill{\cftsectiondotsep} 
		\hyperpage{75}
  \item compactly supported distributional section\cftdotfill{\cftsectiondotsep} 
		\hyperpage{70}
  \item complete pseudometric space\cftdotfill{\cftsectiondotsep} 
		\hyperpage{54}
  \item complete vector field\cftdotfill{\cftsectiondotsep} 
		\hyperpage{13}
  \item components in a frame\cftdotfill{\cftsectiondotsep} 
		\hyperpage{17}
  \item concrete index notation\cftdotfill{\cftsectiondotsep} 
		\hyperpage{18}
  \item conformal
    \subitem $\sim $ embedding\cftdotfill{\cftsectiondotsep} 
		\hyperpage{16}
    \subitem $\sim $ immersion\cftdotfill{\cftsectiondotsep} 
		\hyperpage{16}
    \subitem $\sim $ isometry\cftdotfill{\cftsectiondotsep} 
		\hyperpage{16}
  \item conformal curvature coupling\cftdotfill{\cftsectiondotsep} 
		\hyperpage{112}
  \item conformal Hubble function\cftdotfill{\cftsectiondotsep} 
		\hyperpage{49}
  \item conformal Killing vector field\cftdotfill{\cftsectiondotsep} 
		\hyperpage{22}
  \item conformal time\cftdotfill{\cftsectiondotsep} \hyperpage{47}
    \subitem $\sim $ of de Sitter spacetime\cftdotfill{\cftsectiondotsep} 
		\hyperpage{44}
  \item connected set\cftdotfill{\cftsectiondotsep} \hyperpage{52}
  \item connection\cftdotfill{\cftsectiondotsep} \hyperpage{19}
    \subitem curvature of a $\sim $\cftdotfill{\cftsectiondotsep} 
		\hyperpage{21}
    \subitem dual of a $\sim $\cftdotfill{\cftsectiondotsep} 
		\hyperpage{20}
    \subitem flat $\sim $\cftdotfill{\cftsectiondotsep} \hyperpage{21}
    \subitem Levi-Civita $\sim $\cftdotfill{\cftsectiondotsep} 
		\hyperpage{20}
    \subitem metric $\sim $\cftdotfill{\cftsectiondotsep} 
		\hyperpage{19}
    \subitem pullback of a $\sim $\cftdotfill{\cftsectiondotsep} 
		\hyperpage{20}
    \subitem torsion-free $\sim $\cftdotfill{\cftsectiondotsep} 
		\hyperpage{19}
  \item conservation
    \subitem $\sim $ of energy\cftdotfill{\cftsectiondotsep} 
		\hyperpage{41}
    \subitem $\sim $ of momentum\cftdotfill{\cftsectiondotsep} 
		\hyperpage{41}
  \item continuous\cftdotfill{\cftsectiondotsep} \hyperpage{53}
    \subitem equicontinuous\cftdotfill{\cftsectiondotsep} 
		\hyperpage{56}
    \subitem uniformly $\sim $\cftdotfill{\cftsectiondotsep} 
		\hyperpage{54}
  \item contractible manifold\cftdotfill{\cftsectiondotsep} 
		\hyperpage{24}
  \item contraction\cftdotfill{\cftsectiondotsep} \hyperpage{15}
    \subitem $\sim $ with a metric\cftdotfill{\cftsectiondotsep} 
		\hyperpage{15}
  \item contravariant tensor field\cftdotfill{\cftsectiondotsep} 
		\hyperpage{14}
  \item convergence\cftdotfill{\cftsectiondotsep} \hyperpage{54}
    \subitem $\sim $ pointwise\cftdotfill{\cftsectiondotsep} 
		\hyperpage{55}
    \subitem $\sim $ uniform\cftdotfill{\cftsectiondotsep} 
		\hyperpage{56}
    \subitem uniform $\sim $ in compacta\cftdotfill{\cftsectiondotsep} 
		\hyperpage{56}
  \item convolution\cftdotfill{\cftsectiondotsep} \hyperpage{72}
  \item convolution theorem\cftdotfill{\cftsectiondotsep} 
		\hyperpage{73}
  \item coordinate
    \subitem $\sim $ neighbourhood\cftdotfill{\cftsectiondotsep} 
		\hyperpage{8}
    \subitem $\sim $ vector\cftdotfill{\cftsectiondotsep} 
		\hyperpage{17}
  \item cosmological chart\cftdotfill{\cftsectiondotsep} \hyperpage{44}
  \item cosmological constant\cftdotfill{\cftsectiondotsep} 
		\hyperpage{42}
  \item cosmological spacetime\cftdotfill{\cftsectiondotsep} 
		\see{FLRW spacetime}{title-1}
  \item cosmological time\cftdotfill{\cftsectiondotsep} \hyperpage{46}
  \item cosmology\cftdotfill{\cftsectiondotsep} \hyperpage{44--50}
    \subitem FLRW spacetime\cftdotfill{\cftsectiondotsep} 
		\hyperpage{45}
    \subitem gauge problem\cftdotfill{\cftsectiondotsep} \hyperpage{47}
    \subitem metric perturbation\cftdotfill{\cftsectiondotsep} 
		\hyperpage{49}
  \item cotangent bundle\cftdotfill{\cftsectiondotsep} \hyperpage{13}
    \subitem dual of the $\sim $\cftdotfill{\cftsectiondotsep} 
		\hyperpage{11}
  \item cotangent map\cftdotfill{\cftsectiondotsep} \hyperpage{13}
    \subitem global $\sim $\cftdotfill{\cftsectiondotsep} 
		\hyperpage{13}
  \item cotangent space\cftdotfill{\cftsectiondotsep} \hyperpage{13}
    \subitem dual of the $\sim $\cftdotfill{\cftsectiondotsep} 
		\hyperpage{11}
  \item countable
    \subitem first-$\sim $\cftdotfill{\cftsectiondotsep} \hyperpage{53}
    \subitem second-$\sim $\cftdotfill{\cftsectiondotsep} 
		\hyperpage{52}
  \item covariant derivative
    \subitem $\sim $ along a curve\cftdotfill{\cftsectiondotsep} 
		\hyperpage{21}
    \subitem $\sim $ along a vector field\cftdotfill{\cftsectiondotsep} 
		\hyperpage{19}
    \subitem spatial $\sim $\cftdotfill{\cftsectiondotsep} 
		\hyperpage{38}
    \subitem temporal $\sim $\cftdotfill{\cftsectiondotsep} 
		\hyperpage{38}
  \item covariant expansion\cftdotfill{\cftsectiondotsep} 
		\hyperpage{29}
  \item covariant splitting\cftdotfill{\cftsectiondotsep} 
		\hyperpage{37--39}
  \item covariant tensor field\cftdotfill{\cftsectiondotsep} 
		\hyperpage{14}
  \item covector\cftdotfill{\cftsectiondotsep} \hyperpage{13}
  \item covector field\cftdotfill{\cftsectiondotsep} \hyperpage{13}
    \subitem components of a $\sim $\cftdotfill{\cftsectiondotsep} 
		\hyperpage{17}
    \subitem pullback of a $\sim $\cftdotfill{\cftsectiondotsep} 
		\hyperpage{13}
    \subitem pushforward of a $\sim $\cftdotfill{\cftsectiondotsep} 
		\hyperpage{13}
  \item curvature coupling\cftdotfill{\cftsectiondotsep} 
		\hyperpage{112}
    \subitem conformal $\sim $\cftdotfill{\cftsectiondotsep} 
		\hyperpage{112}
    \subitem minimal $\sim $\cftdotfill{\cftsectiondotsep} 
		\hyperpage{112}
  \item curvature of a connection\cftdotfill{\cftsectiondotsep} 
		\hyperpage{21}
  \item curvature tensor
    \subitem Einstein $\sim $\cftdotfill{\cftsectiondotsep} 
		\hyperpage{42}
    \subitem Ricci $\sim $\cftdotfill{\cftsectiondotsep} \hyperpage{21}
    \subitem Riemann $\sim $\cftdotfill{\cftsectiondotsep} 
		\hyperpage{21}
  \item curve
    \subitem auto-parallel $\sim $\cftdotfill{\cftsectiondotsep} 
		\see{geodesic}{title-1}
    \subitem inextendible $\sim $\cftdotfill{\cftsectiondotsep} 
		\hyperpage{12}
    \subitem integral $\sim $\cftdotfill{\cftsectiondotsep} 
		\hyperpage{12}
    \subitem maximal integrable $\sim $\cftdotfill{\cftsectiondotsep} 
		\hyperpage{13}
    \subitem parametrized $\sim $\cftdotfill{\cftsectiondotsep} 
		\hyperpage{12}
    \subitem velocity of a $\sim $\cftdotfill{\cftsectiondotsep} 
		\hyperpage{12}
  \item cyclic ${}^*$-representation\cftdotfill{\cftsectiondotsep} 
		\hyperpage{60}
  \item cyclic order\cftdotfill{\cftsectiondotsep} \hyperpage{84}
  \item cyclic vector\cftdotfill{\cftsectiondotsep} \hyperpage{60}

  \indexspace
{\bfseries\sffamily D}\nopagebreak

  \item d'Alembert operator\cftdotfill{\cftsectiondotsep} 
		\hyperpage{26}
  \item de Rham cohomology\cftdotfill{\cftsectiondotsep} 
		\hyperpage{24, 25}
    \subitem $\sim $ with compact support\cftdotfill{\cftsectiondotsep} 
		\hyperpage{24, 25}
  \item de Sitter spacetime\cftdotfill{\cftsectiondotsep} 
		\hyperpage{43}
    \subitem chart of $\sim $\cftdotfill{\cftsectiondotsep} 
		\hyperpage{43}
    \subitem conformal time of $\sim $\cftdotfill{\cftsectiondotsep} 
		\hyperpage{44}
    \subitem cosmological chart of $\sim $\cftdotfill{\cftsectiondotsep} 
		\hyperpage{44}
    \subitem metric of $\sim $\cftdotfill{\cftsectiondotsep} 
		\hyperpage{43}
    \subitem world function of $\sim $\cftdotfill{\cftsectiondotsep} 
		\hyperpage{44}
  \item decomposition of a tensor field\cftdotfill{\cftsectiondotsep} 
		\hyperpage{48}
  \item $\delta $-distribution\cftdotfill{\cftsectiondotsep} 
		\hyperpage{77}
  \item dense set\cftdotfill{\cftsectiondotsep} \hyperpage{52}
  \item derivative
    \subitem covariant $\sim $ along a
      \subsubitem $\sim $ curve\cftdotfill{\cftsectiondotsep} 
		\hyperpage{21}
      \subsubitem $\sim $ vector field\cftdotfill{\cftsectiondotsep} 
		\hyperpage{19}
    \subitem directional $\sim $\cftdotfill{\cftsectiondotsep} 
		\hyperpage{62}
    \subitem exterior $\sim $\cftdotfill{\cftsectiondotsep} 
		\hyperpage{23}
    \subitem Fr\'echet $\sim $\cftdotfill{\cftsectiondotsep} 
		\hyperpage{64}
    \subitem G\^ateaux $\sim $\cftdotfill{\cftsectiondotsep} 
		\see{direct. deriv.}{title-1}
    \subitem Lie $\sim $\cftdotfill{\cftsectiondotsep} \hyperpage{11}, 
		\hyperpage{14}
    \subitem spatial covariant $\sim $\cftdotfill{\cftsectiondotsep} 
		\hyperpage{38}
    \subitem temporal covariant $\sim $\cftdotfill{\cftsectiondotsep} 
		\hyperpage{38}
  \item descending permutation\cftdotfill{\cftsectiondotsep} 
		\hyperpage{85}
  \item descent set\cftdotfill{\cftsectiondotsep} \hyperpage{86}
  \item diagonal distribution\cftdotfill{\cftsectiondotsep} 
		\hyperpage{77}
  \item diffeomorphism\cftdotfill{\cftsectiondotsep} \hyperpage{9}
  \item differential form\cftdotfill{\cftsectiondotsep} \hyperpage{23}
    \subitem closed $\sim $\cftdotfill{\cftsectiondotsep} 
		\hyperpage{24}
    \subitem coclosed $\sim $\cftdotfill{\cftsectiondotsep} 
		\hyperpage{25}
    \subitem coexact $\sim $\cftdotfill{\cftsectiondotsep} 
		\hyperpage{25}
    \subitem exact $\sim $\cftdotfill{\cftsectiondotsep} \hyperpage{24}
  \item differential geometry\cftdotfill{\cftsectiondotsep} 
		\hyperpage{7--33}
  \item differential of a smooth map\cftdotfill{\cftsectiondotsep} 
		\hyperpage{11}
  \item differential operator\cftdotfill{\cftsectiondotsep} 
		\hyperpage{17}
    \subitem formally adjoint $\sim $\cftdotfill{\cftsectiondotsep} 
		\hyperpage{25}, \hyperpage{79}
    \subitem formally self-adjoint $\sim $\cftdotfill{\cftsectiondotsep} 
		\hyperpage{80}
    \subitem Green-hyperbolic $\sim $\cftdotfill{\cftsectiondotsep} 
		\hyperpage{79}
    \subitem normally hyperbolic $\sim $\cftdotfill{\cftsectiondotsep} 
		\hyperpage{18}
    \subitem pre-normally hyperbolic $\sim $\cftdotfill{\cftsectiondotsep} 
		\hyperpage{18}
    \subitem principal symbol of a $\sim $\cftdotfill{\cftsectiondotsep} 
		\hyperpage{18}
    \subitem total symbol of a $\sim $\cftdotfill{\cftsectiondotsep} 
		\hyperpage{18}
  \item Dirac $\delta $-distribution\cftdotfill{\cftsectiondotsep} 
		\hyperpage{77}
  \item direct sum topology\cftdotfill{\cftsectiondotsep} 
		\hyperpage{53}
  \item direction
    \subitem regular $\sim $\cftdotfill{\cftsectiondotsep} 
		\hyperpage{73}
    \subitem singular $\sim $\cftdotfill{\cftsectiondotsep} 
		\hyperpage{74}
  \item directional derivative\cftdotfill{\cftsectiondotsep} 
		\hyperpage{62}
  \item discrete topology\cftdotfill{\cftsectiondotsep} \hyperpage{52}
  \item distribution\cftdotfill{\cftsectiondotsep} \hyperpage{69}
    \subitem compactly supported $\sim $\cftdotfill{\cftsectiondotsep} 
		\hyperpage{69}
    \subitem kernel of a $\sim $\cftdotfill{\cftsectiondotsep} 
		\hyperpage{71}
    \subitem multiplication of $\sim $\cftdotfill{\cftsectiondotsep} 
		\hyperpage{76}
    \subitem pullback of a $\sim $\cftdotfill{\cftsectiondotsep} 
		\hyperpage{75}
    \subitem restriction of a $\sim $\cftdotfill{\cftsectiondotsep} 
		\hyperpage{69}
    \subitem Schwartz $\sim $\cftdotfill{\cftsectiondotsep} 
		\see{tempered distribution}{title-1}
    \subitem smooth $\sim $\cftdotfill{\cftsectiondotsep} 
		\hyperpage{73}
    \subitem support of a $\sim $\cftdotfill{\cftsectiondotsep} 
		\hyperpage{69}
    \subitem tempered $\sim $\cftdotfill{\cftsectiondotsep} 
		\hyperpage{69}
    \subitem wavefront set in a cone\cftdotfill{\cftsectiondotsep} 
		\hyperpage{74}
    \subitem wavefront set of a $\sim $\cftdotfill{\cftsectiondotsep} 
		\hyperpage{74}
  \item distributional section\cftdotfill{\cftsectiondotsep} 
		\hyperpage{70}
    \subitem compactly supported $\sim $\cftdotfill{\cftsectiondotsep} 
		\hyperpage{70}
    \subitem restriction of a $\sim $\cftdotfill{\cftsectiondotsep} 
		\hyperpage{70}
    \subitem wavefront set of a $\sim $\cftdotfill{\cftsectiondotsep} 
		\hyperpage{76}
  \item dual
    \subitem $\sim $ frame\cftdotfill{\cftsectiondotsep} \hyperpage{16}
    \subitem $\sim $ metric\cftdotfill{\cftsectiondotsep} 
		\hyperpage{15}
    \subitem $\sim $ of a Banach space\cftdotfill{\cftsectiondotsep} 
		\hyperpage{57}
    \subitem $\sim $ of a connection\cftdotfill{\cftsectiondotsep} 
		\hyperpage{20}
    \subitem $\sim $ of a Fr\'echet space\cftdotfill{\cftsectiondotsep} 
		\hyperpage{57}
    \subitem $\sim $ of a vector bundle\cftdotfill{\cftsectiondotsep} 
		\hyperpage{15}
    \subitem $\sim $ of the cotangent bundle\cftdotfill{\cftsectiondotsep} 
		\hyperpage{11}
    \subitem $\sim $ of the cotangent space\cftdotfill{\cftsectiondotsep} 
		\hyperpage{11}
    \subitem $\sim $ of the tangent bundle\cftdotfill{\cftsectiondotsep} 
		\hyperpage{13}
    \subitem $\sim $ of the tangent space\cftdotfill{\cftsectiondotsep} 
		\hyperpage{13}
    \subitem $\sim $ seminorm\cftdotfill{\cftsectiondotsep} 
		\hyperpage{57}
    \subitem strong $\sim $\cftdotfill{\cftsectiondotsep} 
		\hyperpage{57}
    \subitem topological $\sim $\cftdotfill{\cftsectiondotsep} 
		\hyperpage{56}
    \subitem weak $\sim $\cftdotfill{\cftsectiondotsep} \hyperpage{57}
  \item duality\cftdotfill{\cftsectiondotsep} \hyperpage{56}
    \subitem distributions and test functions\cftdotfill{\cftsectiondotsep} 
		\hyperpage{73}
    \subitem dual vector bundle\cftdotfill{\cftsectiondotsep} 
		\hyperpage{15}
    \subitem musical isomorphisms\cftdotfill{\cftsectiondotsep} 
		\hyperpage{15}
    \subitem Poincar\'e $\sim $\cftdotfill{\cftsectiondotsep} 
		\hyperpage{26}

  \indexspace
{\bfseries\sffamily E}\nopagebreak

  \item Einstein curvature tensor\cftdotfill{\cftsectiondotsep} 
		\hyperpage{42}
  \item Einstein equation\cftdotfill{\cftsectiondotsep} \hyperpage{42}
    \subitem semiclassical $\sim $\cftdotfill{\cftsectiondotsep} 
		\hyperpage{131--137}
    \subitem vacuum solution of $\sim $\cftdotfill{\cftsectiondotsep} 
		\hyperpage{43}
  \item element of a ${}^*$-algebra
    \subitem adjoint $\sim $\cftdotfill{\cftsectiondotsep} 
		\hyperpage{59}
    \subitem normal $\sim $\cftdotfill{\cftsectiondotsep} 
		\hyperpage{59}
    \subitem self-adjoint $\sim $\cftdotfill{\cftsectiondotsep} 
		\hyperpage{59}
    \subitem unitary $\sim $\cftdotfill{\cftsectiondotsep} 
		\hyperpage{59}
  \item embedding\cftdotfill{\cftsectiondotsep} \hyperpage{12}
    \subitem Cauchy $\sim $\cftdotfill{\cftsectiondotsep} 
		\hyperpage{106}
    \subitem conformal $\sim $\cftdotfill{\cftsectiondotsep} 
		\hyperpage{16}
    \subitem hyperbolic-$\sim $\cftdotfill{\cftsectiondotsep} 
		\hyperpage{106}
    \subitem isometric $\sim $\cftdotfill{\cftsectiondotsep} 
		\hyperpage{16}
    \subitem $\sim $ of a submanifold\cftdotfill{\cftsectiondotsep} 
		\hyperpage{12}
  \item energy condition\cftdotfill{\cftsectiondotsep} \hyperpage{41}
    \subitem dominated (DEC)\cftdotfill{\cftsectiondotsep} 
		\hyperpage{41}
    \subitem null (NEC)\cftdotfill{\cftsectiondotsep} \hyperpage{41}
    \subitem strong (SEC)\cftdotfill{\cftsectiondotsep} \hyperpage{42}
    \subitem weak (WEC)\cftdotfill{\cftsectiondotsep} \hyperpage{41}
  \item energy conservation\cftdotfill{\cftsectiondotsep} 
		\hyperpage{41}
  \item energy density\cftdotfill{\cftsectiondotsep} \hyperpage{41}
  \item energy-density
    \subitem $\sim $ per mode\cftdotfill{\cftsectiondotsep} 
		\hyperpage{125}
    \subitem quantum $\sim $\cftdotfill{\cftsectiondotsep} 
		\hyperpage{135}
  \item enumerative combinatorics\cftdotfill{\cftsectiondotsep} 
		\hyperpage{83--100}
  \item equicontinuous\cftdotfill{\cftsectiondotsep} \hyperpage{56}
  \item Euclidean space\cftdotfill{\cftsectiondotsep} \hyperpage{16}
    \subitem chart of $\sim $\cftdotfill{\cftsectiondotsep} 
		\hyperpage{16}
    \subitem metric of $\sim $\cftdotfill{\cftsectiondotsep} 
		\hyperpage{16}
  \item event\cftdotfill{\cftsectiondotsep} \hyperpage{37}
  \item exact differential form\cftdotfill{\cftsectiondotsep} 
		\hyperpage{24}
  \item existence of a fixed-point\cftdotfill{\cftsectiondotsep} 
		\hyperpage{66}
  \item expansion scalar\cftdotfill{\cftsectiondotsep} \hyperpage{38}
  \item expansion tensor\cftdotfill{\cftsectiondotsep} \hyperpage{38}
  \item exponential map\cftdotfill{\cftsectiondotsep} \hyperpage{22}
  \item exterior derivative\cftdotfill{\cftsectiondotsep} 
		\hyperpage{23}
  \item exterior tensor product\cftdotfill{\cftsectiondotsep} 
		\hyperpage{15}

  \indexspace
{\bfseries\sffamily F}\nopagebreak

  \item falling atomic permutation\cftdotfill{\cftsectiondotsep} 
		\hyperpage{85}
  \item fibre of a vector bundle\cftdotfill{\cftsectiondotsep} 
		\hyperpage{9}
  \item field algebra
    \subitem off-shell $\sim $\cftdotfill{\cftsectiondotsep} 
		\hyperpage{109}
    \subitem on-shell $\sim $\cftdotfill{\cftsectiondotsep} 
		\hyperpage{110}
  \item field equation\cftdotfill{\cftsectiondotsep} \hyperpage{109}, 
		\hyperpage{116}
    \subitem Klein--Gordon $\sim $\cftdotfill{\cftsectiondotsep} 
		\hyperpage{112}
    \subitem Proca $\sim $\cftdotfill{\cftsectiondotsep} 
		\hyperpage{113}
  \item field operator\cftdotfill{\cftsectiondotsep} \hyperpage{62}
  \item final topology\cftdotfill{\cftsectiondotsep} \hyperpage{53}
  \item finer topology\cftdotfill{\cftsectiondotsep} \hyperpage{52}
  \item first Bianchi identity\cftdotfill{\cftsectiondotsep} 
		\hyperpage{21}
  \item first Friedmann equation\cftdotfill{\cftsectiondotsep} 
		\hyperpage{46}, \hyperpage{134}
  \item first-countable\cftdotfill{\cftsectiondotsep} \hyperpage{53}
  \item fixed-point
    \subitem existence of a $\sim $\cftdotfill{\cftsectiondotsep} 
		\hyperpage{66}
    \subitem regularity of a $\sim $\cftdotfill{\cftsectiondotsep} 
		\hyperpage{66}
    \subitem uniqueness of a $\sim $\cftdotfill{\cftsectiondotsep} 
		\hyperpage{66}
  \item fixed-point theorem\cftdotfill{\cftsectiondotsep} 
		\hyperpage{65--68}
    \subitem Banach $\sim $\cftdotfill{\cftsectiondotsep} 
		\hyperpage{65}, \hyperpage{148}
    \subitem Lipschitz continuity criterium\cftdotfill{\cftsectiondotsep} 
		\hyperpage{67--68}, \hyperpage{148}
  \item flat connection\cftdotfill{\cftsectiondotsep} \hyperpage{21}
  \item flat $\flat $\cftdotfill{\cftsectiondotsep} \hyperpage{15}
  \item flow of a vector field\cftdotfill{\cftsectiondotsep} 
		\hyperpage{13}
    \subitem local $\sim $\cftdotfill{\cftsectiondotsep} \hyperpage{13}
  \item FLRW spacetime\cftdotfill{\cftsectiondotsep} \hyperpage{45}
    \subitem conformal metric of a $\sim $\cftdotfill{\cftsectiondotsep} 
		\hyperpage{47}
    \subitem conformal time of a $\sim $\cftdotfill{\cftsectiondotsep} 
		\hyperpage{47}
    \subitem cosmological time of a $\sim $\cftdotfill{\cftsectiondotsep} 
		\hyperpage{46}
    \subitem metric of a $\sim $\cftdotfill{\cftsectiondotsep} 
		\hyperpage{47}
    \subitem perturbed metric of a $\sim $\cftdotfill{\cftsectiondotsep} 
		\hyperpage{49}
    \subitem world function of $\sim $\cftdotfill{\cftsectiondotsep} 
		\hyperpage{32}
  \item fluid
    \subitem barotropic $\sim $\cftdotfill{\cftsectiondotsep} 
		\hyperpage{41}
    \subitem perfect $\sim $\cftdotfill{\cftsectiondotsep} 
		\hyperpage{41}
  \item foliation\cftdotfill{\cftsectiondotsep} \hyperpage{27}
    \subitem leaves of a $\sim $\cftdotfill{\cftsectiondotsep} 
		\hyperpage{27}
  \item formally
    \subitem $\sim $ adjoint\cftdotfill{\cftsectiondotsep} 
		\hyperpage{25}, \hyperpage{79}
    \subitem $\sim $ self-adjoint\cftdotfill{\cftsectiondotsep} 
		\hyperpage{80}
  \item Fourier transform\cftdotfill{\cftsectiondotsep} \hyperpage{71}
    \subitem convolution theorem\cftdotfill{\cftsectiondotsep} 
		\hyperpage{73}
    \subitem inverse $\sim $\cftdotfill{\cftsectiondotsep} 
		\hyperpage{72}
    \subitem Plancherel--Parseval identity\cftdotfill{\cftsectiondotsep} 
		\hyperpage{72}
  \item Fr\'chet space
    \subitem smooth sections\cftdotfill{\cftsectiondotsep} 
		\hyperpage{69}
  \item Fr\'echet derivative\cftdotfill{\cftsectiondotsep} 
		\hyperpage{64}
  \item Fr\'echet space\cftdotfill{\cftsectiondotsep} \hyperpage{55}, 
		\hyperpage{64}
    \subitem dual of a $\sim $\cftdotfill{\cftsectiondotsep} 
		\hyperpage{57}
    \subitem rapidly decreasing functions\cftdotfill{\cftsectiondotsep} 
		\hyperpage{69}
    \subitem smooth functions\cftdotfill{\cftsectiondotsep} 
		\hyperpage{68}
  \item frame
    \subitem comoving $\sim $\cftdotfill{\cftsectiondotsep} 
		\hyperpage{37}
    \subitem components in a $\sim $\cftdotfill{\cftsectiondotsep} 
		\hyperpage{17}
    \subitem dual $\sim $\cftdotfill{\cftsectiondotsep} \hyperpage{16}
    \subitem local $\sim $\cftdotfill{\cftsectiondotsep} \hyperpage{16}
    \subitem orthogonal $\sim $\cftdotfill{\cftsectiondotsep} 
		\hyperpage{17}
    \subitem unitary $\sim $\cftdotfill{\cftsectiondotsep} 
		\hyperpage{17}
  \item Friedmann equation
    \subitem first $\sim $\cftdotfill{\cftsectiondotsep} \hyperpage{46}, 
		\hyperpage{134}
    \subitem second $\sim $\cftdotfill{\cftsectiondotsep} 
		\hyperpage{46}, \hyperpage{134}
    \subitem semiclassical $\sim $\cftdotfill{\cftsectiondotsep} 
		\hyperpage{134--137}, \hyperpage{140}
  \item Frobenius' theorem\cftdotfill{\cftsectiondotsep} \hyperpage{26}
    \subitem global $\sim $\cftdotfill{\cftsectiondotsep} 
		\hyperpage{27}
  \item funct
    \subitem topologies\cftdotfill{\cftsectiondotsep} 
		\hyperpage{55--56}
  \item function
    \subitem $\sim $ on a manifold\cftdotfill{\cftsectiondotsep} 
		\hyperpage{9}
    \subitem rapidly decreasing $\sim $\cftdotfill{\cftsectiondotsep} 
		\hyperpage{69}
    \subitem space of $\sim $\cftdotfill{\cftsectiondotsep} 
		\hyperpage{55}
    \subitem space of continuous $\sim $\cftdotfill{\cftsectiondotsep} 
		\hyperpage{53}
    \subitem test $\sim $\cftdotfill{\cftsectiondotsep} \hyperpage{69}
  \item future
    \subitem causal $\sim $\cftdotfill{\cftsectiondotsep} 
		\hyperpage{36}
    \subitem chronological $\sim $\cftdotfill{\cftsectiondotsep} 
		\hyperpage{36}

  \indexspace
{\bfseries\sffamily G}\nopagebreak

  \item G\^ateaux derivative\cftdotfill{\cftsectiondotsep} 
		\see{direct. deriv.}{title-1}
  \item gauge invariant tensor field\cftdotfill{\cftsectiondotsep} 
		\hyperpage{47}
  \item gauge problem in cosmology\cftdotfill{\cftsectiondotsep} 
		\hyperpage{47}
  \item Gauss--Codacci equation\cftdotfill{\cftsectiondotsep} 
		\hyperpage{39}
  \item Gaussian state\cftdotfill{\cftsectiondotsep} 
		\see{quasi-free state}{title-1}
  \item general relativity\cftdotfill{\cftsectiondotsep} 
		\hyperpage{40--44}
  \item geodesic\cftdotfill{\cftsectiondotsep} \hyperpage{22}
  \item geodesic distance\cftdotfill{\cftsectiondotsep} \hyperpage{27}
  \item geodesically
    \subitem $\sim $ complete\cftdotfill{\cftsectiondotsep} 
		\hyperpage{22}
    \subitem $\sim $ convex\cftdotfill{\cftsectiondotsep} 
		\hyperpage{22}
    \subitem $\sim $ starshaped\cftdotfill{\cftsectiondotsep} 
		\hyperpage{22}
  \item geometry
    \subitem differential $\sim $\cftdotfill{\cftsectiondotsep} 
		\hyperpage{7--33}
    \subitem Lorentzian $\sim $\cftdotfill{\cftsectiondotsep} 
		\hyperpage{35--50}
  \item global cotangent map\cftdotfill{\cftsectiondotsep} 
		\hyperpage{13}
  \item global tangent map\cftdotfill{\cftsectiondotsep} \hyperpage{11}
  \item global trivialization\cftdotfill{\cftsectiondotsep} 
		\hyperpage{9}
  \item globally hyperbolic spacetime\cftdotfill{\cftsectiondotsep} 
		\hyperpage{40}
    \subitem category of $\sim $\cftdotfill{\cftsectiondotsep} 
		\hyperpage{106}
  \item GNS construction\cftdotfill{\cftsectiondotsep} \hyperpage{60}
  \item graph topology\cftdotfill{\cftsectiondotsep} \hyperpage{60}
  \item Green's operator
    \subitem advanced $\sim $\cftdotfill{\cftsectiondotsep} 
		\see{advanced prop.}{title-1}
    \subitem retarded $\sim $\cftdotfill{\cftsectiondotsep} 
		\see{retarded prop.}{title-1}
  \item Green-hyperbolic\cftdotfill{\cftsectiondotsep} \hyperpage{79}

  \indexspace
{\bfseries\sffamily H}\nopagebreak

  \item H\"ormander topology\cftdotfill{\cftsectiondotsep} 
		\hyperpage{74}
  \item Hadamard
    \subitem $\sim $ coefficients\cftdotfill{\cftsectiondotsep} 
		\hyperpage{118}, \hyperpage{134, 135}
    \subitem $\sim $ form\cftdotfill{\cftsectiondotsep} \hyperpage{117}
    \subitem $\sim $ parametrix\cftdotfill{\cftsectiondotsep} 
		\hyperpage{118}
    \subitem $\sim $ point-splitting regularization\cftdotfill{\cftsectiondotsep} 
		\hyperpage{133}, \hyperpage{136--137}
    \subitem $\sim $ recursion relations\cftdotfill{\cftsectiondotsep} 
		\hyperpage{118}
    \subitem $\sim $ state\cftdotfill{\cftsectiondotsep} 
		\hyperpage{117}
  \item Hahn--Banach theorem\cftdotfill{\cftsectiondotsep} 
		\hyperpage{55}
  \item Harrison--Zel'dovich spectrum\cftdotfill{\cftsectiondotsep} 
		\hyperpage{154}
  \item Hausdorff space\cftdotfill{\cftsectiondotsep} \hyperpage{52}
  \item Hodge star\cftdotfill{\cftsectiondotsep} \hyperpage{25}
  \item homeomorphism\cftdotfill{\cftsectiondotsep} \hyperpage{53}
  \item homogeneous and isotropic
    \subitem $\sim $ spacetime\cftdotfill{\cftsectiondotsep} 
		\see{FLRW spacetime}{title-1}
    \subitem $\sim $ state\cftdotfill{\cftsectiondotsep} 
		\hyperpage{121}
  \item Hubble
    \subitem conformal $\sim $ function\cftdotfill{\cftsectiondotsep} 
		\hyperpage{49}
    \subitem $\sim $ constant\cftdotfill{\cftsectiondotsep} 
		\hyperpage{43}
    \subitem $\sim $ function\cftdotfill{\cftsectiondotsep} 
		\hyperpage{46}
  \item hyperbolic differential operator
    \subitem Green-$\sim $\cftdotfill{\cftsectiondotsep} \hyperpage{79}
    \subitem normally $\sim $\cftdotfill{\cftsectiondotsep} 
		\hyperpage{18}
    \subitem pre-normally $\sim $\cftdotfill{\cftsectiondotsep} 
		\hyperpage{18}
  \item hyperbolic-embedding\cftdotfill{\cftsectiondotsep} 
		\hyperpage{106}

  \indexspace
{\bfseries\sffamily I}\nopagebreak

  \item immersion\cftdotfill{\cftsectiondotsep} \hyperpage{12}
    \subitem conformal $\sim $\cftdotfill{\cftsectiondotsep} 
		\hyperpage{16}
    \subitem isometric $\sim $\cftdotfill{\cftsectiondotsep} 
		\hyperpage{16}
    \subitem $\sim $ of a submanifold\cftdotfill{\cftsectiondotsep} 
		\hyperpage{12}
  \item index notation
    \subitem abstract $\sim $\cftdotfill{\cftsectiondotsep} 
		\hyperpage{19}
    \subitem concrete $\sim $\cftdotfill{\cftsectiondotsep} 
		\hyperpage{18}
  \item index of a metric\cftdotfill{\cftsectiondotsep} \hyperpage{16}
  \item inextendible curve\cftdotfill{\cftsectiondotsep} \hyperpage{12}
  \item initial topology\cftdotfill{\cftsectiondotsep} \hyperpage{53}
  \item injective tensor product\cftdotfill{\cftsectiondotsep} 
		\hyperpage{58}
  \item integrable
    \subitem Lebesgue $\sim $\cftdotfill{\cftsectiondotsep} 
		\hyperpage{71}
    \subitem square-$\sim $\cftdotfill{\cftsectiondotsep} 
		\hyperpage{71}
  \item integrable plane field\cftdotfill{\cftsectiondotsep} 
		\hyperpage{26}
  \item integral
    \subitem $\sim $ curve\cftdotfill{\cftsectiondotsep} \hyperpage{12}
    \subitem $\sim $ manifold\cftdotfill{\cftsectiondotsep} 
		\hyperpage{26}
  \item integration on manifolds\cftdotfill{\cftsectiondotsep} 
		\hyperpage{24}
  \item involution
    \subitem $\sim $ of a ${}^*$-algebra\cftdotfill{\cftsectiondotsep} 
		\hyperpage{59}
    \subitem $\sim $ of a plane field\cftdotfill{\cftsectiondotsep} 
		\hyperpage{26}
  \item isometric
    \subitem $\sim $ embedding\cftdotfill{\cftsectiondotsep} 
		\hyperpage{16}
    \subitem $\sim $ immersion\cftdotfill{\cftsectiondotsep} 
		\hyperpage{16}
  \item isometry\cftdotfill{\cftsectiondotsep} \hyperpage{16}, 
		\hyperpage{53}
    \subitem conformal $\sim $\cftdotfill{\cftsectiondotsep} 
		\hyperpage{16}

  \indexspace
{\bfseries\sffamily K}\nopagebreak

  \item kernel\cftdotfill{\cftsectiondotsep} \hyperpage{71}
    \subitem properly supported $\sim $\cftdotfill{\cftsectiondotsep} 
		\hyperpage{71}
    \subitem regular $\sim $\cftdotfill{\cftsectiondotsep} 
		\hyperpage{71}
    \subitem semiregular $\sim $\cftdotfill{\cftsectiondotsep} 
		\hyperpage{71}
  \item kernel theorem
    \subitem abstract $\sim $\cftdotfill{\cftsectiondotsep} 
		\hyperpage{58}
    \subitem Schwartz $\sim $\cftdotfill{\cftsectiondotsep} 
		\hyperpage{70}
  \item Killing vector field\cftdotfill{\cftsectiondotsep} 
		\hyperpage{22}
    \subitem conformal $\sim $\cftdotfill{\cftsectiondotsep} 
		\hyperpage{22}
  \item Klein--Gordon field\cftdotfill{\cftsectiondotsep} 
		\hyperpage{112}
    \subitem field equation\cftdotfill{\cftsectiondotsep} 
		\hyperpage{112}
    \subitem stress-energy tensor of the $\sim $\cftdotfill{\cftsectiondotsep} 
		\hyperpage{132}
  \item Koszul formula\cftdotfill{\cftsectiondotsep} \hyperpage{20}

  \indexspace
{\bfseries\sffamily L}\nopagebreak

  \item Lagrangian quantum field theory\cftdotfill{\cftsectiondotsep} 
		\hyperpage{104}
  \item Laplace--{de Rham} operator\cftdotfill{\cftsectiondotsep} 
		\hyperpage{26}
  \item leaves of a foliation\cftdotfill{\cftsectiondotsep} 
		\hyperpage{27}
  \item Lebesgue integrable\cftdotfill{\cftsectiondotsep} 
		\hyperpage{71}
  \item Levi-Civita connection\cftdotfill{\cftsectiondotsep} 
		\hyperpage{20}
  \item Lie derivative\cftdotfill{\cftsectiondotsep} \hyperpage{11}, 
		\hyperpage{14}
  \item lightlike\cftdotfill{\cftsectiondotsep} \hyperpage{36}
  \item linear permutation\cftdotfill{\cftsectiondotsep} \hyperpage{83}
    \subitem $\sim $ group\cftdotfill{\cftsectiondotsep} \hyperpage{83}
    \subitem run structure of $\sim $\cftdotfill{\cftsectiondotsep} 
		\hyperpage{93}
  \item Lipschitz constant\cftdotfill{\cftsectiondotsep} \hyperpage{64}
  \item local flow of a vector field\cftdotfill{\cftsectiondotsep} 
		\hyperpage{13}
  \item local frame\cftdotfill{\cftsectiondotsep} \hyperpage{16}
  \item local section\cftdotfill{\cftsectiondotsep} \hyperpage{10}
  \item local trivialization\cftdotfill{\cftsectiondotsep} 
		\hyperpage{9}
  \item locally compact topological space\cftdotfill{\cftsectiondotsep} 
		\hyperpage{54}
  \item locally convex space\cftdotfill{\cftsectiondotsep} 
		\hyperpage{55}, \hyperpage{64}
    \subitem Banach space\cftdotfill{\cftsectiondotsep} \hyperpage{55}
    \subitem Fr\'echet space\cftdotfill{\cftsectiondotsep} 
		\hyperpage{55}, \hyperpage{64}
    \subitem nuclear $\sim $\cftdotfill{\cftsectiondotsep} 
		\hyperpage{58}
    \subitem tensor product\cftdotfill{\cftsectiondotsep} 
		\hyperpage{57--58}
  \item locally covariant QFT\cftdotfill{\cftsectiondotsep} 
		\hyperpage{107}
  \item locally Lipschitz\cftdotfill{\cftsectiondotsep} \hyperpage{64}
  \item Lorentzian
    \subitem $\sim $ manifold\cftdotfill{\cftsectiondotsep} 
		\hyperpage{16}
    \subitem $\sim $ metric\cftdotfill{\cftsectiondotsep} 
		\hyperpage{16}
  \item Lorentzian geometry\cftdotfill{\cftsectiondotsep} 
		\hyperpage{35--50}
  \item Lorentzian manifold
    \subitem causal structure of a $\sim $\cftdotfill{\cftsectiondotsep} 
		\hyperpage{36--37}
    \subitem covariant splitting of a $\sim $\cftdotfill{\cftsectiondotsep} 
		\hyperpage{37--39}
  \item $L^p$ space\cftdotfill{\cftsectiondotsep} \hyperpage{71}

  \indexspace
{\bfseries\sffamily M}\nopagebreak

  \item Mackey topology\cftdotfill{\cftsectiondotsep} \hyperpage{57}
  \item manifold
    \subitem contractible $\sim $\cftdotfill{\cftsectiondotsep} 
		\hyperpage{24}
    \subitem integral $\sim $\cftdotfill{\cftsectiondotsep} 
		\hyperpage{26}
    \subitem Lorentzian $\sim $\cftdotfill{\cftsectiondotsep} 
		\hyperpage{16}
    \subitem pseudo-Riemannian $\sim $\cftdotfill{\cftsectiondotsep} 
		\hyperpage{15}
    \subitem Riemannian $\sim $\cftdotfill{\cftsectiondotsep} 
		\hyperpage{16}
    \subitem smooth $\sim $\cftdotfill{\cftsectiondotsep} \hyperpage{8}
    \subitem strongly causal $\sim $\cftdotfill{\cftsectiondotsep} 
		\hyperpage{40}
    \subitem topological $\sim $\cftdotfill{\cftsectiondotsep} 
		\hyperpage{8}
  \item maximal atlas\cftdotfill{\cftsectiondotsep} \hyperpage{8}
  \item maximal integrable curve\cftdotfill{\cftsectiondotsep} 
		\hyperpage{13}
  \item metric\cftdotfill{\cftsectiondotsep} 
		\seealso{pseudometric}{title-1}, \hyperpage{53}
    \subitem dual $\sim $\cftdotfill{\cftsectiondotsep} \hyperpage{15}
    \subitem index of a $\sim $\cftdotfill{\cftsectiondotsep} 
		\hyperpage{16}
    \subitem Lorentzian $\sim $\cftdotfill{\cftsectiondotsep} 
		\hyperpage{16}
    \subitem $\sim $ of a FLRW spacetime\cftdotfill{\cftsectiondotsep} 
		\hyperpage{47}
    \subitem $\sim $ of de Sitter spacetime\cftdotfill{\cftsectiondotsep} 
		\hyperpage{43}
    \subitem $\sim $ of Euclidean space\cftdotfill{\cftsectiondotsep} 
		\hyperpage{16}
    \subitem $\sim $ of Minkowski spacetime\cftdotfill{\cftsectiondotsep} 
		\hyperpage{16}
    \subitem $\sim $ on a vector bundle\cftdotfill{\cftsectiondotsep} 
		\hyperpage{15}
    \subitem Riemannian $\sim $\cftdotfill{\cftsectiondotsep} 
		\hyperpage{16}
  \item metric connection\cftdotfill{\cftsectiondotsep} \hyperpage{19}
  \item metric space\cftdotfill{\cftsectiondotsep} 
		\seealso{pseudometric space}{title-1}, \hyperpage{53}
  \item microlocal analysis\cftdotfill{\cftsectiondotsep} 
		\hyperpage{68--78}
  \item microlocal spectrum condition\cftdotfill{\cftsectiondotsep} 
		\hyperpage{117}
  \item minimal curvature coupling\cftdotfill{\cftsectiondotsep} 
		\hyperpage{112}
  \item minimal regular norm\cftdotfill{\cftsectiondotsep} 
		\hyperpage{61}
  \item Minkowski spacetime\cftdotfill{\cftsectiondotsep} 
		\hyperpage{16}, \hyperpage{43}
    \subitem chart of $\sim $\cftdotfill{\cftsectiondotsep} 
		\hyperpage{16}
    \subitem metric of $\sim $\cftdotfill{\cftsectiondotsep} 
		\hyperpage{16}
  \item Minkowski vacuum state\cftdotfill{\cftsectiondotsep} 
		\hyperpage{121}
  \item mixed state\cftdotfill{\cftsectiondotsep} \hyperpage{60}
  \item mode equation\cftdotfill{\cftsectiondotsep} \hyperpage{121}
  \item momentum conservation\cftdotfill{\cftsectiondotsep} 
		\hyperpage{41}
  \item momentum density\cftdotfill{\cftsectiondotsep} \hyperpage{41}
  \item multiplication of distributions\cftdotfill{\cftsectiondotsep} 
		\hyperpage{76}
  \item musical isomorphisms\cftdotfill{\cftsectiondotsep} 
		\hyperpage{15}
    \subitem flat $\flat $\cftdotfill{\cftsectiondotsep} \hyperpage{15}
    \subitem sharp $\sharp $\cftdotfill{\cftsectiondotsep} 
		\hyperpage{15}

  \indexspace
{\bfseries\sffamily N}\nopagebreak

  \item $n$-point distribution\cftdotfill{\cftsectiondotsep} 
		\hyperpage{115}
    \subitem truncated $\sim $\cftdotfill{\cftsectiondotsep} 
		\hyperpage{116}
  \item neighbourhood\cftdotfill{\cftsectiondotsep} \hyperpage{52}
    \subitem coordinate $\sim $\cftdotfill{\cftsectiondotsep} 
		\hyperpage{8}
  \item net of local algebras\cftdotfill{\cftsectiondotsep} 
		\hyperpage{105}
    \subitem isotony\cftdotfill{\cftsectiondotsep} \hyperpage{105}
  \item norm\cftdotfill{\cftsectiondotsep} \seealso{seminorm}{title-1}, 
		\hyperpage{55}
    \subitem minimal regular $\sim $\cftdotfill{\cftsectiondotsep} 
		\hyperpage{61}
    \subitem operator $\sim $\cftdotfill{\cftsectiondotsep} 
		\hyperpage{60}
  \item normal element\cftdotfill{\cftsectiondotsep} \hyperpage{59}
  \item normal topology\cftdotfill{\cftsectiondotsep} \hyperpage{74}
  \item normally hyperbolic\cftdotfill{\cftsectiondotsep} 
		\hyperpage{18}
  \item normed space\cftdotfill{\cftsectiondotsep} \hyperpage{55}
  \item nuclear space\cftdotfill{\cftsectiondotsep} \hyperpage{58}
    \subitem quotient of a $\sim $\cftdotfill{\cftsectiondotsep} 
		\hyperpage{58}
    \subitem subspace of a $\sim $\cftdotfill{\cftsectiondotsep} 
		\hyperpage{58}

  \indexspace
{\bfseries\sffamily O}\nopagebreak

  \item off-shell field algebra\cftdotfill{\cftsectiondotsep} 
		\hyperpage{109}
  \item on-shell field algebra\cftdotfill{\cftsectiondotsep} 
		\hyperpage{110}
  \item one-form\cftdotfill{\cftsectiondotsep} 
		\see{covector field}{title-1}
  \item one-line permutation notation\cftdotfill{\cftsectiondotsep} 
		\hyperpage{83}
  \item open set\cftdotfill{\cftsectiondotsep} \hyperpage{52}
  \item operator norm\cftdotfill{\cftsectiondotsep} \hyperpage{60}
  \item orientable\cftdotfill{\cftsectiondotsep} \hyperpage{24}
  \item orientation\cftdotfill{\cftsectiondotsep} \hyperpage{24}
    \subitem negative $\sim $\cftdotfill{\cftsectiondotsep} 
		\hyperpage{24}
    \subitem positive $\sim $\cftdotfill{\cftsectiondotsep} 
		\hyperpage{24}
    \subitem time-$\sim $\cftdotfill{\cftsectiondotsep} \hyperpage{36}
  \item orthogonal frame\cftdotfill{\cftsectiondotsep} \hyperpage{17}

  \indexspace
{\bfseries\sffamily P}\nopagebreak

  \item parallel
    \subitem $\sim $ propagator\cftdotfill{\cftsectiondotsep} 
		\hyperpage{29}
    \subitem $\sim $ transport\cftdotfill{\cftsectiondotsep} 
		\hyperpage{21}
  \item parallel propagator
    \subitem transport equation of the $\sim $\cftdotfill{\cftsectiondotsep} 
		\hyperpage{29}
  \item parametrized curve\cftdotfill{\cftsectiondotsep} \hyperpage{12}
  \item past
    \subitem causal $\sim $\cftdotfill{\cftsectiondotsep} 
		\hyperpage{36}
    \subitem chronological $\sim $\cftdotfill{\cftsectiondotsep} 
		\hyperpage{36}
  \item Pauli--Jordan distribution\cftdotfill{\cftsectiondotsep} 
		\see{commutator distribution}{title-1}
  \item peak of a permutation\cftdotfill{\cftsectiondotsep} 
		\hyperpage{86}
  \item perfect fluid\cftdotfill{\cftsectiondotsep} \hyperpage{41}
  \item permutation
    \subitem ascending $\sim $\cftdotfill{\cftsectiondotsep} 
		\hyperpage{85}
    \subitem atom of a $\sim $\cftdotfill{\cftsectiondotsep} 
		\hyperpage{85}
    \subitem atomic $\sim $\cftdotfill{\cftsectiondotsep} 
		\hyperpage{85}
    \subitem circular $\sim $\cftdotfill{\cftsectiondotsep} 
		\hyperpage{84}
    \subitem descending $\sim $\cftdotfill{\cftsectiondotsep} 
		\hyperpage{85}
    \subitem linear $\sim $\cftdotfill{\cftsectiondotsep} 
		\hyperpage{83}
    \subitem one-line notation\cftdotfill{\cftsectiondotsep} 
		\hyperpage{83}
    \subitem peak of a $\sim $\cftdotfill{\cftsectiondotsep} 
		\hyperpage{86}
    \subitem principal atom of a $\sim $\cftdotfill{\cftsectiondotsep} 
		\hyperpage{85}
    \subitem two-line notation\cftdotfill{\cftsectiondotsep} 
		\hyperpage{83}
    \subitem valley of a $\sim $\cftdotfill{\cftsectiondotsep} 
		\hyperpage{86}, \hyperpage{95--97}
  \item permutation group
    \subitem circular $\sim $\cftdotfill{\cftsectiondotsep} 
		\hyperpage{84}
    \subitem linear $\sim $\cftdotfill{\cftsectiondotsep} 
		\hyperpage{83}
  \item perturbation
    \subitem scalar $\sim $\cftdotfill{\cftsectiondotsep} 
		\hyperpage{49}
    \subitem tensor $\sim $\cftdotfill{\cftsectiondotsep} 
		\hyperpage{49}
    \subitem vector $\sim $\cftdotfill{\cftsectiondotsep} 
		\hyperpage{49}
  \item perturbed metric of a FLRW spacetime\cftdotfill{\cftsectiondotsep} 
		\hyperpage{49}
  \item physical spacetime\cftdotfill{\cftsectiondotsep} \hyperpage{47}
  \item Plancherel--Parseval identity\cftdotfill{\cftsectiondotsep} 
		\hyperpage{72}
  \item plane field\cftdotfill{\cftsectiondotsep} \hyperpage{26}
    \subitem integrable $\sim $\cftdotfill{\cftsectiondotsep} 
		\hyperpage{26}
    \subitem involution of a $\sim $\cftdotfill{\cftsectiondotsep} 
		\hyperpage{26}
  \item Poincar\'e duality\cftdotfill{\cftsectiondotsep} \hyperpage{26}
  \item point-splitting regularization\cftdotfill{\cftsectiondotsep} 
		\hyperpage{133}, \hyperpage{136--137}
  \item pointwise convergence\cftdotfill{\cftsectiondotsep} 
		\hyperpage{55}
  \item pre-normally hyperbolic\cftdotfill{\cftsectiondotsep} 
		\hyperpage{18}
  \item pressure\cftdotfill{\cftsectiondotsep} \hyperpage{41}
    \subitem quantum $\sim $\cftdotfill{\cftsectiondotsep} 
		\hyperpage{135}
  \item pressure-free matter\cftdotfill{\cftsectiondotsep} 
		\hyperpage{41}
  \item principal atom of a permutation\cftdotfill{\cftsectiondotsep} 
		\hyperpage{85}
  \item principal symbol\cftdotfill{\cftsectiondotsep} \hyperpage{18}
  \item principle
    \subitem $\sim $ of causality\cftdotfill{\cftsectiondotsep} 
		\hyperpage{105}, \hyperpage{107}
    \subitem $\sim $ of covariance\cftdotfill{\cftsectiondotsep} 
		\hyperpage{106}
    \subitem $\sim $ of locality\cftdotfill{\cftsectiondotsep} 
		\hyperpage{105, 106}
    \subitem timeslice axiom\cftdotfill{\cftsectiondotsep} 
		\hyperpage{107}
  \item Proca field\cftdotfill{\cftsectiondotsep} \hyperpage{113}
    \subitem field equation\cftdotfill{\cftsectiondotsep} 
		\hyperpage{113}
  \item product topology\cftdotfill{\cftsectiondotsep} \hyperpage{53}
  \item projected symmetric trace-free part\cftdotfill{\cftsectiondotsep} 
		\hyperpage{38}
  \item projective tensor product\cftdotfill{\cftsectiondotsep} 
		\hyperpage{57}
  \item propagation of singularities\cftdotfill{\cftsectiondotsep} 
		\hyperpage{78}
  \item propagator
    \subitem advanced $\sim $\cftdotfill{\cftsectiondotsep} 
		\hyperpage{79}
    \subitem causal $\sim $\cftdotfill{\cftsectiondotsep} 
		\hyperpage{80}
    \subitem parallel $\sim $\cftdotfill{\cftsectiondotsep} 
		\hyperpage{29}
    \subitem retarded $\sim $\cftdotfill{\cftsectiondotsep} 
		\hyperpage{79}
  \item proper map\cftdotfill{\cftsectiondotsep} \hyperpage{54}
  \item proper time\cftdotfill{\cftsectiondotsep} \hyperpage{36}
  \item properly supported kernel\cftdotfill{\cftsectiondotsep} 
		\hyperpage{71}
  \item pseudo-Riemannian manifold\cftdotfill{\cftsectiondotsep} 
		\hyperpage{15}
    \subitem conformal symmetry of a $\sim $\cftdotfill{\cftsectiondotsep} 
		\hyperpage{23}
    \subitem symmetry of a $\sim $\cftdotfill{\cftsectiondotsep} 
		\hyperpage{23}
  \item pseudometric\cftdotfill{\cftsectiondotsep} 
		\seealso{metric}{title-1}, \hyperpage{53}
  \item pseudometric space\cftdotfill{\cftsectiondotsep} 
		\seealso{metric space}{title-1}, \hyperpage{53}
    \subitem bounded $\sim $\cftdotfill{\cftsectiondotsep} 
		\hyperpage{54}
    \subitem complete $\sim $\cftdotfill{\cftsectiondotsep} 
		\hyperpage{54}
    \subitem completion of a $\sim $\cftdotfill{\cftsectiondotsep} 
		\hyperpage{54}
  \item pullback
    \subitem $\sim $ of a connection\cftdotfill{\cftsectiondotsep} 
		\hyperpage{20}
    \subitem $\sim $ of a covector field\cftdotfill{\cftsectiondotsep} 
		\hyperpage{13}
    \subitem $\sim $ of a distribution\cftdotfill{\cftsectiondotsep} 
		\hyperpage{75}
    \subitem $\sim $ of a section\cftdotfill{\cftsectiondotsep} 
		\hyperpage{10}
    \subitem $\sim $ of a vector bundle\cftdotfill{\cftsectiondotsep} 
		\hyperpage{10}
    \subitem $\sim $ of a vector field\cftdotfill{\cftsectiondotsep} 
		\hyperpage{12}
  \item pure state\cftdotfill{\cftsectiondotsep} \hyperpage{60}
  \item pushforward
    \subitem $\sim $ of a covector field\cftdotfill{\cftsectiondotsep} 
		\hyperpage{13}
    \subitem $\sim $ of a section\cftdotfill{\cftsectiondotsep} 
		\hyperpage{10}
    \subitem $\sim $ of a vector field\cftdotfill{\cftsectiondotsep} 
		\hyperpage{12}

  \indexspace
{\bfseries\sffamily Q}\nopagebreak

  \item quantum energy-density\cftdotfill{\cftsectiondotsep} 
		\hyperpage{135}
  \item quantum field\cftdotfill{\cftsectiondotsep} \hyperpage{109}, 
		\hyperpage{111}
  \item quantum field theory
    \subitem algebraic $\sim $\cftdotfill{\cftsectiondotsep} 
		\hyperpage{105}
    \subitem Lagrangian $\sim $\cftdotfill{\cftsectiondotsep} 
		\hyperpage{104}
    \subitem locally covariant $\sim $\cftdotfill{\cftsectiondotsep} 
		\hyperpage{107}
  \item quantum pressure\cftdotfill{\cftsectiondotsep} \hyperpage{135}
  \item quantum state\cftdotfill{\cftsectiondotsep} 
		\see{state}{title-1}
  \item quantum stress-energy tensor\cftdotfill{\cftsectiondotsep} 
		\hyperpage{132}
  \item quasi-free state\cftdotfill{\cftsectiondotsep} \hyperpage{116}
  \item quotient bundle\cftdotfill{\cftsectiondotsep} \hyperpage{10}
  \item quotient topology\cftdotfill{\cftsectiondotsep} \hyperpage{53}

  \indexspace
{\bfseries\sffamily R}\nopagebreak

  \item rapidly decreasing function\cftdotfill{\cftsectiondotsep} 
		\hyperpage{69}
  \item Raychaudhuri
    \subitem $\sim $ equation\cftdotfill{\cftsectiondotsep} 
		\hyperpage{43}
    \subitem $\sim $ scalar\cftdotfill{\cftsectiondotsep} 
		\hyperpage{38}
  \item regular ${}^*$-representation\cftdotfill{\cftsectiondotsep} 
		\hyperpage{62}
  \item regular direction\cftdotfill{\cftsectiondotsep} \hyperpage{73}
  \item regular kernel\cftdotfill{\cftsectiondotsep} \hyperpage{71}
  \item regularity of a fixed-point\cftdotfill{\cftsectiondotsep} 
		\hyperpage{66}
  \item regularization
    \subitem adiabatic $\sim $\cftdotfill{\cftsectiondotsep} 
		\hyperpage{136--137}, \hyperpage{142--146}
    \subitem Hadamard point-splitting $\sim $\cftdotfill{\cftsectiondotsep} 
		\hyperpage{133}, \hyperpage{136--137}
  \item relative Cauchy evolution\cftdotfill{\cftsectiondotsep} 
		\hyperpage{107}
  \item renormalization of the stress-energy tensor\cftdotfill{\cftsectiondotsep} 
		\hyperpage{132}
  \item rest space\cftdotfill{\cftsectiondotsep} \hyperpage{37}
  \item restriction of a
    \subitem $\sim $ distribution\cftdotfill{\cftsectiondotsep} 
		\hyperpage{69}
    \subitem $\sim $ distributional section\cftdotfill{\cftsectiondotsep} 
		\hyperpage{70}
  \item retarded propagator\cftdotfill{\cftsectiondotsep} 
		\hyperpage{79}
  \item Ricci
    \subitem $\sim $ curvature scalar\cftdotfill{\cftsectiondotsep} 
		\hyperpage{21}
    \subitem $\sim $ curvature tensor\cftdotfill{\cftsectiondotsep} 
		\hyperpage{21}
  \item Riemann curvature tensor\cftdotfill{\cftsectiondotsep} 
		\hyperpage{21}
  \item Riemannian
    \subitem $\sim $ manifold\cftdotfill{\cftsectiondotsep} 
		\hyperpage{16}
    \subitem $\sim $ metric\cftdotfill{\cftsectiondotsep} 
		\hyperpage{16}
  \item rising atomic permutation\cftdotfill{\cftsectiondotsep} 
		\hyperpage{85}
  \item run\cftdotfill{\cftsectiondotsep} \hyperpage{86}
  \item run structure\cftdotfill{\cftsectiondotsep} \hyperpage{86}
    \subitem $\sim $ of atomic permutation\cftdotfill{\cftsectiondotsep} 
		\hyperpage{87--91}
    \subitem $\sim $ of circular permutation\cftdotfill{\cftsectiondotsep} 
		\hyperpage{91--93}
    \subitem $\sim $ of linear permutation\cftdotfill{\cftsectiondotsep} 
		\hyperpage{93}

  \indexspace
{\bfseries\sffamily S}\nopagebreak

  \item scalar field\cftdotfill{\cftsectiondotsep} 
		\see{Klein--Gordon field}{title-1}
  \item scalar perturbation\cftdotfill{\cftsectiondotsep} 
		\hyperpage{49}
  \item Schwartz distribution\cftdotfill{\cftsectiondotsep} 
		\see{tempered distribution}{title-1}
  \item Schwartz kernel theorem\cftdotfill{\cftsectiondotsep} 
		\hyperpage{70}
  \item secant numbers\cftdotfill{\cftsectiondotsep} \hyperpage{91}
  \item second Bianchi identity\cftdotfill{\cftsectiondotsep} 
		\hyperpage{21}
  \item second Friedmann equation\cftdotfill{\cftsectiondotsep} 
		\hyperpage{46}, \hyperpage{134}
  \item second-countable\cftdotfill{\cftsectiondotsep} \hyperpage{52}
  \item section\cftdotfill{\cftsectiondotsep} \hyperpage{10}
    \subitem components of a $\sim $\cftdotfill{\cftsectiondotsep} 
		\hyperpage{17}
    \subitem local $\sim $\cftdotfill{\cftsectiondotsep} \hyperpage{10}
    \subitem pullback of a $\sim $\cftdotfill{\cftsectiondotsep} 
		\hyperpage{10}
    \subitem pushforward of a $\sim $\cftdotfill{\cftsectiondotsep} 
		\hyperpage{10}
    \subitem test $\sim $\cftdotfill{\cftsectiondotsep} \hyperpage{70}
  \item self-adjoint element\cftdotfill{\cftsectiondotsep} 
		\hyperpage{59}
  \item semiclassical Einstein equation\cftdotfill{\cftsectiondotsep} 
		\hyperpage{131--137}
    \subitem solution of the $\sim $\cftdotfill{\cftsectiondotsep} 
		\hyperpage{139--151}
      \subsubitem $\sim $ global\cftdotfill{\cftsectiondotsep} 
		\hyperpage{148--150}
      \subsubitem $\sim $ local\cftdotfill{\cftsectiondotsep} 
		\hyperpage{146--148}
      \subsubitem $\sim $ maximal\cftdotfill{\cftsectiondotsep} 
		\hyperpage{149}
      \subsubitem $\sim $ numerical\cftdotfill{\cftsectiondotsep} 
		\hyperpage{150}
      \subsubitem $\sim $ regular\cftdotfill{\cftsectiondotsep} 
		\hyperpage{148}
  \item semiclassical Friedmann equations\cftdotfill{\cftsectiondotsep} 
		\hyperpage{134--137}, \hyperpage{140}
  \item seminorm\cftdotfill{\cftsectiondotsep} \seealso{norm}{title-1}, 
		\hyperpage{54}
    \subitem dual $\sim $\cftdotfill{\cftsectiondotsep} \hyperpage{57}
    \subitem family of $\sim $\cftdotfill{\cftsectiondotsep} 
		\hyperpage{55}
    \subitem separating $\sim $\cftdotfill{\cftsectiondotsep} 
		\hyperpage{55}
  \item semiregular kernel\cftdotfill{\cftsectiondotsep} \hyperpage{71}
  \item separating seminorm\cftdotfill{\cftsectiondotsep} 
		\hyperpage{55}
  \item set
    \subitem closed $\sim $\cftdotfill{\cftsectiondotsep} 
		\hyperpage{52}
    \subitem closure of a $\sim $\cftdotfill{\cftsectiondotsep} 
		\hyperpage{52}
    \subitem connected $\sim $\cftdotfill{\cftsectiondotsep} 
		\hyperpage{52}
    \subitem dense $\sim $\cftdotfill{\cftsectiondotsep} \hyperpage{52}
    \subitem open $\sim $\cftdotfill{\cftsectiondotsep} \hyperpage{52}
  \item shadow, causal $\sim $\cftdotfill{\cftsectiondotsep} 
		\hyperpage{36}
  \item sharp $\sharp $\cftdotfill{\cftsectiondotsep} \hyperpage{15}
  \item shear tensor\cftdotfill{\cftsectiondotsep} \hyperpage{38}
  \item $\sigma $-compact topological space\cftdotfill{\cftsectiondotsep} 
		\hyperpage{54}
  \item simple ${}^*$-algebra\cftdotfill{\cftsectiondotsep} 
		\hyperpage{62}
  \item singular direction\cftdotfill{\cftsectiondotsep} \hyperpage{74}
    \subitem set of $\sim $\cftdotfill{\cftsectiondotsep} 
		\hyperpage{74}
    \subitem set of localized $\sim $\cftdotfill{\cftsectiondotsep} 
		\hyperpage{74}
  \item singular support\cftdotfill{\cftsectiondotsep} \hyperpage{73}
  \item smooth
    \subitem $\sim $ distribution\cftdotfill{\cftsectiondotsep} 
		\hyperpage{73}
  \item smooth atlas\cftdotfill{\cftsectiondotsep} \hyperpage{8}
  \item smooth functions\cftdotfill{\cftsectiondotsep} \hyperpage{68}
  \item smooth manifold\cftdotfill{\cftsectiondotsep} \hyperpage{8}
  \item smooth sections\cftdotfill{\cftsectiondotsep} \hyperpage{69}
  \item smoothly compatible charts\cftdotfill{\cftsectiondotsep} 
		\hyperpage{8}
  \item space
    \subitem Euclidean $\sim $\cftdotfill{\cftsectiondotsep} 
		\hyperpage{16}
    \subitem Hausdorff $\sim $\cftdotfill{\cftsectiondotsep} 
		\hyperpage{52}
    \subitem locally convex $\sim $\cftdotfill{\cftsectiondotsep} 
		\hyperpage{55}, \hyperpage{64}
    \subitem nuclear $\sim $\cftdotfill{\cftsectiondotsep} 
		\hyperpage{58}
    \subitem $\sim $ of compactly supported distributional sections\cftdotfill{\cftsectiondotsep} 
		\hyperpage{70}
    \subitem $\sim $ of compactly supported distributions\cftdotfill{\cftsectiondotsep} 
		\hyperpage{69}
    \subitem $\sim $ of continuous functions\cftdotfill{\cftsectiondotsep} 
		\hyperpage{53}
    \subitem $\sim $ of distributional sections\cftdotfill{\cftsectiondotsep} 
		\hyperpage{70}
    \subitem $\sim $ of distributions\cftdotfill{\cftsectiondotsep} 
		\hyperpage{69}
    \subitem $\sim $ of functions\cftdotfill{\cftsectiondotsep} 
		\hyperpage{55}
      \subsubitem topologies\cftdotfill{\cftsectiondotsep} 
		\hyperpage{55--56}
    \subitem $\sim $ of $L^p$ functions\cftdotfill{\cftsectiondotsep} 
		\hyperpage{71}
    \subitem $\sim $ of rapidly decreasing functions\cftdotfill{\cftsectiondotsep} 
		\hyperpage{69}
    \subitem $\sim $ of smooth functions\cftdotfill{\cftsectiondotsep} 
		\hyperpage{68}
    \subitem $\sim $ of smooth sections\cftdotfill{\cftsectiondotsep} 
		\hyperpage{69}
    \subitem $\sim $ of tempered distributions\cftdotfill{\cftsectiondotsep} 
		\hyperpage{69}
    \subitem $\sim $ of test functions\cftdotfill{\cftsectiondotsep} 
		\hyperpage{69}
    \subitem $\sim $ of test sections\cftdotfill{\cftsectiondotsep} 
		\hyperpage{70}
    \subitem total $\sim $\cftdotfill{\cftsectiondotsep} \hyperpage{9}
  \item spacelike\cftdotfill{\cftsectiondotsep} \hyperpage{36}
  \item spacetime\cftdotfill{\cftsectiondotsep} \hyperpage{39}
    \subitem anti-de Sitter $\sim $\cftdotfill{\cftsectiondotsep} 
		\hyperpage{43}
    \subitem background $\sim $\cftdotfill{\cftsectiondotsep} 
		\hyperpage{47}, \hyperpage{155}
    \subitem cosmological $\sim $\cftdotfill{\cftsectiondotsep} 
		\see{FLRW spacetime}{title-1}
    \subitem de Sitter $\sim $\cftdotfill{\cftsectiondotsep} 
		\hyperpage{43}
    \subitem $\sim $ FLRW\cftdotfill{\cftsectiondotsep} \hyperpage{45}
    \subitem globally hyperbolic $\sim $\cftdotfill{\cftsectiondotsep} 
		\hyperpage{40}
    \subitem homogeneous and isotropic $\sim $\cftdotfill{\cftsectiondotsep} 
		\see{FLRW spacetime}{title-1}
    \subitem Minkowski $\sim $\cftdotfill{\cftsectiondotsep} 
		\hyperpage{16}, \hyperpage{43}
    \subitem physical $\sim $\cftdotfill{\cftsectiondotsep} 
		\hyperpage{47}
  \item spatial covariant derivative\cftdotfill{\cftsectiondotsep} 
		\hyperpage{38}
  \item square-integrable\cftdotfill{\cftsectiondotsep} \hyperpage{71}
  \item state\cftdotfill{\cftsectiondotsep} \hyperpage{60}, 
		\hyperpage{115}
    \subitem adiabatic $\sim $\cftdotfill{\cftsectiondotsep} 
		\hyperpage{119}, \hyperpage{122--127}
    \subitem Bunch--Davies $\sim $\cftdotfill{\cftsectiondotsep} 
		\hyperpage{119}
    \subitem Gaussian $\sim $\cftdotfill{\cftsectiondotsep} 
		\see{quasi-free state}{title-1}
    \subitem Hadamard $\sim $\cftdotfill{\cftsectiondotsep} 
		\hyperpage{117}
    \subitem homogeneous and isotropic $\sim $\cftdotfill{\cftsectiondotsep} 
		\hyperpage{121}
    \subitem Minkowski vacuum $\sim $\cftdotfill{\cftsectiondotsep} 
		\hyperpage{121}
    \subitem mixed $\sim $\cftdotfill{\cftsectiondotsep} \hyperpage{60}
    \subitem $\sim $ of low energy\cftdotfill{\cftsectiondotsep} 
		\hyperpage{125--127}
    \subitem pure $\sim $\cftdotfill{\cftsectiondotsep} \hyperpage{60}
    \subitem quasi-free $\sim $\cftdotfill{\cftsectiondotsep} 
		\hyperpage{116}
    \subitem strongly regular $\sim $\cftdotfill{\cftsectiondotsep} 
		\hyperpage{62}
  \item Stokes' theorem\cftdotfill{\cftsectiondotsep} \hyperpage{25}
  \item stress-energy tensor\cftdotfill{\cftsectiondotsep} 
		\hyperpage{40}, \hyperpage{131--134}
    \subitem $\sim $ of the Klein--Gordon field\cftdotfill{\cftsectiondotsep} 
		\hyperpage{132}
    \subitem quantum $\sim $\cftdotfill{\cftsectiondotsep} 
		\hyperpage{132}
    \subitem renormalization of the $\sim $\cftdotfill{\cftsectiondotsep} 
		\hyperpage{132}
    \subitem trace of $\sim $\cftdotfill{\cftsectiondotsep} 
		\hyperpage{41}, \hyperpage{134}, \hyperpage{140}
  \item strong
    \subitem $\sim $ dual\cftdotfill{\cftsectiondotsep} \hyperpage{57}
    \subitem $\sim $ topology\cftdotfill{\cftsectiondotsep} 
		\hyperpage{57}
  \item strongly causal
    \subitem $\sim $ at a point\cftdotfill{\cftsectiondotsep} 
		\hyperpage{40}
    \subitem $\sim $ manifold\cftdotfill{\cftsectiondotsep} 
		\hyperpage{40}
  \item strongly cyclic ${}^*$-representation\cftdotfill{\cftsectiondotsep} 
		\hyperpage{60}
  \item strongly regular state\cftdotfill{\cftsectiondotsep} 
		\hyperpage{62}
  \item submanifold
    \subitem embedding of a $\sim $\cftdotfill{\cftsectiondotsep} 
		\hyperpage{12}
    \subitem immersion of a $\sim $\cftdotfill{\cftsectiondotsep} 
		\hyperpage{12}
    \subitem tangent space of a $\sim $\cftdotfill{\cftsectiondotsep} 
		\hyperpage{12}
  \item subspace topology\cftdotfill{\cftsectiondotsep} \hyperpage{53}
  \item summable sequence\cftdotfill{\cftsectiondotsep} \hyperpage{58}
    \subitem absolutely $\sim $\cftdotfill{\cftsectiondotsep} 
		\hyperpage{58}
  \item summation convention\cftdotfill{\cftsectiondotsep} 
		\hyperpage{17}
  \item support
    \subitem $\sim $ of a distribution\cftdotfill{\cftsectiondotsep} 
		\hyperpage{69}
    \subitem singular $\sim $\cftdotfill{\cftsectiondotsep} 
		\hyperpage{73}
  \item symmetric tensor product\cftdotfill{\cftsectiondotsep} 
		\hyperpage{14}
  \item symmetry of a manifold\cftdotfill{\cftsectiondotsep} 
		\hyperpage{23}
    \subitem $\sim $ conformal\cftdotfill{\cftsectiondotsep} 
		\hyperpage{23}
  \item Synge bracket\cftdotfill{\cftsectiondotsep} \hyperpage{27}
  \item Synge's rule\cftdotfill{\cftsectiondotsep} \hyperpage{27}
  \item Synge's world function\cftdotfill{\cftsectiondotsep} 
		\hyperpage{28}
    \subitem transport equation of $\sim $\cftdotfill{\cftsectiondotsep} 
		\hyperpage{28}

  \indexspace
{\bfseries\sffamily T}\nopagebreak

  \item tangent bundle\cftdotfill{\cftsectiondotsep} \hyperpage{11}
    \subitem dual of the $\sim $\cftdotfill{\cftsectiondotsep} 
		\hyperpage{13}
  \item tangent map\cftdotfill{\cftsectiondotsep} \hyperpage{11}
    \subitem global $\sim $\cftdotfill{\cftsectiondotsep} 
		\hyperpage{11}
  \item tangent numbers\cftdotfill{\cftsectiondotsep} \hyperpage{93}
  \item tangent space\cftdotfill{\cftsectiondotsep} \hyperpage{11}
    \subitem dual of the $\sim $\cftdotfill{\cftsectiondotsep} 
		\hyperpage{13}
    \subitem $\sim $ of a submanifold\cftdotfill{\cftsectiondotsep} 
		\hyperpage{12}
  \item tempered distribution\cftdotfill{\cftsectiondotsep} 
		\hyperpage{69}
  \item temporal covariant derivative\cftdotfill{\cftsectiondotsep} 
		\hyperpage{38}
  \item tensor field\cftdotfill{\cftsectiondotsep} \hyperpage{14}
    \subitem components of a $\sim $\cftdotfill{\cftsectiondotsep} 
		\hyperpage{17}
    \subitem contravariant $\sim $\cftdotfill{\cftsectiondotsep} 
		\hyperpage{14}
    \subitem covariant $\sim $\cftdotfill{\cftsectiondotsep} 
		\hyperpage{14}
    \subitem decomposition of a $\sim $\cftdotfill{\cftsectiondotsep} 
		\hyperpage{48}
    \subitem gauge invariant $\sim $\cftdotfill{\cftsectiondotsep} 
		\hyperpage{47}
    \subitem projected symmetric trace-free part of a $\sim $\cftdotfill{\cftsectiondotsep} 
		\hyperpage{38}
  \item tensor perturbation\cftdotfill{\cftsectiondotsep} 
		\hyperpage{49}
  \item tensor product\cftdotfill{\cftsectiondotsep} \hyperpage{14}
    \subitem antisymmetric $\sim $\cftdotfill{\cftsectiondotsep} 
		\hyperpage{14}
    \subitem exterior $\sim $\cftdotfill{\cftsectiondotsep} 
		\hyperpage{15}
    \subitem injective $\sim $ topology\cftdotfill{\cftsectiondotsep} 
		\hyperpage{58}
    \subitem projective $\sim $ topology\cftdotfill{\cftsectiondotsep} 
		\hyperpage{57}
    \subitem symmetric $\sim $\cftdotfill{\cftsectiondotsep} 
		\hyperpage{14}
  \item test function\cftdotfill{\cftsectiondotsep} \hyperpage{69}
    \subitem $\sim $ topology\cftdotfill{\cftsectiondotsep} 
		\hyperpage{69}
  \item test section\cftdotfill{\cftsectiondotsep} \hyperpage{70}
  \item time
    \subitem conformal $\sim $\cftdotfill{\cftsectiondotsep} 
		\hyperpage{47}
    \subitem cosmological $\sim $\cftdotfill{\cftsectiondotsep} 
		\hyperpage{46}
    \subitem $\sim $ function\cftdotfill{\cftsectiondotsep} 
		\hyperpage{36}
    \subitem -orientation\cftdotfill{\cftsectiondotsep} \hyperpage{36}
    \subitem proper $\sim $\cftdotfill{\cftsectiondotsep} 
		\hyperpage{36}
  \item timelike\cftdotfill{\cftsectiondotsep} \hyperpage{36}
  \item timeslice axiom\cftdotfill{\cftsectiondotsep} \hyperpage{107}
  \item topological
    \subitem $\sim $ $*$-algebra\cftdotfill{\cftsectiondotsep} 
		\see{$*$-algebra}{title-1}
    \subitem $\sim $ basis\cftdotfill{\cftsectiondotsep} \hyperpage{52}
    \subitem $\sim $ dual\cftdotfill{\cftsectiondotsep} \hyperpage{56}
    \subitem $\sim $ local
      \subsubitem $\sim $ topological local basis\cftdotfill{\cftsectiondotsep} 
		\hyperpage{52}
    \subitem $\sim $ manifold\cftdotfill{\cftsectiondotsep} 
		\hyperpage{8}
  \item topological space\cftdotfill{\cftsectiondotsep} \hyperpage{52}
    \subitem compact $\sim $\cftdotfill{\cftsectiondotsep} 
		\hyperpage{54}
    \subitem locally compact $\sim $\cftdotfill{\cftsectiondotsep} 
		\hyperpage{54}
    \subitem $\sigma $-compact $\sim $\cftdotfill{\cftsectiondotsep} 
		\hyperpage{54}
  \item topological vector space\cftdotfill{\cftsectiondotsep} 
		\hyperpage{54}
    \subitem convex $\sim $\cftdotfill{\cftsectiondotsep} 
		\hyperpage{54}
    \subitem locally convex $\sim $\cftdotfill{\cftsectiondotsep} 
		\hyperpage{55}, \hyperpage{64}
  \item topology\cftdotfill{\cftsectiondotsep} \hyperpage{52--58}
    \subitem coarser $\sim $\cftdotfill{\cftsectiondotsep} 
		\hyperpage{52}
    \subitem compact-open $\sim $\cftdotfill{\cftsectiondotsep} 
		\hyperpage{56}
    \subitem direct sum $\sim $\cftdotfill{\cftsectiondotsep} 
		\hyperpage{53}
    \subitem discrete $\sim $\cftdotfill{\cftsectiondotsep} 
		\hyperpage{52}
    \subitem final $\sim $\cftdotfill{\cftsectiondotsep} \hyperpage{53}
    \subitem finer $\sim $\cftdotfill{\cftsectiondotsep} \hyperpage{52}
    \subitem graph $\sim $\cftdotfill{\cftsectiondotsep} \hyperpage{60}
    \subitem H\"ormander $\sim $\cftdotfill{\cftsectiondotsep} 
		\hyperpage{74}
    \subitem initial $\sim $\cftdotfill{\cftsectiondotsep} 
		\hyperpage{53}
    \subitem injective tensor product $\sim $\cftdotfill{\cftsectiondotsep} 
		\hyperpage{58}
    \subitem Mackey $\sim $\cftdotfill{\cftsectiondotsep} 
		\hyperpage{57}
    \subitem normal $\sim $\cftdotfill{\cftsectiondotsep} 
		\hyperpage{74}
    \subitem $\sim $ of pointwise convergence\cftdotfill{\cftsectiondotsep} 
		\hyperpage{55}
    \subitem $\sim $ of test functions\cftdotfill{\cftsectiondotsep} 
		\hyperpage{69}
    \subitem $\sim $ of uniform convergence\cftdotfill{\cftsectiondotsep} 
		\hyperpage{56}
      \subsubitem $\sim $ in compacta\cftdotfill{\cftsectiondotsep} 
		\hyperpage{56}
    \subitem product $\sim $\cftdotfill{\cftsectiondotsep} 
		\hyperpage{53}
    \subitem projective tensor product $\sim $\cftdotfill{\cftsectiondotsep} 
		\hyperpage{57}
    \subitem quotient $\sim $\cftdotfill{\cftsectiondotsep} 
		\hyperpage{53}
    \subitem strong $\sim $\cftdotfill{\cftsectiondotsep} 
		\hyperpage{57}
    \subitem subspace $\sim $\cftdotfill{\cftsectiondotsep} 
		\hyperpage{53}
    \subitem trivial $\sim $\cftdotfill{\cftsectiondotsep} 
		\hyperpage{52}
    \subitem uniform operator $\sim $\cftdotfill{\cftsectiondotsep} 
		\hyperpage{60}
    \subitem weak $\sim $\cftdotfill{\cftsectiondotsep} \hyperpage{57}
  \item torsion-free connection\cftdotfill{\cftsectiondotsep} 
		\hyperpage{19}
  \item total space\cftdotfill{\cftsectiondotsep} \hyperpage{9}
  \item total symbol\cftdotfill{\cftsectiondotsep} \hyperpage{18}
  \item trace anomaly\cftdotfill{\cftsectiondotsep} \hyperpage{134}
  \item transition map\cftdotfill{\cftsectiondotsep} \hyperpage{8}
  \item transport
    \subitem $\sim $ operator\cftdotfill{\cftsectiondotsep} 
		\hyperpage{30}
    \subitem parallel $\sim $\cftdotfill{\cftsectiondotsep} 
		\hyperpage{21}
  \item transport equation\cftdotfill{\cftsectiondotsep} \hyperpage{30}
    \subitem $\sim $ of Synge's world function\cftdotfill{\cftsectiondotsep} 
		\hyperpage{28}
    \subitem $\sim $ of the parallel propagator\cftdotfill{\cftsectiondotsep} 
		\hyperpage{29}
    \subitem $\sim $ of the van Vleck--Morette det\cftdotfill{\cftsectiondotsep} 
		\hyperpage{29}
  \item trivial topology\cftdotfill{\cftsectiondotsep} \hyperpage{52}
  \item trivialization
    \subitem global $\sim $\cftdotfill{\cftsectiondotsep} \hyperpage{9}
    \subitem local $\sim $\cftdotfill{\cftsectiondotsep} \hyperpage{9}
  \item truncated $n$-point distribution\cftdotfill{\cftsectiondotsep} 
		\hyperpage{116}
  \item two-line permutation notation\cftdotfill{\cftsectiondotsep} 
		\hyperpage{83}

  \indexspace
{\bfseries\sffamily U}\nopagebreak

  \item uniform convergence\cftdotfill{\cftsectiondotsep} 
		\hyperpage{56}
    \subitem $\sim $ in compacta\cftdotfill{\cftsectiondotsep} 
		\hyperpage{56}
  \item uniform operator topology\cftdotfill{\cftsectiondotsep} 
		\hyperpage{60}
  \item uniformly continuous\cftdotfill{\cftsectiondotsep} 
		\hyperpage{54}
  \item uniqueness of a fixed-point\cftdotfill{\cftsectiondotsep} 
		\hyperpage{66}
  \item unital ${}^*$-algebra\cftdotfill{\cftsectiondotsep} 
		\hyperpage{59}
  \item unitary element\cftdotfill{\cftsectiondotsep} \hyperpage{59}
  \item unitary frame\cftdotfill{\cftsectiondotsep} \hyperpage{17}

  \indexspace
{\bfseries\sffamily V}\nopagebreak

  \item valley of a permutation\cftdotfill{\cftsectiondotsep} 
		\hyperpage{86}, \hyperpage{95--97}
  \item van Vleck--Morette determinant\cftdotfill{\cftsectiondotsep} 
		\hyperpage{28}
    \subitem transport equation of the $\sim $\cftdotfill{\cftsectiondotsep} 
		\hyperpage{29}
  \item vector\cftdotfill{\cftsectiondotsep} \hyperpage{11}
    \subitem coordinate $\sim $\cftdotfill{\cftsectiondotsep} 
		\hyperpage{17}
    \subitem cyclic $\sim $\cftdotfill{\cftsectiondotsep} 
		\hyperpage{60}
  \item vector bundle\cftdotfill{\cftsectiondotsep} \hyperpage{9}
    \subitem antisymm. tensor product of a $\sim $\cftdotfill{\cftsectiondotsep} 
		\hyperpage{14}
    \subitem dual of a $\sim $\cftdotfill{\cftsectiondotsep} 
		\hyperpage{15}
    \subitem exterior tensor product of a $\sim $\cftdotfill{\cftsectiondotsep} 
		\hyperpage{15}
    \subitem fibre of a $\sim $\cftdotfill{\cftsectiondotsep} 
		\hyperpage{9}
    \subitem $\sim $ homomorphism\cftdotfill{\cftsectiondotsep} 
		\hyperpage{9}
    \subitem metric on a $\sim $\cftdotfill{\cftsectiondotsep} 
		\hyperpage{15}
    \subitem $\sim $ projection\cftdotfill{\cftsectiondotsep} 
		\hyperpage{9}
    \subitem pullback of a $\sim $\cftdotfill{\cftsectiondotsep} 
		\hyperpage{10}
    \subitem quotient of a $\sim $\cftdotfill{\cftsectiondotsep} 
		\hyperpage{10}
    \subitem subbundle of a $\sim $\cftdotfill{\cftsectiondotsep} 
		\hyperpage{10}
    \subitem symmetric tensor product of a $\sim $\cftdotfill{\cftsectiondotsep} 
		\hyperpage{14}
    \subitem tensor product of a $\sim $\cftdotfill{\cftsectiondotsep} 
		\hyperpage{14}
  \item vector field\cftdotfill{\cftsectiondotsep} \hyperpage{11}
    \subitem complete $\sim $\cftdotfill{\cftsectiondotsep} 
		\hyperpage{13}
    \subitem components of a $\sim $\cftdotfill{\cftsectiondotsep} 
		\hyperpage{17}
    \subitem conformal Killing $\sim $\cftdotfill{\cftsectiondotsep} 
		\hyperpage{22}
    \subitem flow of a $\sim $\cftdotfill{\cftsectiondotsep} 
		\hyperpage{13}
    \subitem Killing $\sim $\cftdotfill{\cftsectiondotsep} 
		\hyperpage{22}
    \subitem local flow of a $\sim $\cftdotfill{\cftsectiondotsep} 
		\hyperpage{13}
    \subitem pullback of a $\sim $\cftdotfill{\cftsectiondotsep} 
		\hyperpage{12}
    \subitem pushforward of a $\sim $\cftdotfill{\cftsectiondotsep} 
		\hyperpage{12}
  \item vector perturbation\cftdotfill{\cftsectiondotsep} 
		\hyperpage{49}
  \item vector subbundle\cftdotfill{\cftsectiondotsep} \hyperpage{10}
  \item velocity of a curve\cftdotfill{\cftsectiondotsep} 
		\hyperpage{12}
  \item volume form\cftdotfill{\cftsectiondotsep} \hyperpage{24}
    \subitem $\sim $ induced by a metric\cftdotfill{\cftsectiondotsep} 
		\hyperpage{24}
  \item vorticity tensor\cftdotfill{\cftsectiondotsep} \hyperpage{38}

  \indexspace
{\bfseries\sffamily W}\nopagebreak

  \item wave
    \subitem $\sim $ equation\cftdotfill{\cftsectiondotsep} 
		\hyperpage{78}
    \subitem $\sim $ operator\cftdotfill{\cftsectiondotsep} 
		\hyperpage{18}
  \item wavefront set
    \subitem $\sim $ of a distribution\cftdotfill{\cftsectiondotsep} 
		\hyperpage{74}
    \subitem $\sim $ of a distributional section\cftdotfill{\cftsectiondotsep} 
		\hyperpage{76}
  \item weak
    \subitem $\sim $ dual\cftdotfill{\cftsectiondotsep} \hyperpage{57}
    \subitem $\sim $ topology\cftdotfill{\cftsectiondotsep} 
		\hyperpage{57}
  \item Weyl
    \subitem $\sim $ ${}^*$-algebra\cftdotfill{\cftsectiondotsep} 
		\hyperpage{61}
    \subitem $\sim $ $C^*$-algebra\cftdotfill{\cftsectiondotsep} 
		\hyperpage{61}
    \subitem $\sim $ generator\cftdotfill{\cftsectiondotsep} 
		\hyperpage{61}
  \item world function\cftdotfill{\cftsectiondotsep} \hyperpage{28}
    \subitem $\sim $ of de Sitter spacetime\cftdotfill{\cftsectiondotsep} 
		\hyperpage{44}
    \subitem $\sim $ of FLRW spacetime\cftdotfill{\cftsectiondotsep} 
		\hyperpage{32}
  \item world line\cftdotfill{\cftsectiondotsep} \hyperpage{36}
    \subitem acceleration\cftdotfill{\cftsectiondotsep} \hyperpage{38}
    \subitem expansion scalar\cftdotfill{\cftsectiondotsep} 
		\hyperpage{38}
    \subitem expansion tensor\cftdotfill{\cftsectiondotsep} 
		\hyperpage{38}
    \subitem rest space\cftdotfill{\cftsectiondotsep} \hyperpage{37}
    \subitem shear tensor\cftdotfill{\cftsectiondotsep} \hyperpage{38}
    \subitem vorticity tensor\cftdotfill{\cftsectiondotsep} 
		\hyperpage{38}

\end{theindex}